\documentclass[a4paper]{article}
\usepackage{geometry}
\usepackage{amsmath,amsfonts,amssymb,amsthm,stmaryrd,mathrsfs,mathdots,amsbsy}
\usepackage{hyperref}
\usepackage{bm}
\usepackage{tikz}
\usetikzlibrary{tikzmark}
\usepackage[small]{diagrams}
\usepackage{float}
\usetikzlibrary{positioning,calc}

\theoremstyle{definition}
\newtheorem{definition}{Definition}
\newtheorem{example}[definition]{Example}
\theoremstyle{theorem}
\newtheorem{theorem}[definition]{Theorem}

\newtheorem{lemma}[definition]{Lemma}
\newtheorem{corollary}[definition]{Corollary}
\theoremstyle{remark}
\newtheorem*{remark}{Remark}
\newtheorem*{notation}{Notation}
\newtheorem*{convention}{Convention}

\newarrow{Corresponds}<--->

\usepackage{sectsty}
\sectionfont{\sffamily}
\subsectionfont{\sffamily}
\subsubsectionfont{\sffamily}

\newtheoremstyle{myplain}
  {\topsep}   
  {\topsep}   
  {\itshape}  
  {0pt}       
  {\bfseries\sffamily} 
  {.}         
  {5pt plus 1pt minus 1pt} 
  {}          
  
\newtheoremstyle{mydefinition}
  {\topsep}   
  {\topsep}   
  {\normalfont}  
  {0pt}       
  {\bfseries\sffamily} 
  {.}         
  {5pt plus 1pt minus 1pt} 
  {}          
  
\newtheoremstyle{myremark}
  {0.5\topsep}   
  {0.5\topsep}   
  {\normalfont}  
  {0pt}       
  {\sffamily} 
  {.}         
  {5pt plus 1pt minus 1pt} 
  {}          

\begin{document}
\title{\Huge \bf A Game-Semantic \\ Model of Computation}
\author{\Large \bf Norihiro Yamada \\
{\tt norihiro.yamada@cs.ox.ac.uk} \\
University of Oxford
}

\maketitle

\begin{abstract} 
The present paper introduces a novel notion of `(effective) computability', called \emph{viability}, of strategies in game semantics in an \emph{intrinsic} (i.e., without recourse to the standard \emph{Church-Turing computability}), \emph{non-inductive} and \emph{non-axiomatic} manner, and shows, as a main technical achievement, that viable strategies are \emph{Turing complete}.
Consequently, we have given a mathematical foundation of computation in the same sense as Turing machines but \emph{beyond computation on natural numbers}, e.g., higher-order computation, in a more abstract fashion. 
As immediate corollaries, some of the well-known theorems in computability theory such as the smn-theorem and the first recursion theorem are generalized.
Notably, our game-semantic framework distinguishes `high-level' computational processes that operate directly on mathematical objects such as natural numbers (not on their symbolic representations) and their `symbolic implementations' that define their `computability', which sheds new light on the very concept of computation. 
This work is intended to be a stepping stone towards a new mathematical foundation of computation, intuitionistic logic and constructive mathematics.

\end{abstract}






\section{Introduction}
The present work introduces an \emph{intrinsic}, \emph{non-inductive} and \emph{non-axiomatic} formulation of `effectively computable' strategies in game semantics and proves that they are \emph{Turing complete}.
This result leads to a novel mathematical foundation of computation \emph{beyond classical computation}, e.g., higher-order computation, that distinguishes `high-level' and `low-level' computational processes, where the latter defines `effective computability' of the former.

\begin{convention}
We informally use \emph{\bfseries computational processes} and \emph{\bfseries algorithms} almost as synonyms of \emph{computation}, but they put more emphasis on `processes'.
\end{convention}

\subsection{Search for Turing machines beyond classical computation}

\emph{Turing machines (TMs)} introduced in the classic work \cite{turing1936computable} by Alan Turing have been widely accepted as providing a reasonable and highly convincing definition of \emph{`effectivity'} or \emph{`(effective) computability'} of functions on (finite sequences of) natural numbers, which we call \emph{\bfseries recursiveness}, \emph{\bfseries classical computability} or \emph{\bfseries Church-Turing computability} in the present paper, in a mathematically rigorous manner. 
This is because `computability' of a function intuitively means the very existence of an algorithm that `implements' the function's input/output-behavior, and TMs are none other than a mathematical formulation of this informal concept. 
In mathematics, however, there are various kinds of \emph{non-classical computation}, where by \emph{\bfseries classical computation} we mean what merely `implements' a function on (finite sequences of) natural numbers, since there are a variety of mathematical objects other than natural numbers, for which TMs have certain limitations.

As an example of non-classical computation, consider \emph{higher-order computation} \cite{longley2015higher}, i.e., computation that  takes (as an input) or produces (as an output) another computation, which abounds in mathematics, e.g., quantification in mathematical logic, differentiation in analysis or simply an application $(f, a) \mapsto f(a)$ of a function $f : A \to B$ to an argument $a \in A$. 
However, TMs cannot capture higher-order computation in a natural or systematic fashion.
In fact, although TMs may compute on symbols that \emph{encode} other TMs, e.g., consider \emph{universal TMs} \cite{kozen2012automata,sipser2012introduction}, they cannot compute on `external behavior' of an input computation, which implies that the input is limited to a recursive one (to be encoded); however, it makes perfect sense to consider computation on non-recursive objects such as non-recursive real numbers. 
For this point, one may argue that \emph{oracle TMs} \cite{kozen2006theory,sipser2012introduction} may treat an input computation as an \emph{oracle}, a `black-box-like' computation that does not have to be recursive; however, it is like a function (rather than a computational process) that computes just in a single step, which appears conceptually mysterious and technically ad-hoc (another approach is to give an input computation as a potentially infinite sequence of symbols on the input tape \cite{weihrauch2012computable}, but it may be criticized in a similar manner). 

On the other hand, most of the other models of higher-order computation are, unlike TMs, either syntactic (such as $\lambda$-calculi and programming languages \cite{barendregt1984lambda,longley2015higher}), inductive and/or axiomatic (such as \emph{Kleene's schemata \textsf{S1}-\textsf{S9}} \cite{kleene1959recursive,kleene1963recursive}) or extrinsic (i.e., reducing to classical computation by \emph{encoding} whose `effectivity' is usually left imprecise \cite{cutland1980computability,longley2015higher}), thus lacking the semantic, direct and intrinsic natures of TMs.
Also, unlike classical computability, a confluence between different notions of higher-order computability has been rarely established \cite{longley2015higher}.
For this problem, it would be a key step to establish a `TMs-like' model of higher-order computation since it may tell us which notion of higher-order computability is a `correct' one.




\subsection{Search for mathematics of `high-level' computational processes}
Perhaps more crucially than the limitation for non-classical computation mentioned above, one may argue that TMs are not appropriate as \emph{mathematics of computational processes} since computational steps of TMs are often \emph{too `low-level'} to see what they are supposed to compute. 
In other words, we need mathematics of \emph{`high-level'} computational processes that gives a `birds-eye-view' of `low-level' computational processes.\footnote{This idea is similar to that of \emph{denotational semantics} of programming languages \cite{scott1971toward}, but there is an important difference: Denotational semantics models programs usually by (extensional) functions, but we are concerned with (intensional) \emph{processes}.}
Also, what TMs formulate is essentially \emph{symbol manipulations}; however, the content of computation on mathematical, semantic and non-symbolic objects seems completely independent of its symbolic representation, e.g., consider a \emph{process} (not a function) to add numbers or to take the union of sets. 

Therefore, it would be rather appropriate, at least from conceptual and mathematical points of view, to formulate such `high-level' computational processes in a more abstract, in particular \emph{syntax-independent}, manner, in order to explain `low-level' computational processes, and then regard the latter as executable `symbolic implementations' of the former.

\subsection{Our research problem: mathematics of computational processes}
\label{OurResearchProblem}
To summarize, it would be reasonable and meaningful from both conceptual and mathematical points of view to develop mathematics of abstract (in particular syntax-independent) and `high-level' computational processes as well as executable and `low-level' ones beyond classical computation such that the former is defined to be `effectively computable' if it is `implementable' by the latter.

In fact, this (or similar) perspective is nothing new, and shared with various prominent researchers; for instance, Robin Milner stated:
\begin{quote}
... we should have achieved a mathematical model of computation, perhaps highly abstract in contrast with the concrete nature of paper and register machines, but such that programming languages are merely executable fragment of the theory ... \cite{milner2005software}
\end{quote}

We address this problem in the present paper.
However, since there are so many kinds of computation, e.g., parallel, concurrent, probabilistic, non-deterministic, quantum, etc., as the first step, this paper focuses on a certain kind of \emph{higher-order} and \emph{sequential} (i.e., at most one computational step may be performed at a time) computation, which is based on \emph{(sequential) game semantics}\footnote{It is both technically and conceptually simpler to focus on \emph{sequential} games, in which two participants \emph{alternately} take actions, than to consider more general \emph{concurrent} games, in which both participants may be active \emph{simultaneously} \cite{abramsky1999concurrent}.} introduced below.

\subsection{Game semantics}
\label{GameSemantics}
\emph{Game semantics} (of computation) \cite{abramsky1997semantics,abramsky1999game,hyland1997game} is a particular kind of \emph{denotational semantics} of programming languages \cite{amadio1998domains,winskel1993formal,gunter1992semantics}, in which types and terms are modeled as \emph{games} and \emph{strategies} (whose definitions are given in Section~\ref{Preliminaries}), respectively. 
Historically, having its roots in `games-based' approaches in mathematical logic to capture \emph{validity} \cite{requena1980dialogische,felscher1986dialogues}, \emph{higher-order computability} \cite{kleene1978recursive,kleene1980recursive,kleene1982recursive,kleene1985unimonotone,kleene1991recursive,gandy1993dialogues} and \emph{proofs in linear logic} \cite{blass1992game,abramsky1994games,hyland1993fair}, combined with ideas from \emph{sequential algorithms} \cite{berry1982sequential}, \emph{process calculi} \cite{milner1980calculus,hoare1978communicating} and \emph{geometry of interaction} \cite{girard1989geometry,girard1990geometry,girard1995geometry,girard2003geometry,girard2011geometry,girard2013geometry} in computer science, several variants of game semantics in its modern form were developed in the early 1990's to give the first syntax-independent characterization of the programming language PCF \cite{abramsky2000full,hyland2000full,nickau1994hereditarily}; since then a variety of games and strategies have been proposed to model various programming features \cite{abramsky1999game}. 

An advantage of game semantics is this flexibility: It models a wide range of languages by simply varying constraints on strategies \cite{abramsky1999game}, which enables one to systematically compare and relate different languages ignoring syntactic details.
Also, as \emph{full completeness} and \emph{full abstraction} results \cite{curien2007definability} in the literature have demonstrated, game semantics in general has an appropriate degree of abstraction (and thus it has a good potential to be mathematics of `high-level' computational processes).
Finally, yet another strong point of game semantics is its conceptual naturality: It interprets syntax as `dynamic interactions' between the participants of games, providing a computational and intensional explanation of syntax in a natural and intuitive (yet mathematically precise) manner.
Informally, one can imagine that games provide a `high-level' description of interactive computation between a TM and an oracle. 
Note that such an intensional nature stands in sharp contrast to the traditional \emph{domain-theoretic} semantics  \cite{amadio1998domains} which for instance cannot capture \emph{sequentiality} of PCF (but the game models \cite{abramsky2000full,hyland2000full,nickau1994hereditarily} can).

Let us give a brief and informal introduction to games and strategies (as defined in \cite{abramsky1999game}) below in order to sketch the main idea of the present paper.

A \emph{game}, roughly, is a certain kind of a rooted forest whose branches represent possible `developments' or \emph{(valid) positions} of a `game in the usual sense' (such as chess, poker, etc.).
\emph{Moves} of a game are nodes of the game, where some moves are distinguished and called \emph{initial}; only initial moves can be the first element (or occurrence) of a position of the game. 
\emph{Plays} of a game are increasing sequences $\boldsymbol{\epsilon}, m_1, m_1 m_2, \dots$ of positions of the game, where $\boldsymbol{\epsilon}$ is the \emph{empty sequence}. 
For our purpose, it suffices to focus on rather standard \emph{sequential} (as opposed to \emph{concurrent} \cite{abramsky1999concurrent}) and \emph{unpolarized} (as opposed to \emph{polarized }\cite{laurent2004polarized}) games played by two participants, \emph{Player} (\emph{P}), who represents a `computational agent', and \emph{Opponent} (\emph{O}), who represents an `environment', in each of which O always starts a play (i.e., unpolarized), and then they alternately (i.e., sequential) perform moves allowed by the rules of the game.
Strictly speaking, a position of each game is not just a sequence of moves: Each occurrence $m$ of O's or O- (resp. P's or P-) non-initial move in a position points to a previous occurrence $m'$ of P- (resp. O-) move in the position, representing that $m$ is performed specifically as a response to $m'$. 
A \emph{strategy} on a game, on the other hand, is what tells P which move (together with a pointer) she should make at each of her turns in the game.
Hence, a game semantics $\llbracket \_ \rrbracket_{\mathcal{G}}$ of a programming language $\mathscr{L}$ interprets a type $\mathsf{A}$ of $\mathscr{L}$ as a game $\llbracket \mathsf{A} \rrbracket_{\mathcal{G}}$ that specifies possible plays between P and O, and a term $\mathsf{M : A}$\footnote{For simplicity, here we focus on \emph{closed} terms, i.e., ones with the \emph{empty context}.} of $\mathscr{L}$ as a strategy $\llbracket \mathsf{M} \rrbracket_{\mathcal{G}}$ that describes a `strategy' for P on how to play in $\llbracket \mathsf{A} \rrbracket_{\mathcal{G}}$; an `execution' of the term $\mathsf{M}$ is then modeled as a play in $\llbracket \mathsf{A} \rrbracket_{\mathcal{G}}$ for which P follows $\llbracket \mathsf{M} \rrbracket_{\mathcal{G}}$.

Let us consider a simple example. The game $N$ of natural numbers is the following rooted tree (which is infinite in width):
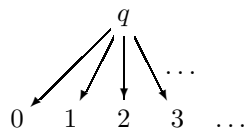
\begin{diagram}
& & q & & \\
& \ldTo(2, 2) \ldTo(1, 2) & \dTo & \rdTo(1, 2) \ \dots & \\
0 & 1 & 2 & 3 & \dots
\end{diagram}
in which a play starts with O's question $q$ (`What is your number?') and ends with P's answer $n \in \mathbb{N}$ (`My number is $n$!'), where $\mathbb{N}$ is the set of all natural numbers, and $n$ points to $q$ (though this pointer is omitted in the diagram). A strategy $\underline{10}$ on $N$, for instance, that corresponds to $10 \in \mathbb{N}$ can be represented by the map $q \mapsto 10$ equipped with a pointer from $10$ to $q$ (though it is the only choice).
In the following, the pointers of most strategies are obvious, and thus we often omit them.

There is a construction $\otimes$ on games, called \emph{tensor (product)}.
Conceptually, a position $\boldsymbol{s}$ of the tensor $A \otimes B$ of games $A$ and $B$ is an interleaving mixture of a position $\boldsymbol{t}$ of $A$ and a position $\boldsymbol{u}$ of $B$ developed `in parallel without communication'; more specifically, $\boldsymbol{t}$ (resp. $\boldsymbol{u}$) is the subsequence of $\boldsymbol{s}$ consisting of moves of $A$ (resp. $B$) such that the change of $AB$-parity (i.e., the switch between $\boldsymbol{t}$ and $\boldsymbol{u}$) in $\boldsymbol{s}$ must be made by O. 
The pointers in $\boldsymbol{t}$ and $\boldsymbol{u}$ are inherited from $\boldsymbol{s}$ in the obvious manner; this point holds also for other constructions on games and strategies in the rest of the introduction, and thus we shall not mention it again.
For instance, a maximal position of the tensor $N \otimes N$ is either of the following forms\footnote{The diagrams are only to make it explicit which component game each move belongs to; the two positions are just finite sequences $q^{[0]} n^{[0]} q^{[1]} m^{[1]}$ and $q^{[1]} m^{[1]} q^{[0]} n^{[0]}$ equipped with the pointers $q^{[i]} \leftarrow n^{[i]}$ and $q^{[i]} \leftarrow m^{[i]}$ ($i = 0, 1$).}:
\begin{center}
\begin{tabular}{ccccccccc}
$N^{[0]}$ & $\otimes$ & $N^{[1]}$ &&&& $N^{[0]}$ & $\otimes$ & $N^{[1]}$ \\ \cline{1-3} \cline{7-9}
\tikzmark{ctensor1} $q^{[0]}$&&&&&&&&\tikzmark{ctensor2} $q^{[1]}$ \\
\tikzmark{dtensor1} $n^{[0]}$&&&&&&&&\tikzmark{dtensor2} $m^{[1]}$ \\
&&\tikzmark{ctensor3} $q^{[1]}$&&&&\tikzmark{ctensor4} $q^{[0]}$&& \\
&&\tikzmark{dtensor3} $m^{[1]}$&&&&\tikzmark{dtensor4} $n^{[0]}$&&
\end{tabular}
\begin{tikzpicture}[overlay, remember picture, yshift=.25\baselineskip]
\draw [->] ({pic cs:dtensor1}) [bend left] to ({pic cs:ctensor1});
\draw [->] ({pic cs:dtensor2}) [bend left] to ({pic cs:ctensor2});
\draw [->] ({pic cs:dtensor3}) [bend left] to ({pic cs:ctensor3});
\draw [->] ({pic cs:dtensor4}) [bend left] to ({pic cs:ctensor4});
\end{tikzpicture}
\end{center}
where $n, m \in \mathbb{N}$, and $(\_)^{[i]}$ ($i = 0, 1$) are (arbitrary and unspecified) `tags' to distinguish the two copies of $N$ (but we often omit them if it does not bring confusion), and the arrows represent pointers in positions (henceforth we employ this notation).

Next, a fundamental construction $!$ on games, called \emph{exponential}, is basically the countably infinite iteration of $\otimes$, i.e., $!A \stackrel{\mathrm{df. }}{=} A \otimes A \otimes \dots$ for each game $A$, where the `tag' for each copy of $A$ is typically given by a natural number $i \in \mathbb{N}$.

Another central construction $\multimap$, called \emph{linear implication}, captures the notion of \emph{linear functions}, i.e., functions that consume exactly one input to produce an output. 
A position of the linear implication $A \multimap B$ from $A$ to $B$ is almost like a position of the tensor $A \otimes B$ except the following three points:
\begin{enumerate}

\item The first occurrence in the position must be a move of $B$; 

\item A change of $AB$-parity in the position must be made by P; 

\item Each occurrence of an initial move (called an \emph{initial occurrence}) of $A$ points to an initial occurrence of $B$.

\end{enumerate}
Thus, a typical position of the game $N \multimap N$ is the following:
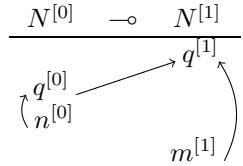
\begin{center}
\begin{tabular}{ccc}
$N^{[0]}$ & $\multimap$ & $N^{[1]}$ \\ \hline 
&&\tikzmark{cmultimap1} $q^{[1]} $\tikzmark{cmultimap3} \\
\tikzmark{cmultimap2} $q^{[0]}$ \tikzmark{dmultimap1}&& \\
\tikzmark{dmultimap2} $n^{[0]}$&& \\
&&$m^{[1]}$ \tikzmark{dmultimap3}
\end{tabular}
\begin{tikzpicture}[overlay, remember picture, yshift=.25\baselineskip]
\draw [->] ({pic cs:dmultimap1}) to ({pic cs:cmultimap1});
\draw [->] ({pic cs:dmultimap2}) [bend left] to ({pic cs:cmultimap2});
\draw [->] ({pic cs:dmultimap3}) [bend right] to ({pic cs:cmultimap3});
\end{tikzpicture}
\end{center}
where $n, m \in \mathbb{N}$, which can be read as follows:
\begin{enumerate}
\item O's question $q^{[1]}$ for an output (`What is your output?');
\item P's question $q^{[0]}$ for an input (`Wait, what is your input?');
\item O's answer, say $n^{[0]}$, to $q^{[0]}$ (`OK, here is an input $n$.');
\item P's answer, say $m^{[1]}$, to $q^{[1]}$ (`Alright, the output is then $m$.').
\end{enumerate}
This play corresponds to any linear function that maps $n \mapsto m$.
The strategy $\mathit{succ}$ (resp. $\mathit{double}$) on $N \multimap N$ for the successor (resp. doubling) function is represented by the map $q^{[1]} \mapsto q^{[0]}, q^{[1]}q^{[0]}n^{[0]} \mapsto n+1^{[1]}$ (resp. $q^{[1]} \mapsto q^{[0]}, q^{[1]}q^{[0]}n^{[0]} \mapsto 2 \cdot n^{[1]}$). 

Let us remark here that the following play, which corresponds to a \emph{constant} linear function that maps $x \mapsto m$ for all $x \in \mathbb{N}$, is also possible: $\boldsymbol{\epsilon}, q^{[1]}, q^{[1]}m^{[1]}$.
Thus, strictly speaking, $A \multimap B$ is the game of \emph{affine functions} from $A$ to $B$, but we follow the standard convention to call it a linear implication. 

Another construction $\&$ on games, called \emph{product}, is similar to yet simpler than tensor: A position $\boldsymbol{s}$ of the product $A \& B$ of $A$ and $B$ is either a position $\boldsymbol{t}^{[0]}$ of $A^{[0]}$ or a position $\boldsymbol{u}^{[1]}$ of $B^{[1]}$.
It is the product in the category $\mathcal{G}$ of games and strategies, e.g., there is the \emph{pairing} $\langle \sigma, \tau \rangle : C \multimap A \& B$ of given strategies $\sigma : C \multimap A$ and $\tau : C \multimap B$ that plays as $\sigma$ (resp. $\tau$) if O initiates a play by a move of $A$ (resp. $B$).
Clearly, we may generalize product and pairing to $n$-ary operations for any $n \in \mathbb{N}$. 

These four constructions $\otimes$, $!$, $\multimap$ and $\&$ come from the corresponding ones in \emph{linear logic} \cite{abramsky1994games}. 
Thus, in particular, the usual \emph{implication} (or the \emph{function space}) $\Rightarrow$ is recovered by \emph{Girard translation} \cite{girard1987linear}: $A \Rightarrow B \stackrel{\mathrm{df. }}{=} \! \ !A \multimap B$.

Girard translation makes explicit the point that some functions need to refer to an input \emph{more than once} to produce an output, i.e., there are non-linear functions.
For instance, consider the game $(N \Rightarrow N) \Rightarrow N$ of higher-order functions, in which the following position is possible: 
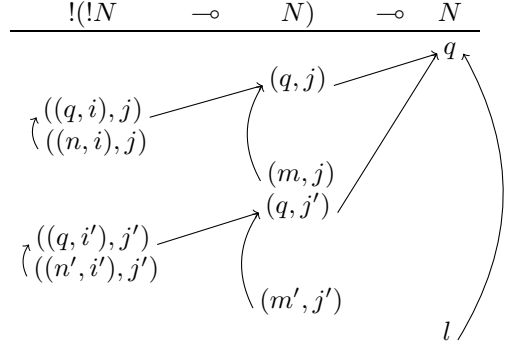
\begin{center}
\begin{tabular}{ccccc}
$!(!N$ & $\multimap$ & $N)$ & $\multimap$ & $N$ \\ \hline 
&&&& \tikzmark{chigher1} $q$ \tikzmark{chigher9} \\
&&\tikzmark{chigher2} $(q, j)$ \tikzmark{dhigher1} && \\
\tikzmark{chigher3} $((q, i), j)$ \tikzmark{dhigher2} && && \\
\tikzmark{dhigher3} $((n, i), j)$&& && \\
&& \tikzmark{dhigher4} $(m, j)$ && \\
&&\tikzmark{chigher8} $(q, j')$ \tikzmark{dhigher5} && \\
\tikzmark{chigher7} $((q, i'), j')$ \tikzmark{dhigher6} && && \\
\tikzmark{dhigher7} $((n', i'), j')$&& && \\
&& \tikzmark{dhigher8} $(m', j')$ && \\
&&&& $l$ \tikzmark{dhigher9}
\end{tabular}
\begin{tikzpicture}[overlay, remember picture, yshift=.25\baselineskip]
\draw [->] ({pic cs:dhigher1}) to ({pic cs:chigher1});
\draw [->] ({pic cs:dhigher2}) to ({pic cs:chigher2});
\draw [->] ({pic cs:dhigher3}) [bend left] to ({pic cs:chigher3});
\draw [->] ({pic cs:dhigher4}) [bend left] to ({pic cs:chigher2});
\draw [->] ({pic cs:dhigher5}) to ({pic cs:chigher1});
\draw [->] ({pic cs:dhigher6}) to ({pic cs:chigher8});
\draw [->] ({pic cs:dhigher7}) [bend left] to ({pic cs:chigher7});
\draw [->] ({pic cs:dhigher8}) [bend left] to ({pic cs:chigher8});
\draw [->] ({pic cs:dhigher9}) [bend right] to ({pic cs:chigher9});
\end{tikzpicture}
\end{center}
where $n, n', m, m', l, i, i', j, j' \in \mathbb{N}$, $i \neq i'$ and $j \neq j'$, which can be read as follows:
\begin{enumerate}
\item O's question $q$ for an output (`What is your output?');
\item P's question $(q, j)$ for an input function (`Wait, your first output please!');
\item O's question $((q, i), j)$ for an input (`What is your first input then?');
\item P's answer, say $((n, i), j)$, to $((q, i), j)$ (`Here is my first input $n$.');
\item O's answer, say $(m, j)$, to $(q, j)$ (`OK, here is my first output $m$.');
\item P's question $(q, j')$ for an input function (`Your second output please!');
\item O's question $((q, i'), j')$ for an input (`What is your second input then?');
\item P's answer, say $((n', i'), j')$, to $((q, i'), j')$ (`Here is my second input $n'$.');
\item O's answer, say $(m', j')$, to $(q, j')$ (`OK, here is my second output $m'$.');
\item P's answer, say $l$, to $q$ (`Alright, my output is then $l$.').
\end{enumerate}
In this play, P asks O \emph{twice} about an input strategy $N \Rightarrow N$.
Clearly, such a play is not possible on the linear implication $(N \multimap N) \multimap N$ or $(N \Rightarrow N) \multimap N$.
The strategy $\mathit{PlusAppToZeroAndOne} : (N \Rightarrow N) \Rightarrow N$ that computes the sum $f(0)+f(1)$ for a given function $f : \mathbb{N} \Rightarrow \mathbb{N}$, for instance, plays as follows: 
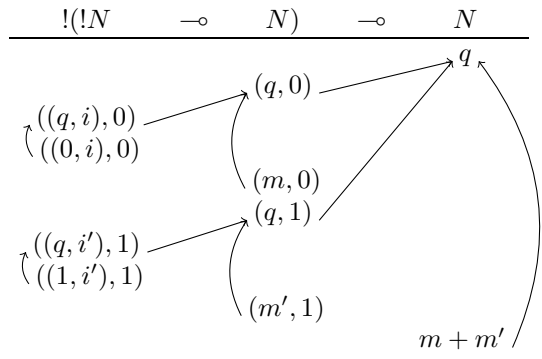
\begin{center}
\begin{tabular}{ccccc}
$!(!N$ & $\multimap$ & $N)$ & $\multimap$ & $N$ \\ \hline 
&&&& \tikzmark{chigher11} $q$ \tikzmark{chigher19} \\
&&\tikzmark{chigher12} $(q, 0)$ \tikzmark{dhigher11} && \\
\tikzmark{chigher13} $((q, i), 0)$ \tikzmark{dhigher12} && && \\
\tikzmark{dhigher13} $((0, i), 0)$ && && \\
&& \tikzmark{dhigher14} $(m, 0)$&& \\
&&\tikzmark{chigher18} $(q, 1)$ \tikzmark{dhigher15} && \\
\tikzmark{chigher17} $((q, i'), 1)$ \tikzmark{dhigher16} && && \\
\tikzmark{dhigher17} $((1, i'), 1)$&& && \\
&& \tikzmark{dhigher18} $(m', 1)$ && \\
&&&& $m+m'$ \tikzmark{dhigher19}
\end{tabular}
\begin{tikzpicture}[overlay, remember picture, yshift=.25\baselineskip]
\draw [->] ({pic cs:dhigher11}) to ({pic cs:chigher11});
\draw [->] ({pic cs:dhigher12}) to ({pic cs:chigher12});
\draw [->] ({pic cs:dhigher13}) [bend left] to ({pic cs:chigher13});
\draw [->] ({pic cs:dhigher14}) [bend left] to ({pic cs:chigher12});
\draw [->] ({pic cs:dhigher15}) to ({pic cs:chigher11});
\draw [->] ({pic cs:dhigher16}) to ({pic cs:chigher18});
\draw [->] ({pic cs:dhigher17}) [bend left] to ({pic cs:chigher17});
\draw [->] ({pic cs:dhigher18}) [bend left] to ({pic cs:chigher18});
\draw [->] ({pic cs:dhigher19}) [bend right] to ({pic cs:chigher19});
\end{tikzpicture}
\end{center}
where the `tags' $j = 0$ and $j' = 1$ are arbitrarily chosen, i.e., any $j, j' \in \mathbb{N}$ work.

Finally, let us point out that any strategy $\phi : \ !A \multimap B$ induces its \emph{promotion} $\phi^\dagger : \ !A \multimap \ !B$ such that if $\phi$ plays, for instance, as 
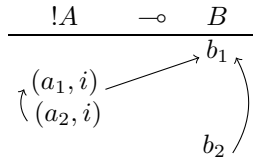
\begin{center}
\begin{tabular}{ccc}
$!A$ & $\multimap$ & $B$ \\ \hline 
&&\tikzmark{cpromotion1} $b_1$ \tikzmark{cpromotion3} \\
\tikzmark{cpromotion2} $(a_1, i)$ \tikzmark{dpromotion1}&& \\
\tikzmark{dpromotion2} $(a_2, i)$ && \\
&&$b_2$ \tikzmark{dpromotion3}
\end{tabular}
\begin{tikzpicture}[overlay, remember picture, yshift=.25\baselineskip]
\draw [->] ({pic cs:dpromotion1}) to ({pic cs:cpromotion1});
\draw [->] ({pic cs:dpromotion2}) [bend left] to ({pic cs:cpromotion2});
\draw [->] ({pic cs:dpromotion3}) [bend right] to ({pic cs:cpromotion3});
\end{tikzpicture}
\end{center}
then $\phi^\dagger$ plays as
\begin{center}
\begin{tabular}{ccc}
$!A$ & $\multimap$ & $!B$ \\ \hline 
&&\tikzmark{cpromotion11} $(b_1, j)$ \tikzmark{cpromotion13} \\
\tikzmark{cpromotion12} $(a_1, \langle i, j \rangle)$ \tikzmark{dpromotion11} && \\
\tikzmark{dpromotion12} $(a_2, \langle i, j \rangle)$ && \\
&& $(b_2, j)$ \tikzmark{dpromotion13} \\
&& \tikzmark{cpromotion14} $(b_1, j')$ \tikzmark{cpromotion16} \\
\tikzmark{cpromotion15} $(a_1, \langle i, j' \rangle)$ \tikzmark{dpromotion14} && \\
\tikzmark{dpromotion15} $(a_2, \langle i, j' \rangle)$ && \\
&& $(b_2, j')$ \tikzmark{dpromotion16}
\end{tabular}
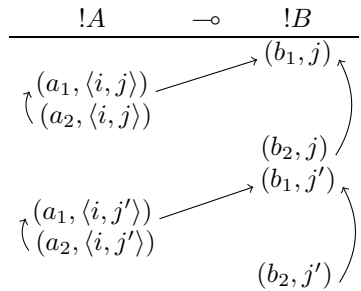
\begin{tikzpicture}[overlay, remember picture, yshift=.25\baselineskip]
\draw [->] ({pic cs:dpromotion11}) to ({pic cs:cpromotion11});
\draw [->] ({pic cs:dpromotion12}) [bend left] to ({pic cs:cpromotion12});
\draw [->] ({pic cs:dpromotion13}) [bend right] to ({pic cs:cpromotion13});
\draw [->] ({pic cs:dpromotion14}) to ({pic cs:cpromotion14});
\draw [->] ({pic cs:dpromotion15}) [bend left] to ({pic cs:cpromotion15});
\draw [->] ({pic cs:dpromotion16}) [bend right] to ({pic cs:cpromotion16});
\end{tikzpicture}
\end{center}
where $\langle \_, \_ \rangle : \mathbb{N} \times \mathbb{N} \stackrel{\sim}{\to} \mathbb{N}$ is an arbitrarily fixed bijection, i.e., $\phi^\dagger$ plays as $\phi$ for each \emph{thread} in a position of $!A \multimap \ !B$ that corresponds to a position of $!A \multimap B$.

\subsection{Towards a game-semantic model of computation}
As seen in the above examples, games and strategies capture higher-order (and sequential) computation in an abstract and conceptually natural fashion, where O plays the role of an oracle as a part of the formalization.
Note also that P computes on `external behavior' of O, and thus O's computation does not have to be recursive. 
Thus, one may expect that games and strategies would be appropriate as mathematics of `high-level' computational processes to solve our research problem defined in Section~\ref{OurResearchProblem}.
However, conventional games and strategies have never been formulated as a mathematical model of computation (in the sense of TMs); rather, a primary focus of the field has been \emph{full abstraction} \cite{curien2007definability,amadio1998domains}, i.e., to characterize \emph{observational equivalences} in syntax.
In other words, game semantics has not been concerned that much with step-by-step processes in computation or their `effective computability', and it has been identifying programs with the same \emph{value} \cite{winskel1993formal,gunter1992semantics}.

For instance, strategies on the game $N \Rightarrow N$ typically play by $q . (q, i) . (n, i) . m$, where $n, m, i \in \mathbb{N}$, as described above, and so they are essentially functions that map $n \mapsto m$; in particular, it is not formulated at all how they calculate the fourth element $m$ from the third one $(n, i)$.\footnote{Nevertheless, they exhibit some simple communication between the participants as seen in the above examples; in other words, game semantics is intensional to some degree but not completely. This is roughly why game semantics has been very successful in the field of denotational semantics.}
Consequently, `effective computability' in game semantics has been \emph{extrinsic}: A strategy has been defined to be `effective' or \emph{recursive} if it is representable by a partial recursive function \cite{abramsky2000full,hyland2000full,feree2017game}.

This situation is in a sense frustrating since games and strategies seem to have a good potential to give a \emph{semantic}, \emph{non-axiomatic}, \emph{non-inductive} and \emph{intrinsic} (i.e., without recourse to an established model of computation) formulation of higher-order computation, but they have not taken advantage of this potential.


For the potential of game semantics, we have decided to employ games and strategies as our basic mathematical framework and extend them to give mathematics of computational processes in the sense described in Section~\ref{OurResearchProblem}.
For this aim, we shall first refine the category $\mathcal{G}$ of games and strategies in such a way that accommodates step-by-step processes in computation, and then define their `effectivity' in terms of their `atomic steps'.
Fortunately, there is already the bicategory $\mathcal{DG}$ of \emph{dynamic games and strategies} \cite{yamada2016dynamic}, which addresses the first point.

\subsection{Dynamic games and strategies}
\label{DynamicGamesAndStrategies}
In the literature, there are several game models \cite{greenland2005game,blum2008concrete,ong2006model,dimovski2005data} that exhibit step-by-step processes in computation to serve as a tool for program verification and analysis (the work \cite{ghica2005slot,dal2008quantitative} may be called `intensional game semantics', but they rather keep track of `costs' in computation, not computational steps themselves). 
However, these variants of games and strategies are just conventional ones, and consequently such step-by-step processes have no official status in their categories.

The problem lies in the point that in conventional game semantics composition of strategies is executed as \emph{parallel composition plus hiding} \cite{abramsky1997semantics}, where \emph{hiding} excludes intermediate steps.
Let us illustrate this point by a simple and informal example as follows.
Consider again strategies $\mathit{succ}$ and $\mathit{double}$, but this time they are adjusted to the game $N \Rightarrow N$. 
Their computations can be described by the following diagrams:
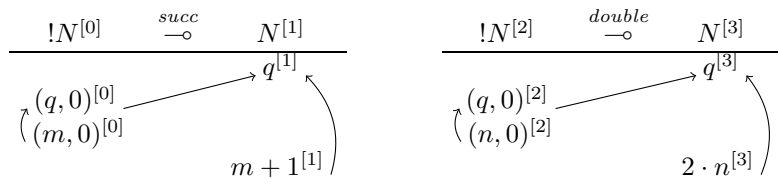
\begin{center}
\begin{tabular}{ccccccccc}
$!N^{[0]}$ & $\stackrel{\mathit{succ}}{\multimap}$ & $N^{[1]}$ &&&& $!N^{[2]}$ & $\stackrel{\mathit{double}}{\multimap}$ & $N^{[3]}$ \\ \cline{1-3} \cline{7-9}
&&\tikzmark{csucc31} $q^{[1]}$ \tikzmark{csucc33}&&&&&&\tikzmark{cdouble31} $q^{[3]}$ \tikzmark{cdouble33} \\
\tikzmark{csucc32} $(q, 0)^{[0]}$ \tikzmark{dsucc31}&&&&&&\tikzmark{cdouble32} $(q, 0)^{[2]}$ \tikzmark{ddouble31}&& \\
\tikzmark{dsucc32} $(m, 0)^{[0]}$ &&&&&&\tikzmark{ddouble32} $(n, 0)^{[2]}$&& \\
&&$m+1^{[1]}$ \tikzmark{dsucc33} &&&&&&$2 \cdot n^{[3]}$ \tikzmark{ddouble33} 
\end{tabular}
\begin{tikzpicture}[overlay, remember picture, yshift=.25\baselineskip]
\draw [->] ({pic cs:dsucc31}) to ({pic cs:csucc31});
\draw [->] ({pic cs:dsucc32}) [bend left] to ({pic cs:csucc32});
\draw [->] ({pic cs:dsucc33}) [bend right] to ({pic cs:csucc33});
\draw [->] ({pic cs:ddouble31}) to ({pic cs:cdouble31});
\draw [->] ({pic cs:ddouble32}) [bend left] to ({pic cs:cdouble32});
\draw [->] ({pic cs:ddouble33}) [bend right] to ({pic cs:cdouble33});
\end{tikzpicture}
\end{center}
where the `tag' $(\_, 0)$ on moves of the domain $!N$ has been arbitrarily chosen (i.e., any natural number $i \in \mathbb{N}$ instead of $0$ works).
The composition $\mathit{double} \bullet \mathit{succ} \stackrel{\mathrm{df. }}{=} \mathit{double} \circ \mathit{succ}^\dagger = \mathit{succ}^\dagger ; \mathit{double} : N \Rightarrow N$ is calculated as follows. 
First, by \emph{internal communication}, we mean that $\mathit{succ}^\dagger$ and $\mathit{double}$ `communicate' to each other by `synchronizing' the codomain $!N^{[1]}$ of $\mathit{succ}^\dagger$ and the domain $!N^{[2]}$ of $\mathit{double}$, for which P plays the role of O in these intermediate games $!N^{[1]}$ and $!N^{[2]}$ by `copying' her last moves\footnote{More precisely, they are the \emph{last occurrences of P-moves}; however, for convenience we shall keep using this kind of abuse of terminology in the rest of the introduction.} of $!N^{[2]}$ and $!N^{[1]}$, respectively, resulting in the following play:
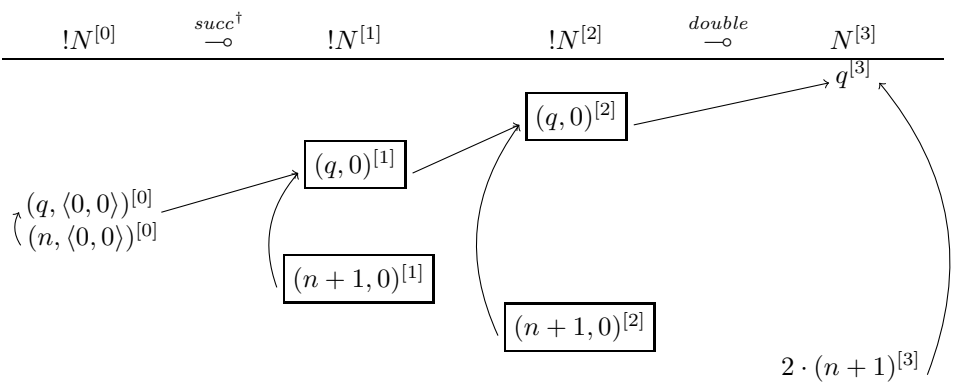
\begin{center}
\begin{tabular}{ccccccc}
$!N^{[0]}$ & $\stackrel{\mathit{succ}^\dagger}{\multimap}$ & $!N^{[1]}$ && $!N^{[2]}$ & $\stackrel{\mathit{double}}{\multimap}$ & $N^{[3]}$ \\ \hline
&&&&&& \tikzmark{cPC2} $q^{[3]}$ \tikzmark{cPC1} \\
&&&&\tikzmark{cPC3} \fbox{$(q, 0)^{[2]}$} \tikzmark{dPC2} && \\
&& \tikzmark{cPC4} \fbox{$(q, 0)^{[1]}$} \tikzmark{dPC3} &&&& \\
\tikzmark{cPC5} $(q, \langle 0, 0 \rangle)^{[0]}$ \tikzmark{dPC4} &&&&&& \\
\tikzmark{dPC5} $(n, \langle 0, 0 \rangle)^{[0]}$ &&&&&& \\
&&\tikzmark{dPC6} \fbox{$(n+1, 0)^{[1]}$} &&&& \\
&&&&\tikzmark{dPC7} \fbox{$(n+1, 0)^{[2]}$} && \\
&&&&&&$2 \cdot (n+1)^{[3]}$ \tikzmark{dPC1}
\end{tabular}
\begin{tikzpicture}[overlay, remember picture, yshift=.25\baselineskip]
\draw [->] ({pic cs:dPC1}) [bend right] to ({pic cs:cPC1});
\draw [->] ({pic cs:dPC2}) to ({pic cs:cPC2});
\draw [->] ({pic cs:dPC3}) to ({pic cs:cPC3});
\draw [->] ({pic cs:dPC4}) to ({pic cs:cPC4});
\draw [->] ({pic cs:dPC5}) [bend left] to ({pic cs:cPC5});
\draw [->] ({pic cs:dPC6}) [bend left] to ({pic cs:cPC4});
\draw [->] ({pic cs:dPC7}) [bend left] to ({pic cs:cPC3});
\end{tikzpicture}
\end{center}
where moves for internal communication are marked by square boxes just for clarity, and a pointer from $(q, 0)^{[1]}$ to $(q, 0)^{[2]}$ is added because the move $(q, 0)^{[1]}$ is no longer initial. 
Importantly, it is assumed that O plays on the `external game' $!N^{[0]} \multimap N^{[3]}$, `seeing' only moves of $!N^{[0]}$ or $N^{[3]}$.
The resulting play is to be read as follows:
\begin{enumerate}

\item O's question $q^{[3]}$ for an output in $!N^{[0]} \multimap N^{[3]}$ (`What is your output?');

\item P's question \fbox{$(q, 0)^{[2]}$} by $\mathit{double}$ for an input in $!N^{[2]} \multimap N^{[3]}$ (`Wait, what is your input?');

\item \fbox{$(q, 0)^{[2]}$} in turn triggers the question \fbox{$(q, 0)^{[1]}$} for an output in $!N^{[0]} \multimap \ !N^{[1]}$ (`What is your output?');

\item P's question $(q, \langle 0, 0 \rangle)^{[0]}$ by $\mathit{succ}^\dagger$ for an input in $!N^{[0]} \multimap \ !N^{[1]}$ (`Wait, what is your input?');

\item O's answer, say $(n, \langle 0, 0 \rangle)^{[1]}$, to the question $(q, \langle 0, 0 \rangle)^{[0]}$ in $!N^{[0]} \multimap \  !N^{[3]}$ (`Here is an input $n$.');

\item P's answer \fbox{$(n+1, 0)^{[1]}$} to the question \fbox{$(q, 0)^{[1]}$} by $\mathit{succ}^\dagger$ in $!N^{[0]} \multimap \ !N^{[1]}$ (`The output is then $n+1$.');

\item \fbox{$(n+1, 0)^{[1]}$} in turn triggers the answer \fbox{$(n+1, 0)^{[2]}$} to the question \fbox{$(q, 0)^{[2]}$} in $!N^{[2]} \multimap N^{[3]}$ (`Here is the input $n+1$.');

\item P's answer $2 \cdot (n+1)^{[3]}$ to the initial question $q^{[3]}$ by $\mathit{double}$ in $!N^{[0]} \multimap N^{[3]}$ (`The output is then $2 \cdot (n+1)$!').

\end{enumerate}

Next, \emph{hiding} means to hide or delete all moves with square boxes from the play, resulting in the strategy for the function $n \mapsto 2 \cdot (n + 1)$ as expected:
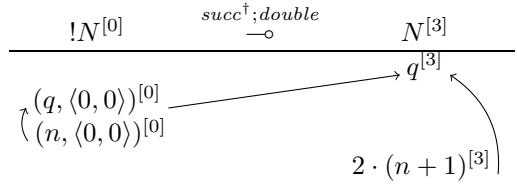
\begin{center}
\begin{tabular}{ccc}
$!N^{[0]}$ & $\stackrel{\mathit{succ}^\dagger ; \mathit{double}}{\multimap}$ & $N^{[3]}$ \\ \hline
&&\tikzmark{cPCH1} $q^{[3]}$ \tikzmark{cPCH3} \\
\tikzmark{cPCH2} $(q, \langle 0, 0 \rangle)^{[0]}$ \tikzmark{dPCH1}&& \\
\tikzmark{dPCH2} $(n, \langle 0, 0 \rangle)^{[0]}$&& \\
&&$2 \cdot (n + 1)^{[3]}$ \tikzmark{dPCH3} 
\end{tabular}
\begin{tikzpicture}[overlay, remember picture, yshift=.25\baselineskip]
\draw [->] ({pic cs:dPCH1}) to ({pic cs:cPCH1});
\draw [->] ({pic cs:dPCH2}) [bend left] to ({pic cs:cPCH2});
\draw [->] ({pic cs:dPCH3}) [bend right] to ({pic cs:cPCH3});
\end{tikzpicture}
\end{center}
By the hiding operation, the resulting play is a legal one on the game $N \Rightarrow N$, but let us point out that the intermediate occurrence of moves, \emph{which represent step-by-step processes} in computation, have been excluded by this operation.  

Nevertheless, the present author and Samson Abramsky have introduced a novel, \emph{dynamic} variant of games and strategies that systematically model `dynamics' and `intensionality' of computation, and also studied their algebraic structures \cite{yamada2016dynamic}.
In contrast to the previous work mentioned above, dynamic strategies \emph{themselves} embody step-by-step processes in computation by retaining intermediate occurrences of moves, and composition of them is parallel composition \emph{without} hiding.
In addition, the categorical structure of game semantics is not lost but rather refined by the \emph{cartesian closed bicategory} \cite{ouaknine1997two} $\mathcal{DG}$ of dynamic games and strategies, forming a categorical `universe' of `high-level' computational processes.

\subsection{Viable strategies}
Now, the remaining problem is to define `effective' dynamic strategies in an \emph{intrinsic} (i.e., solely in terms of games and strategies), \emph{non-inductive} and \emph{non-axiomatic} manner.
Of course, we need to provide a convincing argument that justifies their `effectivity' (though such an argument can never be mathematically precise) as in the case of TMs.
Moreover, to obtain a powerful model of computation, they should be at least \emph{Turing complete}, i.e., they ought to subsume all classically computable partial functions.
This sets up, in addition to the conceptual quest so far, an intriguing mathematical question in its own right: 
\begin{quote}
Is there any \emph{intrinsic}, \emph{non-inductive} and \emph{non-axiomatic} notion of `effective computability' of dynamic strategies that is \emph{Turing complete}?
\end{quote}
Not surprisingly, perhaps, this problem has turned out to be challenging, and the main technical achievement of the present paper is to give a positive answer to it.

As already mentioned, our solution is to give `low-level' computational processes (which are clearly `executable') in order to define `effectivity' of dynamic strategies (or `high-level' computational processes).
This is achieved roughly as follows.
\begin{remark}
The concepts introduced below make sense for conventional (i.e., non-dynamic) games and strategies too, but they do not give rise to a Turing complete model of computation for composition of conventional strategies does not preserve our notion of `effectivity' or \emph{viability} as we shall see.
\end{remark}

First, we give, by an alphabet, a concrete formalization of `tags' for disjoint union of sets of moves for constructions on games in order to rigorously formulate `effectivity' of strategies.
As we have seen in Section~\ref{GameSemantics}, a finite number of `tags' suffice for most constructions on games, but it is not the case for exponential $!$. 
Then, we formalize `tags' for exponential by an unary representation $\underbrace{\ell \ell \dots \ell}_i$ of natural numbers $i \in \mathbb{N}$ (extended by a symbolic implementation of a recursive bijection $\mathbb{N}^\ast \stackrel{\sim}{\rightarrow} \mathbb{N}$ \cite{cutland1980computability}, but here we omit the extension for simplicity), and employ, instead of the game $N$, the \emph{lazy} variant $\mathcal{N}$ of natural number game whose maximal positions are either of the following forms:
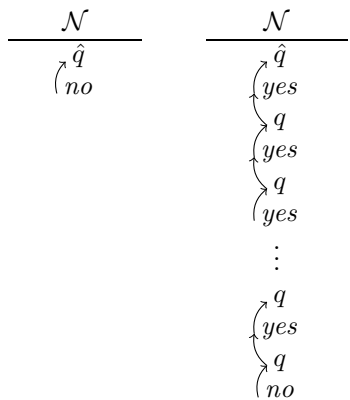
\begin{center}
\begin{tabular}{ccccccccc}
&$\mathcal{N}$ &&&&&$\mathcal{N}$& \\ \cline{1-3} \cline{6-8}
&\tikzmark{cLazyN10} $\hat{q}$&&&&&\tikzmark{cLazyN1} $\hat{q}$& \\
&\tikzmark{dLazyN10} $\mathit{no}$&&&&& \tikzmark{cLazyN2} $\mathit{yes}$ & \\
&&&&&& \tikzmark{dLazyN2} $q$ & \\
&&&&&& \tikzmark{cLazyN4} $\mathit{yes}$ & \\
&&&&&& \tikzmark{dLazyN4} $q$ & \\
&&&&&& \tikzmark{dLazyN5} $\mathit{yes}$& \\
&&&&&&$\vdots$& \\
&&&&&& \tikzmark{cLazyN6} $q$ & \\
&&&&&& \tikzmark{cLazyN7} $\mathit{yes}$ & \\
&&&&&& \tikzmark{dLazyN7} $q$ & \\
&&&&&& \tikzmark{dLazyN8} $\mathit{no}$ & \\
\end{tabular}
\begin{tikzpicture}[overlay, remember picture, yshift=.25\baselineskip]
\draw [->] ({pic cs:cLazyN2}) [bend left] to ({pic cs:cLazyN1});
\draw [->] ({pic cs:dLazyN2}) [bend left] to ({pic cs:cLazyN2});
\draw [->] ({pic cs:cLazyN4}) [bend left] to ({pic cs:dLazyN2});
\draw [->] ({pic cs:dLazyN4}) [bend left] to ({pic cs:cLazyN4});
\draw [->] ({pic cs:dLazyN5}) [bend left] to ({pic cs:dLazyN4});
\draw [->] ({pic cs:cLazyN7}) [bend left] to ({pic cs:cLazyN6});
\draw [->] ({pic cs:dLazyN7}) [bend left] to ({pic cs:cLazyN7});
\draw [->] ({pic cs:dLazyN8}) [bend left] to ({pic cs:dLazyN7});
\draw [->] ({pic cs:dLazyN10}) [bend left] to ({pic cs:cLazyN10});
\end{tikzpicture}
\end{center}
where the number $n$ of $\mathit{yes}$ in the position ranges over all natural numbers, which represents the number intended by P.
In this way, $\mathcal{N}$ gives an unary representation\footnote{This choice is far from canonical; it should also work to employ another representation of natural numbers, e.g., the binary one. We have chosen the unary representation for its (both conceptual and technical) simplicity.} of natural numbers.
Note that the initial question $\hat{q}$ must be distinguished from the non-initial one $q$ for a technical reason, which will be clarified in Section~\ref{Preliminaries}.
This sets up a finitary representation of game-semantic computation on natural numbers. 
Let us write $\underline{n}$ for the position $\hat{q} . \mathit{yes} . \underbrace{q . \mathit{yes} \dots q . \mathit{yes}}_{n-1} . q . \mathit{no}$ such that $n \geqslant 1$, and $\underline{0}$ for the position $\hat{q} . \mathit{no}$, where we omit pointers (we shall often employ this convention).

Next, as we shall see, dynamic strategies modeling PCF only need to refer to \emph{at most three moves} in the history of previous moves which may be `effectively' identified by pointers (specifically the last three moves in the \emph{P-view} \cite{hyland2000full,abramsky1999game}; see Definition~\ref{DefViews}).
Thus, it may seem at first glance that \emph{finitary} dynamic strategies in the following sense suffice: A strategy is \emph{finitary} if its representation by a partial function \cite{mccusker1998games,hyland2000full} that assigns the next P-move to (at most) three previous moves, called its \emph{table}, is finite. 
However, it is not the case: Finitary strategies cannot handle unboundedly many manipulations of `tags' for exponential (more precisely, manipulations such that the length of input or output `tags' is unbounded), but such manipulations seem to be necessary for Turing completeness, e.g., a  strategy that models primitive recursion or minimization has to interact with input strategies unboundedly many times, and thus it must handle unboundedly many `tags'.

Then, the main idea of our solution is to define a strategy to be \emph{viable} if its table is `describable' by a finitary strategy.
To state it more precisely, let us note that there is the \emph{terminal game} $T$ which has only the empty sequence $\boldsymbol{\epsilon}$ as a position, and each game $G$ is isomorphic (or identical up to `tags') to the implication $T \Rightarrow G$.
Hence, we may regard strategies $\sigma : G$ as a one on the implication $T \Rightarrow G$ up to `tags', and vice versa.
Also, we define for each move $m = [m']_{\boldsymbol{e}}$ of a game $G$, where $\boldsymbol{e} = e_1 . e_2 \dots e_{k}$ is a unary representation of the `tag' for exponential on $m$, the strategy $\underline{m}$ on a suitable game $\mathcal{G}(M_G)$ that plays as $\hat{q} . \mathsf{m'} . q_1 . \mathsf{e_1} . q_2 . \mathsf{e_2} \dots q_{k} . \mathsf{e_{k}} . q_{k+1} . \checkmark$ (which is a little bit simplified version of the strategy representation of moves defined in Section~\ref{ViableStrategies}), where the font difference between the moves is just for clarity. 
In this manner, the strategy $\underline{m} : \mathcal{G}(M_G)$ encodes the move $m$ of $G$.
Then, viability of strategies is defined more precisely: A strategy $\sigma : G$ is defined to be \emph{viable} if its partial function representation $(m_3, m_2, m_1) \mapsto m$, where $m_1$, $m_2$ and $m_3$ are the last, second last and third last moves of the current P-view, respectively, is `implementable' by a finitary strategy $\mathcal{A}(\sigma)^\circledS : \mathcal{G}(M_{G}) \& \mathcal{G}(M_{G}) \& \mathcal{G}(M_{G}) \Rightarrow \mathcal{G}(M_{G})$, called an \emph{instruction strategy} for $\sigma$, in the sense that the composition $\mathcal{A}(\sigma)^\circledS \circ \langle \underline{m_3}, \underline{m_2}, \underline{m_1} \rangle^\dagger : \mathcal{G}(M_G)$ coincides with $\underline{m} : \mathcal{G}(M_{G})$ for all quadruples $(m_3, m_2, m_1) \mapsto m$ in the table of $\sigma$, where $\langle \underline{m_3}, \underline{m_2}, \underline{m_1} \rangle : T \Rightarrow \mathcal{G}(M_G) \& \mathcal{G}(M_G) \& \mathcal{G}(M_G)$ is the ternary pairing of the strategies $\underline{m_i} : T \Rightarrow \mathcal{G}(M_G)$ (i = 1, 2, 3).

For instance, consider the successor and doubling strategies modified for the lazy natural number game $\mathcal{N}$, whose plays (on a non-zero input) are as in Figure~\ref{FigSuccAndDoublingStrategies}.
Roughly, $\mathit{succ}$ copies a given input on $!\mathcal{N}^{[0]}$ and repeats it as an output on $\mathcal{N}^{[1]}$, but it adds one more $\mathit{yes}$ to $\mathcal{N}^{[1]}$ before $\mathit{no}$; similarly, $\mathit{double}$ copies an input and repeats it as an output, but it doubles the number of $\mathit{yes}$'s in the output.
It is easy to see that for computing the next P-move (and its pointer) at an odd-length position $\boldsymbol{s} = m_1 m_2 \dots m_{2i+1}$ the strategies only need to refer to at most the last O-move $m_{2i+1}$, the P-move $m_{2j}$ pointed by $m_{2i+1}$ and the O-move $m_{2j-1}$ (they are the last three moves in the P-view of $\boldsymbol{s}$), e.g., $\mathit{succ} : (\square, \square, [\hat{q}^{[1]}]) \mapsto [\hat{q}^{[0]}]$, $([\hat{q}^{[1]}], [\hat{q}^{[0]}], [\mathit{yes}^{[0]}]) \mapsto [\mathit{yes}^{[1]}]$, and so on, where $\square$ denotes `no move'.
Note that these strategies do not need unboundedly many manipulations of `tags'; they are in fact finitary (it is easy to construct finite tables for them, each of which consists of a finite number of quadruples of moves of the form $(m_3, m_2, m_1) \mapsto m$). 
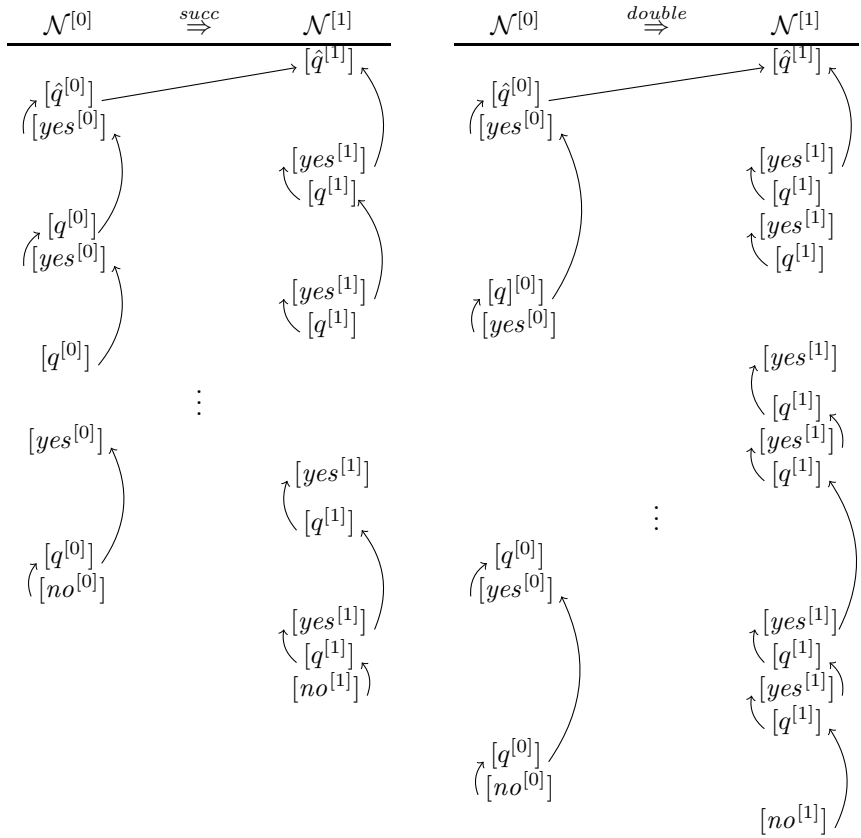
\begin{figure}
\begin{center}
\begin{tabular}{cccccccccccc}
$\mathcal{N}^{[0]}$ && $\stackrel{\mathit{succ}}{\Rightarrow}$ && $\mathcal{N}^{[1]}$ &&& $\mathcal{N}^{[0]}$ && $\stackrel{\mathit{double}}{\Rightarrow}$ && $\mathcal{N}^{[1]}$  \\ \cline{1-5} \cline{8-12}
&&&&\tikzmark{csucc1} $[\hat{q}^{[1]}]$ \tikzmark{csucc3} &&&& &&&\tikzmark{csucc21} $[\hat{q}^{[1]}]$ \tikzmark{csucc23} \\
\tikzmark{csucc2} $[\hat{q}^{[0]}]$ \tikzmark{dsucc1}&&&& &&& \tikzmark{csucc22} $[\hat{q}^{[0]}]$ \tikzmark{dsucc21}&&&& \\
\tikzmark{dsucc2} $[\mathit{yes}^{[0]}]$ \tikzmark{csucc5}&&&& &&& \tikzmark{dsucc22} $[\mathit{yes}^{[0]}]$ \tikzmark{csucc25}&&&& \\
&&&&\tikzmark{csucc4} $[\mathit{yes}^{[1]}]$ \tikzmark{dsucc3} &&&& &&&\tikzmark{csucc24} $[\mathit{yes}^{[1]}]$ \tikzmark{dsucc23} \\
&&&&\tikzmark{dsucc4} $[q^{[1]}]$\tikzmark{csucc9} &&&& &&&\tikzmark{dsucc24} $[q^{[1]}]$ \tikzmark{csucc29} \\
\tikzmark{csucc8} $[q^{[0]}] $\tikzmark{dsucc5}&&&& &&& &&&&\tikzmark{csucc27} $[\mathit{yes}^{[1]}]$ \tikzmark{dsucc29} \\
\tikzmark{dsucc8} $[\mathit{yes}^{[0]}]$ \tikzmark{csucc6}&&&& &&&& &&&\tikzmark{dsucc27} $[q^{[1]}]$ \\
&&&&\tikzmark{csucc7} $[\mathit{yes}^{[1]}]$ \tikzmark{dsucc9} &&& \tikzmark{csucc28} $[q]^{[0]}]$ \tikzmark{dsucc25} &&&& \\
&&&&\tikzmark{dsucc7} $[q^{[1]}]$ &&& \tikzmark{dsucc28} $[\mathit{yes}^{[0]}]$&&&& \\
$[q^{[0]}]$ \tikzmark{dsucc6}&&&& &&& &&&& \tikzmark{csucc211} $[\mathit{yes}^{[1]}]$ \\
&&$\vdots$&& &&&& &&& \tikzmark{dsucc211} $[q^{[1]}]$ \tikzmark{csucc213} \\
$[\mathit{yes}^{[0]}]$ \tikzmark{csucc10}&&&& &&& &&&& \tikzmark{csucc214} $[\mathit{yes}^{[1]}]$ \tikzmark{dsucc213} \\
&&&&\tikzmark{csucc11} $[\mathit{yes}^{[1]}]$ &&& &&&& \tikzmark{dsucc214} $[q^{[1]}]$ \tikzmark{csucc215} \\
&&&&\tikzmark{dsucc11} $[q^{[1]}]$ \tikzmark{csucc13} &&&&& $\vdots$ && \\
\tikzmark{csucc12} $[q^{[0]}]$ \tikzmark{dsucc10}&&&& &&& \tikzmark{csucc212} $[q^{[0]}]$ &&&& \\
\tikzmark{dsucc12} $[\mathit{no}^{[0]}]$&&&& &&& \tikzmark{dsucc212} $[\mathit{yes}^{[0]}]$ \tikzmark{cdouble1}&&&& \\
&&&&\tikzmark{csucc14} $[\mathit{yes}^{[1]}]$ \tikzmark{dsucc13} &&&& &&&\tikzmark{csucc210} $[\mathit{yes}^{[1]}]$\tikzmark{dsucc215} \\
&&&&\tikzmark{dsucc14} $[q^{[1]}]$ \tikzmark{csucc15} &&&& &&&\tikzmark{dsucc210} $[q^{[1]}]$ \tikzmark{csucc26} \\
&&&& $[\mathit{no}^{[1]}]$ \tikzmark{dsucc15} &&&& &&& \tikzmark{csucc29} $[\mathit{yes}^{[1]}]$ \tikzmark{dsucc26} \\
&&&&&&&&&&& \tikzmark{dsucc29} $[q^{[1]}]$ \tikzmark{cdouble3} \\
&&&&&&& \tikzmark{cdouble2} $[q^{[0]}]$ \tikzmark{ddouble1} &&&& \\
&&&&&&& \tikzmark{ddouble2} $[\mathit{no}^{[0]}]$ &&&& \\
&&&&&&&&&&& $[\mathit{no}^{[1]}]$ \tikzmark{ddouble3}
\end{tabular}
\begin{tikzpicture}[overlay, remember picture, yshift=.25\baselineskip]
    \draw [->] ({pic cs:dsucc1}) to ({pic cs:csucc1});
    \draw [->] ({pic cs:dsucc2}) [bend left] to ({pic cs:csucc2});
    \draw [->] ({pic cs:dsucc3}) [bend right] to ({pic cs:csucc3});
    \draw [->] ({pic cs:dsucc4}) [bend left] to ({pic cs:csucc4});
    \draw [->] ({pic cs:dsucc5}) [bend right] to ({pic cs:csucc5});
    \draw [->] ({pic cs:dsucc6}) [bend right] to ({pic cs:csucc6}); 
    \draw [->] ({pic cs:dsucc7}) [bend left] to ({pic cs:csucc7});
    \draw [->] ({pic cs:dsucc8}) [bend left] to ({pic cs:csucc8});
    \draw [->] ({pic cs:dsucc9}) [bend right] to ({pic cs:csucc9});
    \draw [->] ({pic cs:dsucc10}) [bend right] to ({pic cs:csucc10});
    \draw [->] ({pic cs:dsucc11}) [bend left] to ({pic cs:csucc11});
    \draw [->] ({pic cs:dsucc12}) [bend left] to ({pic cs:csucc12});
    \draw [->] ({pic cs:dsucc13}) [bend right] to ({pic cs:csucc13});
    \draw [->] ({pic cs:dsucc14}) [bend left] to ({pic cs:csucc14});
    \draw [->] ({pic cs:dsucc15}) [bend right] to ({pic cs:csucc15});
    \draw [->] ({pic cs:dsucc21}) to ({pic cs:csucc21});
    \draw [->] ({pic cs:dsucc22}) [bend left] to ({pic cs:csucc22});
    \draw [->] ({pic cs:dsucc23}) [bend right] to ({pic cs:csucc23});
    \draw [->] ({pic cs:dsucc24}) [bend left] to ({pic cs:csucc24});
    \draw [->] ({pic cs:dsucc25}) [bend right] to ({pic cs:csucc25});
    \draw [->] ({pic cs:dsucc26}) [bend right] to ({pic cs:csucc26}); 
    \draw [->] ({pic cs:dsucc27}) [bend left] to ({pic cs:csucc27});
    \draw [->] ({pic cs:dsucc28}) [bend left] to ({pic cs:csucc28});
    \draw [->] ({pic cs:dsucc29}) [bend left] to ({pic cs:csucc29});
    \draw [->] ({pic cs:dsucc210}) [bend left] to ({pic cs:csucc210});
    \draw [->] ({pic cs:dsucc211}) [bend left] to ({pic cs:csucc211});
    \draw [->] ({pic cs:dsucc212}) [bend left] to ({pic cs:csucc212});
    \draw [->] ({pic cs:dsucc213}) [bend right] to ({pic cs:csucc213});
    \draw [->] ({pic cs:dsucc214}) [bend left] to ({pic cs:csucc214});
    \draw [->] ({pic cs:dsucc215}) [bend right] to ({pic cs:csucc215});
    \draw [->] ({pic cs:ddouble1}) [bend right] to ({pic cs:cdouble1});
    \draw [->] ({pic cs:ddouble2}) [bend left] to ({pic cs:cdouble2}); 
    \draw [->] ({pic cs:ddouble3}) [bend right] to ({pic cs:cdouble3});
  \end{tikzpicture}
\caption{The `lazy' successor and doubling strategies $\mathit{succ}, \mathit{double} : \mathcal{N} \Rightarrow \mathcal{N}$.}
\label{FigSuccAndDoublingStrategies}
\end{center}
\end{figure}

On the other hand, consider the strategy $\mathit{CountNonZeros} : \mathcal{N} \Rightarrow \mathcal{N}$ that counts the number of non-zero inputs which a given dynamic strategy on $!\mathcal{N}$ provides before it does $\underline{0}$.
As described in Figure~\ref{FigMin}, it simply investigates if an input is $\underline{0}$ just by checking the first digit ($\mathit{yes}^{[0]}$ or $\mathit{no}^{[0]}$) and adds $\mathit{yes}^{[1]}$ to the output if the input is not $\underline{0}$ (i.e., if the first digit is $\mathit{yes}^{[0]}$), where each `round' or thread of $!\mathcal{N}^{[0]}$ corresponding to a position of $\mathcal{N}^{[0]}$ is distinguished by `tags' $[\_]_{\ell^n}$ for exponential.
Note that $\mathit{CountNonZeros}$ only needs to refer to at most the last three moves in each odd-length position (n.b., in this case they are the last three moves in the P-view) since the number $n$ of previous `rounds' is recorded by the `tag' $[\_]_{\ell^n}$.
The point here is that the number of `tags' for exponential that $\mathit{CountNonZeros}$ has to manipulate is \emph{unbounded} (though the manipulation is very simple), and therefore the strategy is not finitary; however, it is easy to show that the strategy is viable as follows.
First, its partial function representation can be given by the following \emph{infinitary} table (n.b., $n$ for $[\_]_{\ell^n}$ given below ranges over \emph{all} natural numbers):
\begin{align*}
&(\square, \square, [\hat{q}^{[1]}]) \mapsto [\hat{q}^{[0]}] \mid ([\hat{q}^{[1]}], [\hat{q}^{[0]}], [\mathit{yes}^{[0]}]) \mapsto [\mathit{yes}^{[1]}] \mid \\ &([q^{[1]}], [q^{[0]}]_{\ell^n}, [\mathit{yes}^{[0]}]_{\ell^n}) \mapsto [\mathit{yes}^{[1]}] \mid ([\mathit{yes}^{[0]}]_{\ell^n}, [\mathit{yes}^{[1]}], [q^{[1]}]) \mapsto [q^{[0]}]_{\ell^{n+1}} \mid \\ &([q^{[1]}], [q^{[0]}]_{\ell^{n+1}}, [\mathit{no}^{[0]}]_{\ell^{n+1}}) \mapsto [\mathit{no}^{[1]}]
\end{align*}
This `high-level' computational process is `implementable' by an instruction strategy $\mathcal{A}(\mathit{CNZ})^\circledS : \mathcal{G}(M_{\mathcal{N} \Rightarrow \mathcal{N}}) \& \mathcal{G}(M_{\mathcal{N} \Rightarrow \mathcal{N}}) \& \mathcal{G}(M_{\mathcal{N} \Rightarrow \mathcal{N}}) \Rightarrow \mathcal{G}(M_{\mathcal{N} \Rightarrow \mathcal{N}})$ which computes as in Figure~\ref{FigAlgMin}, where `tags' for $\mathcal{G}(M_{\mathcal{N} \Rightarrow \mathcal{N}}) \& \mathcal{G}(M_{\mathcal{N} \Rightarrow \mathcal{N}}) \& \mathcal{G}(M_{\mathcal{N} \Rightarrow \mathcal{N}}) \Rightarrow \mathcal{G}(M_{\mathcal{N} \Rightarrow \mathcal{N}})$ are omitted (that is, we do not write $\mathcal{G}(M_{\mathcal{N} \Rightarrow \mathcal{N}})^{[0]} \& \mathcal{G}(M_{\mathcal{N} \Rightarrow \mathcal{N}})^{[1]} \& \mathcal{G}(M_{\mathcal{N} \Rightarrow \mathcal{N}})^{[2]} \Rightarrow \mathcal{G}(M_{\mathcal{N} \Rightarrow \mathcal{N}})^{[3]}$) for brevity.
Note also the font difference between moves in the figure, which is again just for clarity.
Then clearly, there is a \emph{finite} table for $\mathcal{A}(\mathit{CNZ})^\circledS$ that maps the last $k$-moves in the P-view of each odd-length position to the next P-move for some $k \in \mathbb{N}$ (though it is too tedious to write down the table here), proving the viability of $\mathit{CountNonZeros}$.
Observe in particular how the infinitary manipulations of `tags' by $\mathit{CountNonZeros}$ is reduced to finitary ones by $\mathcal{A}(\mathit{CNZ})^\circledS$. 
\begin{figure}
\begin{center}
\begin{tabular}{ccccc}
$\mathcal{N}^{[0]}$ && $\Rightarrow$ && $\mathcal{N}^{[1]}$ \\ \hline
&&&&\tikzmark{cCNZ1} $[\hat{q}^{[1]}]$ \tikzmark{cCNZ3} \\
\tikzmark{cCNZ2} $[\hat{q}^{[0]}]$ \tikzmark{dCNZ1}&&&& \\
\tikzmark{dCNZ2} $[\mathit{yes}^{[0]}]$ \tikzmark{cCNZ5}&&&&  \\
&&&&\tikzmark{cCNZ4} $[\mathit{yes}^{[1]}]$ \tikzmark{dCNZ3} \\
&&&&\tikzmark{dCNZ4} $[q^{[1]}]$ \tikzmark{cCNZ14} \\
\tikzmark{cCNZ6} $[q^{[0]}]_\ell$ \tikzmark{dCNZ5}&&&& \\
\tikzmark{dCNZ6} $[\mathit{yes}^{[0]}]_\ell$ \tikzmark{cCNZ12} &&&& \\
&&&& \tikzmark{cCNZ10} $[\mathit{yes}^{[1]}]$ \tikzmark{dCNZ14} \\
&&&&\tikzmark{dCNZ10} $[q^{[1]}]$ \tikzmark{cCNZ15} \\
\tikzmark{cCNZ7} $[q^{[0]}]_{\ell \ell}$ \tikzmark{dCNZ12}&&&& \\
\tikzmark{dCNZ7} $[\mathit{yes}^{[0]}]_{\ell \ell}$ &&&& \\
&&&& $[\mathit{yes}^{[1]}]$ \tikzmark{dCNZ15} \\
&&$\vdots$&& \\
&&&&$[q^{[1]}]$ \tikzmark{cCNZ16} \\
\tikzmark{cCNZ8} $[q^{[0]}]_{\ell^n}$&&&& \\
\tikzmark{dCNZ8} $[\mathit{yes}^{[0]}]_{\ell^n}$ \tikzmark{cCNZ13} &&&& \\
&&&& \tikzmark{cCNZ11} $[\mathit{yes}^{[1]}]$ \tikzmark{dCNZ16} \\
&&&&\tikzmark{dCNZ11} $[q^{[1]}]$ \tikzmark{cCNZ17} \\
\tikzmark{cCNZ9} $[q^{[0]}]_{\ell^{n+1}}$ \tikzmark{dCNZ13}&&&& \\
\tikzmark{dCNZ9} $[\mathit{no}^{[0]}]_{\ell^{n+1}}$ &&&& \\
&&&& $[\mathit{no}^{[1]}]$ \tikzmark{dCNZ17}
\end{tabular}
\begin{tikzpicture}[overlay, remember picture, yshift=.25\baselineskip]
 \draw [->] ({pic cs:dCNZ1}) to ({pic cs:cCNZ1});
 \draw [->] ({pic cs:dCNZ2}) [bend left] to ({pic cs:cCNZ2});
 \draw [->] ({pic cs:dCNZ3}) [bend right] to ({pic cs:cCNZ3});
 \draw [->] ({pic cs:dCNZ4}) [bend left] to ({pic cs:cCNZ4});
 \draw [->] ({pic cs:dCNZ5}) [bend right] to ({pic cs:cCNZ5});
 \draw [->] ({pic cs:dCNZ6}) [bend left] to ({pic cs:cCNZ6});
 \draw [->] ({pic cs:dCNZ7}) [bend left] to ({pic cs:cCNZ7});
 \draw [->] ({pic cs:dCNZ8}) [bend left] to ({pic cs:cCNZ8});
 \draw [->] ({pic cs:dCNZ9}) [bend left] to ({pic cs:cCNZ9});
 \draw [->] ({pic cs:dCNZ10}) [bend left] to ({pic cs:cCNZ10});
 \draw [->] ({pic cs:dCNZ11}) [bend left] to ({pic cs:cCNZ11});
 \draw [->] ({pic cs:dCNZ12}) [bend right] to ({pic cs:cCNZ12});
 \draw [->] ({pic cs:dCNZ13}) [bend right] to ({pic cs:cCNZ13});
 \draw [->] ({pic cs:dCNZ14}) [bend right] to ({pic cs:cCNZ14});
 \draw [->] ({pic cs:dCNZ15}) [bend right] to ({pic cs:cCNZ15});
 \draw [->] ({pic cs:dCNZ16}) [bend right] to ({pic cs:cCNZ16});
 \draw [->] ({pic cs:dCNZ17}) [bend right] to ({pic cs:cCNZ17});
\end{tikzpicture}
\caption{The non-zero-counting strategy $\mathit{CountNonZeros} : \mathcal{N} \Rightarrow \mathcal{N}$.}
\label{FigMin}
\end{center}
\end{figure}
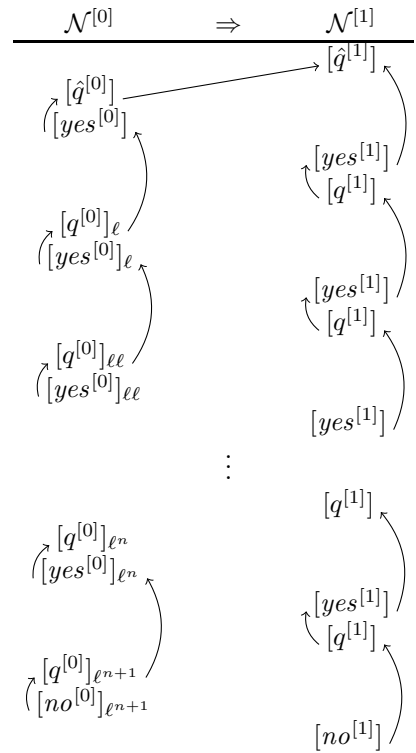

\begin{figure}[p]
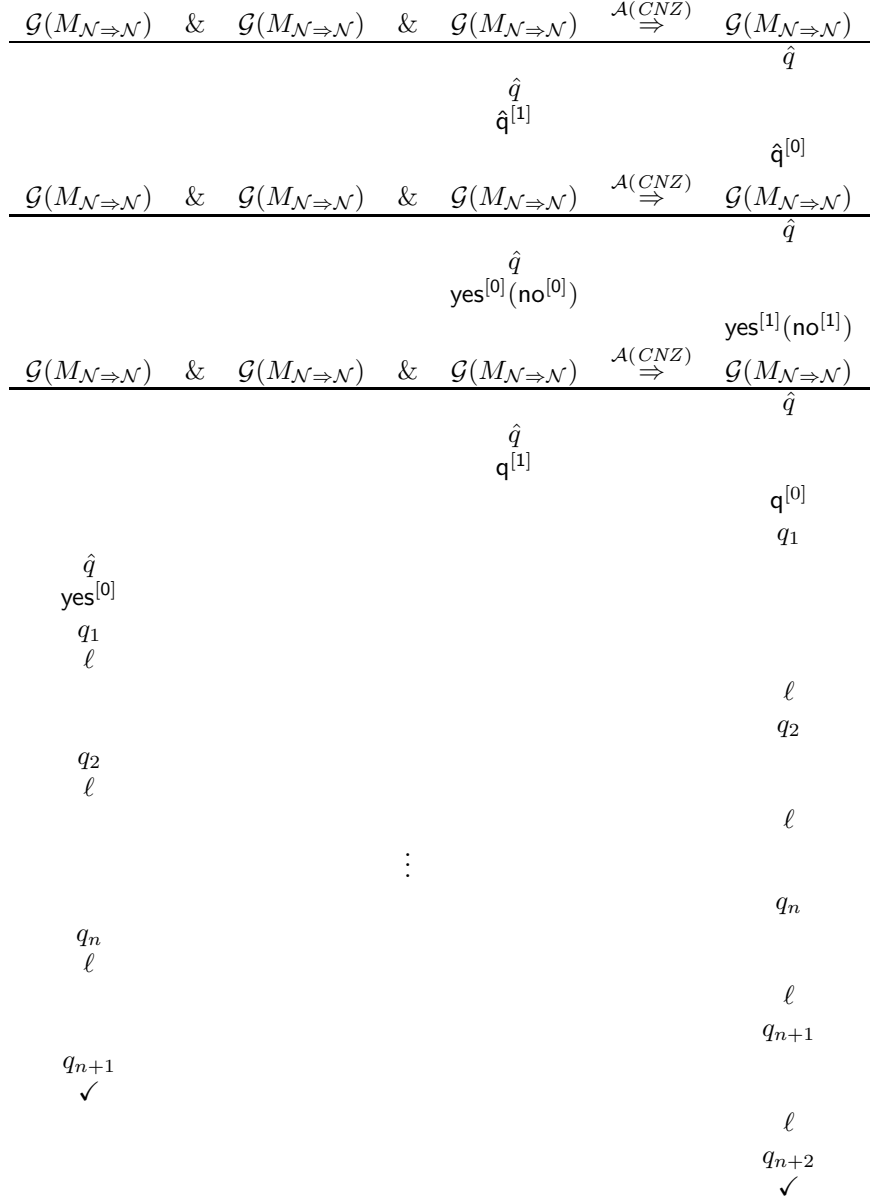

\begin{center}
\begin{tabular}{ccccccc}
$\mathcal{G}(M_{\mathcal{N} \Rightarrow \mathcal{N}})$ & $\&$ & $\mathcal{G}(M_{\mathcal{N} \Rightarrow \mathcal{N}})$ & $\&$ & $\mathcal{G}(M_{\mathcal{N} \Rightarrow \mathcal{N}})$ & $\stackrel{\mathcal{A}(\mathit{CNZ})}{\Rightarrow}$ & $\mathcal{G}(M_{\mathcal{N} \Rightarrow \mathcal{N}})$ \\
\cline{1-7}
&&&&&&$\hat{q}$ \\
&&&&$\hat{q}$&& \\
&&&&$\mathsf{\hat{q}^{[1]}}$&& \\
&&&&&& $\mathsf{\hat{q}^{[0]}}$ \\
$\mathcal{G}(M_{\mathcal{N} \Rightarrow \mathcal{N}})$ & $\&$ & $\mathcal{G}(M_{\mathcal{N} \Rightarrow \mathcal{N}})$ & $\&$ & $\mathcal{G}(M_{\mathcal{N} \Rightarrow \mathcal{N}})$ & $\stackrel{\mathcal{A}(\mathit{CNZ})}{\Rightarrow}$ & $\mathcal{G}(M_{\mathcal{N} \Rightarrow \mathcal{N}})$ \\
\cline{1-7}
&&&&&&$\hat{q}$ \\
&&&&$\hat{q}$&& \\
&&&&$\mathsf{yes^{[0]}} (\mathsf{no^{[0]}})$&& \\
&&&&&& $\mathsf{yes^{[1]}}(\mathsf{no^{[1]}})$ \\
$\mathcal{G}(M_{\mathcal{N} \Rightarrow \mathcal{N}})$ & $\&$ & $\mathcal{G}(M_{\mathcal{N} \Rightarrow \mathcal{N}})$ & $\&$ & $\mathcal{G}(M_{\mathcal{N} \Rightarrow \mathcal{N}})$ & $\stackrel{\mathcal{A}(\mathit{CNZ})}{\Rightarrow}$ & $\mathcal{G}(M_{\mathcal{N} \Rightarrow \mathcal{N}})$ \\
\cline{1-7}
&&&&&&$\hat{q}$ \\
&&&&$\hat{q}$&& \\
&&&&$\mathsf{q^{[1]}}$&& \\
&&&&&&$\mathsf{q}^{[0]}$ \\
&&&&&&$q_1$ \\
$\hat{q}$&&&&&& \\
$\mathsf{yes^{[0]}}$&&&&&& \\
$q_1$&&&&&& \\
$\mathsf{\ell}$&&&&&& \\
&&&&&& $\mathsf{\ell}$ \\
&&&&&&$q_{2}$ \\
$q_2$&&&&&& \\
$\mathsf{\ell}$&&&&&& \\
&&&&&&$\mathsf{\ell}$ \\
&&&$\vdots$&&& \\
&&&&&& $q_n$ \\
$q_n$&&&&&& \\
$\mathsf{\ell}$&&&&&& \\
&&&&&&$\mathsf{\ell}$ \\
&&&&&&$q_{n+1}$ \\
$q_{n+1}$&&&&&& \\
$\checkmark$&&&&&& \\
&&&&&&$\mathsf{\ell}$ \\
&&&&&&$q_{n+2}$ \\
&&&&&&$\checkmark$ 
\end{tabular}
\caption{An instruction strategy $\mathcal{A}(\mathit{CNZ})^\circledS$ for $\mathit{CountNonZeros} : \mathcal{N} \Rightarrow \mathcal{N}$.}
\label{FigAlgMin}
\end{center}
\end{figure}

This example illustrates why we need \emph{viable} (not only finitary) dynamic strategies for Turing completeness, where note that $\mathit{CountNonZero}$ can be seen as a simplification of minimization.
Also, it should be now clear why we employ composition of strategies \emph{without} hiding: An instruction strategy for the composition $\psi \circ \phi^\dagger : A \Rightarrow C$ of strategies $\phi : A \Rightarrow B$ and $\psi : B \Rightarrow C$ can be obtained simply as the disjoint union of instruction strategies for $\phi^\dagger$ and $\psi$, but it is not possible for composition \emph{with} hiding (in fact, there is no obvious way to construct an instruction strategy for the composition of $\phi$ and $\psi$ \emph{with} hiding). 

We advocate that viability of strategies gives a reasonable notion of `effective computability' as finitary strategies are clearly `effective', and so their `descriptions' or instruction strategies can be `effectively read off' by P.
Note also that viability is defined solely in terms of games and strategies without any axiom or induction.
Moreover, viability is at least as strong as Turing computability: As the main result of the present paper, we show that dynamic strategies definable in PCF are all viable (Theorem~\ref{ThmMainTheorem}), and therefore they are Turing complete (Corollary~\ref{CoroTuringCompleteness}).

Also, viable dynamic strategies solve the problem defined in Section~\ref{OurResearchProblem} in the following sense.
First, since dynamic games and strategies are abstract and syntax-independent concepts, they give a formulation of `high-level' computational processes, e.g., the lazy natural number game $\mathcal{N}$ defines natural numbers (not their symbolic representation) as `counting processes' in an abstract and syntax-independent fashion.
Moreover, an instruction strategy for a viable strategy describes a `low-level' computational process that `implements' the strategy.
In this manner, we have obtained a single mathematical framework for both `high-level' and `low-level' computational processes as well as `effective computability' of the former in terms of the latter.

\subsection{Our contribution and related work}
\label{RelatedWorkAndOurContribution}
Our main technical achievement is to define an \emph{intrinsic}, \emph{non-inductive} and \emph{non-axiomatic} notion of `effective computability' of strategies in game semantics, namely \emph{viable dynamic strategies}, and show that they are \emph{Turing complete} (Corollary~\ref{CoroTuringCompleteness}).
We have also shown the converse (though it is not that surprising): The input-output behavior of each viable dynamic strategy that computes on natural numbers coincides with a partial recursive function (Theorem~\ref{ThmConservativeness}).
This immediately implies a \emph{universality} result \cite{plotkin1977lcf,abramsky2000full} as well: Every viable dynamic strategy on a dynamic game interpreting a type of PCF is the denotation of a term of PCF (Corollary~\ref{CoroUniversality}).
In addition, as immediate corollaries, some of the well-known theorems in computability theory \cite{cutland1980computability,rogers1967theory} such as the \emph{smn-theorem} and the \emph{first recursion theorem} are generalized to non-classical computation (Corollaries~\ref{CoroGeneralizedSMN} and \ref{CoroGeneralizedFRT}).
We hope that these technical results would convince the reader that viability of dynamic strategies is a natural and reasonable generalization of Church-Turing computability. 

Another, more conceptual contribution of the paper is to establish a single mathematical framework for both `high-level' and `low-level' computational processes, where the former defines \emph{what} computation does, while the latter describes \emph{how} to execute the former.
In comparison with existing mathematical models of computation, our game-semantic approach has some novel features.
First, in comparison with computation by TMs or programming languages, plays of games are a more abstract concept; in particular they are not necessarily symbol manipulations, which is why they are suitable for abstract and `high-level' computational processes. 
Next, computation in a game proceeds as an interaction between P and O (not an encoding or a realizer of O), which may be seen as a generalization of computation by TMs in which just one interaction occurs (i.e., O gives an input on the infinite tape, and then P returns an output on the tape); this in particular means that O's computation does not have to be recursive, and it is a part of our formalization, which is why game semantics in general captures higher-order computation in a natural and systematic manner.
The present work inherits this interactive nature of game semantics.
Last but not least, games are a semantic counterpart of \emph{types}, where note that types do not a priori exist in TMs, and types in programming languages are syntactic entities.
Hence, our approach provides a deeper clarification of types in the context of theory of computation. 

Moreover, by exploiting the flexibility of game semantics, our approach would be applicable to a wide range of computation though it is left as future work.
Also, game semantics has interpreted various logics as well \cite{abramsky1994games,hyland1997game,abramsky2015games,yamada2016game}, and so it would be possible to employ our framework for a \emph{realizability interpretation} of constructive logic \cite{troelstra1998realizability,van2008realizability}, for which viable dynamic strategies would be more suitable as \emph{realizers} than existing strategies such as \cite{blot2017realizability} since the former contains more `computational contents' and makes more sense as a model of computation than the latter. 
Furthermore, the game models \cite{abramsky2015games,yamada2016game} have interpreted \emph{Martin-L\"{o}f type theory}, one of the most prominent foundations of constructive mathematics, and thus our framework would provide a mathematical and syntax-independent formalization of constructive mathematics too\footnote{This would be seen, if achieved, as a mathematical formalization of Brouwer's \emph{intuitionism} \cite{troelstra2014constructivism}, in which `(mental) constructions' are made precise as game-semantic computational processes, viz., plays by viable dynamic strategies, and moreover extended to \emph{interactive} `constructions'.}.
Of course, we need to work out details for these developments, which is out of the scope of the present paper, but it is in principle clear how to apply our framework to existing game semantics. 
In this sense, the present work would serve as a stepping-stone towards these extensions. 

In the literature, there have been several attempts to provide a mathematical foundation of computation beyond classical or symbolic ones.
We do not claim at all our game-semantic approach is best or canonical in comparison with the previous work; however, our approach certainly has some advantages.
For instance, Robin Gandy proposed in the famous paper \cite{gandy1980church} a notion of `mechanical devices', now known as \emph{Gandy machines (GMs)}, which appear more general than TMs, but showed that TMs are actually as powerful as GMs.
However, since GMs are an \emph{axiomatic} approach to define a general class of `mechanical devices' that are `effectively executable', they do not give a distinction between `high-level' and `low-level' computational processes, where GMs formulate the latter.
The more recent \emph{abstract state machines (ASMs)} \cite{gurevich2004abstract} developed by Yuri Gurevich employ an idea similar to GMs for `effective computability', namely by requiring an upper bound of elements that may change in a single step of computation, utilizing \emph{structures} in the sense of mathematical logic \cite{shoenfield1967mathematical}. 
Notably, ASMs define a very general notion of computation, namely computation as \emph{structure transition}; however, it seems that this framework is in some sense too general. 
For instance, it is possible that an ASM computes a real number in a single step, but then its `effective computability' is questionable; in general, an appropriate notion of `effective computability' of ASMs has been missing. 
Also, the way of computing a function by an ASM is to update input/output-pairs of the function in the \emph{element-wise} fashion, but it does not seem to be a common or natural processes in practice.
Yiannis Moschiovakis has also considered a mathematical foundation of algorithms \cite{moschovakis1998founding} in which, similarly to us, he proposed that algorithms and their `implementations' should be distinguished, where by algorithms he refers to what we call `high-level' computational processes. 
However, his framework, called \emph{recursors}, is also based on structures, and his notion of algorithms are relative to atomic operations given in each structure; thus, it does not give a foundational analysis on the notion of `effective computability'.
Therefore, although the previous work captures broader notions of computation, our approach has the advantage of achieving both of the distinction between `high-level' and `low-level' computational processes and the primitive notion of `effective computability'.
Also, the \emph{interactive} and \emph{typed} natures of game semantics stands in sharp contrast to the previous work as well.

At this point, we need to mention \emph{computability logic} \cite{japaridze2003introduction} developed by Giorgi Japaridze since his idea is similar to ours; he defines `effective computability' via computing machines playing in games.
Nevertheless, there are notable differences between computability logic and the present work.
First, computing machines in computability logic are a variant of TMs, and thus they are less novel as a model of computation than our approach; in fact, the definition of `effective computability' in computability logic can be seen more or less as a consequence of just spelling out the standard notion of recursive strategies \cite{abramsky2000full,hyland2000full,feree2017game}.  
Next, our framework inherits the categorical structure of game semantics, providing a \emph{compositional} formulation of logic and computation, i.e., a compound proof or program is constructed from its components, while there has been no known categorical structure of computability logic. 
Nevertheless, it would be interesting to adopt his `TMs-based' approach in our framework and compare its computational power with that of the present work.

Finally, let us mention some of the precursors of game semantics.
To clarify the notion of higher-order computability, Stephen Cole Kleene considered a model of higher-order computation based on `dialogues' between `computational oracles' in a series of papers \cite{kleene1978recursive,kleene1980recursive,kleene1982recursive}, which can be seen as the first attempt to define a mathematical notion of algorithms in a higher-order setting \cite{longley2015higher}.
Moreover, Gandy and his student Giovanni Pani refined these works by Kleene to obtain a model of PCF that satisfies \emph{universality} \cite{curien2007definability} though this work was not published. These previous works are direct ancestors of game semantics (in particular the so-called \emph{HO-games} \cite{hyland2000full} by Martin Hyland and Luke Ong).
As another line of research (motivated by the problem of full abstraction for PCF \cite{plotkin1977lcf}), Pierre-Louis Curien and Gerard Berry conceive of \emph{sequential algorithms} \cite{berry1982sequential} which was the first attempt to go beyond (extensional) functions to capture sequentiality of PCF. 
Sequential algorithms preceded and became highly influential to the development of game semantics; in fact, sequential algorithms are presented in the style of game semantics in \cite{longley2015higher}, and it is shown in \cite{bucciarelli1994another} that the oracle computation by Kleene can be represented by sequential algorithms (but the converse does not hold).
Nevertheless, neither of the previous work defines `effective computability' in a manner similar to the present work; our definition has an advantage in its intrinsic, non-inductive and non-axiomatic nature.

\subsection{The structure of the paper}
The rest of the paper proceeds roughly as follows. This introduction ends with fixing some notation. Then, recalling dynamic games and strategies in Section~\ref{Preliminaries}, we define viability of strategies and establish, as a main theorem, the fact that viable dynamic strategies may interpret all terms of PCF in Section~\ref{ViableStrategies}, proving their Turing completeness as a corollary.
Finally, we draw a conclusion and propose some future work in Section~\ref{ConclusionAndFutureWork}.

\begin{notation}
We use the following notation throughout the paper:
\begin{itemize}


\item We use bold letters $\boldsymbol{s}, \boldsymbol{t}, \boldsymbol{u}, \boldsymbol{v}$, etc. for sequences, in particular $\boldsymbol{\epsilon}$ for the \emph{empty sequence}, and letters $a, b, c, d, m, n, x, y, z$, etc. for elements of sequences;

\item We often abbreviate a finite sequence $\boldsymbol{s} = (x_1, x_2, \dots, x_n)$ as $x_1 x_2 \dots x_n$, and write $s_i$, where $i \in \{ 1, 2, \dots, n \}$, as another notation for $x_i$;

\item A \emph{concatenation} of sequences is represented by the juxtaposition of them, but we often write $a \boldsymbol{s}$, $\boldsymbol{t} b$, $\boldsymbol{u} c \boldsymbol{v}$ for $(a) \boldsymbol{s}$, $\boldsymbol{t} (b)$, $\boldsymbol{u} (c) \boldsymbol{v}$, etc., and also write $\boldsymbol{s} . \boldsymbol{t}$ for $\boldsymbol{s t}$;

\item We define $\boldsymbol{s}^n \stackrel{\mathrm{df. }}{=} \underbrace{\boldsymbol{s} \boldsymbol{s} \cdots \boldsymbol{s}}_n$ for any sequence $\boldsymbol{s}$ and natural number $n \in \mathbb{N}$;

\item We write $\mathsf{Even}(\boldsymbol{s})$ (resp. $\mathsf{Odd}(\boldsymbol{s})$) iff $\boldsymbol{s}$ is of even-length (resp. odd-length);

\item We define $S^\mathsf{P} \stackrel{\mathrm{df. }}{=} \{ \boldsymbol{s} \in S \mid \mathsf{P}(\boldsymbol{s}) \}$ for a set $S$ of sequences and $\mathsf{P} \in \{ \mathsf{Even}, \mathsf{Odd} \}$;

\item $\boldsymbol{s} \preceq \boldsymbol{t}$ means $\boldsymbol{s}$ is a \emph{prefix} of $\boldsymbol{t}$, and given a set $S$ of sequences, we define $\mathsf{Pref}(S) \stackrel{\mathrm{df. }}{=} \{ \boldsymbol{s} \mid \exists \boldsymbol{t} \in S . \ \! \boldsymbol{s} \preceq \boldsymbol{t} \ \! \}$;

\item For a poset $P$ and a subset $S \subseteq P$, $\mathsf{Sup}(S)$ denotes the \emph{supremum} of $S$;


\item $X^* \stackrel{\mathrm{df. }}{=} \{ x_1 x_2 \dots x_n \mid n \in \mathbb{N}, \forall i \in \{ 1, 2, \dots, n \} . \ \! x_i \in X \ \! \}$ for each set $X$;

\item For a function $f : A \to B$ and a subset $S \subseteq A$, we define $f \upharpoonright S : S \to B$ to be the \emph{restriction} of $f$ to $S$, and $f^\ast : A^\ast \to B^\ast$ by $f^\ast(a_1 a_2 \dots a_n) \stackrel{\mathrm{df. }}{=} f(a_1) f(a_2) \dots f(a_n) \in B^\ast$ for all $a_1 a_2 \dots a_n \in A^\ast$;


\item Given sets $X_1, X_2, \dots, X_n$ and $i \in \{ 1, 2, \dots, n \}$, we write $\pi_i$ (or $\pi_i^{(n)}$) for the \emph{$i^{\text{th}}$-projection function} $X_1 \times X_2 \times \dots \times X_n \to X_i$ that maps $(x_1, x_2, \dots, x_i, \dots, x_n) \mapsto x_i$ for all $x_j \in X_j$ ($j = 1, 2, \dots, i, \dots, n$).

\item $\simeq$ denote the \emph{Kleene equality}, i.e., $x \simeq y \stackrel{\mathrm{df. }}{\Leftrightarrow} (x \downarrow \wedge \ y \downarrow \wedge \ x = y) \vee (x \uparrow \wedge \ y \uparrow)$, where we write $x \downarrow$ if an element $x$ is defined and $x \uparrow$ otherwise.




\end{itemize}
\end{notation}

\section{Preliminary: games and strategies}
\label{Preliminaries}
Our games and strategies are essentially the `dynamic refinement' of McCusker's variant \cite{abramsky1999game,mccusker1998games}\footnote{Strictly speaking, the variants \cite{abramsky1999game,mccusker1998games} are slightly different, and  \cite{yamada2016dynamic} chooses \cite{abramsky1999game} as the basis, but \cite{mccusker1998games} describes a lot of useful technical details. \cite{abramsky1999game} is chosen in \cite{yamada2016dynamic} because it combines good points of the two best-known variants, \emph{AJM-games} \cite{abramsky2000full} and \emph{HO-games} \cite{hyland2000full} so that it may model \emph{linearity} as in \cite{abramsky1994games} and also characterize various programming features as in \cite{abramsky1999game}, though these points are left as future work in \cite{yamada2016dynamic}.}, which has been proposed under the name of \emph{dynamic games and strategies} by the present author and Abramsky in \cite{yamada2016dynamic} to capture `dynamics' (or \emph{rewriting}) and `intensionality' (or \emph{algorithms}) of computation by mathematical, particularly syntax-independent, concepts.
As already explained in the introduction, we have chosen this variant since, in contrast to conventional games and strategies, dynamic games and strategies capture \emph{step-by-step processes} in computation, which is essential for a `TMs-like' model of computation.

However, we need some modifications of dynamic games and strategies. 
First, although disjoint union of sets of moves (for constructions on games) is usually treated \emph{informally} for brevity, we need to adopt a particular formalization of `tags' for the disjoint union because we are concerned with `effectivity' of strategies, and so we must show that manipulations of `tags' are `effectively executable' by strategies.
In particular, we have to employ \emph{exponential} $!$ in which different `rounds' or \emph{threads} are distinguished by such `effective tags'. 

In addition, we slightly refine the original definition of dynamic games in \cite{yamada2016dynamic} by requiring that an intermediate occurrence of an O-move in a position of a dynamic game must be a mere copy of the last occurrence of a P-move, which reflects the example of composition \emph{without} hiding in the introduction.
This modification is due to our computability-theoretic motivation: Intermediate occurrences of moves are `invisible' to O (as in the example of composition without hiding), and therefore \emph{P has to `effectively' compute intermediate occurrences of O-moves} too (though this point does not matter in \cite{yamada2016dynamic}); note that it is clearly `effective' to just copy the last occurrence.
Also, it conceptually makes sense too: Intermediate occurrences of O-moves are just `dummies' of those of P-moves, and thus what happens in the intermediate part of each play is essentially P's calculation only. 
Technically, this is achieved by introducing \emph{dummy internal O-moves} (in Definition~\ref{DefArenas}), and strengthening the axiom DP2 (in Definition~\ref{DefGames}).
Let us remark, however, that this refinement is technically trivial, and it is not our main contribution. 

This section presents the resulting variant of games and strategies.
Fixing an implementation of `tags' in Section~\ref{OnTags}, we recall dynamic games and strategies in Sections~\ref{DynamicGames} and \ref{DynamicStrategies}, respectively.
To make this paper essentially self-contained, we shall explain motivations and intuitions behind the definitions. 

\subsection{On `tags' for disjoint union of sets}
\label{OnTags}
We begin with fixing `tags' for disjoint union that can be `effectively' manipulated.
We first define \emph{outer tags} or \emph{extended effective tags} (Definition~\ref{DefExtendedEffectiveTags}) for exponential, and then \emph{inner tags} (Definition~\ref{DefInnerTags}) for other constructions on games.

\begin{definition}[Effective tags]
An \emph{\bfseries effective tag} is a finite sequence over the two-element set $\Sigma = \{ \ell, \hbar \}$, where $\ell$ and $\hbar$ are arbitrarily fixed elements such that $\ell \neq \hbar$.
We often abbreviate the sequence $\ell^i = \ell \ell \dotsc \ell$ by $\underline{i}$  for each $i \in \mathbb{N}$.
\end{definition}

\begin{definition}[Decoding and encoding]
The \emph{\bfseries decoding function} $\mathit{de} : \Sigma^\ast \to \mathbb{N}^\ast$ and the \emph{\bfseries encoding function} $\mathit{en} : \mathbb{N}^\ast \to \Sigma^\ast$ are defined respectively by:
\begin{align*}
\mathit{de}(\boldsymbol{\gamma}) &\stackrel{\mathrm{df. }}{=}  (i_1, i_2, \dots, i_k) \\
\mathit{en}(j_1, j_2, \dots, j_l) &\stackrel{\mathrm{df. }}{=} \underline{j_1} \hbar \ \! \underline{j_2} \hbar \dots \underline{j_{l-1}} \hbar \ \! \underline{j_l}
\end{align*}
for all $\boldsymbol{\gamma} \in \Sigma^\ast$ and $(j_1, j_2, \dots, j_l) \in \mathbb{N}^\ast$, where $\boldsymbol{\gamma} = \underline{i_1} \hbar \ \!  \underline{i_2} \hbar \dots \underline{i_{k-1}} \hbar \ \! \underline{i_k}$.
\end{definition}

Clearly, the functions $\mathit{de} : \Sigma^\ast \leftrightarrows \mathbb{N}^\ast : \mathit{en}$ are mutually inverses (n.b., they both map the empty sequence $\boldsymbol{\epsilon}$ to itself).
In fact, each effective tag $\boldsymbol{\gamma} \in \Sigma^\ast$ is intended to be a binary representations of the finite sequence $\mathit{de}(\boldsymbol{\gamma}) \in \mathbb{N}^\ast$ of natural numbers.

However, effective tags are not sufficient for our purpose: For \emph{`nested' exponentials} (Definition~\ref{DefExponential}), we need to \emph{`effectively'} associate a natural number to each finite sequence of natural numbers in an \emph{`effectively'} invertible manner. 
Of course it is possible as there is a recursive bijection $\wp : \mathbb{N}^\ast \stackrel{\sim}{\to} \mathbb{N}$ whose inverse is recursive too, which is an elementary fact in computability theory \cite{cutland1980computability,rogers1967theory}, but we cannot rely on it for we are aiming at developing an \emph{autonomous} foundation of `computability'.

On the other hand, this bijection is necessary only for manipulating effective tags, and so we would like to avoid an involved mechanism to achieve it.
Then, our solution for this problem is to simply introduce elements to \emph{denote} the bijection:
\begin{definition}[Extended effective tags]
\label{DefExtendedEffectiveTags}
\emph{\bfseries Extended effective tags} are expressions $\boldsymbol{e} \in (\Sigma \cup \ \! \{ \Lbag, \Rbag \})^\ast$, where $\Lbag$ and $\Rbag$ are arbitrary elements with $\Lbag \neq \Rbag$ and $\Sigma \cap \{ \Lbag, \Rbag \} = \emptyset$, generated by the grammar $\boldsymbol{e} \stackrel{\mathrm{df. }}{\equiv} \boldsymbol{\gamma} \ \! | \! \ \boldsymbol{e}_1 \hbar \ \! \boldsymbol{e}_2 \ \! | \! \ \Lbag \boldsymbol{e} \Rbag$, where $\boldsymbol{\gamma}$ ranges over effective tags.
\end{definition}

\begin{notation}
Let $\mathcal{T}$ denote the set of all extended effective tags.
\end{notation}

\begin{definition}[Extended decoding]
\label{DefExtendedDecoding}
The \emph{\bfseries extended decoding function} $\mathit{ede} : \mathcal{T} \to \mathbb{N}^\ast$ is recursively defined by: 
\begin{align*}
\mathit{ede}(\boldsymbol{\gamma}) &\stackrel{\mathrm{df. }}{=} \mathit{de}(\boldsymbol{\gamma}) \\
\mathit{ede}(\boldsymbol{e}_1 \hbar \ \! \boldsymbol{e}_2) &\stackrel{\mathrm{df. }}{=} \mathit{ede}(\boldsymbol{e}_1) . \mathit{ede}(\boldsymbol{e}_2) \\
\mathit{ede}(\Lbag \boldsymbol{e} \Rbag) &\stackrel{\mathrm{df. }}{=} (\wp(\mathit{ede}(\boldsymbol{e})))
\end{align*}
where $\wp : \mathbb{N}^\ast \stackrel{\sim}{\to} \mathbb{N}$ is any recursive bijection fixed throughout the present paper such that $\wp (i_1, i_2, \dots, i_k) \neq \wp ( j_1, j_2, \dots, j_l)$ whenever $k \neq l$ (see, e.g.,  \cite{cutland1980computability}).
\end{definition}

Of course, we lose the bijectivity between $\Sigma^\ast$ and $\mathbb{N}^\ast$ for \emph{extended effective tags} (since if $\mathit{ede}(\Lbag \boldsymbol{e} \Rbag) = (i)$, then $\mathit{ede}(\underline{i}) = (i)$ but $\Lbag \boldsymbol{e} \Rbag \neq \underline{i}$), but in return, we may `effectively execute' the bijection $\wp : \mathbb{N}^\ast \stackrel{\sim}{\to} \mathbb{N}$ by just inserting $\Lbag$ and $\Rbag$.

We shall adopt extended effective tags to form countably infinite copies of moves for exponential; see Definition~\ref{DefExponential}.

\begin{convention}
From now on, \emph{\bfseries outer tags} refer to extended effective tags, and we write $\boldsymbol{e}, \boldsymbol{f}, \boldsymbol{g}, \boldsymbol{h}$, etc. for outer tags.
\end{convention}

On the other hand, as `tags' for other constructions on games, i.e., $(\_)^{[i]}$ in the introduction, we simply use a finite number of distinguished symbols; as we shall see, four symbols suffice for this purpose. 
Let us call them \emph{inner tags}:
\begin{definition}[Inner tags]
\label{DefInnerTags}
Let $\mathscr{W}$, $\mathscr{E}$, $\mathscr{N}$ and $\mathscr{S}$ be arbitrarily fixed and pairwise distinct elements, each of which is called an \emph{\bfseries inner tag}.
\end{definition}

We will focus on games whose moves are \emph{tagged elements} in the following sense:
\begin{definition}[Inner elements]
\label{DefInnerElements}
An \emph{\bfseries inner element} is a finitely nested pair $( \dots ((m, t_1), t_2), \dots, t_k)$, usually written $m_{t_1 t_2 \dots t_k}$, such that $m$ is a distinguished element, called the \emph{\bfseries substance} of $m_{t_1 t_2 \dots t_k}$, and $t_i$'s ($i = 1, 2, \dots, k$) are inner tags.
\end{definition}

\begin{definition}[Tagged elements]
\label{DefTaggedElements}
A \emph{\bfseries tagged element} is any pair $(m_{t_1 t_2 \dots t_k}, \boldsymbol{e})$, usually written $[m_{t_1 t_2 \dots t_k}]_{\boldsymbol{e}}$, of an inner element $m_{t_1 t_2 \dots t_k}$ and an outer tag $\boldsymbol{e} \in \mathcal{T}$.
\end{definition}

\begin{convention}
We often abbreviate tagged elements $[m_{t_1 t_2 \dots t_k}]_{\boldsymbol{e}}$ as $[m]_{\boldsymbol{e}}$ if the inner tags are not important, and inner elements $m_{t_1 t_2 \dots t_k}$ as $m$ for similar situations. 
\end{convention}

\subsection{Games}
\label{DynamicGames}
As stated in the introduction, our games are \emph{dynamic games} introduced in \cite{yamada2016dynamic}. 
The main idea of dynamic games is to introduce, in McCusker's games \cite{abramsky1999game,mccusker1998games}, a distinction between \emph{internal} and \emph{external} moves, where internal moves constitute `internal communication' between strategies (i.e., moves with square boxes in the introduction), and they are to be \emph{a posteriori} hidden by the \emph{hiding operation}, in order to capture `intensionality' and `dynamics' of computation by internal moves and the hiding operation, respectively.
Conceptually, internal moves are `invisible' to O as they represent how P `internally' calculates the next external P-move (i.e., step-by-step processes in computation).
In addition, unlike \cite{yamada2016dynamic} we restrict internal O-moves to `dummies' of internal P-moves via \emph{dummy functions} (in Definition~\ref{DefArenas}) for the computability-theoretic motivation mentioned above.

We first review the basic definitions of dynamic games in the present section; see \cite{yamada2016dynamic} for the details, and \cite{abramsky1999game,abramsky1997semantics,hyland1997game} for a general introduction to game semantics.

\begin{convention}
To distinguish our `dynamic concepts' from conventional ones \cite{abramsky1999game,mccusker1998games}, we add the word \emph{static} in front of the latter, e.g., static arenas, static games, etc.
\end{convention}

\subsubsection{Arenas and legal positions}
Similarly to McCusker's games, dynamic games are based on two preliminary concepts: \emph{arenas} and \emph{legal positions}. An arena defines the basic components of a game, which in turn induces legal positions that specify the basic rules of the game. 
Let us begin with recalling these two concepts. 

\begin{definition}[Dynamic arenas \cite{yamada2016dynamic}]
\label{DefArenas}
A \emph{\bfseries dynamic arena} is a quadruple $G = (M_G, \lambda_G, \vdash_G, \Delta_G)$, where:
\begin{itemize}
\item $M_G$ is a set of \emph{tagged elements}, called \emph{\bfseries moves}, such that: \textsc{(M)} The set $\pi_1(M_G)$ of \emph{inner elements} of $G$ is finite;

\item $\lambda_G$ is a function $M_G \to \{ \mathsf{O}, \mathsf{P} \} \times \{ \mathsf{Q}, \mathsf{A} \} \times \mathbb{N}$, called the \emph{\bfseries labeling function}, where $\mathsf{O}$, $\mathsf{P}$, $\mathsf{Q}$ and $\mathsf{A}$ are arbitrarily fixed symbols, called the \emph{\bfseries labels}, that satisfies: \textsc{(L)} $\mu(G) \stackrel{\mathrm{df. }}{=}  \mathsf{Sup}(\{ \lambda_G^{\mathbb{N}}(m) \mid m \in M_G \}) \in \mathbb{N}$;

\item $\vdash_G \ \subseteq (\{ \star \} \cup M_G) \times M_G$ is a relation, where $\star$ is an arbitrarily fixed symbol such that $\star \not \in M_G$, called the \emph{\bfseries enabling relation}, that satisfies:
\begin{itemize}

\item \textsc{(E1)} If $\star \vdash_G m$, then $\lambda_G(m) = (\mathsf{O}, \mathsf{Q}, 0)$ and $n = \star$ whenever $n \vdash_G m$;

\item \textsc{(E2)} If $m \vdash_G n$ and $\lambda_G^{\mathsf{QA}}(n) = \mathsf{A}$, then $\lambda_G^{\mathsf{QA}}(m) = \mathsf{Q}$ and $\lambda_G^{\mathbb{N}}(m) = \lambda_G^{\mathbb{N}}(n)$;

\item \textsc{(E3)} If $m \vdash_G n$ and $m \neq \star$, then $\lambda_G^{\mathsf{OP}}(m) \neq \lambda_G^{\mathsf{OP}}(n)$;

\item \textsc{(E4)} If $m \vdash_G n$, $m \neq \star$ and $\lambda_G^{\mathbb{N}}(m) \neq \lambda_G^{\mathbb{N}}(n)$, then $\lambda_G^{\mathsf{OP}}(m) = \mathsf{O}$;

\end{itemize}

\item $\Delta_G$ is a bijection $M_G^{\mathsf{PInt}} \stackrel{\sim}{\rightarrow} M_G^{\mathsf{OInt}}$, called the \emph{\bfseries dummy function},  that satisfies: \textsc{(D)} For any $[m_{\boldsymbol{t}}]_{\boldsymbol{e}} \in M_G^{\mathsf{PInt}}$ and $[n_{\boldsymbol{u}}]_{\boldsymbol{f}} \in M_G^{\mathsf{OInt}}$, if $\Delta_G([m_{\boldsymbol{t}}]_{\boldsymbol{e}}) = [n_{\boldsymbol{u}}]_{\boldsymbol{f}}$, then $m = n$, $\boldsymbol{e} = \boldsymbol{f}$, $\lambda_G^{\mathsf{QA}}([m_{\boldsymbol{t}}]_{\boldsymbol{e}}) = \lambda_G^{\mathsf{QA}}([n_{\boldsymbol{u}}]_{\boldsymbol{f}})$, $\lambda_G^{\mathbb{N}}([m_{\boldsymbol{t}}]_{\boldsymbol{e}}) = \lambda_G^{\mathbb{N}}([n_{\boldsymbol{u}}]_{\boldsymbol{f}})$, and $\boldsymbol{u} = \delta_G(\boldsymbol{t})$ for some \emph{finite} partial function $\delta_G$ on inner tags

\end{itemize}
in which $\lambda_G^{\mathsf{OP}} \stackrel{\mathrm{df. }}{=} \pi_1 \circ \lambda_G : M_G \to \{ \mathsf{O}, \mathsf{P} \}$, $\lambda_G^{\mathsf{QA}} \stackrel{\mathrm{df. }}{=} \pi_2 \circ \lambda_G : M_G \to \{ \mathsf{Q}, \mathsf{A} \}$, $\lambda_G^{\mathbb{N}} \stackrel{\mathrm{df. }}{=} \pi_3 \circ \lambda_G : M_G \to \mathbb{N}$, $M_G^{\mathsf{PInt}} \stackrel{\mathrm{df. }}{=} \langle \lambda_G^{\mathsf{OP}}, \lambda_G^{\mathbb{N}} \rangle^{-1}(\{ (\mathsf{P}, d) \mid d \geqslant 1 \ \! \})$ and $M_G^{\mathsf{OInt}} \stackrel{\mathrm{df. }}{=} \langle \lambda_G^{\mathsf{OP}}, \lambda_G^{\mathbb{N}} \rangle^{-1}(\{ (\mathsf{O}, d) \mid d \geqslant 1 \ \! \})$. 
A move $m \in M_G$ is \emph{\bfseries initial} if $\star \vdash_G m$, an \emph{\bfseries O-move} (resp. a \emph{\bfseries P-move}) if $\lambda_G^{\mathsf{OP}}(m) = \mathsf{O}$ (resp. if $\lambda_G^{\mathsf{OP}}(m) = \mathsf{P}$), a \emph{\bfseries question} (resp. an \emph{\bfseries answer}) if $\lambda_G^{\mathsf{QA}}(m) = \mathsf{Q}$ (resp. if $\lambda_G^{\mathsf{QA}}(m) = \mathsf{A}$), and \emph{\bfseries internal} or \emph{\bfseries $\boldsymbol{\lambda_G^{\mathbb{N}}(m)}$-internal}  (resp. \emph{\bfseries external}) if $\lambda_G^{\mathbb{N}}(m) > 0$ (resp. if $\lambda_G^{\mathbb{N}}(m) = 0$).
A finite sequence $\boldsymbol{s} \in M_G^\ast$ of moves is \emph{\bfseries $\boldsymbol{d}$-complete} if it ends with a move $m$ such that $\lambda_G^{\mathbb{N}}(m) = 0 \vee \lambda_G^{\mathbb{N}}(m) > d$, where $d \in \mathbb{N} \cup \{ \omega \}$, and $\omega$ is the \emph{least transfinite ordinal}.
For each $m \in M_G^{\mathsf{PInt}}$, $\Delta_G(m) \in M_G^{\mathsf{OInt}}$ is called the \emph{\bfseries dummy} of $m$.
We write $M_G^{\mathsf{Init}}$ (resp. $M_G^{\mathsf{Int}}$, $M_G^{\mathsf{Ext}}$) for the set of all initial (resp. internal, external) moves of $G$.
\end{definition}

\begin{remark}
Dynamic arenas in the sense defined above are dynamic arenas in the sense defined in \cite{yamada2016dynamic} equipped with the additional structure of \emph{dummy functions}.
\end{remark}

A dynamic arena is a \emph{static arena} defined in \cite{abramsky1999game}, equipped with a \emph{labeling} $\lambda_G^{\mathbb{N}}$ on moves and \emph{dummies} of internal P-moves, satisfying additional axioms about them.
From the opposite angle, dynamic arenas are a generalization of static arenas: A static arena is equivalent to a dynamic arena whose moves are all external.

Recall that a static arena $G$ determines possible \emph{moves} of a game, each of which is O's/P's question/answer, and specifies which move $n$ can be performed for each move $m$ by the relation $m \vdash_G n$ (and $\star \vdash_G m$ means that $m$ can initiate a play). Its axioms are E1, E2 and E3 (excluding conditions on priority orders):
\begin{itemize}

\item E1 sets the convention that an initial move must be O's question, and an initial move cannot be performed for a previous move;

\item E2 states that an answer must be performed for a question;

\item E3 mentions that an O-move must be performed for a P-move, and vice versa.

\end{itemize}

Then, as an additional structure for \emph{dynamic arenas} $G$, \cite{yamada2016dynamic} employs all natural numbers for $\lambda_G^{\mathbb{N}}$, not only the \emph{internal/external (I/E)}-parity, to define a \emph{step-by-step} execution of the \emph{hiding operation} $\mathcal{H}$: The operation $\mathcal{H}$ deletes all internal moves $m$ such that $\lambda_G^{\mathbb{N}}(m)$, called the \emph{\bfseries priority order}\footnote{It is rather called \emph{degree of internality} in \cite{yamada2016dynamic}. This change of name is motivated by the hindsight that it should be possible to define another order of the step-by-step execution of the hiding operation $\mathcal{H}$, where \cite{yamada2016dynamic} defines the priority order particularly according to the `depth of composition' or \emph{degree of internality} of moves.} (since it indicates the priority order of $m$ with respect to the execution of $\mathcal{H}$), is $1$.\footnote{Although the main focus of \cite{yamada2016dynamic} is to capture a small-step operational semantics of a programming language by such a step-by-step hiding operation $\mathcal{H}$, such fine-grained steps do not play a main role in this paper. Nevertheless, we keep the structure $\lambda_G^{\mathbb{N}}$ as it makes sense for a model of computation to be equipped with the step-by-step hiding operation, and it would be interesting to consider `effectivity' of the hiding operation as future work.}

In addition, unlike \cite{yamada2016dynamic} we have introduced the additional structure of \emph{dummy functions} for the computability-theoretic motivation mentioned at the beginning of Section~\ref{Preliminaries}.
The idea is that each internal O-move $m \in M_G^{\mathsf{OInt}}$ of a dynamic game $G$ must be a \emph{dummy} of a unique internal P-move $m' \in M_G^{\mathsf{PInt}}$, i.e., $m = \Delta_G(m')$, and $m$ may occur in a play only right after an occurrence of $m'$, which axiomatizes the phenomenon of intermediate occurrences of moves in the composition of $\mathit{succ}$ and $\mathit{double}$ \emph{without} hiding described in the introduction.   
We shall formalize the restriction on occurrences of internal O-moves by the axiom DP2 in Definition~\ref{DefGames}.

Note that the additional axioms for dynamic areas are intuitively natural:
\begin{itemize}

\item M requires the set $\pi_1 (M_G)$ to be finite so that each move is distinguishable, which is not required in \cite{yamada2016dynamic} but necessary to define `effective computability';

\item L requires the least upper bound $\mu(G)$ to be finite as it is conceptually natural and technically necessary for \emph{concatenation} $\ddagger$ of games (Definition~\ref{DefConcatenationOfGames});


\item E1 adds $\lambda_G^{\mathbb{N}}(m) = 0$ for all $m \in M_G^{\mathsf{Init}}$ as O cannot `see' internal moves, and thus he cannot initiate a play with an internal move;

\item E2 additionally requires the priority orders between a `QA-pair' to be the same since otherwise outputs of the hiding operation would not be well-defined;

\item E4 states that only P can perform a move for a previous move if they have different priority orders because internal moves are `invisible' to O (as we shall see, if $\lambda_G^{\mathbb{N}}(m_1) = k_1 < k_2 = \lambda_G^{\mathbb{N}}(m_2)$, then after the $k_1$-many iteration of the hiding operation, $m_1$ and $m_2$ become external and internal, respectively, i.e., the I/E-parity of moves is \emph{relative}, which is why E4 is not only concerned with I/E-parity but more fine-grained priority orders);

\item D requires that each internal P-move $p \in M_G^{\mathsf{PInt}}$ and its dummy $\Delta_G(p) \in M_G^{\mathsf{OInt}}$ may differ only in their inner tags since the latter is a `dummy' of the former (again, it reflects the informal example in the introduction); also, $\Delta_G(p)$ is `effectively' obtainable from $p$ by a \emph{finitary} calculation $\delta_G$ on inner tags.

\end{itemize}


\begin{convention}
From now on, \emph{\bfseries arenas} refer to \emph{dynamic} arenas by default.
\end{convention}

As explained previously, an interaction between P and O in a game is represented by a finite sequence of moves that satisfies certain axioms (under the name of \emph{(valid) positions}; see Definition~\ref{DefGames}).
However, strictly speaking, we need to equip such sequences with an additional structure, called \emph{justifiers} or \emph{pointers}, to distinguish similar but different computational processes (see, e.g., \cite{abramsky1999game} for this point):

\begin{notation}
The $i^\text{th}$-occurrence of an element $a$ in a sequence $\boldsymbol{s}$ is temporarily (see the convention below) written $a[i]$.
\end{notation}

\begin{definition}[J-sequences \cite{hyland2000full,abramsky1999game}]
\label{DefJSequences}
A \emph{\bfseries justified (j-) sequence} of an arena $G$ is a finite sequence $\boldsymbol{s} \in M_G^\ast$ in which each occurrence $n[i]$ of a non-initial move $n$ is associated with a unique occurrence $m[j]$ of a move $m$ in $\boldsymbol{s}$, called the \emph{\bfseries justifier} of $n[i]$ in $\boldsymbol{s}$, that occurs previously in $\boldsymbol{s}$ and satisfies $m \vdash_G n$.
In this case, we say that $n[i]$ is \emph{\bfseries justified} by $m[j]$, or equivalently there is a \emph{\bfseries pointer} from $n[i]$ to $m[j]$.
\end{definition}

\begin{convention}
In the rest of the paper, we simply say `an occurrence $m$ (of a move)' instead of `an occurrence $m[i]$ of a move $m$'.
This abuse of notation does not bring any serious confusion in practice.
Moreover, we call an occurrence of an initial (resp. non-initial) move an \emph{\bfseries initial occurrence} (resp. a \emph{\bfseries non-initial occurrence}).
\end{convention}

\begin{notation}
We write $\mathcal{J}_{\boldsymbol{s}}(n)$ for the justifier of a non-initial occurrence $n$ in a j-sequence $\boldsymbol{s}$, where $\mathcal{J}_{\boldsymbol{s}}$ may be thought of as the `function of pointers in $\boldsymbol{s}$', and $\mathscr{J}_G$ for the set of all j-sequences of an arena $G$.
We write $\boldsymbol{s} = \boldsymbol{t}$ for any $\boldsymbol{s}, \boldsymbol{t} \in \mathscr{J}_G$ if $\boldsymbol{s}$ and $\boldsymbol{t}$ are the same j-sequence of $G$, i.e., the same sequence with the same pointers.
\end{notation}

The idea is that each non-initial occurrence in a j-sequence must be performed for a specific previous occurrence, viz., its \emph{justifier}.  
Since the present paper is not concerned with a \emph{faithful} interpretation of programs, one may wonder if justifiers would play any  important role in the rest of the paper; however, they \emph{do} in a novel manner: They allow P to `effectively' collect, from the history of previous moves, a bounded number of necessary ones, as we shall see in Section~\ref{SubsectionViableStrategies}.

Note that the first element $m$ of each non-empty j-sequence $m \boldsymbol{s} \in \mathscr{J}_G$ must be initial; we particularly call $m$ the \emph{\bfseries opening occurrence} of $m \boldsymbol{s}$.
Clearly, an opening occurrence must be an initial occurrence, but not necessarily vice versa. 

Let us now consider justifiers, j-sequences and arenas from the `external viewpoint' (Definitions~\ref{DefExternalJustifiers}, \ref{DefHidingOperationOnJsequences} and \ref{DefExternalArenas} below):
\begin{definition}[J-subsequences]
\label{DefJSubsequences}
Given an arena $G$ and a j-sequence $\boldsymbol{s} \in \mathscr{J}_G$, a \emph{\bfseries j-subsequence} of $\boldsymbol{s}$ is a j-sequence $\boldsymbol{t} \in \mathscr{J}_G$ such that $\boldsymbol{t}$ is a subsequence of $\boldsymbol{s}$, and $\mathcal{J}_{\boldsymbol{t}}(n) = m$ iff there exist occurrences $m_1, m_2, \dots, m_k$ ($k \in \mathbb{N}$) in $\boldsymbol{s}$ eliminated in $\boldsymbol{t}$ such that $\mathcal{J}_{\boldsymbol{s}}(n) = m_1 \wedge \mathcal{J}_{\boldsymbol{s}}(m_1) = m_2 \wedge \dots \wedge \mathcal{J}_{\boldsymbol{s}}(m_{k-1}) = m_k \wedge \mathcal{J}_{\boldsymbol{s}}(m_{k}) = m$.
\end{definition}

\begin{definition}[External justifiers \cite{yamada2016dynamic}]
\label{DefExternalJustifiers}
Let $G$ be an arena, $\boldsymbol{s} \in \mathscr{J}_G$ and $d \in \mathbb{N} \cup \{ \omega \}$.
Each non-initial occurrence $n$ in $\boldsymbol{s}$ has a unique sequence of justifiers $m m_1 m_2 \dots m_k n$ $(k \geqslant 0)$, i.e., $\mathcal{J}_{\boldsymbol{s}}(n) = m_k$, $\mathcal{J}_{\boldsymbol{s}}(m_k) = m_{k-1}$, \dots, $\mathcal{J}_{\boldsymbol{s}}(m_{2}) = m_1$ and $\mathcal{J}_{\boldsymbol{s}}(m_1) = m$, such that $\lambda_G^{\mathbb{N}}(m) = 0 \vee \lambda_G^{\mathbb{N}}(m) > d$ and $0 < \lambda_G^{\mathbb{N}}(m_i) \leqslant d$ for $i = 1, 2, \dots, k$. $m$ is called the \emph{\bfseries $\boldsymbol{d}$-external justifier} of $n$ in $\boldsymbol{s}$, and written $\mathcal{J}_{\boldsymbol{s}}^{\circleddash d}(n)$.
\end{definition}

Note that $d$-external justifiers are a simple generalization of justifiers: $0$-external justifiers coincide with justifiers.
$d$-external justifiers are intended to be justifiers after the $d$-times iteration of the hiding operation $\mathcal{H}$, as we shall see shortly.

\begin{definition}[External j-subsequences \cite{yamada2016dynamic}]
\label{DefHidingOperationOnJsequences}
Given an arena $G$, $\boldsymbol{s} \in \mathscr{J}_G$ and $d \in \mathbb{N} \cup \{ \omega \}$, the \emph{\bfseries $\boldsymbol{d}$-external justified (j-) subsequence} $\mathcal{H}^d_G(\boldsymbol{s})$ of $\boldsymbol{s}$ is obtained from $\boldsymbol{s}$ by deleting occurrences of internal moves $m$ such that $0 < \lambda_G^{\mathbb{N}}(m) \leqslant d$ and equipping the resulting subsequence of $\boldsymbol{s}$ with pointers $\mathcal{J}_{\boldsymbol{s}}^{\circleddash d}$\footnote{Strictly speaking, $\mathcal{J}_{\mathcal{H}^d_G(\boldsymbol{s})}$ is the obvious \emph{restriction} of $\mathcal{J}_{\boldsymbol{s}}^{\circleddash d}$.}.
\end{definition}

\begin{definition}[External arenas \cite{yamada2016dynamic}]
\label{DefExternalArenas}
Let $G$ be an arena, and assume $d \in \mathbb{N} \cup \{ \omega \}$. 
The \emph{\bfseries $\boldsymbol{d}$-external arena} $\mathcal{H}^d(G)$ of $G$ is defined by:
\begin{itemize}

\item $M_{\mathcal{H}^d(G)} \stackrel{\mathrm{df. }}{=} \{ m \in M_G \mid \lambda_G^{\mathbb{N}}(m) = 0 \vee  \lambda_G^{\mathbb{N}}(m) > d \ \! \}$;

\item $\lambda_{\mathcal{H}^d(G)} \stackrel{\mathrm{df. }}{=} \lambda_G^{\circleddash d} \upharpoonright M_{\mathcal{H}^d(G)}$, where $\lambda_G^{\circleddash d} \stackrel{\mathrm{df. }}{=}  \langle \lambda_G^{\mathsf{OP}}, \lambda_G^{\mathsf{QA}}, m \mapsto \lambda_G^{\mathbb{N}} (m) \circleddash d \rangle$ and $n \circleddash d \stackrel{\mathrm{df. }}{=} \begin{cases} n - d &\text{if $n \geqslant d$} \\ 0 &\text{otherwise} \end{cases}$ for all $n \in \mathbb{N}$;

\item $m \vdash_{\mathcal{H}^d(G)} \! n \stackrel{\mathrm{df. }}{\Leftrightarrow} \exists k \in \mathbb{N}, m_1, m_2, \dots, m_{2k-1}, m_{2k} \in M_G \setminus M_{\mathcal{H}^d(G)} . \ \! m \vdash_G m_1 \wedge m_1 \vdash_G m_2 \wedge \dots \wedge m_{2k-1} \vdash_G m_{2k} \wedge m_{2k} \vdash_G n$ ($\Leftrightarrow m \vdash_G n$ if $k = 0$);

\item $\Delta_{\mathcal{H}^d(G)}  \stackrel{\mathrm{df. }}{=} \Delta_G \upharpoonright M_{\mathcal{H}^d(G)}$.

\end{itemize}
\end{definition}
Thus, $\mathcal{H}^d(G)$ is obtained from $G$ by deleting all internal moves $m$ such that $0 < \lambda_G^{\mathbb{N}}(m) \leqslant d$, decreasing by $d$ the priority orders of the remaining moves, and `concatenating' the enabling relation to form the `$d$-external' one. 

\begin{convention}
For each $d \in \mathbb{N} \cup \{ \omega \}$, we regard $\mathcal{H}^d$ as an operation on arenas $G$, called the \emph{\bfseries $\boldsymbol{d}$-hiding operation} on arenas, and $\mathcal{H}_G^d$ as an operation on j-sequences of $G$, called the \emph{\bfseries $\boldsymbol{d}$-hiding operation} on j-sequences.
\end{convention}

\begin{lemma}[Closure of arenas and j-sequences under hiding \cite{yamada2016dynamic}]
\label{LemClosureOfArenasAndJSequencesUnderHiding}
If $G$ is an arena, then, for all $d \in \mathbb{N} \cup \{ \omega \}$, so is $\mathcal{H}^d(G)$ and $\mathcal{H}_G^d(\boldsymbol{s}) \in \mathscr{J}_{\mathcal{H}^d(G)}$ for all $\boldsymbol{s} \in \mathscr{J}_G$. 
Also, $\underbrace{\mathcal{H}^1 \circ \mathcal{H}^1 \cdots \circ \mathcal{H}^1}_i(G) = \mathcal{H}^i (G)$ and $\mathcal{H}^1_{\mathcal{H}^{i-1}(G)} \circ \mathcal{H}^1_{\mathcal{H}^{i-2}(G)} \circ \cdots \circ \mathcal{H}^1_{\mathcal{H}^{1}(G)} \circ \mathcal{H}^1_{G}(\boldsymbol{s}) = \mathcal{H}_G^i(\boldsymbol{s})$ for any $i \in \mathbb{N}$ (it means $G = G$ and $\boldsymbol{s} = \boldsymbol{s}$ if $i = 0$), arena $G$ and $\boldsymbol{s} \in \mathscr{J}_G$.
\end{lemma}
\begin{proof}
We need to consider the additional structure of dummy functions; everything else has been proved in \cite{yamada2016dynamic}.
Let $G$ be an arena, and $d \in \mathbb{N} \cup \{ \omega \}$.
For each $p \in M_G^{\mathsf{PInt}}$, we clearly have $p \in M_{\mathcal{H}^d(G)} \Leftrightarrow \Delta_G(p) \in M_{\mathcal{H}^d(G)}$; thus, $\Delta_{\mathcal{H}^d(G)}$ is a well-defined bijection $M_{\mathcal{H}^d(G)}^{\mathsf{PInt}} \stackrel{\sim}{\to} M_{\mathcal{H}^d(G)}^{\mathsf{OInt}}$.
Finally, the axiom D on $\Delta_{\mathcal{H}^d(G)}$ clearly follows from that on $G$, completing the proof. 
\end{proof}

\begin{convention}
Thanks to Lemma~\ref{LemClosureOfArenasAndJSequencesUnderHiding}, henceforth we regard the $i$-hiding operations $\mathcal{H}^i$ and $\mathcal{H}^i_G$ as the $i$-times iteration of the $1$-hiding operations $\mathcal{H}^1$ and $\mathcal{H}^1_G$, respectively, for all $i \in \mathbb{N}$. 
For this reason, we write $\mathcal{H}$ and $\mathcal{H}_G$ for $\mathcal{H}^1$ and $\mathcal{H}^1_G$, respectively, and call them the \emph{\bfseries hiding operations} (on arenas and j-sequences, respectively). 
\end{convention}

Next, let us recall the notion of `relevant part' of previous moves, called \emph{views}:
\begin{definition}[Views \cite{hyland2000full,abramsky1999game,mccusker1998games}] 
\label{DefViews}
Given a j-sequence $\boldsymbol{s}$ of an arena $G$, the \emph{\bfseries Player (P-) view} $\lceil \boldsymbol{s} \rceil_G$ and the \emph{\bfseries Opponent (O-) view} $\lfloor \boldsymbol{s} \rfloor_G$ (we often omit the subscript $G$) are defined by induction on the length of $\boldsymbol{s}$ as follows: 
\begin{itemize}

\item$\lceil \boldsymbol{\epsilon} \rceil_G \stackrel{\mathrm{df. }}{=} \boldsymbol{\epsilon}$;

\item $\lceil \boldsymbol{s} m \rceil_G \stackrel{\mathrm{df. }}{=} \lceil \boldsymbol{s} \rceil_G . m$ if $m$ is a P-move;

\item $\lceil \boldsymbol{s} m \rceil_G \stackrel{\mathrm{df. }}{=} m$ if $m$ is initial;

\item $\lceil \boldsymbol{s} m \boldsymbol{t} n \rceil_G \stackrel{\mathrm{df. }}{=} \lceil \boldsymbol{s} \rceil_G . m n$ if $n$ is an O-move with $\mathcal{J}_{\boldsymbol{s} m \boldsymbol{t} n}(n) = m$;

\item $\lfloor \boldsymbol{\epsilon} \rfloor_G \stackrel{\mathrm{df. }}{=} \boldsymbol{\epsilon}$;

\item $\lfloor \boldsymbol{s} m \rfloor_G \stackrel{\mathrm{df. }}{=} \lfloor \boldsymbol{s} \rfloor_G . m$ if $m$ is an O-move;

\item $\lfloor \boldsymbol{s} m \boldsymbol{t} n \rfloor_G \stackrel{\mathrm{df. }}{=} \lfloor \boldsymbol{s} \rfloor_G . m n$ if $n$ is a P-move with $\mathcal{J}_{\boldsymbol{s} m \boldsymbol{t} n}(n) = m$

\end{itemize}
where the justifiers of the remaining occurrences in $\lceil \boldsymbol{s} \rceil_G$ (resp. $\lfloor \boldsymbol{s} \rfloor_G$) are unchanged if they occur in $\lceil \boldsymbol{s} \rceil_G$ (resp. $\lfloor \boldsymbol{s} \rfloor_G$) and undefined otherwise. 
A \emph{\bfseries view} is either a P-view or an O-view. 
\end{definition}

The idea behind this definition is as follows.
For a j-sequence $\boldsymbol{t}m$ of an arena $G$ such that $m$ is a P-move (resp. an O-move), the P-view $\lceil \boldsymbol{t} \rceil_G$ (resp. the O-view $\lfloor \boldsymbol{t} \rfloor_G$) is intended to be the currently `relevant part' of $\boldsymbol{t}$ for P (resp. O). That is, P (resp. O) is concerned only with the last O-move (resp. P-move), its justifier and that justifier's `concern', i.e., P-view (resp. O-view), which then recursively proceeds.
As explained in \cite{abramsky1999game}, strategies (Definition~\ref{DefStrategies}) that model computation without \emph{state} see only P-views, not entire histories of previous moves, as inputs; they are called \emph{innocent} strategies (Definition~\ref{DefInnocence}).
In this sense, innocence captures \emph{state-freeness} of strategies.
In this paper, however, P-views play a different yet fundamental role for our notion of `effective computability' in Section~\ref{SubsectionViableStrategies}.

We are now ready to define:
\begin{definition}[Dynamic legal positions \cite{yamada2016dynamic}]
A \emph{\bfseries dynamic legal position} of an arena $G$ is a j-sequence $\boldsymbol{s} \in \mathscr{J}_G$ that satisfies:
\begin{itemize}

\item \textsc{(Alternation)} If $\boldsymbol{s} = \boldsymbol{s}_1 m n \boldsymbol{s}_2$, then $\lambda_G^\mathsf{OP} (m) \neq \lambda_G^\mathsf{OP} (n)$;


\item \textsc{(Generalized Visibility)} If $\boldsymbol{s} = \boldsymbol{t} m \boldsymbol{u}$ with $m$ non-initial and $d \in \mathbb{N} \cup \{ \omega \}$ satisfy $\lambda_G^{\mathbb{N}}(m) = 0 \vee \lambda_G^{\mathbb{N}}(m) > d$, then $\mathcal{J}^{\circleddash d}_{\boldsymbol{s}}(m)$ occurs in $\lceil \mathcal{H}_G^d(\boldsymbol{t}) \rceil_{\mathcal{H}^d(G)}$ if $m$ is a P-move, and it occurs in $\lfloor \mathcal{H}_G^d(\boldsymbol{t}) \rfloor_{\mathcal{H}^d(G)}$ if $m$ is an O-move;

\item \textsc{(IE-switch)} If $\boldsymbol{s} = \boldsymbol{s}_1 m n \boldsymbol{s}_2$ with $\lambda_G^{\mathbb{N}}(m) \neq \lambda_G^{\mathbb{N}}(n)$, then $m$ is an O-move.

\end{itemize} 
\end{definition}

\begin{notation}
We write $\mathscr{L}_G$ for the set of all dynamic legal positions of an arena $G$.
\end{notation}

Recall that a \emph{static legal position} defined in \cite{abramsky1999game} is a j-sequence that satisfies alternation and \emph{visibility}, i.e., generalized visibility only for $d = 0$ \cite{hyland2000full,abramsky1999game,mccusker1998games}, which is technically to guarantee that the P-view and the O-view of a j-sequence are again j-sequences and conceptually to ensure that the justifier of each non-initial occurrence belongs to the `relevant part' of the history of previous moves.
Static legal positions are to specify the basic rules of a \emph{static game} in the sense that every \emph{(valid) position} of the game must be a static legal position \cite{abramsky1999game}:
\begin{itemize}

\item In a position of the game, O always performs the first move by a question, and then P and O alternately play (by alternation), in which every non-initial move is made for a specific previous move (by justification);


\item The justifier of each non-initial occurrence belongs to the `relevant part', i.e., the view, of the previous moves (by visibility) in a position. 
\end{itemize}

Similarly, dynamic legal positions are to specify the basic rules of a \emph{dynamic game} (Definition~\ref{DefGames}). They are static legal positions that satisfy additional axioms:
\begin{itemize}

\item Generalized visibility is a generalization of visibility; it requires that visibility holds after any iteration of the \emph{hiding operations} on arenas and j-sequences for Theorem~\ref{ThmClosureOfGamesUnderHiding} below;

\item IE-switch states that only P can change a priority order during a play because internal moves are `invisible' to O, where the same remark as in E4 is applied for its finer distinction of priority orders than the mere I/E-parity.

\end{itemize}

Note that a dynamic legal position in which no internal move occurs is equivalent to a static legal position.

\begin{convention}
From now on, \emph{\bfseries legal positions} refer to \emph{dynamic} legal positions by default.
\end{convention}

\subsubsection{Games}
We are now ready to recall \emph{dynamic games}:
\begin{definition}[Dynamic games \cite{yamada2016dynamic}]
\label{DefGames}
A \emph{\bfseries dynamic game} is a quintuple $G = (M_G, \lambda_G, \vdash_G, \Delta_G, P_G)$ such that the quadruple $(M_G, \lambda_G, \vdash_G, \Delta_G)$ is an arena, and $P_G$ is a subset of $\mathscr{L}_G$ whose elements are called \emph{\bfseries (valid) positions} of $G$ that satisfies:
\begin{itemize}

\item \textsc{(P1)} $P_G$ is non-empty and \emph{prefix-closed} (i.e., $\boldsymbol{s}m \in P_G \Rightarrow \boldsymbol{s} \in P_G$);


\item \textsc{(DP2)} If $\boldsymbol{s} = \boldsymbol{t} . p . o' . \boldsymbol{u} . p' . o$ (resp. $\boldsymbol{s} = \boldsymbol{t} . o' . \boldsymbol{u} . p' . o$), where $o$ is an internal O-move, and $o'$ is an internal (resp. external) O-move such that $o' = \mathcal{J}_{\boldsymbol{s}}(p')$, then $o = \Delta_G(p')$ and $\mathcal{J}_{\boldsymbol{s}}(o) = p$ (resp. $\mathcal{J}_{\boldsymbol{s}}(o) = p'$).



\end{itemize}
A \emph{\bfseries play} of $G$ is a (finitely or infinitely) increasing (with respect to $\preceq$) sequence $(\boldsymbol{\epsilon}, m_1, m_1m_2, \dots)$ of valid positions of $G$.
\end{definition}

\begin{remark}
In \cite{yamada2016dynamic}, each dynamic game $G$ is equipped with an equivalence relation $\simeq_G$ on its positions in order to ignore permutations of `tags' for exponential $!$ as in \cite{abramsky2000full} and Section 3.6 of \cite{mccusker1998games}.
Naturally, dynamic strategies $\sigma : G$ are identified up to $\simeq_G$, i.e., the equivalence class $[\sigma]$ of $\sigma$ with respect to $\simeq_G$ is a morphism in the category of dynamic games and strategies \cite{yamada2016dynamic}, which matches the syntactic equality on terms.
However, our notion of `effective computability' or \emph{viability} (Definition~\ref{DefViability}) is defined on dynamic strategies, not their equivalence classes, and our focus is not a (fully complete) interpretation of a programming language; thus, we do not have to take such equivalence classes at all.
Hence, for simplicity, we exclude such equivalence relations on positions from the structure of dynamic games in the present paper.

Of course, we may easily adopt the full definition of dynamic games $G$ (i.e., with $\simeq_G$) and equivalence classes $[\sigma]$ of dynamic strategies $\sigma : G$ as in \cite{yamada2016dynamic}: We may simply define $[\sigma]$ to be \emph{viable} if there is some viable representative $\tau \in [\sigma]$.

\end{remark}

Thus, dynamic games are \emph{static games} defined in \cite{abramsky1999game,mccusker1998games} except that their arenas are dynamic ones and they additionally satisfy the axiom DP2.
The axiom P1 talks about the natural phenomenon that each non-empty `moment' of a play must have the previous `moment'. 
In addition, by the axiom DP2 for dynamic games, internal O-moves must be performed as \emph{dummies} of the last internal P-moves, where the pointers specified by the axiom would make sense if one considers the example of composition of $\mathit{succ}$ and $\mathit{double}$ without hiding in the introduction. 
Conceptually, we impose the axiom because O cannot `see' internal moves, and thus the internal part of each play must be essentially P's calculation only; technically, it is to ensure \emph{external consistency} of dynamic strategies: Dynamic strategies behave always in the same way from the viewpoint of O, i.e., the external part of each play by a dynamic strategy does not depend on the internal part (see \cite{yamada2016dynamic} for the detail).

\begin{remark}
The axiom DP2 defined in \cite{yamada2016dynamic} just requires \emph{determinacy} of internal O-moves in each play (it is similar to determinacy of P-moves for strategies), which works for the purpose of the paper.
However, in the present paper, we are concerned with `effective computability' of strategies, and thus in particular computation of internal O-moves by P must be `effective' (since O cannot compute them). 
For this point, we have strengthened the axiom DP2 as above so that computation of internal O-moves becomes virtually nothing.
\end{remark}

\begin{convention}
Henceforth, \emph{\bfseries games} refer to \emph{dynamic} games by default. 
\end{convention}

\begin{definition}[Subgames \cite{yamada2016dynamic}]
A \emph{\bfseries subgame} of a game $G$ is a game $H$ that satisfies $M_H \subseteq M_G$, $\lambda_H = \lambda_G \upharpoonright M_H$, $\vdash_H \ \subseteq \ \vdash_G \cap \ (\{ \star \} \cup M_H) \times M_H$, $\Delta_H = \Delta_G \upharpoonright M_H$, $P_H \subseteq P_G$ and $\mu(H) = \mu(G)$.
In this case, we write $H \trianglelefteqslant G$.
\end{definition}

Below, let us give some simple examples of games:

\begin{example}
The \emph{\bfseries terminal game} $T$ is defined by $T \stackrel{\mathrm{df. }}{=} (\emptyset, \emptyset, \emptyset, \emptyset, \{ \boldsymbol{\epsilon} \})$.
$T$ is the simplest game as its position is only the empty sequence $\boldsymbol{\epsilon}$.
\end{example}

\begin{example}
The \emph{\bfseries boolean game} $\boldsymbol{2}$ is defined by:
\begin{itemize}

\item $M_{\boldsymbol{2}} \stackrel{\mathrm{df. }}{=} \{ [\hat{q}], [\mathit{tt}], [\mathit{ff}] \}$;

\item $\lambda_{\boldsymbol{2}} : [\hat{q}] \mapsto (\mathsf{O}, \mathsf{Q}, 0), [\mathit{tt}] \mapsto (\mathsf{P}, \mathsf{A}, 0), [\mathit{ff}] \mapsto (\mathsf{P}, \mathsf{A}, 0)$;

\item $\vdash_{\boldsymbol{2}} \stackrel{\mathrm{df. }}{=} \{ (\star, [\hat{q}]), ([\hat{q}], [\mathit{tt}]), ([\hat{q}], [\mathit{ff}]) \}$;

\item $\Delta_{\boldsymbol{2}} \stackrel{\mathrm{df. }}{=} \emptyset$;

\item $P_{\boldsymbol{2}} \stackrel{\mathrm{df. }}{=} \mathsf{Pref}(\{ [\hat{q}] . [\mathit{tt}], [\hat{q}] . [\mathit{ff}] \})$, where each non-initial move is justified by $[\hat{q}]$.


\end{itemize}
There are just two maximal positions of $\boldsymbol{2}$: $[\hat{q}] . [\mathit{tt}]$ and $[\hat{q}] . [\mathit{ff}]$, where each answer (i.e., $[\mathit{tt}]$ or $[\mathit{ff}]$) is made for the respective initial question $[\hat{q}]$. 
The former (resp. the latter) is to represent the truth value \emph{true} (resp. \emph{false}).
\end{example}

\begin{example}
The \emph{\bfseries natural number game} $N$ is defined by:
\begin{itemize}

\item $M_{N} \stackrel{\mathrm{df. }}{=} \{ [\hat{q}] \} \cup \{ [n] \ \! | \ \! n \in \mathbb{N} \ \! \}$;

\item $\lambda_{N} : [\hat{q}] \mapsto (\mathsf{O}, \mathsf{Q}, 0), [n] \mapsto (\mathsf{P}, \mathsf{A}, 0)$;

\item $\vdash_{N} \stackrel{\mathrm{df. }}{=} \{ (\star, [\hat{q}]) \} \cup \{ ([\hat{q}], [n]) \ \! | \ \! n \in \mathbb{N} \ \! \}$;

\item $\Delta_{N} \stackrel{\mathrm{df. }}{=} \emptyset$;

\item $P_{N} \stackrel{\mathrm{df. }}{=} \mathsf{Pref}(\{ [\hat{q}] . [n] \ \! | \ \! n \in \mathbb{N} \ \! \})$, where $n$ is justified by $[\hat{q}]$.


\end{itemize}
The position $[q] . [n]$ is to represent the natural number $n \in \mathbb{N}$, where the answer $[n]$ is made for the question $[q]$.
This is the formal definition of $N$ sketched in the introduction though we have slightly changed the notation for the moves.
\end{example}

This game $N$ of natural numbers is, though standard in the literature, almost the same as the set $\mathbb{N}$ of natural numbers except the trivial one-round communication between the participants. This point is unsatisfactory for at least two reasons:
\begin{enumerate}

\item It is difficult to define an intrinsic, non-inductive and non-axiomatic notion of `effective computability' of strategies on games generated from $N$ via the construction $\Rightarrow$ of \emph{function space} defined below since there is no intensional or `low-level' structure in $N$ (see, e.g.,  \cite{abramsky2014intensionality} for this point);

\item The game $N$ has almost the same mathematical structure as that of $\mathbb{N}$, and thus it contributes almost nothing to foundations of mathematics.

\end{enumerate}

Motivated by these points, we adopt the following `lazy' variant:
\begin{example}
\label{ExLazyNaturalNumbers}
The \emph{\bfseries lazy natural number game} $\mathcal{N}$ is defined by:
\begin{itemize}

\item $M_{\mathcal{N}} \stackrel{\mathrm{df. }}{=} \{ [\hat{q}], [q], [\mathit{yes}], [\mathit{no}] \}$;

\item $\lambda_{\mathcal{N}} : [\hat{q}] \mapsto (\mathsf{O}, \mathsf{Q}, 0), [q] \mapsto (\mathsf{O}, \mathsf{Q}, 0), [\mathit{yes}] \mapsto (\mathsf{P}, \mathsf{A}, 0), [\mathit{no}] \mapsto (\mathsf{P}, \mathsf{A}, 0)$;

\item $\vdash_{\mathcal{N}} \stackrel{\mathrm{df. }}{=} \{ (\star, [\hat{q}]), ([\hat{q}], [\mathit{no}]), ([\hat{q}], [\mathit{yes}]), ([q], [\mathit{no}]), ([q], [\mathit{yes}]), ([\mathit{yes}], [q]) \}$;

\item $\Delta_{\mathcal{N}} \stackrel{\mathrm{df. }}{=} \emptyset$;

\item $P_{\mathcal{N}} \stackrel{\mathrm{df. }}{=} \mathsf{Pref}(\{ [\hat{q}] . ([\mathit{yes}] . [q])^n . [\mathit{no}] \ \! | \ \! n \in \mathbb{N} \ \! \})$, where each non-initial occurrence is justified by the last occurrence.


\end{itemize}
We need two questions $[\hat{q}]$ and $[q]$, the initial and non-initial ones, respectively, for the axiom E1.
Intuitively, O asks by a question $[\hat{q}]$ or $[q]$ if P would like to `count one more', and P replies it by $[\mathit{yes}]$ if she would like to and by $[\mathit{no}]$ otherwise. Thus, for each $n \in \mathbb{N}$, the position $[\hat{q}] . ([\mathit{yes}] . [q])^n . [\mathit{no}]$ is to represent the number $n$, where each occurrence of an answer (i.e., $[\mathit{yes}]$ or $[\mathit{no}]$) is made for the last occurrence of a question (i.e., $[\hat{q}]$ or $[q]$).
Let us define $\underline{n} \stackrel{\mathrm{df. }}{=} \mathsf{Pref}(\{ [\hat{q}] . ([\mathit{yes}] . [q])^n . [\mathit{no}] \})^{\mathsf{Even}}$.

The game $\mathcal{N}$ defines natural numbers in an intuitively natural manner, namely as `counting processes', where our choice of notation for moves is inessential, i.e., $\mathcal{N}$ is \emph{syntax-independent}. 
Moreover, we may actually define it \emph{intrinsically}, i.e., without recourse to the set $\mathbb{N}$, by specifying its positions inductively: $[\hat{q}] . [\mathit{no}] \in P_{\mathcal{N}} \wedge ([\hat{q}] . \boldsymbol{s} . [\mathit{no}] \in P_{\mathcal{N}} \Rightarrow  [\hat{q}] . \boldsymbol{s} . [\mathit{yes}] . [q] . [\mathit{no}] \in P_{\mathcal{N}})$.
Thus, we may \emph{define} (rather than \emph{represent}) natural numbers to be the positions of $\mathcal{N}$ though we will not investigate foundational consequences of this definition in the present paper.

As we shall see, such step-by-step processes underlying natural numbers allow us to define `effective computability' of strategies in an intrinsic, non-inductive and non-axiomatic manner in Section~\ref{SubsectionViableStrategies}.
\end{example}

\begin{convention}
Henceforth, the \emph{\bfseries natural number game} refers to the \emph{lazy} variant $\mathcal{N}$.
\end{convention}

The following game will play a key role in Section~\ref{SubsectionViableStrategies} when we define `effective computability' of strategies:
\begin{example}
\label{ExTagGame}
The \emph{\bfseries (outer) tag game} $\mathcal{G}(\mathcal{T})$ is defined by:
\begin{itemize}

\item $M_{\mathcal{G}(\mathcal{T})} \stackrel{\mathrm{df. }}{=} \{ [\hat{q}_{\mathcal{T}}], [q_{\mathcal{T}}], [\sharp], [\prime], [\checkmark], [\langle], [\rangle] \}$, for which we define a function $\mathscr{D} : \sharp \mapsto \hbar, \prime \mapsto \ell, \langle \ \! \mapsto \Lbag, \rangle \mapsto \Rbag$ (n.b. this tells which outer tag each move represents);

\item $\lambda_{\mathcal{G}(\mathcal{T})} : [q] \mapsto (\mathsf{O}, \mathsf{Q}, 0), [a] \mapsto (\mathsf{P}, \mathsf{A}, 0)$, where $q \in \{ \hat{q}_{\mathcal{T}}, q_{\mathcal{T}} \}, a \in \{ \sharp, \prime, \checkmark, \langle, \rangle \}$;

\item $\vdash_{\mathcal{G}(\mathcal{T})} \stackrel{\mathrm{df. }}{=} \{ (\star, [\hat{q}_{\mathcal{T}}]) \} \cup \{ ([q], [a]) \mid q \in \{ \hat{q}_{\mathcal{T}}, q_{\mathcal{T}} \}, a \in \{ \sharp, \prime, \checkmark, \langle, \rangle \} \} \\ \cup \{ ([b], [q_{\mathcal{T}}]) \mid b \in \{ \sharp, \prime, \langle, \rangle \} \}$;

\item $\Delta_{\mathcal{G}(\mathcal{T})} \stackrel{\mathrm{df. }}{=} \emptyset$;

\item $P_{\mathcal{G}(\mathcal{T})} \stackrel{\mathrm{df. }}{=} \mathsf{Pref}(\{ [\hat{q}_{\mathcal{T}}] . [b_1] . [q_{\mathcal{T}}] . [b_2] \dots [q_{\mathcal{T}}] . [b_k] . [q_{\mathcal{T}}] . [\checkmark] \mid k \in \mathbb{N}, b_1, b_2, \dots, b_k \in \{ \sharp, \prime, \langle, \rangle \}, \mathscr{D}^\ast(b_1 b_2 \dots b_k) \in \mathcal{T} \ \! \})$, where each non-initial occurrence is justified by the last occurrence.


\end{itemize}
Similarly to $\mathcal{N}$, each position $[\hat{q}_{\mathcal{T}}] . [b_1] . [q_{\mathcal{T}}] .  [b_2] \dots [q_{\mathcal{T}}] . [b_k] . [q_{\mathcal{T}}] . [\checkmark]$ of $\mathcal{G}(\mathcal{T})$ is to represent the sequence $b_1 b_2 \dots b_k$, which denotes the outer tag $\mathscr{D}^\ast(b_1 b_2 \dots b_k) \in \mathcal{T}$.
\end{example}

Now, let us now recall the \emph{hiding operation} on games:
\begin{definition}[Hiding on games \cite{yamada2016dynamic}]
For each $d \in \mathbb{N} \cup \{ \omega \}$, we define the \emph{\bfseries $\boldsymbol{d}$-hiding operation} $\mathcal{H}^d$ on games as follows. Given a game $G$, the \emph{\bfseries $\boldsymbol{d}$-external game} $\mathcal{H}^d(G)$ of $G$ consists of the $d$-external arena $(M_{\mathcal{H}^d(G)}, \lambda_{\mathcal{H}^d(G)}, \vdash_{\mathcal{H}^d(G)}, \Delta_{\mathcal{H}^d(G)})$ of the arena $G$, and the set $P_{\mathcal{H}^d(G)} \stackrel{\mathrm{df. }}{=} \{ \mathcal{H}^d_G(\boldsymbol{s}) \mid \boldsymbol{s} \in P_G \ \! \}$ of positions. 
A game $H$ is \emph{\bfseries normalized} if $\mathcal{H}^\omega(H) = H$, i.e., if $M_H$ does not contain any internal move.
\end{definition}


\begin{theorem}[Closure of games under hiding \cite{yamada2016dynamic}]
\label{ThmClosureOfGamesUnderHiding}
For any game $G$, $\mathcal{H}^d(G)$ forms a well-defined game for all $d \in \mathbb{N} \cup \{ \omega \}$.
Moreover, if $G \trianglelefteqslant H$, then $\mathcal{H}^d(G) \trianglelefteqslant \mathcal{H}^d(H)$ for all $d \in \mathbb{N} \cup \{ \omega \}$.
Furthermore, $\underbrace{\mathcal{H}^1 \circ \mathcal{H}^1 \cdots \circ \mathcal{H}^1}_i(G) = \mathcal{H}^i(G)$ for all $i \in \mathbb{N}$.
\end{theorem}

\begin{proof}
Based on Lemma~\ref{LemClosureOfArenasAndJSequencesUnderHiding}; see \cite{yamada2016dynamic}.
\end{proof}

\begin{convention}
Thanks to Theorem~\ref{ThmClosureOfGamesUnderHiding}, the $i$-hiding operation $\mathcal{H}^i$ on games for each $i \in \mathbb{N}$ can be thought of as the $i$-times iteration of the $1$-hiding operation $\mathcal{H}^1$, which we call the \emph{\bfseries hiding operation} on games and write $\mathcal{H}$ for it.
\end{convention}

\subsubsection{Constructions on games}
\label{ConstructionsOnGames}
Now, let us recall the constructions on games in \cite{yamada2016dynamic} with `tags' formalized by outer and inner tags defined in Section~\ref{OnTags}.
A \emph{\bfseries tag} refers to an outer or internal tag.

On the other hand, for readers who are not familiar with game semantics, we first give a rather standard presentation of each construction, which keeps `tags' unspecified, before its formal definition. 
For this aim, we employ:
\begin{notation}
We write $x \in S + T$ iff $x \in S$ or $x \in T$, where we cannot have both $x \in S$ and $x \in T$ by the implicit `tag' for the disjoint union $S+T$. 
Also, given functions $f : S \to U$ and $g : T \to U$, we write $[f, g]$ for the function $S + T \to U$ that maps $x \in S + T$ to $f(x) \in U$ iff $x \in S$, and to $g(x) \in U$ otherwise. 
Moreover, given relations $R_S \subseteq S \times S$ and $R_T \subseteq T \times T$, we write $R_S + R_T$ for the relation on $S + T$ such that $(x, y) \in R_S + R_T \stackrel{\mathrm{df. }}{\Leftrightarrow} (x, y) \in R_S \vee (x, y) \in R_T$. 
\end{notation}

Let us begin with \emph{tensor (product)} $\otimes$. 
As already mentioned in the introduction, a position of the tensor product $A \otimes B$ consists of a position of $A$ and a position of $B$ played `in parallel without communication'.
That is, given games $A$ and $B$, their tensor $A \otimes B$ is defined by:
\begin{itemize}

\item $M_{A \otimes B} \stackrel{\mathrm{df. }}{=} M_A + M_B$;

\item $\lambda_{A \otimes B} \stackrel{\mathrm{df. }}{=} [\lambda_A, \lambda_B]$;

\item $\vdash_{A \otimes B} \ \stackrel{\mathrm{df. }}{=} \ \vdash_A + \ \! \vdash_B$;

\item $\Delta_{A \otimes B} \stackrel{\mathrm{df. }}{=} [\Delta_A, \Delta_B]$;

\item $P_{A \otimes B} \stackrel{\mathrm{df. }}{=} \{ \boldsymbol{s} \in \mathscr{L}_{A \otimes B} \mid \boldsymbol{s} \upharpoonright A \in P_A, \boldsymbol{s} \upharpoonright B \in P_B \ \! \}$

\end{itemize}
where $\boldsymbol{s} \upharpoonright A$ (resp. $\boldsymbol{s} \upharpoonright B$) denotes the j-subsequence of $\boldsymbol{s}$ that consists of moves of $A$ (resp. $B$).
As an illustration, recall the example $N \otimes N$ in the introduction, in which the `tags' are informally written as $(\_)^{[i]}$ ($i = 0, 1$).

As explained in \cite{abramsky1997semantics}, it is easy to see that in a play of a tensor $A \otimes B$ only O can switch between the component games $A$ and $B$ (by alternation). 

Let us now give the formal definition of tensor, for which the `tags' $(\_)^{[0]}$ and $(\_)^{[1]}$ are formalized by inner tags $(\_, \mathscr{W})$ and $(\_, \mathscr{E})$, respectively:
\begin{definition}[Tensor of games \cite{abramsky1999game}]
The \emph{\bfseries tensor (product)} $A \otimes B$ of games $A$ and $B$ is defined by:
\begin{itemize}

\item $M_{A \otimes B} \stackrel{\mathrm{df. }}{=} \{ [(a, \mathscr{W})]_{\boldsymbol{e}} \mid [a]_{\boldsymbol{e}} \in M_A \} \cup \{ [(b, \mathscr{E})]_{\boldsymbol{f}} \mid [b]_{\boldsymbol{f}} \in M_B \}$;

\item $\lambda_{A \otimes B} ([(m, X)]_{\boldsymbol{e}}) \stackrel{\mathrm{df. }}{=} \begin{cases} \lambda_A([m]_{\boldsymbol{e}}) &\text{if $X = \mathscr{W}$;} \\ \lambda_B([m]_{\boldsymbol{e}}) &\text{if $X = \mathscr{E}$;} \end{cases}$

\item $\star \vdash_{A \otimes B} [(m, X)]_{\boldsymbol{e}} \stackrel{\mathrm{df. }}{\Leftrightarrow} (X = \mathscr{W} \wedge \star \vdash_A [m]_{\boldsymbol{e}}) \vee (X = \mathscr{E} \wedge \star \vdash_B [m]_{\boldsymbol{e}})$;

\item $[(m, X)]_{\boldsymbol{e}} \vdash_{A \otimes B} [(n, Y)]_{\boldsymbol{f}} \stackrel{\mathrm{df. }}{\Leftrightarrow} (X = \mathscr{W} = Y \wedge [m]_{\boldsymbol{e}} \vdash_A [n]_{\boldsymbol{f}}) \vee (X = \mathscr{E} = Y \wedge [m]_{\boldsymbol{e}} \vdash_B [n]_{\boldsymbol{f}})$;

\item $\Delta_{A \otimes B} ([(m, X)]_{\boldsymbol{e}}) \stackrel{\mathrm{df. }}{=} \begin{cases} [(m', \mathscr{W})]_{\boldsymbol{e}} &\text{if $X = \mathscr{W}$, where $\Delta_A([m]_{\boldsymbol{e}}) = [m']_{\boldsymbol{e}}$;} \\ [(m'', \mathscr{E})]_{\boldsymbol{e}} &\text{if $X = \mathscr{E}$, where $\Delta_B([m]_{\boldsymbol{e}}) = [m'']_{\boldsymbol{e}}$;} \end{cases}$

\item $P_{A \otimes B} \stackrel{\mathrm{df. }}{=} \{ \boldsymbol{s} \in \mathscr{L}_{A \otimes B} \mid \boldsymbol{s} \upharpoonright \mathscr{W} \in P_A, \boldsymbol{s} \upharpoonright \mathscr{E} \in P_B \ \! \}$, where $\boldsymbol{s} \upharpoonright X$ (with $X \in \{ \mathscr{W}, \mathscr{E} \}$) is the j-subsequence of $\boldsymbol{s}$ that consists of moves of the form $[(m, X)]_{\boldsymbol{e}}$ changed into $[m]_{\boldsymbol{e}}$.


\end{itemize}
\end{definition}

\begin{example}
Some typical plays in the tensor $\mathcal{N} \otimes \mathcal{N}$ are as follows:
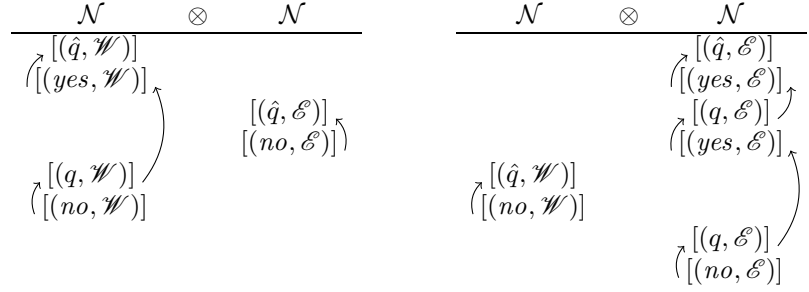
\begin{center}
\begin{tabular}{ccccccccc}
$\mathcal{N}$ & $\otimes$ & $\mathcal{N}$ &&&& $\mathcal{N}$ & $\otimes$ & $\mathcal{N}$ \\ \cline{1-3} \cline{7-9}
\tikzmark{LazyTensorC1} $[(\hat{q}, \mathscr{W})]$ &&&&&&&& \tikzmark{LazyTensorC6}  $[(\hat{q}, \mathscr{E})]$ \\
\tikzmark{LazyTensorD1} $[(\mathit{yes}, \mathscr{W})]$ \tikzmark{LazyTensorC2}&&&&&&&&\tikzmark{LazyTensorD6} $[(\mathit{yes}, \mathscr{E})]$ \tikzmark{LazyTensorC7} \\
&&$[(\hat{q}, \mathscr{E})]$ \tikzmark{LazyTensorC4}&&&&&& \tikzmark{LazyTensorC8} $[(q, \mathscr{E})]$ \tikzmark{LazyTensorD7} \\
&&$[(\mathit{no}, \mathscr{E})]$ \tikzmark{LazyTensorD4}&&&&&& \tikzmark{LazyTensorD8} $[(\mathit{yes}, \mathscr{E})]$ \tikzmark{LazyTensorC9} \\
\tikzmark{LazyTensorC3} $[(q, \mathscr{W})]$ \tikzmark{LazyTensorD2}&&&&&& \tikzmark{LazyTensorC5} $[(\hat{q}, \mathscr{W})]$&& \\
\tikzmark{LazyTensorD3} $[(\mathit{no}, \mathscr{W})]$ &&&&&& \tikzmark{LazyTensorD5} $[(\mathit{no}, \mathscr{W})]$ && \\
&&&&&&&&\tikzmark{LazyTensorC10} $[(q, \mathscr{E})]$ \tikzmark{LazyTensorD9} \\
&&&&&&&&\tikzmark{LazyTensorD10} $[(\mathit{no}, \mathscr{E})]$ \\
\end{tabular}
\begin{tikzpicture}[overlay, remember picture, yshift=.25\baselineskip]
\draw [->] ({pic cs:LazyTensorD1}) [bend left] to ({pic cs:LazyTensorC1});
\draw [->] ({pic cs:LazyTensorD2}) [bend right] to ({pic cs:LazyTensorC2});
\draw [->] ({pic cs:LazyTensorD3}) [bend left] to ({pic cs:LazyTensorC3});
\draw [->] ({pic cs:LazyTensorD4}) [bend right] to ({pic cs:LazyTensorC4});
\draw [->] ({pic cs:LazyTensorD5}) [bend left] to ({pic cs:LazyTensorC5});
\draw [->] ({pic cs:LazyTensorD6}) [bend left] to ({pic cs:LazyTensorC6});
\draw [->] ({pic cs:LazyTensorD7}) [bend right] to ({pic cs:LazyTensorC7});
\draw [->] ({pic cs:LazyTensorD8}) [bend left] to ({pic cs:LazyTensorC8});
\draw [->] ({pic cs:LazyTensorD9}) [bend right] to ({pic cs:LazyTensorC9});
\draw [->] ({pic cs:LazyTensorD10}) [bend left] to ({pic cs:LazyTensorC10});
\end{tikzpicture}
\end{center}
\end{example}

Next, the \emph{linear implication} $A \multimap B$ is a space of \emph{linear functions} from $A$ to $B$ in the sense of \emph{linear logic} \cite{girard1987linear}, i.e., they consume exactly one input in $A$ to produces an output in $B$ (n.b., strictly speaking, it is an \emph{affine implication} for its functions are not necessarily \emph{strict} as explained in the introduction).
Usually, the linear implication $A \multimap B$ is given by: 
\begin{itemize}

\item $M_{A \multimap B} \stackrel{\mathrm{df. }}{=} M_{\mathcal{H}^\omega(A)} + M_B$;

\item $\lambda_{A \multimap B} \stackrel{\mathrm{df. }}{=} [\overline{\lambda_{\mathcal{H}^\omega(A)}}, \lambda_B]$, where $\overline{\lambda_{\mathcal{H}^\omega(A)}} \stackrel{\mathrm{df. }}{=} \langle \overline{\lambda_{\mathcal{H}^\omega(A)}^\textsf{OP}}, \lambda_{\mathcal{H}^\omega(A)}^\textsf{QA} \rangle$ and $\overline{\lambda_{\mathcal{H}^\omega(A)}^\textsf{OP}} (m) \stackrel{\mathrm{df. }}{=} \begin{cases} \mathsf{P} \ &\text{if $\lambda_{\mathcal{H}^\omega(A)}^\textsf{OP} (m) = \textsf{O}$} \\ \textsf{O} \ &\text{otherwise} \end{cases}$;

\item $\star \vdash_{A \multimap B} m \stackrel{\mathrm{df. }}{\Leftrightarrow} \star \vdash_B m$;

\item $m \vdash_{A \multimap B} n \ (m \neq \star) \stackrel{\mathrm{df. }}{\Leftrightarrow} (m \vdash_{\mathcal{H}^\omega(A)} n) \vee (m \vdash_B n) \vee (\star \vdash_B m \wedge \star \vdash_{\mathcal{H}^\omega(A)} n)$;

\item $\Delta_{A \multimap B} \stackrel{\mathrm{df. }}{=} [\emptyset, \Delta_B]$ (where note that $\Delta_{\mathcal{H}^\omega(A)} = \emptyset$);

\item $P_{A \multimap B} \stackrel{\mathrm{df. }}{=} \{ \boldsymbol{s} \in \mathscr{L}_{\mathcal{H}^\omega(A) \multimap B} \mid \boldsymbol{s} \upharpoonright \mathcal{H}^\omega(A) \in P_{\mathcal{H}^\omega(A)}, \boldsymbol{s} \upharpoonright B \in P_B \ \! \}$.

\end{itemize}
As an illustration, recall the example $N \multimap N$ in the introduction. 

Note that $A$ must be normalized to $\mathcal{H}^\omega(A)$ since otherwise the linear implication $A \multimap B$ may not satisfy the axiom DP2 (and the dummy function $\Delta_{A \multimap B}$ would not be well-defined).
It conceptually makes sense too for the roles of P and O in the domain $A$ are exchanged, and thus P should not be able to `see' internal moves of $A$.
Note also that $A \multimap B$ is almost $A \otimes B$ if $A$ is normalized except the switch of the roles in $A$; dually to $A \otimes B$, in a play of $A \multimap B$ only P can switch between $A$ and $B$ by alternation (see, e.g., \cite{abramsky1997semantics} for the proof).
Surprisingly, this simple point turns $A \multimap B$ into a game for linear functions from $A$ to $B$.

Similarly to tensor, the formal definition of linear implication is as follows:
\begin{definition}[Linear implication \cite{abramsky1999game}]
\label{DefLinearImplication}
The \emph{\bfseries linear implication} $A \multimap B$ from a normalized game $A$ to another (not necessarily normalized) game $B$ is defined by:
\begin{itemize}

\item $M_{A \multimap B} \stackrel{\mathrm{df. }}{=} \{ [(a, \mathscr{W})]_{\boldsymbol{e}} \mid [a]_{\boldsymbol{e}} \in M_A \} \cup \{ [(b, \mathscr{E})]_{\boldsymbol{f}} \mid [b]_{\boldsymbol{f}} \in M_B \}$;

\item $\lambda_{A \multimap B} ([(m, X)]_{\boldsymbol{e}}) \stackrel{\mathrm{df. }}{=} \begin{cases} \overline{\lambda_A}([m]_{\boldsymbol{e}}) &\text{if $X = \mathscr{W}$;} \\ \lambda_B([m]_{\boldsymbol{e}}) &\text{if $X = \mathscr{E}$} \end{cases}$, where $\overline{\lambda_A} \stackrel{\mathrm{df. }}{=} \langle \overline{\lambda_A^\textsf{OP}}, \lambda_A^\textsf{QA}, \lambda_A^{\mathbb{N}} \rangle$ and $\overline{\lambda_A^\textsf{OP}} ([a]_{\boldsymbol{e}}) \stackrel{\mathrm{df. }}{=} \begin{cases} \mathsf{P} \ &\text{if $\lambda_A^\textsf{OP} ([a]_{\boldsymbol{e}}) = \textsf{O}$;} \\ \textsf{O} \ &\text{otherwise} \end{cases}$ for all $[a]_{\boldsymbol{e}} \in M_A$;

\item $\star \vdash_{A \multimap B} [(m, X)]_{\boldsymbol{e}} \stackrel{\mathrm{df. }}{\Leftrightarrow} X = \mathscr{E} \wedge \star \vdash_B [m]_{\boldsymbol{e}}$;

\item $[(m, X)]_{\boldsymbol{e}} \vdash_{A \multimap B} [(n, Y)]_{\boldsymbol{f}} \stackrel{\mathrm{df. }}{\Leftrightarrow} \begin{cases} \begin{aligned} &(X = \mathscr{W} = Y \wedge [m]_{\boldsymbol{e}} \vdash_A [n]_{\boldsymbol{f}}) \\ &\vee (X = \mathscr{E} = Y \wedge [m]_{\boldsymbol{e}} \vdash_B [n]_{\boldsymbol{f}}) \\ &\vee (X = \mathscr{E} \wedge Y = \mathscr{W} \wedge \star \vdash_B [m]_{\boldsymbol{e}} \wedge \star \vdash_A [n]_{\boldsymbol{f}}); \end{aligned} \end{cases}$

\item $\Delta_{A \multimap B}([(b, \mathscr{E})]_{\boldsymbol{f}}) \stackrel{\mathrm{df. }}{=} [(b', \mathscr{E})]_{\boldsymbol{f}}$, where $\Delta_B([b]_{\boldsymbol{f}}) = [b']_{\boldsymbol{f}}$;

\item $P_{A \multimap B} \stackrel{\mathrm{df. }}{=} \{ \boldsymbol{s} \in \mathscr{L}_{A \multimap B} \mid \boldsymbol{s} \upharpoonright \mathscr{W} \in P_A, \boldsymbol{s} \upharpoonright \mathscr{E} \in P_B \}$, where pointers from initial occurrences in $\boldsymbol{s} \upharpoonright \mathscr{W}$ to initial occurrences in $\boldsymbol{s} \upharpoonright \mathscr{E}$ are simply deleted.


\end{itemize}
This construction $\multimap$ is then extended to any pair $(C, B)$ of games: For a not necessarily normalized game $C$, we define $C \multimap B \stackrel{\mathrm{df. }}{=} \mathcal{H}^\omega(C) \multimap B$. 
\end{definition}

\begin{example}
Some typical plays of the linear implication $\boldsymbol{2} \multimap \boldsymbol{2}$ are as follows:
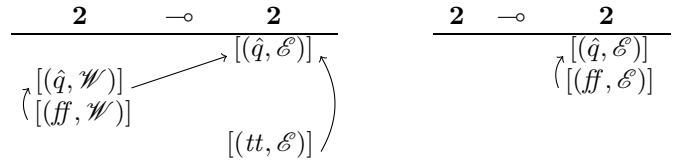
\begin{center}
\begin{tabular}{ccccccccc}
$\boldsymbol{2}$ & $\multimap$ & $\boldsymbol{2}$ &&&& $\boldsymbol{2}$ & $\multimap$ & $\boldsymbol{2}$ \\ \cline{1-3} \cline{7-9}
&&\tikzmark{LIC1} $[(\hat{q}, \mathscr{E})]$ \tikzmark{LIC3}&&&&&&\tikzmark{LIC4} $[(\hat{q}, \mathscr{E})]$ \\
\tikzmark{LIC2} $[(\hat{q}, \mathscr{W})]$ \tikzmark{LID1}&&&&&&&&\tikzmark{LID4} $[(\mathit{ff}, \mathscr{E})]$ \\
\tikzmark{LID2} $[(\mathit{ff}, \mathscr{W})]$&&&&&&&& \\
&&$[(\mathit{tt}, \mathscr{E})]$ \tikzmark{LID3}&&&&&&
\end{tabular}
\begin{tikzpicture}[overlay, remember picture, yshift=.25\baselineskip]
\draw [->] ({pic cs:LID1}) to ({pic cs:LIC1});
\draw [->] ({pic cs:LID2}) [bend left] to ({pic cs:LIC2});
\draw [->] ({pic cs:LID3}) [bend right] to ({pic cs:LIC3});
\draw [->] ({pic cs:LID4}) [bend left] to ({pic cs:LIC4});
\end{tikzpicture}
\end{center}
Note that the left diagram describes a \emph{strict} linear function, while the right diagram does a non-strict one.
\end{example} 

Next, let us recall \emph{product} $\&$ of games.
As stated in the introduction, a position of the product $A \& B$ is essentially a position of $A$ or a one of $B$:
Given games $A$ and $B$, their product $A \& B$ is given by:

\begin{itemize}

\item $M_{A \& B} \stackrel{\mathrm{df. }}{=} M_A + M_B$;

\item $\lambda_{A \& B} \stackrel{\mathrm{df. }}{=} [\lambda_A, \lambda_B]$;

\item $\vdash_{A \& B} \ \stackrel{\mathrm{df. }}{=} \ \vdash_A + \ \! \vdash_B$;

\item $\Delta_{A \& B} \stackrel{\mathrm{df. }}{=} [\Delta_A, \Delta_B]$;

\item $P_{A \& B} \stackrel{\mathrm{df. }}{=} \{ \boldsymbol{s} \in \mathscr{L}_{A \& B} \mid (\boldsymbol{s} \upharpoonright A \in P_A \wedge \boldsymbol{s} \upharpoonright B = \boldsymbol{\epsilon}) \vee (\boldsymbol{s} \upharpoonright A = \boldsymbol{\epsilon} \wedge \boldsymbol{s} \upharpoonright B \in P_B) \ \! \}$.

\end{itemize}

Similarly to the case of tensor, we formalize product as follows: 
\begin{definition}[Product \cite{abramsky1999game}]
The \emph{\bfseries product} $A \& B$ of games $A$ and $B$ is given by:
\begin{itemize}

\item $M_{A \& B} \stackrel{\mathrm{df. }}{=} \{ [(a, \mathscr{W})]_{\boldsymbol{e}} \mid [a]_{\boldsymbol{e}} \in M_A \ \! \} \cup \{ [(b, \mathscr{E})]_{\boldsymbol{f}} \mid [b]_{\boldsymbol{f}} \in M_B \}$;

\item $\lambda_{A \& B} ([(m, X)]_{\boldsymbol{e}}) \stackrel{\mathrm{df. }}{=} \begin{cases} \lambda_A([m]_{\boldsymbol{e}}) &\text{if $X = \mathscr{W}$;} \\ \lambda_B([m]_{\boldsymbol{e}}) &\text{if $X = \mathscr{E}$;} \end{cases}$

\item $\star \vdash_{A \& B} [(m, X)]_{\boldsymbol{e}} \stackrel{\mathrm{df. }}{\Leftrightarrow} (X = \mathscr{W} \wedge \star \vdash_A [m]_{\boldsymbol{e}}) \vee (X = \mathscr{E} \wedge \star \vdash_B [m]_{\boldsymbol{e}})$;

\item $[(m, X)]_{\boldsymbol{e}} \vdash_{A \& B} [(n, Y)]_{\boldsymbol{f}} \stackrel{\mathrm{df. }}{\Leftrightarrow} (X = \mathscr{W} = Y \wedge [m]_{\boldsymbol{e}} \vdash_A [n]_{\boldsymbol{f}}) \vee (X = \mathscr{E} = Y \wedge [m]_{\boldsymbol{e}} \vdash_B [n]_{\boldsymbol{f}})$;

\item $\Delta_{A \& B} ([(m, X)]_{\boldsymbol{e}}) \stackrel{\mathrm{df. }}{=} \begin{cases} [(m', \mathscr{W})]_{\boldsymbol{e}} &\text{if $X = \mathscr{W}$, where $\Delta_A([m]_{\boldsymbol{e}}) = [m']_{\boldsymbol{e}}$;} \\ [(m'', \mathscr{E})]_{\boldsymbol{e}} &\text{if $X = \mathscr{E}$, where $\Delta_B([m]_{\boldsymbol{e}}) = [m'']_{\boldsymbol{e}}$;} \end{cases}$

\item $P_{A \& B} \stackrel{\mathrm{df. }}{=} \{ \boldsymbol{s} \in \mathscr{L}_{A \& B} \mid (\boldsymbol{s} \upharpoonright \mathscr{W} \in P_A \wedge \boldsymbol{s} \upharpoonright \mathscr{E} = \boldsymbol{\epsilon}) \vee (\boldsymbol{s} \upharpoonright \mathscr{W} = \boldsymbol{\epsilon} \wedge \boldsymbol{s} \upharpoonright \mathscr{E} \in P_B) \ \! \}$.


\end{itemize}
\end{definition}

For the \emph{cartesian closed bicategory} $\mathcal{DG}$ of dynamic games and strategies defined in \cite{yamada2016dynamic}, however, we need to generalize the construction $C \multimap A \& B$ on normalized games $A$, $B$ and $C$, where $\&$ precedes $\multimap$, because we need to pair strategies $\sigma : L$ and $\tau : R$ such that $\mathcal{H}^\omega(L) \trianglelefteqslant C \multimap A$ and $\mathcal{H}^\omega(R) \trianglelefteqslant C \multimap B$, and the ambient game of the pairing $\langle \sigma, \tau \rangle$ would be such a generalization of $C \multimap A \& B$.

For this point, \cite{yamada2016dynamic} defines the \emph{pairing} $\langle L, R \rangle$ of any games $L$ and $R$ with $\mathcal{H}^\omega(L) \trianglelefteqslant C \multimap A$ and $\mathcal{H}^\omega(R) \trianglelefteqslant C \multimap B$ for some normalized games $A$, $B$ and $C$ by:
\begin{itemize}

\item $M_{\langle L, R \rangle} \stackrel{\mathrm{df. }}{=} M_C + (M_L \setminus M_C) + (M_R \setminus M_C)$;

\item $\lambda_{\langle L, R \rangle} \stackrel{\mathrm{df. }}{=} [\overline{\lambda_C}, \lambda_L \downharpoonright M_C, \lambda_R \downharpoonright M_C]$;

\item $m \vdash_{\langle L, R \rangle} n \stackrel{\mathrm{df. }}{\Leftrightarrow} m \vdash_L n \vee m \vdash_R n$;

\item $\Delta_{\langle L, R \rangle} \stackrel{\mathrm{df. }}{=} [\Delta_L \downharpoonright M_C, \Delta_R \downharpoonright M_C]$;

\item $P_{\langle L, R \rangle} \stackrel{\mathrm{df. }}{=} \{ \boldsymbol{s} \in \mathscr{L}_{L \& R} \mid (\boldsymbol{s} \upharpoonright L \in P_L \wedge \boldsymbol{s} \upharpoonright R = \boldsymbol{\epsilon}) \vee (\boldsymbol{s} \upharpoonright L = \boldsymbol{\epsilon} \wedge \boldsymbol{s} \upharpoonright R \in P_R) \ \! \}$

\end{itemize}
where given a function $f : X \to Y$ and a subset $Z \subseteq X$ we write $f \downharpoonright Z : X \setminus Z \to Y$ for the restrictions of $f$ to the subset $X \setminus Z \subseteq X$. 

Note that the definition of a pairing $\langle L, R \rangle$ does not depend on the choice of the normalized games $A$, $B$ and $C$ such that $\mathcal{H}^\omega(L) \trianglelefteqslant C \multimap A$ and $\mathcal{H}^\omega(R) \trianglelefteqslant C \multimap B$.
Also, we have $\langle C \multimap A, C \multimap B \rangle = C \multimap A \& B$; pairing of games is to serve as a generalization of this phenomenon in the sense that $\mathcal{H}^\omega(\langle L, R \rangle) \trianglelefteqslant C \multimap A \& B$ holds (see \cite{yamada2016dynamic} for the proof), where note that the slightly involved disjoint union of sets of moves for $M_{\langle L, R \rangle}$ is to establish this subgame relation.

Let us now formalize `tags' for the disjoint union $M_C + (M_L \setminus M_C) + (M_R \setminus M_C)$ as follows:
\begin{itemize}

\item Adding no tags for external moves of the form $[(c, \mathscr{W})]_{\boldsymbol{e}}$ of $L$ or $R$, where $[c]_{\boldsymbol{e}}$ must be a move of $C$ by the definition of tags for $\multimap$ (see Definition~\ref{DefLinearImplication}); 

\item Changing external moves of the form $[(a, \mathscr{E})]_{\boldsymbol{f}}$ of $L$, where $[a]_{\boldsymbol{f}}$ must be a move of $A$, into $[((a, \mathscr{W}), \mathscr{E})]_{\boldsymbol{f}}$; 

\item Changing external moves of the form $[(b, \mathscr{E})]_{\boldsymbol{g}}$ of $R$, where $[b]_{\boldsymbol{g}}$ must be a move of $B$, into $[((b, \mathscr{E}), \mathscr{E})]_{\boldsymbol{g}}$; 

\item Changing internal moves $[l]_{\boldsymbol{h}}$ of $L$ into $[(l, \mathscr{S})]_{\boldsymbol{h}}$; 

\item Changing internal moves $[r]_{\boldsymbol{k}}$ of $R$ into $[(r, \mathscr{N})]_{\boldsymbol{k}}$

\end{itemize} 
where of course these tags are not canonical at all.

Then, we formalize the labeling function, the enabling relation and the dummy function of $\langle L, R \rangle$ by the obvious pattern matching on internal tags; positions of $\langle L, R \rangle$ are formalized in the obvious manner. 
However, the enabling relation is rather involved; thus, for convenience, we define the \emph{peeling} $\mathit{peel}_{\langle L, R \rangle}(m) \in M_L \cup M_R$ of each move $m \in M_{\langle L, R \rangle}$ such that changing the inner tag of $\mathit{peel}_{\langle L, R \rangle}(m)$ as defined above results in $m$, and also the \emph{attribute} $\mathit{att}_{\langle L, R \rangle}(m) \in \{ L, R, C \}$ of $m$ by: 
\begin{equation*}
\mathit{att}_{\langle L, R \rangle}(m) \stackrel{\mathrm{df. }}{=} \begin{cases} L &\text{if $\mathit{peel}_{\langle L, R \rangle}(m) \in M_L \setminus M_C$;} \\ R &\text{if $\mathit{peel}_{\langle L, R \rangle}(m) \in M_R \setminus M_C$;} \\ C &\text{otherwise (i.e., if $\mathit{peel}_{\langle L, R \rangle}(m) \in M_C$).} \end{cases}
\end{equation*}
The enabling relation $m \vdash_{\langle L, R \rangle} n$ is then easily defined as the conjunction of:
\begin{itemize}

\item $\mathit{att}_{\langle L, R \rangle}(m) = \mathit{att}_{\langle L, R \rangle}(n) \vee \mathit{att}_{\langle L, R \rangle}(m) = C \vee \mathit{att}_{\langle L, R \rangle}(n) = C$;

\item $\mathit{peel}_{\langle L, R \rangle}(m) \vdash_L \mathit{peel}_{\langle L, R \rangle}(n) \vee \mathit{peel}_{\langle L, R \rangle}(m) \vdash_R \mathit{peel}_{\langle L, R \rangle}(n)$.

\end{itemize}

Formally, we define pairing of games as follows:
\begin{definition}[Pairing of games \cite{yamada2016dynamic}]
\label{DefPairingOfGames}
The \emph{\bfseries pairing} $\langle L, R \rangle$ of games $L$ and $R$ such that $\mathcal{H}^\omega(L) \trianglelefteqslant C \multimap A$ and $\mathcal{H}^\omega(R) \trianglelefteqslant C \multimap B$ for any normalized games $A$, $B$ and $C$ is given by:
\begin{itemize}

\item $M_{\langle L, R \rangle} \stackrel{\mathrm{df. }}{=} \{ [(c, \mathscr{W})]_{\boldsymbol{e}} \mid [(c, \mathscr{W})]_{\boldsymbol{e}} \in M_L^{\mathsf{Ext}} \cup M_R^{\mathsf{Ext}}, [c]_{\boldsymbol{e}} \in M_C \} \\ \cup \{ [((a, \mathscr{W}), \mathscr{E})]_{\boldsymbol{f}} \mid [(a, \mathscr{E})]_{\boldsymbol{f}} \in M_L^{\mathsf{Ext}}, [a]_{\boldsymbol{f}} \in M_A \} \\ \cup \{ [((b, \mathscr{E}), \mathscr{E})]_{\boldsymbol{g}} \mid [(b, \mathscr{E})]_{\boldsymbol{g}} \in M_R^{\mathsf{Ext}}, [b]_{\boldsymbol{g}} \in M_B \} \\ \cup \{ [(l, \mathscr{S})]_{\boldsymbol{h}} \mid [l]_{\boldsymbol{h}} \in M_L^{\mathsf{Int}} \} \cup \{ [(r, \mathscr{N})]_{\boldsymbol{k}} \mid [r]_{\boldsymbol{k}} \in M_R^{\mathsf{Int}} \}$;

\item $\lambda_{\langle L, R \rangle} ([(m, X)]_{\boldsymbol{e}}) \stackrel{\mathrm{df. }}{=} \begin{cases} \overline{\lambda_C} ([m]_{\boldsymbol{e}}) &\text{if $X = \mathscr{W}$;} \\ \lambda_A([a]_{\boldsymbol{e}}) &\text{if $X = \mathscr{E}$ and $m$ is of the form $(a, \mathscr{W})$;} \\ \lambda_B([b]_{\boldsymbol{e}}) &\text{if $X = \mathscr{E}$ and $m$ is of the form $(b, \mathscr{E})$;} \\ \lambda_L([m]_{\boldsymbol{e}}) &\text{if $X = \mathscr{S}$;} \\ \lambda_R([m]_{\boldsymbol{e}}) &\text{if $X = \mathscr{N}$;} \end{cases}$

\item $\star \vdash_{\langle L, R \rangle} [(m, X)]_{\boldsymbol{e}} \stackrel{\mathrm{df. }}{\Leftrightarrow} X = \mathscr{E} \wedge (\exists [a]_{\boldsymbol{e}} \in M_A^{\mathsf{Init}} . \ \! m = (a, \mathscr{W}) \vee \exists [b]_{\boldsymbol{e}} \in M_B^{\mathsf{Init}} . \ \! m = (b, \mathscr{E}))$;

\item $[(m, X)]_{\boldsymbol{e}} \vdash_{\langle L, R \rangle} [(n, Y)]_{\boldsymbol{f}} \stackrel{\mathrm{df. }}{\Leftrightarrow} (\mathit{att}_{\langle L, R \rangle}([(m, X)]_{\boldsymbol{e}}) = \mathit{att}_{\langle L, R \rangle}([(n, Y)]_{\boldsymbol{f}}) \vee \mathit{att}_{\langle L, R \rangle}([(m, X)]_{\boldsymbol{e}}) = C \vee \mathit{att}_{\langle L, R \rangle}([(n, Y)]_{\boldsymbol{f}}) = C) \wedge (\mathit{peel}_{\langle L, R \rangle}([(m, X)]_{\boldsymbol{e}}) \vdash_L \mathit{peel}_{\langle L, R \rangle}([(n, Y)]_{\boldsymbol{f}}) \vee \mathit{peel}_{\langle L, R \rangle}([(m, X)]_{\boldsymbol{e}}) \vdash_R \mathit{peel}_{\langle L, R \rangle}([(n, Y)]_{\boldsymbol{f}}))$, where the function $\mathit{att}_{\langle L, R \rangle} : M_{\langle L, R \rangle} \to \{ L, R, C \}$ is defined by $[(c, \mathscr{W})]_{\boldsymbol{e}} \mapsto C$, $[((a, \mathscr{W}), \mathscr{E})]_{\boldsymbol{f}} \mapsto L$, $[((b, \mathscr{E}), \mathscr{E})]_{\boldsymbol{g}} \mapsto R$, $[(l, \mathscr{S})]_{\boldsymbol{h}} \mapsto L$, $[(r, \mathscr{N})]_{\boldsymbol{k}} \mapsto R$, and the function $\mathit{peel}_{\langle L, R \rangle} : M_{\langle L, R \rangle} \to M_L \cup M_R$ by $[(c, \mathscr{W})]_{\boldsymbol{e}} \mapsto [(c, \mathscr{W})]_{\boldsymbol{e}}$, $[((a, \mathscr{W}), \mathscr{E})]_{\boldsymbol{f}} \mapsto [(a, \mathscr{E})]_{\boldsymbol{f}}$, $[((b, \mathscr{E}), \mathscr{E})]_{\boldsymbol{g}} \mapsto [(b, \mathscr{E})]_{\boldsymbol{g}}$, $[(l, \mathscr{S})]_{\boldsymbol{h}} \mapsto [l]_{\boldsymbol{h}}$, $[(r, \mathscr{N})]_{\boldsymbol{k}} \mapsto [r]_{\boldsymbol{k}}$;

\item $\Delta_{\langle L, R \rangle} ([(m, X)]_{\boldsymbol{e}}) \stackrel{\mathrm{df. }}{=} \begin{cases} [(l', \mathscr{S})]_{\boldsymbol{e}} &\text{if $X = \mathscr{S}$, where $\Delta_L([l]_{\boldsymbol{e}}) = [l']_{\boldsymbol{e}}$;} \\ [(r', \mathscr{N})]_{\boldsymbol{e}} &\text{if $X = \mathscr{N}$, where $\Delta_R([r]_{\boldsymbol{e}}) = [r']_{\boldsymbol{e}}$;} \end{cases}$

\item $P_{\langle L, R \rangle} \stackrel{\mathrm{df. }}{=} \{ \boldsymbol{s} \in \mathscr{L}_{\langle L, R \rangle} \mid (\boldsymbol{s} \upharpoonright L \in P_L \wedge \boldsymbol{s} \upharpoonright B = \boldsymbol{\epsilon}) \vee (\boldsymbol{s} \upharpoonright R \in P_R \wedge \boldsymbol{s} \upharpoonright A = \boldsymbol{\epsilon}) \ \! \}$, where $\boldsymbol{s} \upharpoonright L$ (resp. $\boldsymbol{s} \upharpoonright R$) is the j-subsequence of $\boldsymbol{s}$ that consists of moves $x$ such that $\mathit{peel}_{\langle L, R \rangle}(x) \in M_L$ (resp. $\mathit{peel}_{\langle L, R \rangle}(x) \in M_R$) changed into $\mathit{peel}_{\langle L, R \rangle}(x)$, and $\boldsymbol{s} \upharpoonright B$ (resp. $\boldsymbol{s} \upharpoonright A$) is the j-subsequence of $\boldsymbol{s}$ that consists of moves of the form $[((b, \mathscr{E}), \mathscr{E})]_{\boldsymbol{g}}$ with $[b]_{\boldsymbol{g}} \in M_B$ (resp. $[((a_0, \mathscr{W}), \mathscr{E})]_{\boldsymbol{f}}$ with $[a]_{\boldsymbol{g}} \in M_A$).


\end{itemize}
\end{definition}

\begin{example}
Some typical plays of the pairing $\langle \boldsymbol{2} \multimap \boldsymbol{2}, \boldsymbol{2} \multimap \boldsymbol{2} \rangle$ are as follows:
\begin{center}
\begin{tabular}{ccccccc}
$\langle \boldsymbol{2}$ && $\multimap$ & $\boldsymbol{2},$ & $\boldsymbol{2}$ & $\multimap$ & $\boldsymbol{2} \rangle$ \\ \cline{1-7}
&&&&&&\tikzmark{PairingC1} $[((\hat{q}, \mathscr{E}), \mathscr{E})]$ \tikzmark{PairingC3} \\
&&&&\tikzmark{PairingC2} $[(\hat{q}, \mathscr{W})]$ \tikzmark{PairingD1}&& \\
&&&&\tikzmark{PairingD2} $[(\mathit{ff}, \mathscr{W})]$&& \\
&&&&&&$[((\mathit{tt}, \mathscr{E}), \mathscr{E})]$ \tikzmark{PairingD3} \\ \\
$\langle \boldsymbol{2}$ && $\multimap$ & $\boldsymbol{2},$ & $\boldsymbol{2}$ & $\multimap$ & $\boldsymbol{2} \rangle $ \\ \cline{1-7}
&&&\tikzmark{PairingC4} $[((\hat{q}, \mathscr{W}), \mathscr{E})]$&&& \\
&&&\tikzmark{PairingD4} $[((\mathit{tt}, \mathscr{W}), \mathscr{E})]$&&&
\end{tabular}
\begin{tikzpicture}[overlay, remember picture, yshift=.25\baselineskip]
\draw [->] ({pic cs:PairingD1}) to ({pic cs:PairingC1});
\draw [->] ({pic cs:PairingD2}) [bend left] to ({pic cs:PairingC2});
\draw [->] ({pic cs:PairingD3}) [bend right] to ({pic cs:PairingC3});
\draw [->] ({pic cs:PairingD4}) [bend left] to ({pic cs:PairingC4});
\end{tikzpicture}
\end{center}
\end{example}

Next, let us recall \emph{exponential} $!$ in the sense of linear logic, i.e., the exponential $!A$ of a given game $A$ is essentially a countably infinite tensor of $A$, i.e., $!A \stackrel{\mathrm{df. }}{=} A \otimes A \otimes \dots$
The exponential $!A$ is usually given by:
\begin{itemize}

\item $M_{!A} \stackrel{\mathrm{df. }}{=} M_A \times \mathbb{N}$;

\item $\lambda_{!A} : (a, i) \mapsto \lambda_A(a)$;

\item $\star \vdash_{!A} (a, i) \stackrel{\mathrm{df. }}{\Leftrightarrow} \star \vdash_A a$;

\item $(a, i) \vdash_{!A} (a', j) \stackrel{\mathrm{df. }}{\Leftrightarrow} i = j \wedge a \vdash_A a'$;

\item $\Delta_{!A} : (a, i) \mapsto (\Delta_A(a), i)$;

\item $P_{!A} \stackrel{\mathrm{df. }}{=} \{ \boldsymbol{s} \in \mathscr{L}_{!A} \mid \forall i \in \mathbb{N} . \ \! \boldsymbol{s} \upharpoonright i \in P_A \ \! \}$

\end{itemize}
where $\boldsymbol{s} \upharpoonright i$ is the j-subsequence of $\boldsymbol{s}$ that consists of occurrences of the form $(a, i)$ but changed into $a$.

The naive idea is then to formalize each `tag' $(\_, i)$ by an effective tag $[\_]_{\underline{i}}$, but as already mentioned before, we need to generalize it to an extended effective tag $[\_]_{\boldsymbol{f}}$.
Therefore, we formalize exponential as follows: 
\begin{definition}[Exponential \cite{mccusker1998games}]
\label{DefExponential}
The \emph{\bfseries exponential} $!A$ of a game $A$ is defined by:
\begin{itemize}

\item $M_{!A} \stackrel{\mathrm{df. }}{=} \{ [m]_{\Lbag \boldsymbol{f} \Rbag \hbar \boldsymbol{e}} \ \! | \ \! [m]_{\boldsymbol{e}} \in M_A, \boldsymbol{f} \in \mathcal{T} \ \! \}$; 

\item  $\lambda_{!A} ([m]_{\Lbag \boldsymbol{f} \Rbag \hbar \boldsymbol{e}}) \stackrel{\mathrm{df. }}{=} \lambda_A([m]_{\boldsymbol{e}})$;

\item $\star \vdash_{!A} [m]_{\Lbag \boldsymbol{f} \Rbag \hbar \boldsymbol{e}} \stackrel{\mathrm{df. }}{\Leftrightarrow} \star \vdash_A [m]_{\boldsymbol{e}}$; 

\item $[m]_{\Lbag \boldsymbol{f} \Rbag \hbar \boldsymbol{e}} \vdash_{!A} [m']_{\Lbag \boldsymbol{f'} \Rbag \hbar \boldsymbol{e'}} \stackrel{\mathrm{df. }}{\Leftrightarrow} \boldsymbol{f} = \boldsymbol{f'} \wedge [m]_{\boldsymbol{e}} \vdash_A [m']_{\boldsymbol{e'}}$;

\item  $\Delta_{!A} ([m]_{\Lbag \boldsymbol{f} \Rbag \hbar \boldsymbol{e}}) \stackrel{\mathrm{df. }}{=} [m']_{\Lbag \boldsymbol{f} \Rbag \hbar \boldsymbol{e}}$, where $\Delta_A([m]_{\boldsymbol{e}}) = [m']_{\boldsymbol{e}}$;

\item $P_{!A} \stackrel{\mathrm{df. }}{=} \{ \boldsymbol{s} \in \mathscr{L}_{!A} \mid \forall \boldsymbol{f} \in \mathcal{T} . \ \! \boldsymbol{s} \upharpoonright \boldsymbol{f} \in P_A \wedge (\boldsymbol{s} \upharpoonright \boldsymbol{f} \neq \boldsymbol{\epsilon} \Rightarrow \forall \boldsymbol{g} \in \mathcal{T} . \ \! \boldsymbol{s} \upharpoonright \boldsymbol{g} \neq \boldsymbol{\epsilon} \Rightarrow \mathit{ede}(\boldsymbol{f}) \neq \mathit{ede}(\boldsymbol{g})) \ \! \}$, where $\boldsymbol{s} \upharpoonright \boldsymbol{f}$ is the j-subsequence of $\boldsymbol{s}$ that consists of moves of the form $[m]_{\Lbag \boldsymbol{f} \Rbag \hbar \boldsymbol{e}}$ changed into $[m]_{\boldsymbol{e}}$.


\end{itemize}
\end{definition}
Thus, our exponential $!A$ is essentially a slight modification of the one in \cite{mccusker1998games,hyland1997game} which generalizes moves $[m]_{\underline{i} \hbar \boldsymbol{e}}$ to $[m]_{\Lbag \boldsymbol{f} \Rbag \hbar \boldsymbol{e}}$, where $[m]_{\boldsymbol{e}} \in M_A$, $i \in \mathbb{N}$ and $\boldsymbol{f} \in \mathcal{T}$.
By the condition on positions of $!A$, an element $\boldsymbol{f}$ in an outer tag $[\_]_{\Lbag \boldsymbol{f} \Rbag \hbar \boldsymbol{e}}$ that represents a natural number $i \in \mathbb{N}$, i.e., $\mathit{ede}(\boldsymbol{f}) = i$, is unique in each $\boldsymbol{s} \in P_{!A}$.

\begin{notation}
We often write $A \Rightarrow B$ for the linear implication $!A \multimap B$ for any games $A$ and $B$, which we call the \emph{\bfseries implication} or \emph{\bfseries function space} from $A$ to $B$.
The constructions $\multimap$ and $\Rightarrow$ are both right associative.
\end{notation}

\begin{example}
Some typical plays in the exponential $! \boldsymbol{2}$ are as follows:
\begin{center}
\begin{tabular}{ccccc}
$! \boldsymbol{2}$ &&&& $!\boldsymbol{2}$ \\ \cline{1-1} \cline{5-5}
\tikzmark{ExC1} $[\hat{q}]_{\Lbag \underline{10} \Rbag \hbar}$&&&&\tikzmark{ExC3} $[\hat{q}]_{\Lbag \underline{2} \hbar \underline{3} \hbar \underline{5} \Rbag \hbar}$ \\
\tikzmark{ExD1} $[\mathit{tt}]_{\Lbag \underline{10} \Rbag \hbar}$&&&&\tikzmark{ExD3} $[\mathit{tt}]_{\Lbag \underline{2} \hbar \underline{3} \hbar \underline{5} \Rbag \hbar}$ \\
\tikzmark{ExC2} $[\hat{q}]_{\Lbag \underline{100} \Rbag \hbar}$&&&&\tikzmark{ExC4} $[\hat{q}]_{\Lbag \Lbag \underline{2} \hbar \underline{3} \Rbag \hbar \underline{5} \Rbag \hbar}$ \\
\tikzmark{ExD2} $[\mathit{ff}]_{\Lbag \underline{100} \Rbag \hbar}$&&&&\tikzmark{ExD4} $[\mathit{tt}]_{\Lbag \Lbag \underline{2} \hbar \underline{3} \Rbag \hbar \underline{5} \Rbag \hbar}$ 
\end{tabular}
\begin{tikzpicture}[overlay, remember picture, yshift=.25\baselineskip]
\draw [->] ({pic cs:ExD1}) [bend left] to ({pic cs:ExC1});
\draw [->] ({pic cs:ExD2}) [bend left] to ({pic cs:ExC2});
\draw [->] ({pic cs:ExD3}) [bend left] to ({pic cs:ExC3});
\draw [->] ({pic cs:ExD4}) [bend left] to ({pic cs:ExC4});
\end{tikzpicture}
\end{center}
\end{example}

Similarly to the case of pairing, exponential is generalized in \cite{yamada2016dynamic}: Given a game $G$ such that $\mathcal{H}^\omega(G) \trianglelefteqslant \ !A \multimap B$ for some normalized games $A$ and $B$, there is the \emph{promotion} $G^\dagger$ of $G$ such that $\mathcal{H}^\omega(G^\dagger) \trianglelefteqslant \ !A \multimap \ !B$.
In fact, promotion is a generalization of exponential for $(!T \multimap B)^\dagger \simeq \ !B$ holds for any normalized game $B$, where note that $!T \multimap B \simeq B$; see \cite{yamada2016dynamic} for the proof.
Promotion of games has been defined in \cite{yamada2016dynamic} because a morphism $A \to B$ in the bicategory $\mathcal{DG}$ is a strategy $\phi : G$ such that $\mathcal{H}^\omega(G) \trianglelefteqslant \ !A \multimap B$, and therefore it is necessary to take a generalized promotion $\phi^\dagger$ (Definition~\ref{DefPromotionOfStrategies}) for composition of strategies in $\mathcal{DG}$, whose ambient game is $G^\dagger$.
The promotion $G^\dagger$ is simply given by:
\begin{itemize}

\item $M_{G^\dagger} \stackrel{\mathrm{df. }}{=} ((M_G \setminus M_{!A}) \times \mathbb{N}) + M_{!A}$;

\item $\lambda_{G^\dagger} : ((m, i) \in (M_G \setminus M_{!A}) \times \mathbb{N}) \mapsto \lambda_G(m), ((a, j) \in M_{!A}) \mapsto \lambda_G(a, j)$;

\item $\star \vdash_{G^\dagger} (m, i) \stackrel{\mathrm{df. }}{\Leftrightarrow} \star \vdash_G m$ for all $i \in \mathbb{N}$;

\item $(m, i) \vdash_{G^\dagger} (n, j) \stackrel{\mathrm{df. }}{\Leftrightarrow} (i = j \wedge m, n \in M_G \setminus M_{!A} \wedge m \vdash_G n) \\ \vee (i = j \wedge m \vdash_A n) \vee (m \in M_G \setminus M_{!A} \wedge (n, j) \in M_{!A} \wedge m \vdash_G (n, j))$;

\item $\Delta_{G^\dagger} : (m, i) \mapsto (\Delta_G(m), i)$;

\item $P_{G^\dagger} \stackrel{\mathrm{df. }}{=} \{ \boldsymbol{s} \in \mathscr{L}_{G^\dagger} \mid \forall i \in \mathbb{N} . \ \! \boldsymbol{s} \upharpoonright i \in P_G \ \! \}$, where $\boldsymbol{s} \upharpoonright i$ is the j-subsequence of $\boldsymbol{s}$ that consists of moves of the form $(m, i)$ with $m \in M_G \setminus M_{!A}$ or $(a, \langle i, j \rangle)$ with $a \in M_A$ and $j \in \mathbb{N}$ but changed into $m$ or $(a, j)$, respectively.

\end{itemize}

Let us formalize `tags' of moves of $G^\dagger$ as follows:
\begin{itemize}

\item We do not change tags of moves of $G$ coming from $!A$, i.e., ones of the form $[(a, \mathscr{W})]_{\Lbag \boldsymbol{f} \Rbag \hbar \boldsymbol{e}}$;

\item We duplicate moves of $G$ coming from $B$, i.e., ones of the form $[(b, \mathscr{E})]_{\boldsymbol{e}}$, as $[(b, \mathscr{E})]_{\Lbag \boldsymbol{f} \Rbag \hbar \boldsymbol{e}}$ for each $\boldsymbol{f} \in \mathcal{T}$;

\item We duplicate internal moves $[m]_{\boldsymbol{e}}$ of $G$ as $[(m, \mathscr{S})]_{\Lbag \boldsymbol{f} \Rbag \hbar \boldsymbol{e}}$ for each $\boldsymbol{f} \in \mathcal{T}$.

\end{itemize}
where again this way of formalizing `tags' is far from canonical. 

Then, the labeling function, the enabling relation and the dummy function of $G^\dagger$ are again defined by pattern matching on inner tags in the obvious manner, for which unlike the case of pairing we do not use peeling or attributes as promotion is simpler.
Also, positions of $G^\dagger$ are given by a straightforward generalization of those of exponential defined in Definition~\ref{DefExponential}. 

Formally, we define promotion on games as follows: 
\begin{definition}[Promotion of games \cite{yamada2016dynamic}]
Given a game $G$ with $\mathcal{H}^\omega(G) \trianglelefteqslant \ !A \multimap B$ for some normalized games $A$ and $B$, its \emph{\bfseries promotion} $G^\dagger$ is given by:
\begin{itemize}

\item $M_{G^\dagger} \stackrel{\mathrm{df. }}{=} \{ [(a, \mathscr{W})]_{\Lbag \boldsymbol{f} \Rbag \hbar \boldsymbol{e}} \ \! | \ \! [(a, \mathscr{W})]_{\Lbag \boldsymbol{f} \Rbag \hbar \boldsymbol{e}} \in M_G, [a]_{\boldsymbol{e}} \in M_A \} \\
\cup \{ [(b, \mathscr{E})]_{\Lbag \boldsymbol{f} \Rbag \hbar \boldsymbol{e}} \ \! | \ \! [(b, \mathscr{E})]_{\boldsymbol{e}} \in M_G, [b]_{\boldsymbol{e}} \in M_B, \boldsymbol{f} \in \mathcal{T} \ \! \} \\ \cup \{ [(m, \mathscr{S})]_{\Lbag \boldsymbol{f} \Rbag \hbar \boldsymbol{e}} \ \! | \ \! [m]_{\boldsymbol{e}} \in M_G^{\mathsf{Int}}, \boldsymbol{f} \in \mathcal{T} \ \! \}$;

\item $\lambda_{G^\dagger} : [(a, \mathscr{W})]_{\Lbag \boldsymbol{f} \Rbag \hbar \boldsymbol{e}} \mapsto \lambda_G([(a, \mathscr{W})]_{\Lbag \boldsymbol{f} \Rbag \hbar \boldsymbol{e}}), [(b, \mathscr{E})]_{\Lbag \boldsymbol{f} \Rbag \hbar \boldsymbol{e}} \mapsto \lambda_G([(b, \mathscr{E})]_{\boldsymbol{e}}), [(m, \mathscr{S})]_{\Lbag \boldsymbol{f} \Rbag \hbar \boldsymbol{e}} \mapsto \lambda_G([m]_{\boldsymbol{e}})$;

\item $\star \vdash_{G^\dagger} [(b, X)]_{\Lbag \boldsymbol{f} \Rbag \hbar \boldsymbol{e}} \stackrel{\mathrm{df. }}{\Leftrightarrow} X = \mathscr{E} \wedge \star \vdash_B [b]_{\boldsymbol{e}}$;

\item $[(m, X)]_{\Lbag \boldsymbol{f} \Rbag \hbar \boldsymbol{e}} \vdash_{G^\dagger} [(n, Y)]_{\Lbag \boldsymbol{g} \Rbag \hbar \boldsymbol{h}} \stackrel{\mathrm{df. }}{\Leftrightarrow} (X = \mathscr{W} = Y \wedge [(m, \mathscr{W})]_{\Lbag \boldsymbol{f} \Rbag \hbar \boldsymbol{e}} \vdash_{G} [(n, \mathscr{W})]_{\Lbag \boldsymbol{g} \Rbag \hbar \boldsymbol{h}}) \vee (X = \mathscr{E} \wedge Y = \mathscr{W} \wedge \star \vdash_B [m]_{\boldsymbol{e}} \wedge \star \vdash_A [n]_{\boldsymbol{h}}) \\ \vee (X = \mathscr{E} = Y \wedge [(m, \mathscr{E})]_{\boldsymbol{e}} \vdash_{G} [(n, \mathscr{E})]_{\boldsymbol{h}}) \vee (X = \mathscr{S} = Y \wedge [m]_{\boldsymbol{e}} \vdash_{G} [n]_{\boldsymbol{h}})$;

\item $\Delta_{G^\dagger} : [(m, \mathscr{S})]_{\Lbag \boldsymbol{f} \Rbag \hbar \boldsymbol{e}} \stackrel{\mathrm{df. }}{=} [(m', \mathscr{S})]_{\Lbag \boldsymbol{f} \Rbag \hbar \boldsymbol{e}}$, where $\Delta_G([m]_{\boldsymbol{e}}) = [m']_{\boldsymbol{e}}$;

\item $P_{G^\dagger} \stackrel{\mathrm{df. }}{=} \{ \boldsymbol{s} \in \mathscr{L}_{G^\dagger} \mid \forall \boldsymbol{f} \in \mathcal{T} . \ \! \boldsymbol{s} \upharpoonright \boldsymbol{f} \in P_G \wedge (\boldsymbol{s} \upharpoonright \boldsymbol{f} \neq \boldsymbol{\epsilon} \Rightarrow \forall \boldsymbol{g} \in \mathcal{T} . \ \! \boldsymbol{s} \upharpoonright \boldsymbol{g} \neq \boldsymbol{\epsilon} \Rightarrow \mathit{ede}(\boldsymbol{f}) \neq \mathit{ede}(\boldsymbol{g})) \ \! \}$, where $\boldsymbol{s} \upharpoonright \boldsymbol{f}$ is a j-subsequence of $\boldsymbol{s}$ that consists of moves of the form $[(a, \mathscr{W})]_{\Lbag \Lbag \boldsymbol{f} \Rbag \hbar \Lbag \boldsymbol{g} \Rbag \Rbag \hbar \boldsymbol{e}}$ with $[a]_{\boldsymbol{e}} \in M_A$, $[(b, \mathscr{E})]_{\Lbag \boldsymbol{f} \Rbag \hbar \boldsymbol{e}}$ with $[b]_{\boldsymbol{e}} \in M_B$ or $[(m, \mathscr{S})]_{\Lbag \boldsymbol{f} \Rbag \hbar \boldsymbol{e}}$ with $[m]_{\boldsymbol{e}} \in M_G^{\mathsf{Int}}$ changed into $[(a, \mathscr{W})]_{\Lbag \boldsymbol{g} \Rbag \hbar \boldsymbol{e}}$, $[(b, \mathscr{E})]_{\boldsymbol{e}}$ and $[m]_{\boldsymbol{e}}$, respectively.

\end{itemize}

\end{definition}

Now, let us recall \emph{concatenation} of games, which was first introduced in \cite{yamada2016dynamic}:
Given games $J$ and $K$ such that $\mathcal{H}^\omega(J) \trianglelefteqslant A \multimap B$ and $\mathcal{H}^\omega(K) \trianglelefteqslant B \multimap C$ for some normalized games $A$, $B$ and $C$, their concatenation $J \ddagger K$ is given by:
\begin{itemize}

\item $M_{J \ddagger K} \stackrel{\mathrm{df. }}{=} M_J + M_K$;

\item $\lambda_{J \ddagger K} \stackrel{\mathrm{df. }}{=} [\lambda_J \downharpoonright M_{B^{[1]}}, \lambda^{+\mu}_J \upharpoonright M_{B^{[1]}}, \lambda^{+\mu}_K \upharpoonright M_{B^{[2]}}, \lambda_K \downharpoonright M_{B^{[2]}}]$, where $\lambda_G^{+ \mu} \stackrel{\mathrm{df. }}{=}  \langle \lambda_G^{\mathsf{OP}}, \lambda_G^{\mathsf{QA}}, n \mapsto \lambda_G^{\mathbb{N}} (n) + \mu \rangle$ ($G$ is $J$ or $K$), and $\mu \stackrel{\mathrm{df. }}{=} \mathsf{Max}(\mu(J), \mu(K)) + 1$;

\item $\star \vdash_{J \ddagger K} m \stackrel{\mathrm{df. }}{\Leftrightarrow} \star \vdash_K m$;

\item $m \vdash_{J \ddagger K} n \ (m \neq \star) \stackrel{\mathrm{df. }}{\Leftrightarrow} m \vdash_J n \vee m \vdash_K n \vee (\star \vdash_{B^{[2]}} m \wedge \star \vdash_{B^{[1]}} n)$;

\item $\Delta_{J \ddagger K} \stackrel{\mathrm{df. }}{=} [\Delta_J, \Delta_K] \upharpoonright M_{J \ddagger K}$;

\item $P_{J \ddagger K} \stackrel{\mathrm{df. }}{=} \{ \boldsymbol{s} \in \mathscr{J}_{J \ddagger K} \mid \boldsymbol{s} \upharpoonright J \in P_J, \boldsymbol{s} \upharpoonright K \in P_K, \boldsymbol{s}  \upharpoonright B^{[1]}, B^{[2]} \in \mathit{pr}_B \ \! \}$, where $\mathit{pr}_B \stackrel{\mathrm{df. }}{=} \{ \bm{s} \in P_{B^{[1]} \multimap B^{[2]}} \mid \forall \bm{t} \preceq{\bm{s}}. \ \mathsf{Even}(\bm{t}) \Rightarrow \bm{t} \upharpoonright B^{[1]} = \bm{t} \upharpoonright B^{[2]} \}$.

\end{itemize}
Note that moves of $B$ (in $J$ or $K$) become internal in $J \ddagger K$, and therefore they would be deleted by the hiding operation $\mathcal{H}$ on games.

Concatenation corresponds to \emph{composition without hiding} in the introduction, and it plays a central role in \cite{yamada2016dynamic}.
We shall see later that \emph{concatenation $\sigma \ddagger \tau$ of strategies} $\sigma : J$ and $\tau : K$ (Definition~\ref{DefConcatenationAndCompositionOfStrategies}), where $\mathcal{H}^\omega(J) \trianglelefteqslant A \multimap B$ and $\mathcal{H}^\omega(K) \trianglelefteqslant B \multimap C$ for some normalized games $A$, $B$ and $C$, forms a strategy on the game $J \ddagger K$, and also that $\mathcal{H}^\omega(\sigma) : A \multimap B$, $\mathcal{H}^\omega(\tau) : B \multimap C$ and $\mathcal{H}^\omega(\sigma) ; \mathcal{H}^\omega(\tau) = \mathcal{H}^\omega(\sigma \ddagger \tau) : \mathcal{H}^\omega(J \ddagger K) \trianglelefteqslant A \multimap C$, whence $\mathcal{H}^\omega(\sigma) ; \mathcal{H}^\omega(\tau) : A \multimap C$ (for $\phi : G \trianglelefteqslant H$ implies $\phi : H$ for any strategy $\phi$ and games $G$ and $H$; see \cite{yamada2016dynamic} for the proof).
Note that $\mathcal{H}^\omega(\sigma) ; \mathcal{H}^\omega(\tau) : \mathcal{H}^\omega(J \ddagger K) \trianglelefteqslant A \multimap C$ becomes just the familiar relation $\sigma ; \tau : A \multimap C$ when $\sigma : J = A \multimap B$ and $\tau : K = B \multimap C$; concatenation is to capture a generalization of this phenomenon.

Let us formalize `tags' for concatenation as follows:
\begin{itemize}

\item We do not change moves of $A$ or $C$, i.e., ones of the form $[(a, \mathscr{W})]_{\boldsymbol{e}} \in M_J^{\mathsf{Ext}}$ or $[(c, \mathscr{E})]_{\boldsymbol{f}} \in M_K^{\mathsf{Ext}}$;

\item We change moves of $B$ in $J$, i.e., external ones of the form $[(b, \mathscr{E})]_{\boldsymbol{g}}$, into $[((b, \mathscr{E}), \mathscr{S})]_{\boldsymbol{g}}$;

\item We change moves of $B$ in $K$, i.e., external ones of the form $[(b, \mathscr{W})]_{\boldsymbol{g}}$, into $[((b, \mathscr{W}), \mathscr{N})]_{\boldsymbol{g}}$;

\item We change internal moves $[m]_{\boldsymbol{l}}$ of $J$ into $[(m, \mathscr{S})]_{\boldsymbol{l}}$;

\item We change internal moves $[n]_{\boldsymbol{r}}$ of $K$ into $[(n, \mathscr{N})]_{\boldsymbol{r}}$.

\end{itemize}

Then again, the labeling function, the enabling relation and the dummy function of $J \ddagger K$ are defined by the obvious pattern matching, and positions of $J \ddagger K$ are defined as usual.

Formally, concatenation of games is defined as follows, where it should be now clear how the peeling $\mathit{peel}_{J \ddagger K}$ and the attributes $\mathit{att}_{J \ddagger K}$ work:
\begin{definition}[Concatenation of games \cite{yamada2016dynamic}]
\label{DefConcatenationOfGames}
Given games $J$, $K$, $A$, $B$ and $C$ with $\mathcal{H}^\omega(J) \trianglelefteqslant A \multimap B$ and $\mathcal{H}^\omega(K) \trianglelefteqslant B \multimap C$, the \emph{\bfseries concatenation} $J \ddagger K$ of $J$ and $K$ is defined by:
\begin{itemize}

\item $M_{J \ddagger K} \stackrel{\mathrm{df. }}{=} \{ [(a, \mathscr{W})]_{\boldsymbol{e}} \mid [(a, \mathscr{W})]_{\boldsymbol{e}} \in M_J^{\mathsf{Ext}}, [a]_{\boldsymbol{e}} \in M_A \} \\ \cup \{ [(c, \mathscr{E})]_{\boldsymbol{f}} \mid [(c, \mathscr{E})]_{\boldsymbol{f}} \in M_K^{\mathsf{Ext}}, [c]_{\boldsymbol{f}} \in M_C \} \\ \cup \{ [((b, \mathscr{E}), \mathscr{S})]_{\boldsymbol{g}} \mid [(b, \mathscr{E})]_{\boldsymbol{g}} \in M_J^{\mathsf{Ext}}, [b]_{\boldsymbol{g}} \in M_B \} \\ \cup \{ [((b, \mathscr{W}), \mathscr{N})]_{\boldsymbol{g}} \mid [(b, \mathscr{W})]_{\boldsymbol{g}} \in M_K^{\mathsf{Ext}}, [b]_{\boldsymbol{g}} \in M_B \} \\ \cup \{ [(m, \mathscr{S})]_{\boldsymbol{l}} \mid [m]_{\boldsymbol{l}} \in M_J^{\mathsf{Int}} \} \cup \{ [(n, \mathscr{N})]_{\boldsymbol{r}} \mid [n]_{\boldsymbol{r}} \in M_K^{\mathsf{Int}} \}$;

\item $\lambda_{J \ddagger K}([(m, X)]_{\boldsymbol{e}}) \stackrel{\mathrm{df. }}{=} \begin{cases} \lambda_J^{+\mu}([m]_{\boldsymbol{e}}) &\text{if $X = \mathscr{S} \wedge \exists [b]_{\boldsymbol{e}} \in M_B . \ \! [m]_{\boldsymbol{e}} = [(b, \mathscr{E})]_{\boldsymbol{e}} \in M_J^{\mathsf{Ext}}$;} \\ \lambda_J([m]_{\boldsymbol{e}}) &\text{if $X = \mathscr{W} \vee (X = \mathscr{S} \wedge [m]_{\boldsymbol{e}} \in M_J^{\mathsf{Int}})$;} \\ \lambda_K^{+\mu}([m]_{\boldsymbol{e}}) &\text{if $X = \mathscr{N} \wedge \exists [b]_{\boldsymbol{e}} \in M_B . \ \! [m]_{\boldsymbol{e}} = [(b, \mathscr{W})]_{\boldsymbol{e}} \in M_K^{\mathsf{Ext}}$;} \\ \lambda_K([m]_{\boldsymbol{e}}) &\text{if $X = \mathscr{E} \vee (X = \mathscr{N} \wedge [m]_{\boldsymbol{e}} \in M_K^{\mathsf{Int}});$} \end{cases}$ 

\item $\star \vdash_{J \ddagger K} [(m, X)]_{\boldsymbol{e}} \stackrel{\mathrm{df. }}{\Leftrightarrow} X = \mathscr{E} \wedge \star \vdash_C [m]_{\boldsymbol{e}}$;

\item $[(m, X)]_{\boldsymbol{e}} \vdash_{J \ddagger K} [(n, Y)]_{\boldsymbol{f}} \stackrel{\mathrm{df. }}{\Leftrightarrow} \begin{cases} \begin{aligned} & (\mathit{att}_{J \ddagger K}([(m, X)]_{\boldsymbol{e}}) = J = \mathit{att}_{J \ddagger K}([(n, Y)]_{\boldsymbol{f}}) \\ & \ \wedge \mathit{peel}_{J \ddagger K}([(m, X)]_{\boldsymbol{e}}) \vdash_J\mathit{peel}_{J \ddagger K}([(n, Y)]_{\boldsymbol{f}})) \\ &\vee (\mathit{att}_{J \ddagger K}([(m, X)]_{\boldsymbol{e}}) = K = \mathit{att}_{J \ddagger K}([(n, Y)]_{\boldsymbol{f}}) \\ & \ \ \ \ \wedge \mathit{peel}_{J \ddagger K}([(m, X)]_{\boldsymbol{e}})  \vdash_K \mathit{peel}_{J \ddagger K}([(n, Y)]_{\boldsymbol{f}})) \\ &\vee (X = \mathscr{N} \wedge Y = \mathscr{S} \wedge \exists [b]_{\boldsymbol{e}}, [b']_{\boldsymbol{f}} \in M_B^{\mathsf{Init}} . \\ & \ \ \ \ \ m = (b, \mathscr{W}) \wedge n = (b, \mathscr{E})) \end{aligned} \end{cases}$ where the function $\mathit{att}_{J \ddagger K} : M_{J \ddagger K} \to \{ J, K \}$ is defined by $[(a, \mathscr{W})]_{\boldsymbol{e}} \mapsto J$, $[(m, \mathscr{S})]_{\boldsymbol{l}} \mapsto J$, $[((b, \mathscr{E}), \mathscr{S})]_{\boldsymbol{g}} \mapsto J$, $[(c, \mathscr{E})]_{\boldsymbol{f}} \mapsto K$, $[(n, \mathscr{N})]_{\boldsymbol{r}} \mapsto K$, $[((b, \mathscr{W}), \mathscr{N})]_{\boldsymbol{g}} \mapsto K$, and the function $\mathit{peel}_{J \ddagger K} : M_{J \ddagger K} \to M_J \cup M_K$ by $[(a, \mathscr{W})]_{\boldsymbol{e}} \mapsto [(a, \mathscr{W})]_{\boldsymbol{e}}$, $[(c, \mathscr{E})]_{\boldsymbol{f}} \mapsto [(c, \mathscr{E})]_{\boldsymbol{f}}$, $[((b, \mathscr{E}), \mathscr{S})]_{\boldsymbol{g}} \mapsto [(b, \mathscr{E})]_{\boldsymbol{g}}$, $[((b, \mathscr{W}), \mathscr{N})]_{\boldsymbol{g}} \mapsto [(b, \mathscr{W})]_{\boldsymbol{g}}$, $[(m, \mathscr{S})]_{\boldsymbol{l}} \mapsto [m]_{\boldsymbol{l}}$, $[(n, \mathscr{N})]_{\boldsymbol{r}} \mapsto [n]_{\boldsymbol{r}}$;

\item $\Delta_{J \ddagger K} ([(m, X)]_{\boldsymbol{e}}) \stackrel{\mathrm{df. }}{=} \begin{cases} [(m', \mathscr{S})]_{\boldsymbol{e}} &\text{if $X = \mathscr{S}$ and $\Delta_J([m]_{\boldsymbol{e}}) = [m']_{\boldsymbol{e}}$;} \\ [(m'', \mathscr{N})]_{\boldsymbol{e}} &\text{if $X = \mathscr{N}$ and $\Delta_K([m]_{\boldsymbol{e}}) = [m'']_{\boldsymbol{e}}$;} \\ [((b, \mathscr{W}), \mathscr{N})]_{\boldsymbol{e}} &\text{if $X = \mathscr{S}$, $\Delta_J([m]_{\boldsymbol{e}}) \uparrow$ and $m = (b, \mathscr{E})$;} \\ [((b, \mathscr{E}), \mathscr{S})]_{\boldsymbol{e}} &\text{if $X = \mathscr{N}$, $\Delta_K([m]_{\boldsymbol{e}}) \uparrow$ and $m = (b, \mathscr{W})$;} \end{cases}$

\item $P_{J \ddagger K} \stackrel{\mathrm{df. }}{=} \{ \boldsymbol{s} \in \mathscr{J}_{J \ddagger K} \mid \boldsymbol{s} \upharpoonright J \in P_J, \boldsymbol{s} \upharpoonright K \in P_K, \boldsymbol{s} \upharpoonright B^{[1]}, B^{[2]} \in \mathit{pr}_B \}$, where $\boldsymbol{s} \upharpoonright J$ (resp. $\boldsymbol{s} \upharpoonright K$) is the j-subsequence of $\boldsymbol{s}$ that consists of moves $m$ such that $\mathit{att}_{J \ddagger K}(m) = J$ (resp. $\mathit{att}_{J \ddagger K}(m) = K$) changed into $\mathit{peel}_{J \ddagger K}(m)$,  $B^{[1]}$ and $B^{[2]}$ are two copies of $B$, $\boldsymbol{s} \upharpoonright B^{[1]}, B^{[2]}$ is the j-subsequence of $\boldsymbol{s}$ that consists of moves of $B^{[1]}$ or $B^{[2]}$, i.e., moves $[((b, X), Y)]_{\boldsymbol{e}}$ such that $[b]_{\boldsymbol{e}} \in M_B \wedge ((X = \mathscr{E} \wedge Y = \mathscr{S}) \vee (X = \mathscr{W} \wedge Y = \mathscr{N}))$ changed into $[(b, \overline{X})]_{\boldsymbol{e}}$, for which $\overline{\mathscr{E}} \stackrel{\mathrm{df. }}{=} \mathscr{W}$ and $\overline{\mathscr{W}} \stackrel{\mathrm{df. }}{=} \mathscr{E}$, and $\mathit{pr}_B \stackrel{\mathrm{df. }}{=} \{ \boldsymbol{t} \in P_{B^{[1]} \multimap B^{[2]}} \ \! | \! \ \forall \boldsymbol{u} \preceq{\boldsymbol{t}}. \ \mathsf{Even}(\boldsymbol{u}) \Rightarrow \boldsymbol{u} \upharpoonright \mathscr{W} = \boldsymbol{u} \upharpoonright \mathscr{E} \ \! \}$.


\end{itemize}

\end{definition}

\begin{example}
Typical plays of the concatenation $(\mathcal{N} \multimap \mathcal{N}) \ddagger (\mathcal{N} \multimap \mathcal{N})$ are:
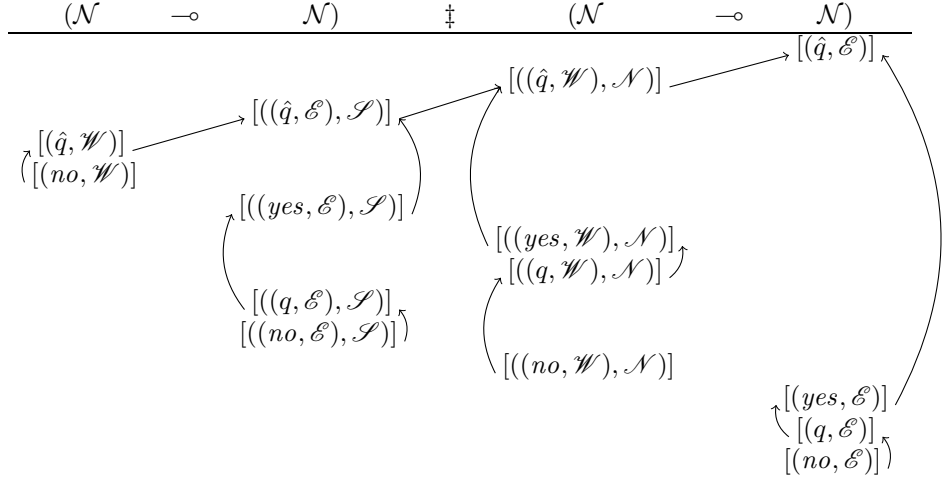
\begin{center}
\begin{tabular}{ccccccc}
$(\mathcal{N}$ & $\multimap$ & $\mathcal{N})$ & $\ddagger$ & $(\mathcal{N}$ & $\multimap$ & $\mathcal{N})$ \\ \cline{1-7} 
&&&&&&\tikzmark{ConC1} $[(\hat{q}, \mathscr{E})]$ \tikzmark{ConC10} \\
&&&&\tikzmark{ConC5} $[((\hat{q}, \mathscr{W}), \mathscr{N})]$ \tikzmark{ConD1}&& \\
&&\tikzmark{ConC2} $[((\hat{q}, \mathscr{E}), \mathscr{S})]$ \tikzmark{ConC4}&&&& \\
\tikzmark{ConC3} $[(\hat{q}, \mathscr{W})]$ \tikzmark{ConD2}&&&&&& \\
\tikzmark{ConD3} $[(\mathit{no}, \mathscr{W})]$&&&&&& \\
&&\tikzmark{ConC7} $[((\mathit{yes}, \mathscr{E}), \mathscr{S})]$ \tikzmark{ConD4}&&&& \\
&&&&\tikzmark{ConD5} $[((\mathit{yes}, \mathscr{W}), \mathscr{N})]$ \tikzmark{ConC6}&& \\
&&&&\tikzmark{ConC9} $[((q, \mathscr{W}), \mathscr{N})]$ \tikzmark{ConD6}&& \\
&&\tikzmark{ConD7} $[((q, \mathscr{E}), \mathscr{S})]$ \tikzmark{ConC8}&&&& \\
&&$[((\mathit{no}, \mathscr{E}), \mathscr{S})]$ \tikzmark{ConD8}&&&& \\
&&&&\tikzmark{ConD9} $[((\mathit{no}, \mathscr{W}), \mathscr{N})]$&& \\
&&&&&&\tikzmark{ConC11} $[(\mathit{yes}, \mathscr{E})]$ \tikzmark{ConD10} \\
&&&&&&\tikzmark{ConD11} $[(q, \mathscr{E})]$ \tikzmark{ConC12} \\
&&&&&&$[(\mathit{no}, \mathscr{E})]$ \tikzmark{ConD12}
\end{tabular}
\begin{tikzpicture}[overlay, remember picture, yshift=.25\baselineskip]
\draw [->] ({pic cs:ConD1}) to ({pic cs:ConC1});
\draw [->] ({pic cs:ConD2}) to ({pic cs:ConC2});
\draw [->] ({pic cs:ConD3}) [bend left] to ({pic cs:ConC3});
\draw [->] ({pic cs:ConD4}) [bend right] to ({pic cs:ConC4});
\draw [->] ({pic cs:ConD5}) [bend left] to ({pic cs:ConC5});
\draw [->] ({pic cs:ConD6}) [bend right] to ({pic cs:ConC6});
\draw [->] ({pic cs:ConD7}) [bend left] to ({pic cs:ConC7});
\draw [->] ({pic cs:ConD8}) [bend right] to ({pic cs:ConC8});
\draw [->] ({pic cs:ConD9}) [bend left] to ({pic cs:ConC9});
\draw [->] ({pic cs:ConD10}) [bend right] to ({pic cs:ConC10});
\draw [->] ({pic cs:ConD11}) [bend left] to ({pic cs:ConC11});
\draw [->] ({pic cs:ConD12}) [bend right] to ({pic cs:ConC12});
\draw [->] ({pic cs:ConC4}) to ({pic cs:ConC5});
\end{tikzpicture}
\end{center}
\end{example}


Finally, let us recall the rather trivial \emph{currying} $\Lambda$ and \emph{uncurrying} $\Lambda^{\circleddash}$ \cite{abramsky1999game}).
Roughly, currying $\Lambda$ is a generalization of the map $A \otimes B \multimap C \mapsto A \multimap (B \multimap C)$, and uncurrying $\Lambda^\circleddash$ is a generalization of the inverse $A \multimap (B \multimap C) \mapsto A \otimes B \multimap C$, where $A$, $B$ and $C$ are arbitrary normalized games.
For their simplicity, let us skip their informal definitions and just present the formal ones: 
\begin{definition}[Currying of games \cite{abramsky1999game}]
If a game $G$ satisfies $\mathcal{H}^\omega(G) \trianglelefteqslant A \otimes B \multimap C$ for some normalized games $A$, $B$ and $C$, then the \emph{\bfseries currying} $\Lambda(G)$ of $G$ is given by:
\begin{itemize}

\item $M_{\Lambda(G)} \stackrel{\mathrm{df. }}{=} \{ [(a, \mathscr{W})]_{\boldsymbol{e}} \mid [((a, \mathscr{W}), \mathscr{W})]_{\boldsymbol{e}} \in M_G^{\mathsf{Ext}}, [a]_{\boldsymbol{e}} \in M_A \} \\ \cup \{ [((b, \mathscr{W}), \mathscr{E})]_{\boldsymbol{f}} \mid [((b, \mathscr{E}), \mathscr{W})]_{\boldsymbol{f}} \in M_G^{\mathsf{Ext}}, [b]_{\boldsymbol{f}} \in M_B \} \\ \cup \{ [((c, \mathscr{E}), \mathscr{E})]_{\boldsymbol{g}} \mid [(c, \mathscr{E})]_{\boldsymbol{g}} \in M_G^{\mathsf{Ext}}, [c]_{\boldsymbol{g}} \in M_C \} \cup \{ [(m, \mathscr{N})]_{\boldsymbol{h}} \mid [m]_{\boldsymbol{h}} \in M_G^{\mathsf{Int}} \}$;

\item $\lambda_{\Lambda(G)} : [(a, \mathscr{W})]_{\boldsymbol{e}} \mapsto \lambda_G([((a, \mathscr{W}), \mathscr{W})]_{\boldsymbol{e}})$, $[((b, \mathscr{W}), \mathscr{E})]_{\boldsymbol{f}} \mapsto \lambda_G([((b, \mathscr{E}), \mathscr{W})]_{\boldsymbol{f}})$, $[((c, \mathscr{E}), \mathscr{E})]_{\boldsymbol{g}} \mapsto \lambda_G([(c, \mathscr{E})]_{\boldsymbol{g}})$, $[(m, \mathscr{N})]_{\boldsymbol{h}} \mapsto \lambda_G([m]_{\boldsymbol{h}})$;

\item $\star \vdash_{\Lambda(G)} [m]_{\boldsymbol{e}} \stackrel{\mathrm{df. }}{\Leftrightarrow} \exists [c]_{\boldsymbol{e}} \in M_C^{\mathsf{Init}} . \ \! m = ((c, \mathscr{E}), \mathscr{E})$;

\item $[m]_{\boldsymbol{e}} \vdash_{\Lambda(G)} [n]_{\boldsymbol{f}} \stackrel{\mathrm{df. }}{\Leftrightarrow} \mathit{peel}_{\Lambda(G)}([m]_{\boldsymbol{e}}) \vdash_{G} \mathit{peel}_{\Lambda(G)}([n]_{\boldsymbol{f}})$, where the function $\mathit{peel}_{\Lambda(G)} : M_{\Lambda(G)} \to M_G$ is defined by $[(a, \mathscr{W})]_{\boldsymbol{e}} \mapsto [((a, \mathscr{W}), \mathscr{W})]_{\boldsymbol{e}}$, $[((b, \mathscr{W}), \mathscr{E})]_{\boldsymbol{f}} \mapsto [((b, \mathscr{E}), \mathscr{W})]_{\boldsymbol{f}}$, $[((c, \mathscr{E}), \mathscr{E})]_{\boldsymbol{g}} \mapsto [(c, \mathscr{E})]_{\boldsymbol{g}}$, $[(m, \mathscr{N})]_{\boldsymbol{h}} \mapsto [m]_{\boldsymbol{h}}$;

\item $\Delta_{\Lambda(G)} : [(m, \mathscr{N})]_{\boldsymbol{e}} \mapsto [(m', \mathscr{N})]_{\boldsymbol{e}}$, where $\Delta_G : [m]_{\boldsymbol{e}} \mapsto [m']_{\boldsymbol{e}}$;

\item $P_{\Lambda(G)} \stackrel{\mathrm{df. }}{=} \{ \boldsymbol{s} \in \mathscr{L}_{\Lambda(G)} \mid \mathit{peel}_{\Lambda(G)}^\ast(\boldsymbol{s}) \in P_G \}$, where the structure of justifiers in $\mathit{peel}_{\Lambda(G)}^\ast(\boldsymbol{s})$ is the same as the one in $\boldsymbol{s}$.

\end{itemize}
\end{definition}
It is easy to see that $\mathcal{H}^\omega(\Lambda(G)) \trianglelefteqslant A \multimap (B \multimap C)$ if $\mathcal{H}^\omega(G) \trianglelefteqslant A \otimes B \multimap C$, which is a generalization of the equation $\Lambda(A \otimes B \multimap C) = A \multimap (B \multimap C)$.

\begin{definition}[Uncurrying of games \cite{abramsky1999game}]
If games $H$, $A$, $B$ and $C$ satisfy $\mathcal{H}^\omega(H) \trianglelefteqslant A \multimap (B \multimap C)$, then the \emph{\bfseries uncurrying} $\Lambda^{\circleddash}(H)$ of $H$ is defined by:
\end{definition}
\begin{itemize}

\item $M_{\Lambda^\circleddash(H)} \stackrel{\mathrm{df. }}{=} \{ [((a, \mathscr{W}), \mathscr{W})]_{\boldsymbol{e}} \mid [(a, \mathscr{W})]_{\boldsymbol{e}} \in M_H^{\mathsf{Ext}}, [a]_{\boldsymbol{e}} \in M_A \} \\ \cup \{ [((b, \mathscr{E}), \mathscr{W})]_{\boldsymbol{f}} \mid [((b, \mathscr{W}), \mathscr{E})]_{\boldsymbol{f}} \in M_H^{\mathsf{Ext}}, [b]_{\boldsymbol{f}} \in M_B \} \\ \cup \{ [(c, \mathscr{E})]_{\boldsymbol{g}} \mid [((c, \mathscr{E}), \mathscr{E})]_{\boldsymbol{g}} \in M_H^{\mathsf{Ext}}, [c]_{\boldsymbol{g}} \in M_C \} \cup \{ [(m, \mathscr{S})]_{\boldsymbol{h}} \mid [m]_{\boldsymbol{h}} \in M_H^{\mathsf{Int}} \}$;

\item $\lambda_{\Lambda^\circleddash(H)} : [((a, \mathscr{W}), \mathscr{W})]_{\boldsymbol{e}} \mapsto \lambda_H([(a, \mathscr{W})]_{\boldsymbol{e}})$, $[((b, \mathscr{E}), \mathscr{W})]_{\boldsymbol{f}} \mapsto \lambda_H([((b, \mathscr{W}), \mathscr{E})]_{\boldsymbol{f}})$, $[(c, \mathscr{E})]_{\boldsymbol{g}} \mapsto \lambda_H([((c, \mathscr{E}), \mathscr{E})]_{\boldsymbol{g}}), [(m, \mathscr{S})]_{\boldsymbol{h}} \mapsto \lambda_H([m]_{\boldsymbol{h}})$;

\item $\star \vdash_{\Lambda^\circleddash(H)} [m]_{\boldsymbol{e}} \stackrel{\mathrm{df. }}{\Leftrightarrow} \exists [c]_{\boldsymbol{e}} \in M_C^{\mathsf{Init}} . \ \! m = (c, \mathscr{E})$;

\item $[m]_{\boldsymbol{e}} \vdash_{\Lambda^\circleddash(H)} [n]_{\boldsymbol{f}} \stackrel{\mathrm{df. }}{\Leftrightarrow} \mathit{peel}_{\Lambda^\circleddash(H)}([m]_{\boldsymbol{e}}) \vdash_{H} \mathit{peel}_{\Lambda^\circleddash(H)}([n]_{\boldsymbol{f}})$, where the function $\mathit{peel}_{\Lambda^\circleddash(H)} : M_{\Lambda^\circleddash(H)} \to M_H$ is defined by $[((a, \mathscr{W}), \mathscr{W})]_{\boldsymbol{e}} \mapsto [(a, \mathscr{W})]_{\boldsymbol{e}}$, $[((b, \mathscr{E}), \mathscr{W})]_{\boldsymbol{f}} \mapsto [((b, \mathscr{W}), \mathscr{E})]_{\boldsymbol{f}}$, $[(c, \mathscr{E})]_{\boldsymbol{g}} \mapsto [((c, \mathscr{E}), \mathscr{E})]_{\boldsymbol{g}}$, $[(m, \mathscr{S})]_{\boldsymbol{h}} \mapsto [m]_{\boldsymbol{h}}$;

\item $\Delta_{\Lambda^\circleddash(H)} : [(m, \mathscr{S})]_{\boldsymbol{h}} \mapsto  [(m', \mathscr{S})]_{\boldsymbol{h}}$, where $\Delta_H :  [m]_{\boldsymbol{h}} \mapsto  [m']_{\boldsymbol{h}}$;

\item $P_{\Lambda^\circleddash(H)} \stackrel{\mathrm{df. }}{=} \{ \boldsymbol{s} \in \mathscr{L}_{\Lambda^\circleddash(H)} \mid \mathit{peel}_{\Lambda^\circleddash(H)}^\ast(\boldsymbol{s}) \in P_H \}$, where the structure of justifiers in $\mathit{peel}_{\Lambda^\circleddash(H)}^\ast(\boldsymbol{s})$ is the same as the one in $\boldsymbol{s}$.

\end{itemize}
Dually to currying, $\mathcal{H}^\omega(\Lambda^\circleddash(H)) \trianglelefteqslant A \otimes B \multimap C$ if $\mathcal{H}^\omega(H) \trianglelefteqslant A \multimap (B \multimap C)$, which is a generalization of the equation $\Lambda^\circleddash(A \multimap (B \multimap C)) = A \otimes B \multimap C$.
It is easy to see that $\Lambda$ and $\Lambda^\circleddash$ are inverses to each other for normalized games, i.e., $\Lambda \circ \Lambda^\circleddash(A \multimap (B \multimap C)) = A \multimap (B \multimap C)$ and $\Lambda^\circleddash \circ \Lambda(A \otimes B \multimap C) = A \otimes B \multimap C$ for any normalized games $A$, $B$ and $C$.
Note also that currying $\Lambda(G)$ and uncurrying $\Lambda^\circleddash(H)$ do not depend on the choice of the normalized games $A$, $B$ and $C$ such that $\mathcal{H}^\omega(G) \trianglelefteqslant A \otimes B \multimap C$ and $\mathcal{H}^\omega(\Lambda^\circleddash(H)) \trianglelefteqslant A \otimes B \multimap C$. 

The following are two of the technical highlights of \cite{yamada2016dynamic}, where it is straightforward to see that the additional structure of dummy functions and the strengthened axiom DP2 are preserved under the operations on games introduced so far (and therefore the two results still hold).

\begin{theorem}[Constructions on games \cite{yamada2016dynamic}]
\label{ThmConstructionsOnGames}
Games are closed under $\otimes$, $\multimap$, $\&$, $\langle \_, \_ \rangle$, $!$, $(\_)^\dagger$, $\ddagger$, $\Lambda$ and $\Lambda^{\circleddash}$.
Moreover, the subgame relation is preserved under these constructions, i.e., $\clubsuit_{i \in I} G_i \trianglelefteqslant \clubsuit_{i \in I} H_i$ if $G_i \trianglelefteqslant H_i$ for all $i \in I$, where $\clubsuit_{i \in I}$ is either $\otimes$, $\multimap$, $\&$, $\langle \_, \_ \rangle$, $!$, $(\_)^\dagger$, $\ddagger$, $\Lambda$ or $\Lambda^{\circleddash}$ (and thus, I is either $\{ 1 \}$ or $\{ 1, 2 \}$).
\end{theorem}

\begin{lemma}[Hiding lemma on games \cite{yamada2016dynamic}]
\label{LemHidingLemmaOnGames}
Let $\clubsuit_{i \in I}$ be either $\otimes$, $\multimap$, $\&$, $\langle \_, \_ \rangle$, $!$, $(\_)^\dagger$, $\ddagger$, $\Lambda$ or $\Lambda^\circleddash$, and $d \in \mathbb{N} \cup \{ \omega \}$. Then, for any family $(G_i)_{i \in I}$ of games, 
\begin{enumerate}

\item $\mathcal{H}^d(\clubsuit_{i \in I}G_i) \trianglelefteqslant \clubsuit_{i \in I} \mathcal{H}^d(G_i)$ if $\clubsuit_{i \in I}$ is not $\ddagger$;

\item $\mathcal{H}^d (G_1 \ddagger G_2) = \mathcal{H}^d(G_1) \ddagger \mathcal{H}^d(G_2)$ if $\mathcal{H}^d(G_1 \ddagger G_2)$ is not yet normalized, and $\mathcal{H}^d(G_1 \ddagger G_2) \trianglelefteqslant A \multimap C$ otherwise, where $\mathcal{H}^{d-1}(G_1) \trianglelefteqslant A \multimap B$ and $\mathcal{H}^{d-1}(G_2) \trianglelefteqslant B \multimap C$ for some normalized games $A$, $B$ and $C$.

\end{enumerate}
\end{lemma}

\subsection{Strategies}
\label{DynamicStrategies}
Next, let us recall another central notion of \emph{strategies}.

\subsubsection{Strategies}
Our strategies are the \emph{dynamic} variant defined in \cite{yamada2016dynamic}.
However, there is nothing special in the definition: A dynamic strategy on a (dynamic) game is a strategy on the game in the conventional sense \cite{abramsky1999game,mccusker1998games}.
Formally:
\begin{definition}[Dynamic strategies \cite{abramsky1999game,mccusker1998games,yamada2016dynamic}]
\label{DefStrategies}
A \emph{\bfseries dynamic strategy} $\sigma$ on a (dynamic) game $G$, written $\sigma : G$, is a subset $\sigma \subseteq P_G^{\mathsf{Even}}$ that satisfies:
\begin{itemize}

\item \textsc{(S1)} It is non-empty and \emph{even-prefix-closed}: $\sigma \neq \emptyset \wedge (\boldsymbol{s}mn \in \sigma \Rightarrow \boldsymbol{s} \in \sigma)$;

\item \textsc{(S2)} It is \emph{deterministic}: $\boldsymbol{s}mn, \boldsymbol{s}mn' \in \sigma \Rightarrow n = n' \wedge \mathcal{J}_{\boldsymbol{s}mn}(n) = \mathcal{J}_{\boldsymbol{s}mn'}(n')$.

\end{itemize}
\end{definition}



\begin{convention}
From now on, \emph{\bfseries strategies} refer to \emph{dynamic} strategies as defined above by default.
\end{convention}

Next, let us recall two constraints on strategies: \emph{innocence} and \emph{well-bracketing}. 
One of the highlights of HO-games \cite{hyland2000full} is to establish a one-to-one correspondence between terms of PCF in a certain \emph{$\eta$-long normal form}, known as \emph{PCF B\"{o}hm trees} \cite{amadio1998domains}, and innocent and well-bracketed strategies (on games modeling types of PCF).
That is, the two conditions limit the codomain of the interpretation of PCF, i.e., the category of HO-games, in such a way that the interpretation becomes \emph{full}.

Roughly, a strategy is \emph{innocent} if its computation depends only on the P-view of the current odd-length position (rather than the current position itself), and \emph{well-bracketed} if every `question-answering' by the strategy is done in the `last-question-first-answered' fashion. 
Formally:
\begin{definition}[Innocence \cite{hyland2000full,abramsky1999game,mccusker1998games}]
\label{DefInnocence}
A strategy $\sigma : G$ is \emph{\bfseries innocent} if $\boldsymbol{s}mn, \boldsymbol{t} \in \sigma \wedge \boldsymbol{t} m \in P_G \wedge \lceil \boldsymbol{t} m \rceil_G = \lceil \boldsymbol{s} m \rceil_G \Rightarrow \boldsymbol{t}mn \in \sigma$.
\end{definition}

\begin{definition}[Well-bracketing \cite{abramsky1999game,mccusker1998games}]
\label{DefWellBracketing}
A strategy $\sigma : G$ is \emph{\bfseries well-bracketed} if, whenever there is a position $\boldsymbol{s} q \boldsymbol{t} a \in \sigma$, where $q$ is a question that justifies an answer $a$, every question in $\boldsymbol{t'}$ defined by $\lceil \boldsymbol{s} q \boldsymbol{t} \rceil_G = \lceil \boldsymbol{s} q \rceil_G . \boldsymbol{t'}$ justifies an answer in $\boldsymbol{t'}$.
\end{definition}

The bijective correspondence holds also for the game model \cite{abramsky1999game}, on which our games and strategies are based.
Moreover, it corresponds respectively to modeling \emph{states} and \emph{control operators} to relax innocence and well-bracketing in the model; in this sense, the two conditions characterize `purely functional' languages \cite{abramsky1999game}.

Note that innocence and well-bracketing have been imposed on strategies in order to establish \emph{full abstraction} and/or \emph{definability} \cite{curien2007definability}, but neither is our main concern in this paper.
However, we would like P to be able to collect a \emph{bounded} number of `relevant' moves from the current odd-length position in an `effective' fashion; for this point, it is convenient to focus on \emph{innocent} strategies since it then suffices for P to trace back the chain of justifiers. 
In fact, we shall define our notion of `computability' only on \emph{innocent} strategies in this paper.

On the other hand, we do not impose well-bracketing on strategies (thus, control operators are `effective' in our sense); nevertheless, we shall consider only strategies modeling terms of PCF which are all well-bracketed.

\begin{remark}
We conjecture that it is possible to define `effectivity' of \emph{non-innocent} strategies in a fashion similar to the case of innocent ones defined in Section~\ref{SubsectionViableStrategies}.
For this, however, we need to modify the procedure for P to collect a bounded number of moves from the current odd-length position (defined in Section~\ref{SubsectionViableStrategies}) so that she may refer to moves outside of the P-view, which is left as future work.
\end{remark}

\begin{convention}
From now on, a \emph{\bfseries strategy} refers to an \emph{innocent} strategy by default. 
We may clearly regard innocent strategies $\sigma : G$ as \emph{\bfseries (partial) view functions} $f_\sigma : \lceil P_G^{\mathsf{Odd}} \rceil_G \rightharpoonup M_G$ (see \cite{mccusker1998games} for the detail); we shall freely exchange the tree-representation $\sigma$ and the function representation $f_\sigma$, and write $\sigma$ for $f_\sigma$.
\end{convention}

\begin{notation}
Henceforth, we abbreviate inner elements $(\dots((m, X_1), X_2), \dots, X_k)$, where $X_i \in \{ \mathscr{W}, \mathscr{E}, \mathscr{N}, \mathscr{S} \}$ for $i = 1, 2, \dots, k$, as $m_{X_1 X_2 \dots X_k}$.
Also, we often indicate the form of tags of moves $[m_{X_1 X_2 \dots X_k}]_{\boldsymbol{e}}$ of a game $G$ informally by $[G_{X_1 X_2 \dots X_k}]_{\boldsymbol{e}}$. 
\end{notation}

\begin{example}
\label{ExSuccPred}
The \emph{\bfseries successor strategy} $\mathit{succ} : [\mathcal{N}_\mathscr{W}]_{\Lbag \boldsymbol{e} \Rbag \hbar} \Rightarrow [\mathcal{N}_\mathscr{E}]$ is defined by:
\begin{equation*}
\mathit{succ} \stackrel{\mathrm{df. }}{=} \mathsf{Pref}(\{ [\hat{q}_\mathscr{E}] [\hat{q}_\mathscr{W}]_{\Lbag \Rbag \hbar} ([\mathit{y}_\mathscr{W}]_{\Lbag \Rbag \hbar} [\mathit{y}_\mathscr{E}] [q_\mathscr{E}] [q_\mathscr{W}]_{\Lbag \Rbag \hbar})^i [\mathit{n}_\mathscr{W}]_{\Lbag \Rbag \hbar} [\mathit{y}_\mathscr{E}] [q_\mathscr{E}]  [\mathit{n}_\mathscr{E}] \mid i \in \mathbb{N} \ \! \})^{\mathsf{Even}}.
\end{equation*} 
where $\mathit{y}$ and $\mathit{n}$ are abbreviates $\mathit{yes}$ and $\mathit{no}$, respectively. 
That is, $\mathit{succ}$ copies an input `counting process' $\underline{n}$ in the domain $[!\mathcal{N}_\mathscr{W}]_{\Lbag \boldsymbol{e} \Rbag \hbar}$ and repeats it in the codomain $[\mathcal{N}_\mathscr{E}]$ adding `one more counting' as already mentioned in the introduction; thus, it makes sense to take $\mathit{succ}$ as a strategy for the successor function $n \mapsto n+1$.

Similarly, the \emph{\bfseries predecessor strategy} $\mathit{pred} : [\mathcal{N}_\mathscr{W}]_{\Lbag \boldsymbol{e} \Rbag \hbar} \Rightarrow [\mathcal{N}_\mathscr{E}]$ is defined by:
\begin{align*}
\mathit{pred} \stackrel{\mathrm{df. }}{=} \ &\mathsf{Pref}(\{ [\hat{q}_\mathscr{E}] [\hat{q}_\mathscr{W}]_{\Lbag \Rbag \hbar} [\mathit{y}_\mathscr{W}]_{\Lbag \Rbag \hbar} [q_\mathscr{W}]_{\Lbag \Rbag \hbar} ([\mathit{y}_\mathscr{W}]_{\Lbag \Rbag \hbar} [\mathit{y}_\mathscr{E}] [q_\mathscr{E}] [q_\mathscr{W}]_{\Lbag \Rbag \hbar})^i [\mathit{n}_\mathscr{W}]_{\Lbag \Rbag \hbar} [\mathit{n}_\mathscr{E}] \mid i \in \mathbb{N} \ \! \} \\
&\cup \{ [\hat{q}_\mathscr{E}] [\hat{q}_\mathscr{W}]_{\Lbag \Rbag \hbar} [\mathit{n}_\mathscr{W}]_{\Lbag \Rbag \hbar} [\mathit{n}_\mathscr{E}] \})^{\mathsf{Even}}.
\end{align*}
Thus, if the input in the domain $[!N_\mathscr{W}]_{\Lbag \boldsymbol{e} \Rbag \hbar}$ is $\underline{0}$, then $\mathit{pred}$ simply `copy-cats' it in the codomain $[\mathcal{N}_{\mathscr{E}}]$; otherwise, i.e., if the input is $\underline{n+1}$ for some $n \in \mathbb{N}$, then $\mathit{pred}$ `copy-cats' it except the first `counting'.
Therefore $\mathit{pred}$ `implements' the predecessor function $0 \mapsto 0$, $n+1 \mapsto n$.
\end{example}

As in the case of games, we now define the \emph{hiding operation on strategies}.
Note that an even-length position is not necessarily preserved under the hiding operation on j-sequences.
For instance, let $\boldsymbol{s} m n \boldsymbol{t}$ be an even-length position of a game $G$ such that $\boldsymbol{s} m$ (resp. $\boldsymbol{t} n$) consists of external (resp. internal) moves only.
By IE-switch on $G$, $m$ is an O-move, and so $\mathcal{H}^\omega (\boldsymbol{s} m n \boldsymbol{t}) = \boldsymbol{s} m$ is of odd-length.
Thus, we define:

\begin{definition}[Hiding on strategies \cite{yamada2016dynamic}]
Given a game $G$, a position $\boldsymbol{s} \in P_G$ and a number $d \in \mathbb{N} \cup \{ \omega \}$, let $\boldsymbol{s} \natural \mathcal{H}_G^d \stackrel{\mathrm{df. }}{=} \begin{cases} \mathcal{H}_G^d(\boldsymbol{s}) &\text{if $\boldsymbol{s}$ is $d$-complete (Definition~\ref{DefArenas});} \\ \boldsymbol{t} &\text{otherwise, where $\mathcal{H}_G^d(\boldsymbol{s}) = \boldsymbol{t}m$.} \end{cases}$
We define the \emph{\bfseries $\boldsymbol{d}$-hiding operation} $\mathcal{H}^d$ on strategies by $\mathcal{H}^d : (\sigma : G) \mapsto \{ \boldsymbol{s} \natural \mathcal{H}_G^d \mid \boldsymbol{s} \in \sigma \ \! \}$.
A strategy $\sigma : G$ is \emph{\bfseries normalized} if $\mathcal{H}^\omega(\sigma) = \sigma$.
\end{definition}


The following beautiful theorem in a sense implies that the above definition is a reasonable one.
Also, it induces the \emph{hiding functor} $\mathcal{H}^\omega$ from the bicategory $\mathcal{DG}$ of dynamic games and strategies to the category $\mathcal{G}$ of static games and strategies \cite{yamada2016dynamic}.
\begin{theorem}[Hiding theorem \cite{yamada2016dynamic}]
\label{ThmHidingTheorem}
If $\sigma : G$, then $\mathcal{H}^d(\sigma) : \mathcal{H}^d(G)$ for all $d \in \mathbb{N} \cup \{ \omega \}$, and $\underbrace{\mathcal{H}^1 \circ \mathcal{H}^1 \cdots \circ \mathcal{H}^1}_i (\sigma) = \mathcal{H}^i(\sigma) : \mathcal{H}^i(G)$ for all $i \in \mathbb{N}$.
\end{theorem}

\subsubsection{Constructions on strategies}
\label{ConstructionsOnStrategies}
Next, let us review standard constructions on strategies \cite{abramsky1999game,mccusker1998games}, for which we need to adopt a formalization of `tags'.
Having introduced our formalization of `tags' for constructions on games in Section~\ref{ConstructionsOnGames}, let us just present formalized constructions on strategies (without standard and informal versions) as it should be clear enough.

\emph{Copy-cat strategies} are the most basic strategy, which, as the name indicates, just `copy-cats' the last occurrence of an O-move in the current odd-length position:
\begin{definition}[Copy-cats \cite{abramsky1994games,abramsky2000full,hyland2000full}]
\label{CopyCatStrategies}
The \emph{\bfseries copy-cat (strategy)} $\mathit{cp}_A : A \multimap A$ on a normalized game $A$ is defined by:
\begin{equation*}
\mathit{cp}_A \stackrel{\mathrm{df. }}{=} \{ \boldsymbol{s} \in P_{A \multimap A}^\mathsf{Even} \mid \forall \boldsymbol{t} \preceq{\boldsymbol{s}}. \ \mathsf{Even}(\boldsymbol{t}) \Rightarrow \boldsymbol{t} \upharpoonright \mathscr{W} = \boldsymbol{t} \upharpoonright \mathscr{E} \ \! \}.
\end{equation*}
\end{definition}

The following \emph{dereliction} plays essentially in the same way as copy-cats:
\begin{definition}[Derelicition \cite{abramsky2000full, mccusker1998games}]
\label{Derelictions}
Let $A$ be a normalized game. The \emph{\bfseries dereliction} $\mathit{der}_A : A \Rightarrow A$ on $A$ is defined by:
\begin{equation*}
\mathit{der}_A \stackrel{\mathrm{df. }}{=} \{ \boldsymbol{s} \in P_{A \Rightarrow A}^\mathsf{Even} \mid \forall \boldsymbol{t} \preceq{\boldsymbol{s}}. \ \mathsf{Even}(\boldsymbol{t}) \Rightarrow \boldsymbol{t} \upharpoonright (\mathscr{W})_{\Lbag \Rbag \hbar \_} = \boldsymbol{t} \upharpoonright (\mathscr{E})_{\_} \ \! \}
\end{equation*}
where $\boldsymbol{t} \upharpoonright (\mathscr{W})_{\Lbag \Rbag \hbar \_}$ (resp. $\boldsymbol{t} \upharpoonright (\mathscr{E})_{\_}$) is the j-subsequence of $\boldsymbol{t}$ that consists of moves of the form $[(a, \mathscr{W})]_{(\Lbag \boldsymbol{f} \Rbag \hbar \boldsymbol{e})}$ (resp. $[(a', \mathscr{E})]_{\boldsymbol{e'}}$) changed into $[a]_{\boldsymbol{e}}$ (resp. $[a']_{\boldsymbol{e'}}$).
\end{definition}

Next, as in the case of the tensor of games, we have:
\begin{definition}[Tensor of strategies \cite{abramsky1994games, mccusker1998games}]
Given static strategies $\sigma : A \multimap C$ and $\tau : B \multimap D$, their \emph{\bfseries tensor (product)} $\sigma \! \otimes \! \tau : A \otimes B \multimap C \otimes D$ is defined by: 
\begin{equation*}
\sigma \! \otimes \! \tau \stackrel{\mathrm{df. }}{=} \{ \boldsymbol{s} \in \mathscr{L}_{A \otimes B \multimap C \otimes D} \mid \boldsymbol{s} \upharpoonright (\mathscr{W}, \_) \in \sigma, \boldsymbol{s} \upharpoonright (\mathscr{E}, \_) \in \tau \ \! \}
\end{equation*}
where $\boldsymbol{s} \upharpoonright (\mathscr{W}, \_)$ (resp. $\boldsymbol{s} \upharpoonright (\mathscr{E}, \_)$) is the j-subsequence of $\boldsymbol{s}$ consisting of moves of the form $[((m, \mathscr{W}), X)]_{\boldsymbol{e}}$ (resp. $[((m', \mathscr{E}), Y)]_{\boldsymbol{e}}$) changed into $[(m, X)]_{\boldsymbol{e}}$ (resp. $[(m', Y)]_{\boldsymbol{e}}$).
\end{definition}

Intuitively, the tensor $\sigma \! \otimes \! \tau : A \otimes B \multimap C \otimes D$ plays by $\sigma$ if the last occurrence of an O-move belongs to $A \multimap C$, and by $\tau$ otherwise.

\begin{example}
Consider the tensor $\mathit{succ} \otimes \mathit{pred} : [!\mathcal{N}_{\mathscr{W} \mathscr{W}}]_{\Lbag \boldsymbol{e} \Rbag \hbar} \otimes [!\mathcal{N}_{\mathscr{E} \mathscr{W}}]_{\Lbag \boldsymbol{e'} \Rbag \hbar} \multimap [\mathcal{N}_{\mathscr{W} \mathscr{E}}] \otimes [\mathcal{N}_{\mathscr{E} \mathscr{E}}]$.
A typical play by $\mathit{succ} \otimes \mathit{pred}$ is as follows:
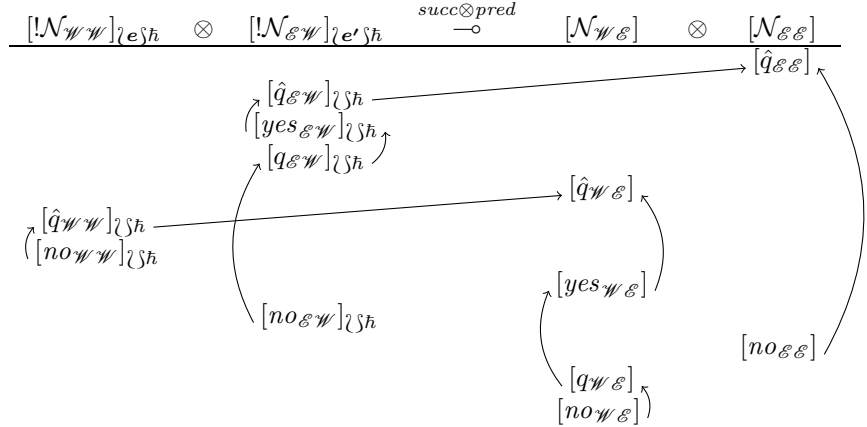
\begin{center}
\begin{tabular}{ccccccc}
$[!\mathcal{N}_{\mathscr{W} \mathscr{W}}]_{\Lbag \boldsymbol{e} \Rbag \hbar}$ & $\otimes$ & $[!\mathcal{N}_{\mathscr{E} \mathscr{W}}]_{\Lbag \boldsymbol{e'} \Rbag \hbar}$ & $\stackrel{\mathit{succ} \otimes \mathit{pred}}{\multimap}$ & $[\mathcal{N}_{\mathscr{W} \mathscr{E}}]$ & $\otimes$ & $[\mathcal{N}_{\mathscr{E} \mathscr{E}}]$ \\
\cline{1-7}
&&&&&&\tikzmark{SuccTensorDoubleC1} $[\hat{q}_{\mathscr{E} \mathscr{E}}]$ \tikzmark{SuccTensorDoubleC8} \\
&&\tikzmark{SuccTensorDoubleC2} $[\hat{q}_{\mathscr{E} \mathscr{W}}]_{\Lbag \Rbag \hbar}$ \tikzmark{SuccTensorDoubleD1}&&&& \\
&&\tikzmark{SuccTensorDoubleD2} $[\mathit{yes}_{\mathscr{E} \mathscr{W}}]_{\Lbag \Rbag \hbar}$ \tikzmark{SuccTensorDoubleC3}&&&& \\
&&\tikzmark{SuccTensorDoubleC7} $[q_{\mathscr{E} \mathscr{W}}]_{\Lbag \Rbag \hbar}$ \tikzmark{SuccTensorDoubleD3}&&&& \\
&&&&\tikzmark{SuccTensorDoubleC4} $[\hat{q}_{\mathscr{W} \mathscr{E}}]$ \tikzmark{SuccTensorDoubleC6}&& \\
\tikzmark{SuccTensorDoubleC5} $[\hat{q}_{\mathscr{W} \mathscr{W}}]_{\Lbag \Rbag \hbar}$ \tikzmark{SuccTensorDoubleD4}&&&&&& \\
\tikzmark{SuccTensorDoubleD5} $[\mathit{no}_{\mathscr{W} \mathscr{W}}]_{\Lbag \Rbag \hbar}$&&&&&& \\
&&&&\tikzmark{SuccTensorDoubleC9} $[\mathit{yes}_{\mathscr{W} \mathscr{E}}]$ \tikzmark{SuccTensorDoubleD6}&& \\
&&\tikzmark{SuccTensorDoubleD7} $[\mathit{no}_{\mathscr{E} \mathscr{W}}]_{\Lbag \Rbag \hbar}$&&&& \\
&&&&&&$[\mathit{no}_{\mathscr{E} \mathscr{E}}]$ \tikzmark{SuccTensorDoubleD8} \\
&&&&\tikzmark{SuccTensorDoubleD9} $[q_{\mathscr{W} \mathscr{E}}]$ \tikzmark{SuccTensorDoubleC10}&& \\
&&&&$[\mathit{no}_{\mathscr{W} \mathscr{E}}]$ \tikzmark{SuccTensorDoubleD10}&& 
\end{tabular}
\begin{tikzpicture}[overlay, remember picture, yshift=.25\baselineskip]
\draw [->] ({pic cs:SuccTensorDoubleD1}) to ({pic cs:SuccTensorDoubleC1});
\draw [->] ({pic cs:SuccTensorDoubleD2}) [bend left] to ({pic cs:SuccTensorDoubleC2});
\draw [->] ({pic cs:SuccTensorDoubleD3}) [bend right] to ({pic cs:SuccTensorDoubleC3});
\draw [->] ({pic cs:SuccTensorDoubleD4}) to ({pic cs:SuccTensorDoubleC4});
\draw [->] ({pic cs:SuccTensorDoubleD5}) [bend left] to ({pic cs:SuccTensorDoubleC5});
\draw [->] ({pic cs:SuccTensorDoubleD6}) [bend right] to ({pic cs:SuccTensorDoubleC6});
\draw [->] ({pic cs:SuccTensorDoubleD7}) [bend left] to ({pic cs:SuccTensorDoubleC7});
\draw [->] ({pic cs:SuccTensorDoubleD8}) [bend right] to ({pic cs:SuccTensorDoubleC8});
\draw [->] ({pic cs:SuccTensorDoubleD9}) [bend left] to ({pic cs:SuccTensorDoubleC9});
\draw [->] ({pic cs:SuccTensorDoubleD10}) [bend right] to ({pic cs:SuccTensorDoubleC10});
\end{tikzpicture}
\end{center}
\end{example}

The next construction is the categorical \emph{pairing} in the category $\mathcal{G}$ of static games and strategies \cite{abramsky1999game,mccusker1998games}:
\begin{definition}[Pairing of strategies \cite{abramsky2000full, mccusker1998games}]
\label{Pairing}
Given strategies $\sigma : C \multimap A$ and $\tau : C \multimap B$, their \emph{\bfseries pairing} $\langle \sigma, \tau \rangle : C \multimap A \& B$ is defined by: 
\begin{align*}
\langle \sigma, \tau \rangle \stackrel{\mathrm{df. }}{=} & \ \! \{ \boldsymbol{s} \in \mathscr{L}_{C \multimap A \& B} \mid \boldsymbol{s} \upharpoonright (\mathscr{W} \multimap \mathscr{W} \mathscr{E}) \in \sigma, \boldsymbol{s} \upharpoonright (\mathscr{W} \multimap \mathscr{E} \mathscr{E}) = \boldsymbol{\epsilon} \ \! \} \\ &\cup \ \! \{ \boldsymbol{s} \in \mathscr{L}_{C \multimap A \& B} \mid \boldsymbol{s} \upharpoonright (\mathscr{W} \multimap \mathscr{E} \mathscr{E}) \in \tau, \boldsymbol{s} \upharpoonright (\mathscr{W} \multimap \mathscr{W} \mathscr{E}) = \boldsymbol{\epsilon} \ \! \}
\end{align*}
where $\boldsymbol{s} \upharpoonright (\mathscr{W} \multimap \mathscr{W} \mathscr{E})$ (resp. $\boldsymbol{s} \upharpoonright (\mathscr{W} \multimap \mathscr{E} \mathscr{E})$) is the j-subsequence of $\boldsymbol{s}$ that consists of moves of the form $[(c, \mathscr{W})]_{\boldsymbol{e}}$ or $[((a, \mathscr{W}), \mathscr{E})]_{\boldsymbol{f}}$ with $[a] \in M_A$ (resp. or $[((b, \mathscr{E}), \mathscr{E})]_{\boldsymbol{f}}$ with $[b] \in M_B$) with the latter changed into $[(a, \mathscr{E})]_{\boldsymbol{f}}$ (resp. $[(b, \mathscr{E})]_{\boldsymbol{f}}$).
\end{definition}

However, for the bicategory $\mathcal{DG}$ of dynamic games and strategies \cite{yamada2016dynamic}, we need the following generalization (for the reason explained right before introducing pairing of games in Section~\ref{ConstructionsOnGames}):
\begin{definition}[Generalized pairing \cite{yamada2016dynamic}]
Given strategies $\sigma : L$ and $\tau : R$ such that $\mathcal{H}^\omega(L) \trianglelefteqslant C \multimap A$, $\mathcal{H}^\omega(R) \trianglelefteqslant C \multimap B$ for some normalized games $A$, $B$ and $C$, their \emph{\bfseries (generalized) pairing} $\langle \sigma, \tau \rangle : \langle L, R \rangle$ is defined by:
\begin{align*}
\langle \sigma, \tau \rangle \stackrel{\mathrm{df. }}{=} \{ \boldsymbol{s} \in \mathscr{L}_{\langle L, R \rangle} \mid (\boldsymbol{s} \upharpoonright L \in \sigma \wedge \boldsymbol{s} \upharpoonright B = \boldsymbol{\epsilon}) \vee (\boldsymbol{s} \upharpoonright R \in \tau \wedge \boldsymbol{s} \upharpoonright A = \boldsymbol{\epsilon}) \ \! \}.
\end{align*}
\end{definition}

It is clearly a generalization of pairing; consider the case where $L = C \multimap A$ and $R = C \multimap B$.

\begin{convention}
Henceforth, \emph{\bfseries pairing} refers to \emph{generalized} pairing by default. 
\end{convention}

\begin{example}
Consider the pairing $\langle \mathit{succ}, \mathit{pred} \rangle : [!\mathcal{N}_\mathscr{W}]_{\Lbag \boldsymbol{e} \Rbag \hbar} \multimap [\mathcal{N}_{\mathscr{WE}}] \& [\mathcal{N}_{\mathscr{EE}}]$.
Its typical plays are as follows:
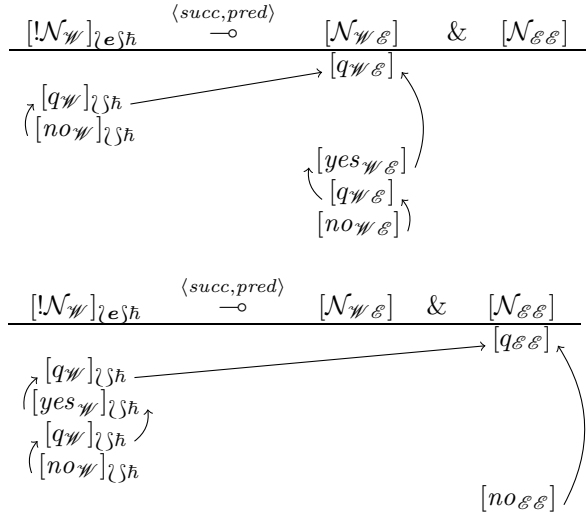
\begin{center}
\begin{tabular}{ccccc}
$[!\mathcal{N}_{\mathscr{W}}]_{\Lbag \boldsymbol{e} \Rbag \hbar}$ & $\stackrel{\langle \mathit{succ}, \mathit{pred} \rangle}{\multimap}$ & $[\mathcal{N}_{\mathscr{W} \mathscr{E}}]$ & $\&$ & $[\mathcal{N}_{\mathscr{E} \mathscr{E}}]$  \\
\cline{1-5} 
&&\tikzmark{SuccPairPredC1} $[q_{\mathscr{W} \mathscr{E}}]$ \tikzmark{SuccPairPredC3}&& \\
\tikzmark{SuccPairPredC2} $[q_{\mathscr{W}}]_{\Lbag \Rbag \hbar}$ \tikzmark{SuccPairPredD1}&&&& \\
\tikzmark{SuccPairPredD2} $[\mathit{no}_{\mathscr{W}}]_{\Lbag \Rbag \hbar}$&&&& \\
&&\tikzmark{SuccPairPredC4} $[\mathit{yes}_{\mathscr{W} \mathscr{E}}]$ \tikzmark{SuccPairPredD3}&& \\
&&\tikzmark{SuccPairPredD4} $[q_{\mathscr{W} \mathscr{E}}]$ \tikzmark{SuccPairPredC5}&& \\
&&$[\mathit{no}_{\mathscr{W} \mathscr{E}}]$ \tikzmark{SuccPairPredD5}&& \\ \\
\end{tabular}
\begin{tabular}{ccccc}
$[!\mathcal{N}_{\mathscr{W}}]_{\Lbag \boldsymbol{e} \Rbag \hbar}$ & $\stackrel{\langle \mathit{succ}, \mathit{pred} \rangle}{\multimap}$ & $[\mathcal{N}_{\mathscr{W} \mathscr{E}}]$ & $\&$ & $[\mathcal{N}_{\mathscr{E} \mathscr{E}}]$ \\
\cline{1-5}
&&&&\tikzmark{SuccPairPredC6} $[q_{\mathscr{E} \mathscr{E}}]$ \tikzmark{SuccPairPredC10} \\
\tikzmark{SuccPairPredC7} $[q_{\mathscr{W}}]_{\Lbag \Rbag \hbar}$ \tikzmark{SuccPairPredD6}&&&& \\
\tikzmark{SuccPairPredD7} $[\mathit{yes}_{\mathscr{W}}]_{\Lbag \Rbag \hbar}$ \tikzmark{SuccPairPredC8}&&&& \\
\tikzmark{SuccPairPredC9} $[q_{\mathscr{W}}]_{\Lbag \Rbag \hbar}$ \tikzmark{SuccPairPredD8}&&&& \\
\tikzmark{SuccPairPredD9} $[\mathit{no}_{\mathscr{W}}]_{\Lbag \Rbag \hbar}$&&&& \\
&&&& $[\mathit{no}_{\mathscr{E} \mathscr{E}}]$ \tikzmark{SuccPairPredD10} \\
\end{tabular}
\begin{tikzpicture}[overlay, remember picture, yshift=.25\baselineskip]
\draw [->] ({pic cs:SuccPairPredD1}) to ({pic cs:SuccPairPredC1});
\draw [->] ({pic cs:SuccPairPredD2}) [bend left] to ({pic cs:SuccPairPredC2});
\draw [->] ({pic cs:SuccPairPredD3}) [bend right] to ({pic cs:SuccPairPredC3});
\draw [->] ({pic cs:SuccPairPredD4}) [bend left] to ({pic cs:SuccPairPredC4});
\draw [->] ({pic cs:SuccPairPredD5}) [bend right] to ({pic cs:SuccPairPredC5});
\draw [->] ({pic cs:SuccPairPredD6}) to ({pic cs:SuccPairPredC6});
\draw [->] ({pic cs:SuccPairPredD7}) [bend left] to ({pic cs:SuccPairPredC7});
\draw [->] ({pic cs:SuccPairPredD8}) [bend right] to ({pic cs:SuccPairPredC8});
\draw [->] ({pic cs:SuccPairPredD9}) [bend left] to ({pic cs:SuccPairPredC9});
\draw [->] ({pic cs:SuccPairPredD10}) [bend right] to ({pic cs:SuccPairPredC10});
\end{tikzpicture}
\end{center}
\end{example}

Next, let us recall \emph{promotion} of strategies:
\begin{definition}[Promotion of strategies \cite{abramsky2000full, mccusker1998games}]
\label{DefPromotionOfStrategies}
Given a strategy $\phi : \ !A \multimap B$, its \emph{\bfseries promotion} $\phi^{\dagger} : \ !A \multimap \ !B$ is defined by: 
\begin{equation*}
\phi^{\dagger} \stackrel{\mathrm{df. }}{=} \{ \boldsymbol{s} \in \mathscr{L}_{!A \multimap !B} \mid \ \forall e \in \mathcal{T} . \ \! \boldsymbol{s} \upharpoonright \boldsymbol{e} \in \phi \ \! \}
\end{equation*}
where $\boldsymbol{s} \upharpoonright \boldsymbol{e}$ is the j-subsequence of $\boldsymbol{s}$ that consists of moves of the form $[(b, \mathscr{E})]_{\Lbag \boldsymbol{e} \Rbag \hbar \boldsymbol{e'}}$ with $[b]_{\boldsymbol{e'}} \in M_B$ or $[(a, \mathscr{W})]_{\Lbag \Lbag \boldsymbol{e} \Rbag \hbar \Lbag \boldsymbol{f} \Rbag \Rbag \hbar \boldsymbol{f'}}$ with $[a]_{\boldsymbol{f'}} \in M_A$, which are respectively changed into $[(b, \mathscr{E})]_{\boldsymbol{e'}}$ and $[(a, \mathscr{W})]_{\Lbag \boldsymbol{f} \Rbag \hbar \boldsymbol{f'}}$.
\end{definition}

As stated before, \cite{yamada2016dynamic} generalizes promotion of strategies (for the reason explained right before introducing promotion of games in Section~\ref{ConstructionsOnGames}) as follows: 
\begin{definition}[Generalized promotion of strategies \cite{abramsky2000full, mccusker1998games}]
\label{DefGeneralizedPromotionOfStrategies}
Given a strategy $\phi : G$ such that $\mathcal{H}^\omega(G) \trianglelefteqslant \ !A \multimap \ !B$ for some normalized games $A$ and $B$, its \emph{\bfseries (generalized) promotion} $\phi^{\dagger} : G^\dagger$ is defined by: 
\begin{equation*}
\phi^{\dagger} \stackrel{\mathrm{df. }}{=} \{ \boldsymbol{s} \in \mathscr{L}_{G^\dagger} \mid \ \forall e \in \mathcal{T} . \ \! \boldsymbol{s} \upharpoonright \boldsymbol{e} \in \phi \ \! \}
\end{equation*}
where $\boldsymbol{s} \upharpoonright \boldsymbol{e}$ is the j-subsequence of $\boldsymbol{s}$ that consists of moves of the form $[(b, \mathscr{E})]_{\Lbag \boldsymbol{e} \Rbag \hbar \boldsymbol{e'}}$ with $[b]_{\boldsymbol{e'}} \in M_B$, $[(a, \mathscr{W})]_{\Lbag \Lbag \boldsymbol{e} \Rbag \hbar \Lbag \boldsymbol{f} \Rbag \Rbag \hbar \boldsymbol{f'}}$ with $[a]_{\boldsymbol{f'}} \in M_A$, or $[(m, \mathscr{S})]_{\Lbag \boldsymbol{e} \Rbag \hbar \boldsymbol{e'}}$ with $[m]_{\boldsymbol{e'}} \in M_G^{\mathsf{Int}}$, which are respectively changed into $[(b, \mathscr{E})]_{\boldsymbol{e'}}$, $[(a, \mathscr{W})]_{\Lbag \boldsymbol{f} \Rbag \hbar \boldsymbol{f'}}$ and $[m]_{\boldsymbol{e'}}$.
\end{definition}

It is clearly a generalization of promotion; consider the case $G = \ !A \multimap B$.

\begin{convention}
Henceforth, \emph{\bfseries promotion} refers to \emph{generalized} promotion by default. 
\end{convention}

\begin{example}
\label{ExPromotion}
Consider the promotion $\mathit{succ}^\dagger : [!\mathcal{N}_\mathscr{W}]_{\Lbag \boldsymbol{e} \Rbag \hbar} \multimap [!\mathcal{N}_\mathscr{E}]_{\Lbag \boldsymbol{e'} \Rbag \hbar}$.
Its typical play is as depicted in the following diagram:
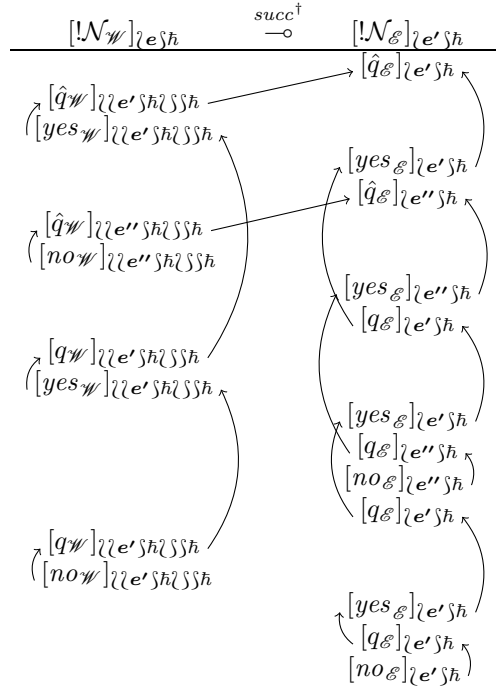
\begin{center}
\begin{tabular}{ccc}
$[!\mathcal{N}_\mathscr{W}]_{\Lbag \boldsymbol{e} \Rbag \hbar}$ & $\stackrel{\mathit{succ}^\dagger}{\multimap}$ & $[!\mathcal{N}_\mathscr{E}]_{\Lbag \boldsymbol{e'} \Rbag \hbar}$ \\
\cline{1-3}
&&\tikzmark{SuccDaggerC1} $[\hat{q}_{\mathscr{E}}]_{\Lbag \boldsymbol{e'} \Rbag \hbar}$ \tikzmark{SuccDaggerC3} \\
\tikzmark{SuccDaggerC2} $[\hat{q}_{\mathscr{W}}]_{\Lbag \Lbag \boldsymbol{e'} \Rbag \hbar \Lbag \Rbag \Rbag \hbar}$ \tikzmark{SuccDaggerD1}&& \\
\tikzmark{SuccDaggerD2} $[\mathit{yes}_{\mathscr{W}}]_{\Lbag \Lbag \boldsymbol{e'} \Rbag \hbar \Lbag \Rbag \Rbag \hbar}$ \tikzmark{SuccDaggerC8}&& \\
&&\tikzmark{SuccDaggerC7} $[\mathit{yes}_{\mathscr{E}}]_{\Lbag \boldsymbol{e'} \Rbag \hbar}$ \tikzmark{SuccDaggerD3} \\
&&\tikzmark{SuccDaggerC4} $[\hat{q}_{\mathscr{E}}]_{\Lbag \boldsymbol{e''} \Rbag \hbar}$ \tikzmark{SuccDaggerC6} \\
\tikzmark{SuccDaggerC5} $[\hat{q}_{\mathscr{W}}]_{\Lbag \Lbag \boldsymbol{e''} \Rbag \hbar \Lbag \Rbag \Rbag \hbar}$ \tikzmark{SuccDaggerD4}&& \\
\tikzmark{SuccDaggerD5} $[\mathit{no}_{\mathscr{W}}]_{\Lbag \Lbag \boldsymbol{e''} \Rbag \hbar \Lbag \Rbag \Rbag \hbar}$&& \\
&&\tikzmark{SuccDaggerC11} $[\mathit{yes}_{\mathscr{E}}]_{\Lbag \boldsymbol{e''} \Rbag \hbar}$ \tikzmark{SuccDaggerD6} \\
&&\tikzmark{SuccDaggerD7} $[q_{\mathscr{E}}]_{\Lbag \boldsymbol{e'} \Rbag \hbar}$ \tikzmark{SuccDaggerC10} \\
\tikzmark{SuccDaggerC9} $[q_{\mathscr{W}}]_{\Lbag \Lbag \boldsymbol{e'} \Rbag \hbar \Lbag \Rbag \Rbag \hbar}$ \tikzmark{SuccDaggerD8}&& \\
\tikzmark{SuccDaggerD9} $[\mathit{yes}_{\mathscr{W}}]_{\Lbag \Lbag \boldsymbol{e'} \Rbag \hbar \Lbag \Rbag \Rbag \hbar}$ \tikzmark{SuccDaggerC14}&& \\
&&\tikzmark{SuccDaggerC13} $[\mathit{yes}_{\mathscr{E}}]_{\Lbag \boldsymbol{e'} \Rbag \hbar}$ \tikzmark{SuccDaggerD10} \\
&&\tikzmark{SuccDaggerD11} $[q_{\mathscr{E}}]_{\Lbag \boldsymbol{e''} \Rbag \hbar}$ \tikzmark{SuccDaggerC12} \\
&&$[\mathit{no}_{\mathscr{E}}]_{\Lbag \boldsymbol{e''} \Rbag \hbar}$ \tikzmark{SuccDaggerD12} \\
&&\tikzmark{SuccDaggerD13} $[q_{\mathscr{E}}]_{\Lbag \boldsymbol{e'} \Rbag \hbar}$ \tikzmark{SuccDaggerC16} \\
\tikzmark{SuccDaggerC15} $[q_{\mathscr{W}}]_{\Lbag \Lbag \boldsymbol{e'} \Rbag \hbar \Lbag \Rbag \Rbag \hbar}$ \tikzmark{SuccDaggerD14}&& \\
\tikzmark{SuccDaggerD15} $[\mathit{no}_{\mathscr{W}}]_{\Lbag \Lbag \boldsymbol{e'} \Rbag \hbar \Lbag \Rbag \Rbag \hbar}$&& \\
& & \tikzmark{SuccDaggerC17} $[\mathit{yes}_{\mathscr{E}}]_{\Lbag \boldsymbol{e'} \Rbag \hbar}$ \tikzmark{SuccDaggerD16} \\
&&\tikzmark{SuccDaggerD17} $[q_{\mathscr{E}}]_{\Lbag \boldsymbol{e'} \Rbag \hbar}$ \tikzmark{SuccDaggerC18} \\
&&$[\mathit{no}_{\mathscr{E}}]_{\Lbag \boldsymbol{e'} \Rbag \hbar}$ \tikzmark{SuccDaggerD18}
\end{tabular}
\begin{tikzpicture}[overlay, remember picture, yshift=.25\baselineskip]
\draw [->] ({pic cs:SuccDaggerD1}) to ({pic cs:SuccDaggerC1});
\draw [->] ({pic cs:SuccDaggerD2}) [bend left] to ({pic cs:SuccDaggerC2});
\draw [->] ({pic cs:SuccDaggerD3}) [bend right] to ({pic cs:SuccDaggerC3});
\draw [->] ({pic cs:SuccDaggerD4}) to ({pic cs:SuccDaggerC4});
\draw [->] ({pic cs:SuccDaggerD5}) [bend left] to ({pic cs:SuccDaggerC5});
\draw [->] ({pic cs:SuccDaggerD6}) [bend right] to ({pic cs:SuccDaggerC6});
\draw [->] ({pic cs:SuccDaggerD7}) [bend left] to ({pic cs:SuccDaggerC7});
\draw [->] ({pic cs:SuccDaggerD8}) [bend right] to ({pic cs:SuccDaggerC8});
\draw [->] ({pic cs:SuccDaggerD9}) [bend left] to ({pic cs:SuccDaggerC9});
\draw [->] ({pic cs:SuccDaggerD10}) [bend right] to ({pic cs:SuccDaggerC10});
\draw [->] ({pic cs:SuccDaggerD11}) [bend left] to ({pic cs:SuccDaggerC11});
\draw [->] ({pic cs:SuccDaggerD12}) [bend right] to ({pic cs:SuccDaggerC12});
\draw [->] ({pic cs:SuccDaggerD13}) [bend left] to ({pic cs:SuccDaggerC13});
\draw [->] ({pic cs:SuccDaggerD14}) [bend right] to ({pic cs:SuccDaggerC14});
\draw [->] ({pic cs:SuccDaggerD15}) [bend left] to ({pic cs:SuccDaggerC15});
\draw [->] ({pic cs:SuccDaggerD16}) [bend right] to ({pic cs:SuccDaggerC16});
\draw [->] ({pic cs:SuccDaggerD17}) [bend left] to ({pic cs:SuccDaggerC17});
\draw [->] ({pic cs:SuccDaggerD18}) [bend right] to ({pic cs:SuccDaggerC18});
\end{tikzpicture}
\end{center}
Note that there are two \emph{threads}\footnote{A \emph{thread} in a j-sequence $\boldsymbol{s}$ is a j-subsequence of $\boldsymbol{s}$ that consists of moves \emph{hereditarily justified} by the same initial occurrence in $\boldsymbol{s}$; see \cite{abramsky1999game,mccusker1998games} for its precise definition.} in the above play, and the strategy $\mathit{succ}^\dagger$ behaves as $\mathit{succ}$ in both of the threads. 
\end{example}

Now, let us recall a central construction of strategies in \cite{yamada2016dynamic}, which reformulates \emph{composition of strategies} as follows:
\begin{definition}[Concatenation and composition of strategies \cite{yamada2016dynamic}]
\label{DefConcatenationAndCompositionOfStrategies}
Let $\sigma : J$ and $\tau : K$; and assume that $\mathcal{H}^\omega(J) \trianglelefteqslant A \multimap B$ and $\mathcal{H}^\omega(K) \trianglelefteqslant B \multimap C$ for some normalized games $A$, $B$ and $C$. Their \emph{\bfseries concatenation} $\sigma \ddagger \tau : J \ddagger K$ is defined by: 
\begin{equation*}
\sigma \ddagger \tau \stackrel{\mathrm{df. }}{=} \{ \boldsymbol{s} \in \mathscr{J}_{J \ddagger K} \mid \boldsymbol{s} \upharpoonright J \in \sigma, \boldsymbol{s} \upharpoonright K \in \tau, \boldsymbol{s} \upharpoonright B^{[1]}, B^{[2]} \in \mathit{pr}_B \}
\end{equation*}
and their \emph{\bfseries composition} $\sigma ; \tau : \mathcal{H}^\omega(J \ddagger K)$ by $\sigma ; \tau \stackrel{\mathrm{df. }}{=} \mathcal{H}^\omega (\sigma \ddagger \tau)$ (see Theorem~\ref{ThmHidingTheorem}).
\end{definition}

We also write $\tau \circ \sigma$ for $\sigma ; \tau$.
If $J = A \multimap B$, $K = B \multimap C$, then our composition $\sigma ; \tau : \mathcal{H}^\omega (A \multimap B \ddagger B \multimap C) \trianglelefteqslant A \multimap C$ coincides with the standard one \cite{hyland2000full,abramsky1999game,mccusker1998games}; see \cite{yamada2016dynamic} for the detail.
In this sense, our composition generalizes the standard one, and it is decomposed into \emph{concatenation plus hiding}. 

\begin{example}
Consider the concatenation $\mathit{succ}^\dagger \ddagger \mathit{pred} : ([!\mathcal{N}_{\mathscr{W}}]_{\Lbag \Lbag \Rbag \hbar \Lbag \Rbag \Rbag \hbar} \multimap [!\mathcal{N}_{\mathscr{E} \mathscr{S}}]_{\Lbag \Rbag \hbar}) \ddagger ([!\mathcal{N}_{\mathscr{W} \mathscr{N}}]_{\Lbag \Rbag \hbar} \multimap [\mathcal{N}_{\mathscr{E}}])$.
Its typical play is as follows:
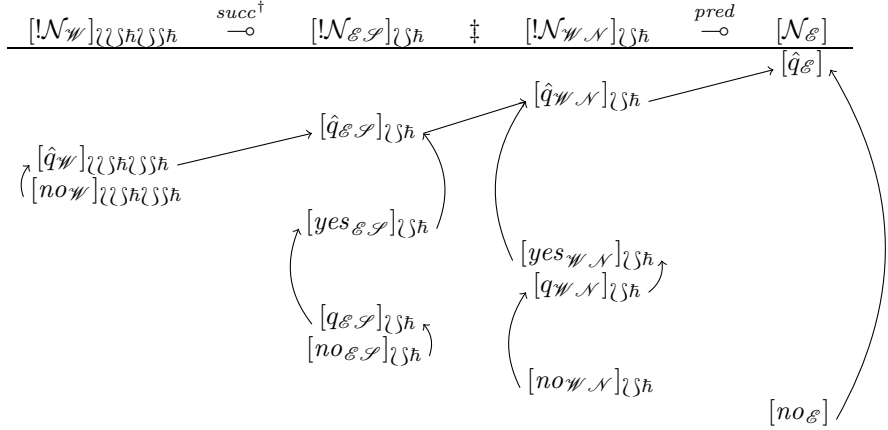
\begin{center}
\begin{tabular}{ccccccc}
$[!\mathcal{N}_{\mathscr{W}}]_{\Lbag \Lbag \Rbag \hbar \Lbag \Rbag \Rbag \hbar}$ & $\stackrel{\mathit{succ}^\dagger}{\multimap}$ & $[!\mathcal{N}_{\mathscr{E} \mathscr{S}}]_{\Lbag \Rbag \hbar}$ & $\ddagger$ & $[!\mathcal{N}_{\mathscr{W} \mathscr{N}}]_{\Lbag \Rbag \hbar}$ & $\stackrel{\mathit{pred}}{\multimap}$ & $[\mathcal{N}_{\mathscr{E}}]$ \\
\cline{1-7}
&&&&&&\tikzmark{SuccConcatPredC1} $[\hat{q}_{\mathscr{E}}]$ \tikzmark{SuccConcatPredC10} \\
&&&&\tikzmark{SuccConcatPredC5} $[\hat{q}_{\mathscr{W} \mathscr{N}}]_{\Lbag \Rbag \hbar}$ \tikzmark{SuccConcatPredD1}&& \\
&&\tikzmark{SuccConcatPredC2} $[\hat{q}_{\mathscr{E} \mathscr{S}}]_{\Lbag \Rbag \hbar}$ \tikzmark{SuccConcatPredC4}&&&& \\
\tikzmark{SuccConcatPredC3} $[\hat{q}_{\mathscr{W}}]_{\Lbag \Lbag \Rbag \hbar \Lbag \Rbag \Rbag \hbar}$ \tikzmark{SuccConcatPredD2}&&&&&& \\
\tikzmark{SuccConcatPredD3} $[\mathit{no}_{\mathscr{W}}]_{\Lbag \Lbag \Rbag \hbar \Lbag \Rbag \Rbag \hbar}$&&&&&& \\
&&\tikzmark{SuccConcatPredC7} $[\mathit{yes}_{\mathscr{E} \mathscr{S}}]_{\Lbag \Rbag \hbar}$ \tikzmark{SuccConcatPredD4}&&&& \\
&&&&\tikzmark{SuccConcatPredD5} $[\mathit{yes}_{\mathscr{W} \mathscr{N}}]_{\Lbag \Rbag \hbar}$ \tikzmark{SuccConcatPredC6}&& \\
&&&&\tikzmark{SuccConcatPredC9} $[q_{\mathscr{W} \mathscr{N}}]_{\Lbag \Rbag \hbar}$ \tikzmark{SuccConcatPredD6}&& \\
&&\tikzmark{SuccConcatPredD7} $[q_{\mathscr{E} \mathscr{S}}]_{\Lbag \Rbag \hbar}$ \tikzmark{SuccConcatPredC8}&&&& \\
&&$[\mathit{no}_{\mathscr{E} \mathscr{S}}]_{\Lbag \Rbag \hbar}$ \tikzmark{SuccConcatPredD8}&&&& \\
&&&&\tikzmark{SuccConcatPredD9} $[\mathit{no}_{\mathscr{W} \mathscr{N}}]_{\Lbag \Rbag \hbar}$&& \\
&&&&&&$[\mathit{no}_{\mathscr{E}}]$ \tikzmark{SuccConcatPredD10}
\end{tabular}
\begin{tikzpicture}[overlay, remember picture, yshift=.25\baselineskip]
\draw [->] ({pic cs:SuccConcatPredD1}) to ({pic cs:SuccConcatPredC1});
\draw [->] ({pic cs:SuccConcatPredD2}) to ({pic cs:SuccConcatPredC2});
\draw [->] ({pic cs:SuccConcatPredD3}) [bend left] to ({pic cs:SuccConcatPredC3});
\draw [->] ({pic cs:SuccConcatPredD4}) [bend right] to ({pic cs:SuccConcatPredC4});
\draw [->] ({pic cs:SuccConcatPredD5}) [bend left] to ({pic cs:SuccConcatPredC5});
\draw [->] ({pic cs:SuccConcatPredD6}) [bend right] to ({pic cs:SuccConcatPredC6});
\draw [->] ({pic cs:SuccConcatPredD7}) [bend left] to ({pic cs:SuccConcatPredC7});
\draw [->] ({pic cs:SuccConcatPredD8}) [bend right] to ({pic cs:SuccConcatPredC8});
\draw [->] ({pic cs:SuccConcatPredD9}) [bend left] to ({pic cs:SuccConcatPredC9});
\draw [->] ({pic cs:SuccConcatPredD10}) [bend right] to ({pic cs:SuccConcatPredC10});
\draw [->] ({pic cs:SuccConcatPredC4}) to ({pic cs:SuccConcatPredC5});
\end{tikzpicture}
\end{center}
\end{example}

Finally, we introduce the \emph{currying} and the \emph{uncurrying} of strategies:
\begin{definition}[Currying and uncurrying of strategies \cite{abramsky1999game}]
If $\phi :  G$ (resp. $\psi :  H$) and $\mathcal{H}^\omega(G) \trianglelefteqslant A \otimes B \multimap C$ (resp. $\mathcal{H}^\omega(H) \trianglelefteqslant A \multimap (B \multimap C)$) for some normalized games $A$, $B$ and $C$, then its \emph{\bfseries currying} $\Lambda(\phi) : \Lambda(G)$ (resp. \emph{\bfseries uncurrying} $\Lambda^\circleddash(\psi) : \Lambda^\circleddash(H)$) is defined by:
\begin{align*}
\Lambda(\phi) &\stackrel{\mathrm{df. }}{=} \{ \boldsymbol{s} \in \mathscr{L}_{\Lambda(G)} \mid \mathit{peel}_{\Lambda(G)}^\ast(\boldsymbol{s}) \in \phi \ \! \} \\
\Lambda^\circleddash(\psi) &\stackrel{\mathrm{df. }}{=} \{ \boldsymbol{s} \in \mathscr{L}_{\Lambda^\circleddash(H)} \mid \mathit{peel}_{\Lambda^\circleddash(H)}^\ast(\boldsymbol{s}) \in \psi \ \! \}.
\end{align*}
\end{definition}

\begin{theorem}[Constructions on strategies \cite{yamada2016dynamic}]
The constructions $\otimes$, $\langle \_, \_ \rangle$, $(\_)^\dagger$, $\ddagger$, $;$, $\Lambda$ and $\Lambda^\circleddash$ on strategies are all well-defined.
\end{theorem}

\begin{lemma}[Hiding lemma on strategies \cite{yamada2016dynamic}]
\label{LemHidingLemmaOnStrategies}
Let $\spadesuit_{i \in I}$ be either $\otimes$, $\langle \_, \_ \rangle$, $(\_)^\dagger$, $\ddagger$, $\Lambda$ or $\Lambda^\circleddash$, and $d \in \mathbb{N} \cup \{ \omega \}$. Then, for any family $(\sigma_i)_{i \in I}$ of strategies, we have:
\begin{enumerate}

\item $\mathcal{H}^d(\spadesuit_{i \in I}\sigma_i) = \spadesuit_{i \in I} \mathcal{H}^d(\sigma_i)$ if $\spadesuit_{i \in I}$ is not $\ddagger$;

\item $\mathcal{H}^d (\sigma_1 \ddagger \sigma_2) = \mathcal{H}^d(\sigma_1) \ddagger \mathcal{H}^d(\sigma_2)$ if $\mathcal{H}^d(\sigma_1 \ddagger \sigma_2)$ is not yet normalized, and $\mathcal{H}^d(\sigma_1 \ddagger \sigma_2) = \mathcal{H}^d(\sigma_1) ; \mathcal{H}^d(\sigma_2)$ otherwise.

\end{enumerate}
\end{lemma}

\section{Viable strategies}
\label{ViableStrategies}
We have defined our games and strategies in the previous section.
In this main section of the paper, we introduce a novel notion of `effective' or \emph{viable} strategies, and show that they subsume all computations of the programming language \emph{PCF} \cite{plotkin1977lcf,mitchell1996foundations}, and therefore they are  \emph{Turing complete}.

In Section~\ref{SubsectionViableStrategies}, we define viability of strategies and show that it is preserved under the constructions on strategies defined in Section~\ref{ConstructionsOnStrategies}. We then describe various examples of viable strategies in Section~\ref{Examples}, based on which we finally prove in Section~\ref{TuringCompleteness} that viable strategies may interpret all terms of PCF.

\subsection{Viable strategies}
\label{SubsectionViableStrategies}
The idea of viable strategies is as follows.
First, it seems a reasonable idea to restrict the number of previous occurrences of moves which P is allowed to look at to calculate the next P-move to a \emph{bounded} one.\footnote{Note that this is analogous to the computation of a TM which looks at only one cell of an infinite tape at a time.}
Fortunately, to model the language PCF, it turns out that strategies only need to read off at most \emph{last three moves} in the P-view (and possibly a few initial and internal moves) as we shall see, which is clearly `effectively' achievable in an informal sense.
Thus, it remains to formalize how strategies `effectively' compute the next P-move from such a bounded number of previous occurrences. 
Note that (as already mentioned) computation of internal O-moves should be done by P, but it is rather trivial by the axiom DP2, and therefore we do not have to take into account it.
We focus on \emph{innocent} strategies as a means to narrow down previous occurrences to be concerned with.\footnote{Of course, there would be another means to `effectively' eliminate irrelevant moves from the history of previous moves; in fact, we need more than P-views in order to model languages with \emph{states} \cite{abramsky1999game}, which is left as future work.}

As the set $\pi_1(M_G)$ is finite for any game $G$, innocent strategies that are \emph{finitary} in the sense that their view functions are finite seem sufficient at first glance.
However, to model \emph{fixed-point combinators} in PCF, strategies need to initiate new threads \emph{unboundedly many times} \cite{hyland2000full,abramsky2000full}; also, they have to model promotion in which possible outer tags are infinitely many.
Thus, finitary strategies are not enough. 

Then, how can we define a stronger notion of `effective computability' of the next P-move from (a bounded number of) previous occurrences \emph{solely in terms of games and strategies}? 
Our solution is, which is the main achievement of the present paper, to define a strategy $\sigma : G$ to be `effective' or \emph{viable} if it is `describable' by a finitary strategy, called an \emph{instruction strategy}, on the \emph{instruction game} of $G$:
\begin{notation}
For conceptual clarity, we assign a unique symbol $\mathsf{m}$ to each inner element $m \in \pi_1(M_G)$ for a game $G$, for which we may assume that these symbols are \emph{pairwise distinct} since the set $\pi_1 (M_G)$ is finite.
Also, we assign symbols to elements of outer tags by a function $\mathscr{C} : \ell \mapsto \prime$, $\hbar \mapsto \sharp$, $\Lbag \mapsto \langle$ and $\Rbag \mapsto \rangle$.
Let us define $\mathsf{Sym}(\pi_1 (M_G)) \stackrel{\mathrm{df. }}{=} \{ \mathsf{m} \mid m \in \pi_1(M_G) \}$, $\mathsf{Sym}(\mathcal{T})  \stackrel{\mathrm{df. }}{=} \{ \prime, \sharp, \langle, \rangle \}$ and $\mathsf{Sym}(M_G) \stackrel{\mathrm{df. }}{=} \mathsf{Sym}(\pi_1(M_G)) \times \mathsf{Sym}(\mathcal{T})$.
(N.b., technically, these new symbols are not necessary; they are just for clarity.)
\end{notation}

\begin{definition}[Instruction games]
The \emph{\bfseries instruction game} $\mathcal{G}(M_G)$ on a game $G$ is the product $\mathcal{G}(\pi_1(M_G)) \& \mathcal{G}(\mathcal{T})$, where $\mathcal{G}(\mathcal{T})$ is the tag game (Example~\ref{ExTagGame}), and the game $\mathcal{G}(\pi_1(M_G))$ is defined by: 
\begin{itemize}

\item $M_{\mathcal{G}(\pi_1(M_G))} \stackrel{\mathrm{df. }}{=} \{ [\hat{q}_G], [\square] \} \cup \{ [\mathsf{m}] \ \! | \ \! \mathsf{m} \in \mathsf{Sym}(\pi_1(M_G)) \}$, where $\hat{q}_G$ and $\square$ are arbitrary chosen elements with $\hat{q}_G \neq \square \wedge \{ \hat{q}_G, \square \} \cap \mathsf{Sym}(\pi_1(M_G)) = \emptyset$;

\item $\lambda_{\mathcal{G}(\pi_1(M_G))} : [\hat{q}_G] \mapsto (\mathsf{O}, \mathsf{Q}, 0)$, $[\square] \mapsto (\mathsf{P}, \mathsf{A}, 0)$, $[\mathsf{m}] \mapsto (\mathsf{P}, \mathsf{A}, 0)$;

\item $\vdash_{\mathcal{G}(\pi_1(M_G))} \stackrel{\mathrm{df. }}{=} \{ (\star, [\hat{q}_G]), ([\hat{q}_G], [\square]) \} \cup \{ ([\hat{q}_G], [\mathsf{m}]) \ \! | \ \! \mathsf{m} \in \mathsf{Sym}(\pi_1(M_G)) \}$;

\item $\Delta_{\mathcal{G}(\pi_1(M_G))} \stackrel{\mathrm{df. }}{=} \emptyset$;

\item $P_{\mathcal{G}(\pi_1(M_G))} \stackrel{\mathrm{df. }}{=} \mathsf{Pref} (\{ [\hat{q}_G] . [\square] \} \cup \{ [\hat{q}_G] . [\mathsf{m}] \ \! | \ \! \mathsf{m} \in \mathsf{Sym}(\pi_1(M_G)) \})$, where $[\hat{q}_G]$ justifies $[\square]$ and $[\mathsf{m}]$.


\end{itemize}
\end{definition}
The positions $[\hat{q}_G] . [\mathsf{m}]$ and $[\hat{q}_G] . [\square]$ are to represent the inner element $m \in \pi_1(M_G)$ and `no element', respectively, which will be clearer shortly.

\begin{convention}
Pointers in positions of instruction games are rather trivial, and thus we usually omit them. 
\end{convention}

\begin{notation}
Let $G$ be a game, and $[m]_{\boldsymbol{e}} \in M_G$ with $\boldsymbol{e} = e_1 e_2 \dots e_k \in \mathcal{T}$. We write $\underline{[m]_{\boldsymbol{e}}}$ for the strategy $\langle \underline{m}, \underline{\boldsymbol{e}}\rangle : \mathcal{G}(M_G)$, where $\underline{m} : \mathcal{G}(\pi_1(M_G))$ and $\underline{\boldsymbol{e}} : \mathcal{G}(\mathcal{T})$ are defined respectively by $\underline{m} \stackrel{\mathrm{df. }}{=} \mathsf{Pref}(\{ [\hat{q}_G] . [\mathsf{m}] \})^{\mathsf{Even}}$ and $\underline{\boldsymbol{e}} \stackrel{\mathrm{df. }}{=} \mathsf{Pref}(\{ [\hat{q}_{\mathcal{T}}] . [\mathscr{C}(e_1)] . [q_{\mathcal{T}}] . [\mathscr{C}(e_2)] \dots [q_{\mathcal{T}}] . [\mathscr{C}(e_k)] . [q_{\mathcal{T}}] . [\checkmark] \})^{\mathsf{Even}}$. 
Similarly, we define $\underline{\square} \stackrel{\mathrm{df. }}{=} \mathsf{Pref}(\{ [\hat{q}_G] . [\square] \})^{\mathsf{Even}} : \mathcal{G}(\pi_1(M_G))$ and $\underline{[\square]} \stackrel{\mathrm{df. }}{=} \langle \underline{\square}, \underline{\boldsymbol{\epsilon}} \rangle : \mathcal{G}(M_G)$.
For any finite sequence $\boldsymbol{s} = [m_l]_{\boldsymbol{e}_l} [m_{l-1}]_{\boldsymbol{e}_{l-1}} \dots [m_1]_{\boldsymbol{e}_1} \in M_G^\ast$ of moves and number $n \geqslant l$, we define $\underline{\boldsymbol{s}}_n \stackrel{\mathrm{df. }}{=} \langle \underbrace{\underline{[\square]}, \dots, \underline{[\square]}}_{n-l}, \underline{[m_l]_{\boldsymbol{e}_l}}, \underline{[m_{l-1}]_{\boldsymbol{e}_{l-1}}}, \dots, \underline{[m_1]_{\boldsymbol{e}_1}} \rangle : \mathcal{G}(M_G)^n$ and $\mathcal{G}(M_G)^n \stackrel{\mathrm{df. }}{=} \underbrace{\mathcal{G}(M_G) \& \mathcal{G}(M_G) \dots \& \mathcal{G}(M_G)}_{n}$, where the $n$-ary pairing and product are just abbreviations of the $(n-1)$-times iteration of the usual (binary) ones from the left.
Given a strategy $\sigma : \mathcal{G}(M_G)$, we define $\mathcal{M}(\sigma)$ to be the unique move in $M_G$ such that $\underline{\mathcal{M}(\sigma)} = \sigma$ if it exists, and undefined otherwise.
\end{notation}

The idea is to `describe' or \emph{realize} an `effective' strategy $\sigma$ on a game $G$ by a finitary strategy $\mathcal{A}(\sigma)^\circledS$ on the game  $\mathcal{G}(M_G)^3 \Rightarrow \mathcal{G}(M_G)$.
For instance, recall the successor strategy $\mathit{succ} : \mathcal{N} \Rightarrow \mathcal{N}$ in Example~\ref{ExSuccPred}, which computes as in Figure~\ref{FormalSucc} below.
It is then easy to see that this computation can be described by a finite partial function $(m_3, m_2, m_1) \mapsto m$ as already mentioned in the introduction, where $m_1$, $m_2$, $m_3$ are the last, second last and third last occurrences in the P-view of the current odd-length position, respectively, and $m$ is the next P-move.
Concretely, the finite partial function is given by the following table:
\begin{align*}
&(\square, \square, [\hat{q}_{\mathscr{E}}]) \mapsto [\hat{q}_{\mathscr{W}}]_{\Lbag \Rbag \hbar} \mid ([\mathit{no}_{\mathscr{W}}]_{\Lbag \Rbag \hbar}, \_, [q_{\mathscr{E}}]) \mapsto [\mathit{no}_{\mathscr{E}}] \mid \\ &(\_, \_, [\mathit{yes}_{\mathscr{W}}]_{\Lbag \Rbag \hbar}) \mapsto [\mathit{yes}_{\mathscr{E}}] \mid (\_, \_, [\mathit{no}_{\mathscr{W}}]_{\Lbag \Rbag \hbar}) \mapsto [\mathit{yes}_{\mathscr{E}}] \mid \\ &([\mathit{yes}_{\mathscr{W}}]_{\Lbag \Rbag \hbar}, \_, [q]_{\mathscr{E}}) \mapsto [q_{\mathscr{W}}]_{\Lbag \Rbag \hbar}
\end{align*}
where $\square$ and $\_$ mean `no move' and `any move', respectively.
Therefore, it can be realized by a finitary strategy $\mathcal{A}(\mathit{succ})^\circledS : \mathcal{G}(M_{\mathcal{N} \Rightarrow \mathcal{N}})^3 \Rightarrow \mathcal{G}(M_{\mathcal{N} \Rightarrow \mathcal{N}})$ in the sense that it satisfies $\mathcal{A}(\mathit{succ})^\circledS \circ \langle \underline{m_3}, \underline{m_2}, \underline{m_1} \rangle^\dagger = \underline{m}$ for all input-output pairs $(m_3, m_2, m_1) \mapsto m$ of $\mathit{succ}$.
Diagrammatically, $\mathcal{A}(\mathit{succ})^\circledS$ computes as follows:
\begin{figure}
\begin{center}
\begin{tabular}{cccccccccccc}
$[\mathcal{N}_{\mathscr{W}}]_{\Lbag \boldsymbol{e} \Rbag \hbar}$ && $\stackrel{\mathit{succ}}{\Rightarrow}$ && $[\mathcal{N}_{\mathscr{E}}]$ &&& $[\mathcal{N}_{\mathscr{W}}]_{\Lbag \boldsymbol{e} \Rbag \hbar}$ && $\stackrel{\mathit{succ}}{\Rightarrow}$ && $[\mathcal{N}_{\mathscr{E}}]$ \\ \cline{1-5} \cline{8-12}
&&&&\tikzmark{csuccF1} $[\hat{q}_{\mathscr{E}}]$ \tikzmark{csuccF3} &&&& &&&\tikzmark{csuccF21} $[\hat{q}_{\mathscr{E}}]$ \tikzmark{csuccF23} \\
\tikzmark{csuccF2} $[\hat{q}_{\mathscr{W}}]_{\Lbag \Rbag \hbar}$ \tikzmark{dsuccF1}&&&& &&& \tikzmark{csuccF22} $[\hat{q}_{\mathscr{W}}]_{\Lbag \Rbag \hbar}$ \tikzmark{dsuccF21}&&&& \\
\tikzmark{dsuccF2} $[\mathit{yes}_{\mathscr{W}}]_{\Lbag \Rbag \hbar}$ \tikzmark{csuccF5}&&&& &&& \tikzmark{dsuccF22} $[\mathit{no}_{\mathscr{W}}]_{\Lbag \Rbag \hbar}$ \tikzmark{csuccF25}&&&& \\
&&&&\tikzmark{csuccF4} $[\mathit{yes}_{\mathscr{E}}]$ \tikzmark{dsuccF3} &&&& &&&\tikzmark{csuccF24} $[\mathit{yes}_{\mathscr{E}}]$ \tikzmark{dsuccF23} \\
&&&&\tikzmark{dsuccF4} $[q_{\mathscr{E}}]$ \tikzmark{csuccF9} &&&& &&&\tikzmark{dsuccF24} $[q_{\mathscr{E}}]$ \tikzmark{csuccF29} \\
\tikzmark{csuccF8} $[q_{\mathscr{W}}]_{\Lbag \Rbag \hbar}$ \tikzmark{dsuccF5}&&&& &&& &&&&\tikzmark{csuccF27} $[\mathit{no}_{\mathscr{E}}]$ \tikzmark{dsuccF29} \\
\tikzmark{dsuccF8} $[\mathit{yes}_{\mathscr{W}}]_{\Lbag \Rbag \hbar}$ \tikzmark{csuccF6}&&&& &&&& &&& \\
&&&&\tikzmark{csuccF7} $[\mathit{yes}_{\mathscr{E}}]$ \tikzmark{dsuccF9} &&& &&&& \\
&&&&\tikzmark{dsuccF7} $[q_{\mathscr{E}}]$ &&& &&&& \\
$[q_{\mathscr{W}}]_{\Lbag \Rbag \hbar}$ \tikzmark{dsuccF6}&&&& &&& &&&& \\
&&$\vdots$&& &&&& &&& \\
$[\mathit{yes}_{\mathscr{W}}]_{\Lbag \Rbag \hbar}$ \tikzmark{csuccF10}&&&& &&& &&&& \\
&&&&\tikzmark{csuccF11} $[\mathit{yes}_{\mathscr{E}}]$ &&& &&&& \\
&&&&\tikzmark{dsuccF11} $[q_{\mathscr{E}}]$ \tikzmark{csuccF13} &&&&& && \\
\tikzmark{csuccF12} $[q_{\mathscr{W}}]_{\Lbag \Rbag \hbar}$ \tikzmark{dsuccF10}&&&& &&& &&&& \\
\tikzmark{dsuccF12} $[\mathit{no}_{\mathscr{W}}]_{\Lbag \Rbag \hbar}$&&&& &&& &&&& \\
&&&&\tikzmark{csuccF14} $[\mathit{yes}_{\mathscr{E}}]$ \tikzmark{dsuccF13} &&&& &&& \\
&&&&\tikzmark{dsuccF14} $[q_{\mathscr{E}}]$ \tikzmark{csuccF15} &&&& &&& \\
&&&& $[\mathit{no}_{\mathscr{E}}]$ \tikzmark{dsuccF15} &&&& &&& 
\end{tabular}
\begin{tikzpicture}[overlay, remember picture, yshift=.25\baselineskip]
    \draw [->] ({pic cs:dsuccF1}) to ({pic cs:csuccF1});
    \draw [->] ({pic cs:dsuccF2}) [bend left] to ({pic cs:csuccF2});
    \draw [->] ({pic cs:dsuccF3}) [bend right] to ({pic cs:csuccF3});
    \draw [->] ({pic cs:dsuccF4}) [bend left] to ({pic cs:csuccF4});
    \draw [->] ({pic cs:dsuccF5}) [bend right] to ({pic cs:csuccF5});
    \draw [->] ({pic cs:dsuccF6}) [bend right] to ({pic cs:csuccF6}); 
    \draw [->] ({pic cs:dsuccF7}) [bend left] to ({pic cs:csuccF7});
    \draw [->] ({pic cs:dsuccF8}) [bend left] to ({pic cs:csuccF8});
    \draw [->] ({pic cs:dsuccF9}) [bend right] to ({pic cs:csuccF9});
    \draw [->] ({pic cs:dsuccF10}) [bend right] to ({pic cs:csuccF10});
    \draw [->] ({pic cs:dsuccF11}) [bend left] to ({pic cs:csuccF11});
    \draw [->] ({pic cs:dsuccF12}) [bend left] to ({pic cs:csuccF12});
    \draw [->] ({pic cs:dsuccF13}) [bend right] to ({pic cs:csuccF13});
    \draw [->] ({pic cs:dsuccF14}) [bend left] to ({pic cs:csuccF14});
    \draw [->] ({pic cs:dsuccF15}) [bend right] to ({pic cs:csuccF15});
    \draw [->] ({pic cs:dsuccF21}) to ({pic cs:csuccF21});
    \draw [->] ({pic cs:dsuccF22}) [bend left] to ({pic cs:csuccF22});
    \draw [->] ({pic cs:dsuccF23}) [bend right] to ({pic cs:csuccF23});
    \draw [->] ({pic cs:dsuccF24}) [bend left] to ({pic cs:csuccF24});
    \draw [->] ({pic cs:dsuccF29}) [bend right] to ({pic cs:csuccF29});
  \end{tikzpicture}
  \caption{The successor strategy $\mathit{succ} : \mathcal{N} \Rightarrow \mathcal{N}$}
  \label{FormalSucc}
  \end{center}
\end{figure}

\begin{center}
\begin{tabular}{ccccccc}
$\mathcal{G}(M_{\mathcal{N} \Rightarrow \mathcal{N}})^{[0]}$ & $\&$ & $\mathcal{G}(M_{\mathcal{N} \Rightarrow \mathcal{N}})^{[1]}$ & $\&$ & $\mathcal{G}(M_{\mathcal{N} \Rightarrow \mathcal{N}})^{[2]}$ & $\stackrel{\mathcal{A}(\mathit{succ})^\circledS}{\Rightarrow}$ & $\mathcal{G}(M_{\mathcal{N} \Rightarrow \mathcal{N}})^{[3]}$ \\
\cline{1-7}
&&&&&&$(\hat{q}_{\mathcal{N} \Rightarrow \mathcal{N}})^{[3]}$ \\
&&&&$(\hat{q}_{\mathcal{N} \Rightarrow \mathcal{N}})^{[2]}$&& \\
&&&&$(\mathsf{\hat{q}_E})^{[2]}$&& \\
&&&&&&$(\mathsf{\hat{q}_W})^{[3]}$
\end{tabular}
\end{center}

\begin{center}
\begin{tabular}{ccccccc}
$\mathcal{G}(M_{\mathcal{N} \Rightarrow \mathcal{N}})^{[0]}$ & $\&$ & $\mathcal{G}(M_{\mathcal{N} \Rightarrow \mathcal{N}})^{[1]}$ & $\&$ & $\mathcal{G}(M_{\mathcal{N} \Rightarrow \mathcal{N}})^{[2]}$ & $\stackrel{\mathcal{A}(\mathit{succ})^\circledS}{\Rightarrow}$ & $\mathcal{G}(M_{\mathcal{N} \Rightarrow \mathcal{N}})^{[3]}$ \\
\cline{1-7}
&&&&&&$(\hat{q}_{\mathcal{T}})^{[3]}$ \\
&&&&$(\hat{q}_{\mathcal{N} \Rightarrow \mathcal{N}})^{[2]}$&& \\
&&&&$(\mathsf{\hat{q}_E})^{[2]}$&& \\
&&&&&&$(\langle)^{[3]}$ \\
&&&&&&$(q_{\mathcal{T}})^{[3]}$ \\
&&&&&&$(\rangle)^{[3]}$ \\
&&&&&&$(q_{\mathcal{T}})^{[3]}$ \\
&&&&&&$(\sharp)^{[3]}$ \\
&&&&&&$(q_{\mathcal{T}})^{[3]}$ \\
&&&&&&$(\checkmark)^{[3]}$
\end{tabular}
\end{center}

\begin{center}
\begin{tabular}{ccccccc}
$\mathcal{G}(M_{\mathcal{N} \Rightarrow \mathcal{N}})^{[0]}$ & $\&$ & $\mathcal{G}(M_{\mathcal{N} \Rightarrow \mathcal{N}})^{[1]}$ & $\&$ & $\mathcal{G}(M_{\mathcal{N} \Rightarrow \mathcal{N}})^{[2]}$ & $\stackrel{\mathcal{A}(\mathit{succ})^\circledS}{\Rightarrow}$ & $\mathcal{G}(M_{\mathcal{N} \Rightarrow \mathcal{N}})^{[3]}$ \\
\cline{1-7}
&&&&&&$(\hat{q}_{\mathcal{N} \Rightarrow \mathcal{N}})^{[3]}$ \\
&&&&$(\hat{q}_{\mathcal{N} \Rightarrow \mathcal{N}})^{[2]}$&& \\
&&&&$(\mathsf{q_E})^{[2]}$&& \\
$(\hat{q}_{\mathcal{N} \Rightarrow \mathcal{N}})^{[0]}$&&&&&& \\ 
$(\mathsf{no_W})^{[0]}$&&&&&& \\
&&&&&&$(\mathsf{no_E})^{[3]}$
\end{tabular}
\end{center}

\begin{center}
\begin{tabular}{ccccccc}
$\mathcal{G}(M_{\mathcal{N} \Rightarrow \mathcal{N}})^{[0]}$ & $\&$ & $\mathcal{G}(M_{\mathcal{N} \Rightarrow \mathcal{N}})^{[1]}$ & $\&$ & $\mathcal{G}(M_{\mathcal{N} \Rightarrow \mathcal{N}})^{[2]}$ & $\stackrel{\mathcal{A}(\mathit{succ})^\circledS}{\Rightarrow}$ & $\mathcal{G}(M_{\mathcal{N} \Rightarrow \mathcal{N}})^{[3]}$ \\
\cline{1-7}
&&&&&&$(\hat{q}_{\mathcal{T}})^{[3]}$ \\
&&&&$(\hat{q}_{\mathcal{N} \Rightarrow \mathcal{N}})^{[2]}$&& \\
&&&&$(\mathsf{q_E})^{[2]}$&& \\
$(\hat{q}_{\mathcal{N} \Rightarrow \mathcal{N}})^{[0]}$&&&&&& \\
$(\mathsf{no_W})^{[0]}$&&&&&& \\
&&&&&&$(\checkmark)^{[3]}$ \\
\end{tabular}
\end{center}

\begin{center}
\begin{tabular}{ccccccc}
$\mathcal{G}(M_{\mathcal{N} \Rightarrow \mathcal{N}})^{[0]}$ & $\&$ & $\mathcal{G}(M_{\mathcal{N} \Rightarrow \mathcal{N}})^{[1]}$ & $\&$ & $\mathcal{G}(M_{\mathcal{N} \Rightarrow \mathcal{N}})^{[2]}$ & $\stackrel{\mathcal{A}(\mathit{succ})^\circledS}{\Rightarrow}$ & $\mathcal{G}(M_{\mathcal{N} \Rightarrow \mathcal{N}})^{[3]}$ \\
\cline{1-7}
&&&&&&$(\hat{q}_{\mathcal{N} \Rightarrow \mathcal{N}})^{[3]}$ \\
&&&&$(\hat{q}_{\mathcal{N} \Rightarrow \mathcal{N}})^{[2]}$&& \\
&&&&$(\mathsf{yes_W})^{[2]}$ ($(\mathsf{no_W})^{[2]}$)&& \\
&&&&&&$(\mathsf{yes_E})^{[3]}$
\end{tabular}
\end{center}

\begin{center}
\begin{tabular}{ccccccc}
$\mathcal{G}(M_{\mathcal{N} \Rightarrow \mathcal{N}})^{[0]}$ & $\&$ & $\mathcal{G}(M_{\mathcal{N} \Rightarrow \mathcal{N}})^{[1]}$ & $\&$ & $\mathcal{G}(M_{\mathcal{N} \Rightarrow \mathcal{N}})^{[2]}$ & $\stackrel{\mathcal{A}(\mathit{succ})^\circledS}{\Rightarrow}$ & $\mathcal{G}(M_{\mathcal{N} \Rightarrow \mathcal{N}})^{[3]}$ \\
\cline{1-7}
&&&&&&$(\hat{q}_{\mathcal{T}})^{[3]}$ \\
&&&&$(\hat{q}_{\mathcal{N} \Rightarrow \mathcal{N}})^{[2]}$&& \\
&&&&$(\mathsf{yes_W})^{[2]}$ ($(\mathsf{no_W})^{[2]}$)&& \\
&&&&&&$(\checkmark)^{[3]}$
\end{tabular}
\end{center}

\begin{center}
\begin{tabular}{ccccccc}
$\mathcal{G}(M_{\mathcal{N} \Rightarrow \mathcal{N}})^{[0]}$ & $\&$ & $\mathcal{G}(M_{\mathcal{N} \Rightarrow \mathcal{N}})^{[1]}$ & $\&$ & $\mathcal{G}(M_{\mathcal{N} \Rightarrow \mathcal{N}})^{[2]}$ & $\stackrel{\mathcal{A}(\mathit{succ})^\circledS}{\Rightarrow}$ & $\mathcal{G}(M_{\mathcal{N} \Rightarrow \mathcal{N}})^{[3]}$ \\
\cline{1-7}
&&&&&&$(\hat{q}_{\mathcal{N} \Rightarrow \mathcal{N}})^{[3]}$ \\
&&&&$(\hat{q}_{\mathcal{N} \Rightarrow \mathcal{N}})^{[2]}$&& \\
&&&&$(\mathsf{q_E})^{[2]}$&& \\
$(\hat{q}_{\mathcal{N} \Rightarrow \mathcal{N}})^{[0]}$&&&&&& \\ 
$(\mathsf{yes_W})^{[0]}$&&&&&& \\
&&&&&&$(\mathsf{q_W})^{[3]}$ 
\end{tabular}
\end{center}

\begin{center}
\begin{tabular}{ccccccc}
$\mathcal{G}(M_{\mathcal{N} \Rightarrow \mathcal{N}})^{[0]}$ & $\&$ & $\mathcal{G}(M_{\mathcal{N} \Rightarrow \mathcal{N}})^{[1]}$ & $\&$ & $\mathcal{G}(M_{\mathcal{N} \Rightarrow \mathcal{N}})^{[2]}$ & $\stackrel{\mathcal{A}(\mathit{succ})^\circledS}{\Rightarrow}$ & $\mathcal{G}(M_{\mathcal{N} \Rightarrow \mathcal{N}})^{[3]}$ \\
\cline{1-7}
&&&&&&$(\hat{q}_{\mathcal{T}})^{[3]}$ \\
&&&&$(\hat{q}_{\mathcal{N} \Rightarrow \mathcal{N}})^{[2]}$&& \\
&&&&$(\mathsf{q_E})^{[2]}$&& \\
$(\hat{q}_{\mathcal{N} \Rightarrow \mathcal{N}})^{[0]}$&&&&&& \\ 
$(\mathsf{yes_W})^{[0]}$&&&&&& \\
&&&&&&$(\langle)^{[3]}$ \\
&&&&&&$(q_{\mathcal{T}})^{[3]}$ \\
&&&&&&$(\rangle)^{[3]}$ \\
&&&&&&$(q_{\mathcal{T}})^{[3]}$ \\
&&&&&&$(\sharp)^{[3]}$ \\
&&&&&&$(q_{\mathcal{T}})^{[3]}$ \\
&&&&&&$(\checkmark)^{[3]}$
\end{tabular}
\end{center}

Hence, we are particularly concerned with games of the form $\mathcal{G}(M_G)^3 \Rightarrow \mathcal{G}(M_G)$, where $G$ is a game.
We loosely call $\mathcal{G}(M_G)^3 \Rightarrow \mathcal{G}(M_G)$ \emph{instruction games} as well.

Clearly, the strategy $\mathcal{A}(\mathit{succ})^\circledS$ is \emph{finitary}\footnote{Note that $\mathit{succ}$ is already finitary. $\mathcal{A}(\mathit{succ})^\circledS$ is just for an illustration, and the idea of `realizing strategies by strategies' is necessary only for more complex strategies.}, and it correctly `describes' the computation of $\mathit{succ}$.
In Definition~\ref{DefStAlgorithms} below, we define formalized finite tables $\mathcal{A}(\sigma)$ for such strategies $\mathcal{A}(\sigma)^\circledS : \mathcal{G}(M_G)^3 \Rightarrow \mathcal{G}(M_G)$ on $\sigma : G$ as \emph{st-algorithm} for $\sigma : G$.

\begin{notation}
We write `tags' in $\mathcal{G}(M_G)^3 \Rightarrow \mathcal{G}(M_G)$ \emph{informally} for brevity, e.g., $\mathcal{G}(M_G)^{[0]} \Rightarrow \mathcal{G}(M_G)^{[1]}$, $\mathcal{G}(M_G)^{[0]} \& \mathcal{G}(M_G)^{[1]} \& \mathcal{G}(M_G)^{[2]}  \Rightarrow \mathcal{G}(M_G)^{[3]}$, $(\hat{q}_G)^{[0]}$, $(q_{\mathcal{T}})^{[1]}$.
\end{notation}

However, there remain two points to overcome. 
The first point is the pairing $\langle \sigma, \tau \rangle : \langle L, R \rangle$ of strategies $\sigma : L$ and $\tau : R$ such that $\mathcal{H}^\omega(L) \trianglelefteqslant C \Rightarrow A$ and $\mathcal{H}^\omega(R) \trianglelefteqslant C \Rightarrow B$: Because moves of $C$ are common to $\sigma$ and $\tau$, the last three moves in each P-view may not suffice; the pairing $\langle \sigma, \tau \rangle$ needs to know whether $A$ or $B$ the first occurrence of a move in each play belongs to. We shall record this information as \emph{states} of positions (see Definition~\ref{DefStAlgorithms} below). 

The second point is how to `effectively' calculate the finite and `relevant' part of outer tags.
For this point, we introduce the notion of \emph{m-views} as follows.
Note first that omitting tags for brevity occurrences of $\langle$ and $\rangle$ in any position $\boldsymbol{s} \in P_{\mathcal{G}(M_G)^3 \Rightarrow \mathcal{G}(M_G)}$ form \emph{unique} pairs similarly to `QA-pairs' for \emph{well-bracketing} (Definition~\ref{DefWellBracketing}): Each occurrence of $\rangle$ is paired with the most recent yet unpaired occurrence of $\langle$ in the same component game $\mathcal{G}(\mathcal{T})$ in $\mathcal{G}(M_G)^3 \Rightarrow \mathcal{G}(M_G)$; one in such a pair is called the \emph{\bfseries mate} of the other. 
Then, let us define:
\begin{definition}[M-views]
Let $G$ be a game, and assume $\boldsymbol{s} \in P_{\mathcal{G}(M_G)^3 \Rightarrow \mathcal{G}(M_G)}$, where we omit tags in $\boldsymbol{s}$ for brevity.
The \emph{\bfseries depth} of an occurrence of $\langle$ in $\boldsymbol{s}$ is the number of previous occurrences of $\langle$ in the same component game $\mathcal{G}(\mathcal{T})$ whose mate does not occur before that occurrence; the \emph{\bfseries depth} of an occurrence of $\rangle$ in $\boldsymbol{s}$ is the depth of its mate.
The \emph{\bfseries matching view} (\emph{\bfseries m-view}) $\llbracket \boldsymbol{s} \rrbracket_G^d$ of $\boldsymbol{s}$ up to depth $d \in \mathbb{N}$ is the j-subsequence of $\boldsymbol{s}$ that consists of occurrences of $\langle$ or $\rangle$ of depth $\leqslant d$.
\end{definition}

\begin{notation}
Given a finite sequence $\boldsymbol{s} = x_k x_{k-1} \dots x_1$ and a natural number $l \in \mathbb{N}$, we define $\boldsymbol{s} \downharpoonright l \stackrel{\mathrm{df. }}{=} \begin{cases} \boldsymbol{s} &\text{if $l \geqslant k$;} \\ x_l x_{l-1} \dots x_1 &\text{otherwise.} \end{cases}$
A function $f : \pi_1(M_G) \to \{ \top, \bot \}$, where $G$ is a game, and $\top$ and $\bot$ are any symbols with $\top \neq \bot$, induces another function $f^\star : M_G^\ast \to \pi_1(M_G)^\ast$ defined by $f^\star([m_k]_{\boldsymbol{e}_k} [m_{k-1}]_{\boldsymbol{e}_{k-1}} \dots [m_1]_{\boldsymbol{e}_1}) \stackrel{\mathrm{df. }}{=} m_{i_l} m_{i_{l-1}} \dots m_{i_1}$, where $l \leqslant k$ and $m_{i_l} m_{i_{l-1}} \dots m_{i_1}$ is the subsequence of $m_k m_{k-1} \dots m_1$ that consists of $m_{i_j}$ such that $f(m_{i_j}) = \top$ for $j = 1, 2, \dots, l$.
\end{notation}

It is clearly `effective' to calculate the m-view of a position in an informal sense.
We are now ready to make the notion of `describable by a finitary strategy' precise:
\begin{definition}[St-algorithms]
\label{DefStAlgorithms}
An \emph{\bfseries st-algorithm} $\mathcal{A}$ on a game $G$, written $\mathcal{A} :: G$, is a family $\mathcal{A} = (\mathcal{A}_{\boldsymbol{m}})_{\boldsymbol{m} \in \mathcal{S}_{\mathcal{A}}}$ of finite partial functions $\mathcal{A}_{\boldsymbol{m}} : \partial_{\boldsymbol{m}} (P_{\mathcal{G}(M_G)^3 \Rightarrow \mathcal{G}(M_G)}^{\mathsf{Odd}}) \rightharpoonup M_{\mathcal{G}(M_G)^3 \Rightarrow \mathcal{G}(M_G)}$, which also specifies the justifier of each output in the input (but we usually treat this structure \emph{implicit} as in \cite{abramsky1999game,mccusker1998games}), where: 
\begin{itemize}

\item $\mathcal{S}_{\mathcal{A}} \subseteq \pi_1(M_G)^\ast$ is a finite set of \emph{\bfseries states};

\item $\partial_{\boldsymbol{m}} (\boldsymbol{t}x) \stackrel{\mathrm{df. }}{=} (\boldsymbol{t}x \! \downharpoonright |\mathcal{A}_{\boldsymbol{m}}|, \llbracket \boldsymbol{t}x \rrbracket_G^{\|\mathcal{A}_{\boldsymbol{m}}\|})$ for all $\boldsymbol{t}x \in P_{\mathcal{G}(M_G)^3 \Rightarrow \mathcal{G}(M_G)}^{\mathsf{Odd}}$, where $|\mathcal{A}_{\boldsymbol{m}}|, \| \mathcal{A}_{\boldsymbol{m}} \| \in \mathbb{N}$ are natural numbers assigned to $\mathcal{A}_{\boldsymbol{m}}$, called the \emph{\bfseries view-scope} and the \emph{\bfseries mate-scope} of $\mathcal{A}_{\boldsymbol{m}}$, respectively

\end{itemize}
equipped with the \emph{\bfseries query (function)} $\mathcal{Q}_{\mathcal{A}} : \pi_1(M_G) \to \{ \top, \bot \}$ that satisfies:
\begin{itemize}

\item \textsc{(Q1)} $[m]_{\boldsymbol{e}} \in M_G^{\mathsf{Init}} \Rightarrow \mathcal{Q}_{\mathcal{A}}(m) = \top$;

\item \textsc{(Q2)} $\mathcal{Q}_{\mathcal{A}}(m) = \top \Rightarrow \exists \boldsymbol{e} \in \mathcal{T} . \ \! [m]_{\boldsymbol{e}} \in M_G^{\mathsf{Init}} \cup M_G^{\mathsf{Int}}$.

\end{itemize}

\end{definition}

\begin{remark}
The axiom Q2 is not necessary for the main theorem in Section~\ref{TuringCompleteness}, but it is informative and reasonable to ensure that states are initial or internal moves.
\end{remark}

\begin{definition}[Instruction strategies]
Given a game $G$, an st-algorithm $\mathcal{A} :: G$ and a state $\boldsymbol{m} \in \mathcal{S}_{\mathcal{A}}$, the \emph{\bfseries instruction strategy} $\mathcal{A}_{\boldsymbol{m}}^\circledS$ of $\mathcal{A}$ at $\boldsymbol{m}$ is the strategy on the game $\mathcal{G}(M_G)^3 \Rightarrow \mathcal{G}(M_G)$ defined by: 
\begin{equation*}
\mathcal{A}_{\boldsymbol{m}}^\circledS \stackrel{\mathrm{df. }}{=} \{ \boldsymbol{\epsilon} \} \cup \{ \boldsymbol{t}xy \in P_{\mathcal{G}(M_G)^3 \Rightarrow \mathcal{G}(M_G)} \mid \boldsymbol{t} \in \mathcal{A}_{\boldsymbol{m}}^\circledS, \mathcal{A}_{\boldsymbol{m}} \circ \partial_{\boldsymbol{m}}(\boldsymbol{t}x) \downarrow, y = \mathcal{A}_{\boldsymbol{m}} \circ \partial_{\boldsymbol{m}}(\boldsymbol{t}x) \}.
\end{equation*}
\end{definition}

\begin{convention}
Since an st-algorithm $\mathcal{A} :: G$ refers to m-views \emph{only occasionally}, we treat $\mathcal{A}_{\boldsymbol{m}}$ for each $\boldsymbol{m} \in \mathcal{S}_{\mathcal{A}}$ as a partial function $\mathcal{A}_{\boldsymbol{m}} : \{ \boldsymbol{t}x \! \downharpoonright |\mathcal{A}_{\boldsymbol{m}}| \mid \boldsymbol{t}x \in P_{\mathcal{G}(M_G)^3 \Rightarrow \mathcal{G}(M_G)}^{\mathsf{Odd}} \} \rightharpoonup M_{\mathcal{G}(M_G)^3 \Rightarrow \mathcal{G}(M_G)}$ in most cases. 
Accordingly, $\mathcal{A}_{\boldsymbol{m}}^\circledS$ is mostly a strategy $\mathcal{G}(M_G)^3 \Rightarrow \mathcal{G}(M_G)$ whose function representation $\mathcal{A}_{\boldsymbol{m}}$ is finite.
\end{convention}

\begin{remark}
Note that it does not make difference if each st-algorithm $\mathcal{A} :: G$ focuses on P-views of positions $\boldsymbol{t}x \in P_{\mathcal{G}(M_G)^3 \Rightarrow \mathcal{G}(M_G)}^{\mathsf{Odd}}$ since $\lceil \boldsymbol{t} x \rceil = \boldsymbol{t} x$.
Also, strictly speaking, each instruction strategy $\mathcal{A}_{\boldsymbol{m}}^\circledS : \mathcal{G}(M_G)^3 \Rightarrow \mathcal{G}(M_G)$ has to specify justifiers (in $P_G$) of outputs\footnote{Though it is not complicated at all to specify justifiers since the choice is ternary, i.e., the last, third last or opening occurrence in the P-view.}. 
However, since justifiers occurring in this paper are all obvious ones, we have adopted the abbreviated form of $\mathcal{A}_{\boldsymbol{m}}^\circledS$ as above.
\end{remark}

Thus, an instruction strategy is a strategy on the game $\mathcal{G}(M_G)^3 \Rightarrow \mathcal{G}(M_G)$, where $G$ is a game, that is \emph{finitary} in the sense that it is representable by a finite partial function, and so it is clearly `computable' in an informal sense.
We shall see that the number $3$ on $\mathcal{G}(M_G)^3$ is the least number to achieve Turing completeness in Section~\ref{TuringCompleteness}. 
As already mentioned, our idea is to utilize such an instruction strategy as a `description' of a strategy on $G$, which may be `effectively read-off':

\begin{definition}[Realizability]
The strategy $\mathsf{st}(\mathcal{A})$ \emph{\bfseries realized} by an st-algorithm $\mathcal{A} :: G$ is defined by: 
\begin{equation*}
\mathsf{st}(\mathcal{A}) \stackrel{\mathrm{df. }}{=} \{ \boldsymbol{\epsilon} \} \cup \{ \boldsymbol{s}ab \in P_G \mid \boldsymbol{s} \in \mathsf{st}(\mathcal{A}), \mathcal{A}^\circledS(\lceil \boldsymbol{s}a \rceil \! \downharpoonright 3) \downarrow, b = \mathcal{A}^\circledS(\lceil \boldsymbol{s}a \rceil \! \downharpoonright 3) \ \! \}
\end{equation*}
where $\mathcal{A}^\circledS(\lceil \boldsymbol{s}a \rceil \! \downharpoonright 3) \stackrel{\mathrm{df. }}{\simeq} \mathcal{M}(\mathcal{A}^\circledS_{\mathcal{Q}_{\mathcal{A}}^\star(\lceil \boldsymbol{s}a \rceil)} \circ (\underline{\lceil \boldsymbol{s}a \rceil \! \downharpoonright 3}_3)^\dagger)$, and $\mathcal{A}^\circledS(\lceil \boldsymbol{s}a \rceil \! \downharpoonright 3) \downarrow$ presupposes $\mathcal{Q}_{\mathcal{A}}^\star(\lceil \boldsymbol{s}a \rceil) \in \mathcal{S}_{\mathcal{A}}$.
\end{definition}

Clearly, $\mathcal{A} :: G \Rightarrow \mathsf{st}(\mathcal{A}) : G$ holds.
We are now ready to define the central notion of the paper, namely `effective computability' of strategies:
\begin{definition}[Viable strategies]
\label{DefViability}
A strategy $\sigma : G$ is \emph{\bfseries viable} if there exists an st-algorithm $\mathcal{A} :: G$ that realizes $\sigma$, i.e., $\mathsf{st}(\mathcal{A}) = \sigma$.
\end{definition}

That is, a strategy $\sigma : G$ is viable if there is a finitary strategy on $\mathcal{G}(M_G)^3 \Rightarrow \mathcal{G}(M_G)$ that `describes' the computation of $\sigma$.
The terms \emph{realize} and \emph{realizability} come from mathematical logic, in which a \emph{realizer} refers to some computational information that `realizes' the constructive truth of a mathematical statement \cite{troelstra1998realizability}.

Given an st-algorithm $\mathcal{A} :: G$ that realizes a strategy $\sigma : G$, P may `effectively execute' $\mathcal{A}$ to compute $\sigma$ roughly as follows:
\begin{enumerate}

\item Given $\boldsymbol{s}a \in P_G^{\mathsf{Odd}}$, P calculates the current state $\boldsymbol{m} \stackrel{\mathrm{df. }}{=} \mathcal{Q}_{\mathcal{A}}^\star(\lceil \boldsymbol{s}a \rceil)$ and the last (up to) three moves $\lceil \boldsymbol{s}a \rceil \! \downharpoonright 3$ in the P-view; if $\boldsymbol{m} \not \in \mathcal{S}_{\mathcal{A}}$, then she stops, i.e., the next move is undefined;

\item Otherwise, she composes $(\underline{\lceil \boldsymbol{s}a \rceil \! \downharpoonright 3}_3)^\dagger$ with $\mathcal{A}_{\boldsymbol{m}}^\circledS$, calculating $\mathcal{A}_{\boldsymbol{m}}^\circledS \circ (\underline{\lceil \boldsymbol{s}a \rceil \! \downharpoonright 3}_3)^\dagger$;

\item Finally, she reads off the next move $\mathcal{M}(\mathcal{A}_{\boldsymbol{m}}^\circledS \circ (\underline{\lceil \boldsymbol{s}a \rceil \! \downharpoonright 3}_3)^\dagger)$ (and its justifier) and performes that move.

\end{enumerate}
For conceptual clarity, here we assume that P may write down moves $[m]_{\boldsymbol{e}}$ in P-views as $[\mathsf{m}]_{\mathscr{C}^\ast(\boldsymbol{e})}$ and execute strategies on instruction games symbolically on her `scratch pad', and also she may read off strategies $\sigma : \mathcal{G}(M_G)$ and reproduce them as moves $\mathcal{M}(\sigma) \in M_G$.
This procedure is clearly `effective' in an informal sense, which is our justification of the notion of viable strategies.

Note that there are two kinds of processes in viable strategies $\sigma : G$.
The first kind is the process of $\sigma$ per se whose atomic steps are $(\boldsymbol{s} a \in P_G^{\mathsf{Odd}}) \mapsto \boldsymbol{s}a . \sigma(\lceil \boldsymbol{s}a \rceil)$, and the second kind is the process of its st-algorithm $\mathcal{A}$ whose atomic steps are $\boldsymbol{t} x \in P_{\mathcal{G}(M_G)^3 \Rightarrow \mathcal{G}(M_G)}^{\mathsf{Odd}} \mapsto \boldsymbol{t} x . \mathcal{A}_{\boldsymbol{m}} \circ \partial_{\boldsymbol{m}}(\boldsymbol{t} x)$, where $\boldsymbol{m}$ is the `current state'.
The former is abstract and `high-level', while the latter is symbolic and `low-level'.
In this manner, we have achieved a mathematical formulation of `high-level' and `low-level' processes and `effective computability' of the former in terms of the latter (as promised in the introduction).


In order to establish Theorem~\ref{ThmPreservationOfViability}, which is our main theorem, later, we shall focus on the following st-algorithms:
\begin{definition}[Standard st-algorithms]
\label{DefStandardAlgorithms}
An st-algorithm $\mathcal{A} :: G$ is \emph{\bfseries standard} if it satisfies the following three conditions:
\begin{enumerate}

\item It does not refer to any input outer tag when it computes the inner element, i.e., if $\boldsymbol{s} = [(\hat{q}_G)_\mathscr{WE}] x \boldsymbol{t} \in \mathcal{A}^\circledS_{\boldsymbol{m}}$ for some $\boldsymbol{m} \in \mathcal{S}_{\mathcal{A}}$, then $\hat{q}_{\mathcal{T}}$ does not occur in $\boldsymbol{s}$;

\item If it refers to input outer tags, they must belongs to the last move in the P-view of the current position of $G$, i.e., if $\hat{q}_{\mathcal{T}}$ occurs as a P-move in some $\boldsymbol{s} \in \mathcal{A}^\circledS_{\boldsymbol{m}}$, where $\boldsymbol{m} \in \mathcal{S}_{\mathcal{A}}$, then the inner tag of the move is $\mathscr{EEW}$;

\item The symbol $\square$ does not occur in $\mathcal{A}_{\boldsymbol{m}}$ for any $\boldsymbol{m} \in \mathcal{S}_{\mathcal{A}}$.

\end{enumerate}
\end{definition}

\begin{example}
The \emph{\bfseries zero strategy} $\mathit{zero}_A \stackrel{\mathrm{df. }}{=} \mathsf{Pref}(\{ [\hat{q}_\mathscr{E}] [\mathit{no}_\mathscr{E}] \})^{\mathsf{Even}} : [A_{\mathscr{W}}]_{\Lbag \boldsymbol{e} \Rbag \hbar} \Rightarrow [\mathcal{N}_{\mathscr{E}}]$ on any normalized game $A$ is viable since we may give an st-algorithm $\mathcal{A}(\mathit{zero}_A)$ by $\mathcal{Q}_{\mathcal{A}(\mathit{zero}_A)}(m) \stackrel{\mathrm{df. }}{=} \begin{cases} \top &\text{if $m = \hat{q}_\mathscr{E}$;} \\ \bot &\text{otherwise} \end{cases}$, $\mathcal{S}_{\mathcal{A}(\mathit{zero}_A)} \stackrel{\mathrm{df. }}{=} \{ \hat{q}_\mathscr{E} \}$, $|\mathcal{A}(\mathit{zero}_A)_{\hat{q}_\mathscr{E}}| \stackrel{\mathrm{df. }}{=} 1$, $\|\mathcal{A}(\mathit{zero}_A)_{\hat{q}_\mathscr{E}} \| \stackrel{\mathrm{df. }}{=} 0$ and $\mathcal{A}(\mathit{zero}_A)_{\hat{q}_\mathscr{E} } : (\hat{q}_{A \Rightarrow N})^{[3]} \mapsto (\mathsf{\mathsf{no}_E})^{[3]} \mid (\hat{q}_{\mathcal{T}})^{[3]} \mapsto (\mathsf{\checkmark})^{[3]}$.
Then, the instruction strategy $\mathcal{A}(\mathit{zero}_A)_{\hat{q}_{\mathscr{E}}}^\circledS$ is as depicted in the following diagram:
\begin{center}
\begin{tabular}{ccccccc}
$\mathcal{G}(M_{A \Rightarrow \mathcal{N}})^{[0]}$ & $\&$ & $\mathcal{G}(M_{A \Rightarrow \mathcal{N}})^{[1]}$ & $\&$ & $\mathcal{G}(M_{A \Rightarrow \mathcal{N}})^{[2]}$ & $\stackrel{\mathcal{A}(\mathit{zero}_A)_{\hat{q}_\mathscr{E}}^\circledS}{\Rightarrow}$ & $\mathcal{G}(M_{A \Rightarrow \mathcal{N}})^{[3]}$ \\
\cline{1-7}
&&&&&&$(\hat{q}_{A \Rightarrow N})^{[3]}$ ($(\hat{q}_{\mathcal{T}})^{[3]}$) \\
&&&&&&$(\mathsf{no_E})^{[3]}$ ($(\mathsf{\checkmark})^{[3]}$)
\end{tabular}
\end{center}
Clearly, $\mathit{zero}_A$ is standard, and $\mathsf{st}(\mathcal{A}(\mathit{zero}_A)) = \mathit{zero}_A$.
\end{example}

\begin{example}
Let us complete the example of successor strategy $\mathit{succ} : [\mathcal{N}_\mathscr{W}]_{\Lbag \boldsymbol{e} \Rbag \hbar} \Rightarrow [\mathcal{N}_\mathscr{E}]$ (Example~\ref{ExSuccPred}). 
We give an st-algorithm $\mathcal{A}(\mathit{succ})$ for $\mathit{succ}$ by defining $\mathcal{Q}_{\mathcal{A}(\mathit{succ})}(m) \stackrel{\mathrm{df. }}{=} \begin{cases} \top &\text{if $m = \hat{q}_\mathscr{E}$;} \\ \bot &\text{otherwise} \end{cases}$, $\mathcal{S}_{\mathcal{A}(\mathit{succ})} \stackrel{\mathrm{df. }}{=} \{ \hat{q}_\mathscr{E} \}$, $|\mathcal{A}(\mathit{succ})_{\hat{q}_\mathscr{E}}| \stackrel{\mathrm{df. }}{=} 11$, $\| \mathcal{A}(\mathit{succ})_{\hat{q}_\mathscr{E}} \| \stackrel{\mathrm{df. }}{=} 0$ and the table is as given in Appendix~\ref{AppSucc}.
We clearly have $\mathsf{st}(\mathcal{A}(\mathit{succ})) = \mathit{succ}$, which establishes the viability of $\mathit{succ}$.
Also, it is easy to see that $\mathcal{A}(\mathit{succ})$ is standard. 
\end{example}

\begin{example}
Similarly to $\mathit{zero}$ and $\mathit{succ}$, we may give an st-algorithm $\mathcal{A}(\mathit{pred})$ for the predecessor strategy $\mathit{pred} : [!\mathcal{N}_\mathscr{W}]_{\Lbag \boldsymbol{e} \Rbag \hbar} \multimap [\mathcal{N}_\mathscr{E}]$ (Example~\ref{ExSuccPred}) as follows.
We define the states, the view- and mate-scopes and query of $\mathcal{A}(\mathit{pred})$ to be the same as those of $\mathcal{A}(\mathit{succ})$.
At this point, it should suffice to show diagrams for $\mathcal{A}(\mathit{pred})_{\hat{q}_\mathscr{E}}^\circledS$ since it is clear that there is a finite table $\mathcal{A}(\mathit{pred})_{\hat{q}_\mathscr{E}}$ achieving it:
\begin{center}
\begin{tabular}{ccccccc}
$\mathcal{G}(M_{\mathcal{N} \Rightarrow \mathcal{N}})^{[0]}$ & $\&$ & $\mathcal{G}(M_{\mathcal{N} \Rightarrow \mathcal{N}})^{[1]}$ & $\&$ & $\mathcal{G}(M_{\mathcal{N} \Rightarrow \mathcal{N}})^{[2]}$ & $\stackrel{\mathcal{A}(\mathit{pred})_{\hat{q}_\mathscr{E}}^\circledS}{\Rightarrow}$ & $\mathcal{G}(M_{\mathcal{N} \Rightarrow \mathcal{N}})^{[3]}$ \\
\cline{1-7}
&&&&&&$(\hat{q}_{\mathcal{N} \Rightarrow \mathcal{N}})^{[3]}$ \\
&&&&$(\hat{q}_{\mathcal{N} \Rightarrow \mathcal{N}})^{[2]}$&& \\
&&&&$(\mathsf{\hat{q}_E})^{[2]}$ ($(\mathsf{q_E})^{[2]}$)&& \\
&&&&&&$(\mathsf{\hat{q}_W})^{[3]}$ ($(\mathsf{q_W})^{[3]}$)
\end{tabular}
\end{center}

\begin{center}
\begin{tabular}{ccccccc}
$\mathcal{G}(M_{\mathcal{N} \Rightarrow \mathcal{N}})^{[0]}$ & $\&$ & $\mathcal{G}(M_{\mathcal{N} \Rightarrow \mathcal{N}})^{[1]}$ & $\&$ & $\mathcal{G}(M_{\mathcal{N} \Rightarrow \mathcal{N}})^{[2]}$ & $\stackrel{\mathcal{A}(\mathit{pred})_{\hat{q}_\mathscr{E}}^\circledS}{\Rightarrow}$ & $\mathcal{G}(M_{\mathcal{N} \Rightarrow \mathcal{N}})^{[3]}$ \\
\cline{1-7}
&&&&&&$(\hat{q}_{\mathcal{T}})^{[3]}$ \\
&&&&$(\hat{q}_{\mathcal{N} \Rightarrow \mathcal{N}})^{[2]}$&& \\
&&&& $(\mathsf{\hat{q}_E})^{[2]}$ ($(\mathsf{q_E})^{[2]}$) && \\
&&&&&& $(\langle)^{[3]}$ \\
&&&&&&$(q_{\mathcal{T}})^{[3]}$ \\
&&&&&&$(\rangle)^{[3]}$ \\
&&&&&&$(q_{\mathcal{T}})^{[3]}$ \\
&&&&&&$(\sharp)^{[3]}$ \\
&&&&&&$(q_{\mathcal{T}})^{[3]}$ \\
&&&&&&$(\checkmark)^{[3]}$ 
\end{tabular}
\end{center}

\begin{center}
\begin{tabular}{ccccccc}
$\mathcal{G}(M_{\mathcal{N} \Rightarrow \mathcal{N}})^{[0]}$ & $\&$ & $\mathcal{G}(M_{\mathcal{N} \Rightarrow \mathcal{N}})^{[1]}$ & $\&$ & $\mathcal{G}(M_{\mathcal{N} \Rightarrow \mathcal{N}})^{[2]}$ & $\stackrel{\mathcal{A}(\mathit{pred})_{\hat{q}_\mathscr{E}}^\circledS}{\Rightarrow}$ & $\mathcal{G}(M_{\mathcal{N} \Rightarrow \mathcal{N}})^{[3]}$ \\
\cline{1-7}
&&&&&&$(\hat{q}_{\mathcal{N} \Rightarrow \mathcal{N}})^{[3]}$ \\
&&&&$(\hat{q}_{\mathcal{N} \Rightarrow \mathcal{N}})^{[2]}$&& \\
&&&& $(\mathsf{no_W})^{[2]}$&& \\
&&&&&& $(\mathsf{no_E})^{[3]}$
\end{tabular}
\end{center}

\begin{center}
\begin{tabular}{ccccccc}
$\mathcal{G}(M_{\mathcal{N} \Rightarrow \mathcal{N}})^{[0]}$ & $\&$ & $\mathcal{G}(M_{\mathcal{N} \Rightarrow \mathcal{N}})^{[1]}$ & $\&$ & $\mathcal{G}(M_{\mathcal{N} \Rightarrow \mathcal{N}})^{[2]}$ & $\stackrel{\mathcal{A}(\mathit{pred})_{\hat{q}_\mathscr{E}}^\circledS}{\Rightarrow}$ & $\mathcal{G}(M_{\mathcal{N} \Rightarrow \mathcal{N}})^{[3]}$ \\
\cline{1-7}
&&&&&&$(\hat{q}_{\mathcal{T}})^{[3]}$ \\
&&&&$(\hat{q}_{\mathcal{N} \Rightarrow \mathcal{N}})^{[2]}$&& \\
&&&&$(\mathsf{no_W})^{[2]}$ && \\
&&&&&&$(\checkmark)^{[3]}$ 
\end{tabular}
\end{center}

\begin{center}
\begin{tabular}{ccccccc}
$\mathcal{G}(M_{\mathcal{N} \Rightarrow \mathcal{N}})^{[0]}$ & $\&$ & $\mathcal{G}(M_{\mathcal{N} \Rightarrow \mathcal{N}})^{[1]}$ & $\&$ & $\mathcal{G}(M_{\mathcal{N} \Rightarrow \mathcal{N}})^{[2]}$ & $\stackrel{\mathcal{A}(\mathit{pred})_{\hat{q}_\mathscr{E}}^\circledS}{\Rightarrow}$ & $\mathcal{G}(M_{\mathcal{N} \Rightarrow \mathcal{N}})^{[3]}$ \\
\cline{1-7}
&&&&&&$(\hat{q}_{\mathcal{N} \Rightarrow \mathcal{N}})^{[3]}$ \\
&&&&$(\hat{q}_{\mathcal{N} \Rightarrow \mathcal{N}})^{[2]}$&& \\
&&&&$(\mathsf{yes_W})^{[2]}$&& \\
&&$(\hat{q}_{\mathcal{N} \Rightarrow \mathcal{N}})^{[1]}$&&&& \\
&&$(\mathsf{\hat{q}_W})^{[1]}$ ($(\mathsf{q_W})^{[1]}$)&&&& \\
&&&&&& $(\mathsf{q_W})^{[3]}$ ($(\mathsf{yes_E})^{[3]}$)
\end{tabular}
\end{center}

\begin{center}
\begin{tabular}{ccccccc}
$\mathcal{G}(M_{\mathcal{N} \Rightarrow \mathcal{N}})^{[0]}$ & $\&$ & $\mathcal{G}(M_{\mathcal{N} \Rightarrow \mathcal{N}})^{[1]}$ & $\&$ & $\mathcal{G}(M_{\mathcal{N} \Rightarrow \mathcal{N}})^{[2]}$ & $\stackrel{\mathcal{A}(\mathit{pred})_{\hat{q}_\mathscr{E}}^\circledS}{\Rightarrow}$ & $\mathcal{G}(M_{\mathcal{N} \Rightarrow \mathcal{N}})^{[3]}$ \\
\cline{1-7}
&&&&&&$(\hat{q}_{\mathcal{T}})^{[3]}$ \\
&&&&$(\hat{q}_{\mathcal{N} \Rightarrow \mathcal{N}})^{[2]}$&& \\
&&&&$(\mathsf{yes_W})^{[2]}$&& \\
&&$(\hat{q}_{\mathcal{N} \Rightarrow \mathcal{N}})^{[1]}$&&&& \\
&&$(\mathsf{\hat{q}_W})^{[1]}$ &&&& \\
&&&&&&$(\langle)^{[3]}$ \\
&&&&&&$(q_{\mathcal{T}})^{[3]}$ \\
&&&&&&$(\rangle)^{[3]}$ \\
&&&&&&$(q_{\mathcal{T}})^{[3]}$ \\
&&&&&&$(\sharp)^{[3]}$ \\
&&&&&&$(q_{\mathcal{T}})^{[3]}$ \\
&&&&&&$(\checkmark)^{[3]}$
\end{tabular}
\end{center}

\begin{center}
\begin{tabular}{ccccccc}
$\mathcal{G}(M_{\mathcal{N} \Rightarrow \mathcal{N}})^{[0]}$ & $\&$ & $\mathcal{G}(M_{\mathcal{N} \Rightarrow \mathcal{N}})^{[1]}$ & $\&$ & $\mathcal{G}(M_{\mathcal{N} \Rightarrow \mathcal{N}})^{[2]}$ & $\stackrel{\mathcal{A}(\mathit{pred})_{\hat{q}_\mathscr{E}}^\circledS}{\Rightarrow}$ & $\mathcal{G}(M_{\mathcal{N} \Rightarrow \mathcal{N}})^{[3]}$ \\
\cline{1-7}
&&&&&&$(\hat{q}_{\mathcal{T}})^{[3]}$ \\
&&&&$(\hat{q}_{\mathcal{N} \Rightarrow \mathcal{N}})^{[2]}$&& \\
&&&&$(\mathsf{yes_W})^{[2]}$&& \\
&&$(\hat{q}_{\mathcal{N} \Rightarrow \mathcal{N}})^{[1]}$&&&& \\
&&$(\mathsf{q_W})^{[1]}$ &&&& \\
&&&&&&$(\checkmark)^{[3]}$
\end{tabular}
\end{center}

Clearly $\mathsf{st}(\mathcal{A}(\mathit{pred})) = \mathit{pred}$, establishing the viability of $\mathit{pred}$.
Also, it is easy to see that $\mathcal{A}(\mathit{pred})$ is standard. 
\end{example}

\begin{example}
\label{ExInformalFix}
Consider the \emph{\bfseries fixed-point strategy} $\mathit{fix}_A : ([A_{\mathscr{WW}}]_{\Lbag \boldsymbol{e'} \Rbag \hbar \Lbag \boldsymbol{e} \Rbag \hbar \boldsymbol{f}} \Rightarrow [A_{\mathscr{EW}}]_{\Lbag \boldsymbol{e'} \Rbag \hbar \boldsymbol{f}}) \Rightarrow [A_{\mathscr{E}}]_{\boldsymbol{f}}$ for each normalized game $A$ interpreting the \emph{fixed-point combinator} $\mathsf{fix_A}$ in PCF \cite{abramsky2000full,hyland2000full,mccusker1998games}, where $A$ is the interpretation of a type $\mathsf{A}$ of PCF.
Roughly, $\mathit{fix}_A$ computes as follows (for its detailed description, see \cite{hyland1997game,hyland2000full}):
\begin{itemize}

\item After an opening occurrence $[a_\mathscr{E}]_{\boldsymbol{f}}$, $\mathit{fix}_A$ copies it and performs the second move $[a_{\mathscr{EW}}]_{\Lbag \Rbag \hbar \boldsymbol{f}}$.

\item If O initiates a new thread $[a'_{\mathscr{WW}}]_{\Lbag \boldsymbol{e'} \Rbag \hbar \Lbag \boldsymbol{e} \Rbag \hbar \boldsymbol{f}}$ in the inner implication, then $\mathit{fix}_A$ copies it and launches a new thread in the outer implication by $[a'_{\mathscr{EW}}]_{\Lbag \Lbag \boldsymbol{e'} \Rbag \hbar \Lbag \boldsymbol{e} \Rbag \Rbag \hbar \boldsymbol{f}}$.

\item If O makes a move $[a''_{\mathscr{WW}}]_{\Lbag \boldsymbol{e'} \Rbag \hbar \Lbag \boldsymbol{e} \Rbag \hbar \boldsymbol{f}}$ (resp. $[a''_{\mathscr{EW}}]_{\Lbag \Rbag \hbar \boldsymbol{f}}$, $[a''_{\mathscr{EW}}]_{\Lbag \Lbag \boldsymbol{e'} \Rbag \hbar \Lbag \boldsymbol{e} \Rbag \Rbag \hbar \boldsymbol{f}}$, $[a''_{\mathscr{E}}]_{\boldsymbol{f}}$) in an existing thread, then $\mathit{fix}_A$ copies it and makes the next move $[a''_{\mathscr{EW}}]_{\Lbag \Lbag \boldsymbol{e'} \Rbag \hbar \Lbag \boldsymbol{e} \Rbag \Rbag \hbar \boldsymbol{f}}$ (resp. $[a''_{\mathscr{E}}]_{\boldsymbol{f}}$, $[a''_{\mathscr{WW}}]_{\Lbag \boldsymbol{e'} \Rbag \hbar \Lbag \boldsymbol{e} \Rbag \hbar \boldsymbol{f}}$, $[a''_{\mathscr{EW}}]_{\Lbag \Rbag \hbar \boldsymbol{f}}$) in the \emph{dual thread} to which the third last occurrence in the current P-view belongs.

\end{itemize}
Clearly, $\mathit{fix}_A$ is not finitary for the calculation of a unbounded number of outer tags.
It is, however, viable for any game $A$, which is perhaps surprising to many readers.
Here, let us just informally describe an st-algorithm $\mathcal{A}(\mathit{fix}_A)$ that realizes $\mathit{fix}_A$ as a preparation for Section~\ref{Examples}.
Let $\mathcal{Q}_{\mathcal{A}(\mathit{fix}_A)}(m) = \top \stackrel{\mathrm{df. }}{\Leftrightarrow} m \in \pi_1(M_{(A \Rightarrow A) \Rightarrow A}^{\mathsf{Init}})$ and $\mathcal{S}_{\mathcal{A}(\mathit{fix}_A)} \stackrel{\mathrm{df. }}{=} \pi_1(M_{(A \Rightarrow A) \Rightarrow A}^{\mathsf{Init}})$. Since $\mathcal{A}(\mathit{fix}_A)_{m}$ does not depend on $m$, fix an arbitrary state $m \in \mathcal{S}_{\mathcal{A}(\mathit{fix}_A)}$.
Below, we proceed by a case analysis on the rightmost component of input strategies on $\mathcal{G}(M_{(A \Rightarrow A) \Rightarrow A})^3$ for $\mathcal{A}(\mathit{fix}_A)_{m}^\circledS$ (which corresponds to the last occurrence in the current P-view of $(A \Rightarrow A) \Rightarrow A$).
\begin{itemize}

\item If the rightmost component is of the form $(\underline{[a_\mathscr{E}]_{\boldsymbol{f}}})^\dagger$, then $\mathcal{A}(\mathit{fix}_A)_{m}^\circledS$ recognizes it by the inner tag $\mathscr{E}$, and calculates the next move $[a_{\mathscr{EW}}]_{\Lbag \Rbag \hbar \boldsymbol{f}}$ once and for all for the inner element $a_{\mathscr{EW}}$ and `digit-by-digit' for the outer tag $\Lbag \Rbag \hbar \boldsymbol{f}$ (by producing $\Lbag \Rbag \hbar$ and copying $\boldsymbol{f}$).

\item If the rightmost component is of the form $(\underline{[a_{\mathscr{EW}}]_{\Lbag \Rbag \hbar \boldsymbol{f}}})^\dagger$, then the leftmost component (which corresponds to the third last occurrence in the current P-view of $(A \Rightarrow A) \Rightarrow A$) is of the form $(\underline{[a'_\mathscr{E}]_{\boldsymbol{f}}})^\dagger$, and thus $\mathcal{A}(\mathit{fix}_A)_{m}^\circledS$ recognizes it by the inner tags $\mathscr{EW}$ and $\mathscr{E}$, and calculates the next move $[a_\mathscr{E}]_{\boldsymbol{f}}$ once and for all for the inner element $a_\mathscr{E}$ and `digit-by-digit' for the outer tag $\boldsymbol{f}$ (by ignoring $\Lbag \Rbag \hbar$ and copying $\boldsymbol{f}$).

\item If the rightmost component is of the form $(\underline{[a_{\mathscr{EW}}]_{\Lbag \Lbag \boldsymbol{e'} \Rbag \hbar \Lbag \boldsymbol{e} \Rbag \Rbag \hbar \boldsymbol{f}}})^\dagger$, then $\mathcal{A}(\mathit{fix}_A)_{m}^\circledS$ calculates the next move $[a_{\mathscr{WW}}]_{\Lbag \boldsymbol{e'} \Rbag \hbar \Lbag \boldsymbol{e} \Rbag \hbar \boldsymbol{f}}$ similarly to the above case but with the help of m-views for the outer tag; see Section~\ref{Examples} for the details.

\item If the rightmost component is of the form $(\underline{[a_{\mathscr{WW}}]_{\Lbag \boldsymbol{e'} \Rbag \hbar \Lbag \boldsymbol{e} \Rbag \hbar \boldsymbol{f}}})^\dagger$, then $\mathcal{A}(\mathit{fix}_A)_{m}^\circledS$ calculates the next move $[a_{\mathscr{EW}}]_{\Lbag \Lbag \boldsymbol{e'} \Rbag \hbar \Lbag \boldsymbol{e} \Rbag \Rbag \hbar \boldsymbol{f}}$ in a similar manner to the above case with the help of m-views for the outer tag (n.b. only in this last case the justifier may not be the last or third last occurrence in the current P-view, but it may be the opening occurrence); see Section~\ref{Examples} for the detail.
\end{itemize}

\end{example}

We now turn to establishing a key theorem, which states that viability of strategies are preserved under the constructions on strategies defined in Section~\ref{ConstructionsOnStrategies}:
\begin{theorem}[Preservation of viability]
\label{ThmPreservationOfViability}
Viable strategies are closed under tensor $\otimes$, pairing $\langle \_, \_ \rangle$, promotion $(\_)^\dagger$, concatenation $\ddagger$, currying $\Lambda$ and uncurrying $\Lambda^\circleddash$ if the underlying st-algorithms are standard, where the standardness is also preserved.
\end{theorem}
\begin{proof}
Let us first show that tensor $\otimes$ preserves viability of normalized strategies.
Let $\sigma : [A_\mathscr{W}]_{\boldsymbol{e}} \multimap [C_\mathscr{E}]_{\boldsymbol{e'}}$ and $\tau : [B_\mathscr{W}]_{\boldsymbol{f}} \multimap [D_\mathscr{E}]_{\boldsymbol{f'}}$ be viable strategies with st-algorithms $\mathcal{A}(\sigma)$ and $\mathcal{A}(\tau)$ realizing $\sigma$ and $\tau$, respectively.
We have to construct an st-algorithm $\mathcal{A}(\sigma \otimes \tau)$ such that $\mathsf{st}(\mathcal{A}(\sigma \otimes \tau)) = \sigma \otimes \tau : [A_{\mathscr{WW}}]_{\boldsymbol{e}} \otimes [B_{\mathscr{EW}}]_{\boldsymbol{f}} \multimap [C_{\mathscr{WE}}]_{\boldsymbol{e'}} \otimes [D_{\mathscr{EE}}]_{\boldsymbol{f'}}$.
Let us define the finite set $\mathcal{S}_{\mathcal{A}(\sigma \otimes \tau)}$ of states and the query $\mathcal{Q}_{\mathcal{A}(\sigma \otimes \tau)}$ by:
\begin{align*}
\mathcal{S}_{\mathcal{A}(\sigma \otimes \tau)} \stackrel{\mathrm{df. }}{=} \ & \{ m^{(k)}_{\mathscr{W} X_k} m^{(k-1)}_{\mathscr{W} X_{k-1}} \dots m^{(1)}_{\mathscr{W} X_1} \ \! | \ \! m^{(k)}_{X_k} m^{(k-1)}_{X_{k-1}} \dots m^{(1)}_{X_1} \in \mathcal{S}_{\mathcal{A}(\sigma)} \} \\
&\cup \{ n^{(l)}_{\mathscr{E} Y_l} n^{(l-1)}_{\mathscr{E} Y_{l-1}} \dots n^{(1)}_{\mathscr{E} Y_1} \ \! | \ \! n^{(l)}_{Y_l} n^{(l-1)}_{Y_{l-1}} \dots n^{(1)}_{Y_1} \in \mathcal{S}_{\mathcal{A}(\tau)} \} \\
\mathcal{Q}_{\mathcal{A}(\sigma \otimes \tau)} : \ &a_{\mathscr{WW}} \mapsto \mathcal{Q}_{\mathcal{A}(\sigma)}(a_\mathscr{W}), b_{\mathscr{EW}} \mapsto \mathcal{Q}_{\mathcal{A}(\tau)}(b_\mathscr{W}), c_{\mathscr{WE}} \mapsto \mathcal{Q}_{\mathcal{A}(\sigma)}(c_\mathscr{E}), \\
&d_{\mathscr{EE}} \mapsto \mathcal{Q}_{\mathcal{A}(\tau)}(d_\mathscr{E})
\end{align*}
where $X_k, X_{k-1}, \dots, X_1, Y_l, Y_{l-1}, \dots, Y_1 \in \{ \mathscr{W}, \mathscr{E} \}$ are inner tags that are the rightmost components of $m^{(k)}, m^{(k-1)}, \dots, m^{(1)} \in \pi_1(M_A \cup M_C), n^{(l)}, n^{(l-1)}, \dots, n^{(1)} \in \pi_1(M_B \cup M_D)$, respectively.
Clearly, $\mathcal{Q}_{\mathcal{A}(\sigma \otimes \tau)}$ satisfies the axioms Q1 and Q2.
For each $m^{(k)}_{X_k} m^{(k-1)}_{X_{k-1}} \dots m^{(1)}_{X_1} \in \mathcal{S}_{\mathcal{A}(\sigma)}$ and $n^{(l)}_{Y_l} n^{(l-1)}_{Y_{l-1}} \dots n^{(1)}_{Y_1} \in \mathcal{S}_{\mathcal{A}(\tau)}$, we construct the finite partial functions $\mathcal{A}(\sigma \otimes \tau)_{m^{(k)}_{\mathscr{W}X_k} m^{(k-1)}_{\mathscr{W}X_{k-1}} \dots m^{(1)}_{\mathscr{W}X_1}}$ and $\mathcal{A}(\sigma \otimes \tau)_{n^{(l)}_{\mathscr{E}Y_l} n^{(l-1)}_{\mathscr{E}Y_{l-1}} \dots n^{(1)}_{\mathscr{E}Y_1}}$ from $\mathcal{A}(\sigma)_{m^{(k)}_{X_k} m^{(k-1)}_{X_{k-1}} \dots m^{(1)}_{X_1}}$ and $\mathcal{A}(\tau)_{n^{(l)}_{Y_l} n^{(l-1)}_{Y_{l-1}} \dots n^{(1)}_{Y_1}}$ simply by changing symbols $\mathsf{m_{X}} \in \mathsf{Sym}(\pi_1(M_{A \multimap C}))$ and $\mathsf{n_{Y}} \in \mathsf{Sym}(\pi_1(M_{B \multimap D}))$ into $\mathsf{m_{W X}}$ and $\mathsf{n_{E Y}}$ respectively in their finite tables, where the view-scopes are defined by: 
\begin{align*}
|\mathcal{A}(\sigma \otimes \tau)_{m^{(k)}_{\mathscr{W}X_k} m^{(k-1)}_{\mathscr{W}X_{k-1}} \dots m^{(1)}_{\mathscr{W}X_1}}| &\stackrel{\mathrm{df. }}{=} |\mathcal{A}(\sigma)_{m^{(k)}_{X_k} m^{(k-1)}_{X_{k-1}} \dots m^{(1)}_{X_1}}| \\
|\mathcal{A}(\sigma \otimes \tau)_{n^{(l)}_{\mathscr{E}l} n^{(l-1)}_{\mathscr{E}Y_{l-1}} \dots n^{(1)}_{\mathscr{E}Y_1}}| &\stackrel{\mathrm{df. }}{=} |\mathcal{A}(\tau)_{n^{(l)}_{Y_l} n^{(l-1)}_{Y_{l-1}} \dots n^{(1)}_{Y_1}}|
\end{align*}
and the mate-scopes are defined similarly.
Then, because a P-view in $\sigma \otimes \tau$ is either a P-view in $\sigma$ or $\tau$ (which is shown by induction on the length of positions of $\sigma \otimes \tau$), it is straightforward to see that $\mathsf{st}(\mathcal{A}(\sigma \otimes \tau)) = \sigma \otimes \tau$ holds.
Also, it is clear that $\mathcal{A}(\sigma \otimes \tau)$ is standard if so are $\mathcal{A}(\sigma)$ and $\mathcal{A}(\tau)$.
Intuitively, $\mathcal{A}(\sigma \otimes \tau)$ sees the new digit ($\mathscr{W}$ or $\mathscr{E}$) of the current state $\boldsymbol{s} \in \mathcal{S}_{\mathcal{A}(\sigma \otimes \tau)}$ and decides $\mathcal{A}(\sigma)$ or $\mathcal{A}(\tau)$ to apply (n.b. since $\mathcal{Q}_{\mathcal{A}(\sigma \otimes \tau)}$ `tracks' every initial move by Q1, a state must be non-empty).

It is then clear that pairing of dynamic strategies may be handled in a completely similar manner; currying and uncurrying are even simpler.

Now, consider the concatenation $\iota \ddagger \kappa : J \ddagger K$ of viable dynamic strategies $\iota : J$ and $\kappa : K$ such that $\mathcal{H}^\omega(J) \trianglelefteqslant A \multimap B$ and $\mathcal{H}^\omega(K) \trianglelefteqslant B \multimap C$ for some normalized games $A$, $B$ and $C$. Let $\mathcal{A}(\iota)$ and $\mathcal{A}(\kappa)$ be standard st-algorithms such that $\mathsf{st}(\mathcal{A}(\iota)) = \iota$ and $\mathsf{st}(\mathcal{A}(\kappa)) = \kappa$.
We define the set $\mathcal{S}_{\mathcal{A}(\iota \ddagger \kappa)}$ of states and the query $\mathcal{Q}_{\mathcal{A}(\iota \ddagger \kappa)}$ by: 
\begin{align*}
\mathcal{S}_{\mathcal{A}(\iota \ddagger \kappa)} &\stackrel{\mathrm{df. }}{=} \{ n^{(l)}_{G(Y_l)} n^{(l-1)}_{G(Y_{l-1})} \dots n^{(1)}_{G(Y_1)} m^{(k)}_{F(X_k)} m^{(k-1)}_{F(X_{k-1})} \dots m^{(1)}_{F(X_1)} \ \! | \\ & \ \ \ \ \ \ m^{(k)}_{X_k} m^{(k-1)}_{X_{k-1}} \dots m^{(1)}_{X_1} \in \mathcal{S}_{\mathcal{A}(\iota)}, n^{(l)}_{Y_l} n^{(l-1)}_{Y_{l-1}} \dots n^{(1)}_{Y_1} \in \mathcal{S}_{\mathcal{A}(\kappa)} \} \\
\mathcal{Q}_{\mathcal{A}(\iota \ddagger \kappa)} &: m^{(i)}_{F(X_i)} \mapsto \mathcal{Q}_{\mathcal{A}(\iota)}(m^{(i)}_{X_i}), n^{(j)}_{G(Y_j)} \mapsto \mathcal{Q}_{\mathcal{A}(\kappa)}(n^{(j)}_{Y_j})
\end{align*}
where 
\begin{align*}
F(X_i) &\stackrel{\mathrm{df. }}{=} \begin{cases} X_i &\text{if $m^{(i)}_{X_i} \in M_J^{\mathsf{Ext}} \wedge X_i = \mathscr{W}$;} \\ X_i \mathscr{S} &\text{otherwise} \end{cases} \\
G(Y_j) &\stackrel{\mathrm{df. }}{=} \begin{cases} Y_j &\text{if $n^{(j)}_{Y_j} \in M_K^{\mathsf{Ext}} \wedge Y_j = \mathscr{E}$;} \\ Y_j \mathscr{N} &\text{otherwise.} \end{cases}
\end{align*}

We construct the finite partial function $\mathcal{A}(\iota \ddagger \kappa)_{n^{(l)}_{G(Y_l)} n^{(l-1)}_{G(Y_{l-1})} \dots n^{(1)}_{G(Y_1)} m^{(k)}_{F(X_k)} m^{(k-1)}_{F(X_{k-1})} \dots m^{(1)}_{F(X_1)}}$ (including its view- and mate-scopes) from $\mathcal{A}(\kappa)_{n^{(l)}_{Y_l} n^{(l-1)}_{Y_{l-1}} \dots n^{(l)}_{Y_1}}$ if $k = 0$ and from $\mathcal{A}(\iota)_{m^{(k)}_{X_k} m^{(k-1)}_{X_{k-1}} \dots m^{(1)}_{X_1}}$ otherwise, by modifying the symbols in the table similarly to the case of tensor (where the view- and mate-scopes are just inherited).
Again, $\mathcal{Q}_{\mathcal{A}(\iota \ddagger \kappa)}$ clearly satisfies the axioms Q1 and Q2.
Because a P-view in $\iota \ddagger \kappa$ is a one in $\kappa$ or a one in $\iota$ followed by a one in $\kappa$ (it is crucial here that $\mathcal{Q}_{\mathcal{A}(\iota)}$ `tracks' initial moves by Q1, and $\square$ does not appear in $\mathcal{A}(\iota)$ for it is standard), we may conclude that $\mathsf{st}(\mathcal{A}(\iota \ddagger \kappa)) = \iota \ddagger \kappa$.
Moreover, $\mathcal{A}(\iota \ddagger \kappa)$ is clearly standard as so are $\mathcal{A}(\iota)$ and $\mathcal{A}(\kappa)$.

Finally, assume that $\varphi^\dagger : [!A_\mathscr{W}]_{\Lbag \boldsymbol{e} \Rbag \hbar \boldsymbol{f}} \multimap [!B_\mathscr{E}]_{\Lbag \boldsymbol{e'} \Rbag \hbar \boldsymbol{f'}}$ is the promotion of a viable strategy $\varphi : [!A_\mathscr{W}]_{\Lbag \boldsymbol{e} \Rbag \hbar \boldsymbol{f}} \multimap [B_\mathscr{E}]_{\boldsymbol{f'}}$ with a standard st-algorithm $\mathcal{A}(\varphi)$ that realizes $\varphi$.
As the more general case $\varphi : G$, where $\mathcal{H}^\omega(G) \trianglelefteqslant A \Rightarrow B$, is similar (as internal moves of $G^\dagger$ keep new digits of outer tags on moves of $!B$), we focus on the case $\varphi : A \Rightarrow B$ for simplicity. 
We define $\mathcal{S}_{\mathcal{A}(\varphi^\dagger)} \stackrel{\mathrm{df. }}{=} \mathcal{S}_{\mathcal{A}(\varphi)}$ and $\mathcal{Q}_{\mathcal{A}(\varphi^\dagger)} \stackrel{\mathrm{df. }}{=} \mathcal{Q}_{\mathcal{A}(\varphi)}$.
Then, roughly, the idea is that if $\varphi$ makes the next P-move $[a_\mathscr{W}]_{\Lbag \boldsymbol{e} \Rbag \hbar \boldsymbol{f}}$ (resp. $[b_\mathscr{E}]_{\boldsymbol{f'}}$) at an odd-length position $\boldsymbol{t}x$ of $[!A_\mathscr{W}]_{\Lbag \boldsymbol{e} \Rbag \hbar \boldsymbol{f}} \multimap [B_\mathscr{E}]_{\boldsymbol{f'}}$, then $\varphi^\dagger$ at an odd-length position $\boldsymbol{t'} x'$ of $[!A_\mathscr{W}]_{\Lbag \boldsymbol{e} \Rbag \hbar \boldsymbol{f}} \multimap [!B_\mathscr{E}]_{\Lbag \boldsymbol{e'} \Rbag \hbar \boldsymbol{f'}}$ that begins with an initial move $[b^{(0)}_\mathscr{E}]_{\Lbag \boldsymbol{e^{(0)}} \Rbag \hbar \boldsymbol{f}^{(0)}}$ and satisfies $\boldsymbol{t'} x' \upharpoonright \boldsymbol{e^{(0)}} = \boldsymbol{t} x$ makes the corresponding next P-move $[a_\mathscr{W}]_{\Lbag \Lbag \boldsymbol{e^{(0)}} \Rbag \hbar \Lbag \boldsymbol{e} \Rbag \Rbag \hbar \boldsymbol{f}}$ (resp. $[b_\mathscr{E}]_{\Lbag \boldsymbol{e^{(0)}} \Rbag \hbar \boldsymbol{f'}}$).
As opposed to the other constructions, however, the formal definition of the finite table for each $\mathcal{A}(\varphi^\dagger)_{\boldsymbol{m}}$, where $\boldsymbol{m} \in \mathcal{S}_{\mathcal{A}(\varphi^\dagger)}$, is rather involved; thus, we just informally describe how to obtain these tables from those for $\varphi$, which should suffice for the reader to see how to construct the tables if he or she wishes:

\begin{enumerate}

\item Fix $\boldsymbol{s} \in \mathcal{S}_{\mathcal{A}(\varphi^\dagger)} = \mathcal{S}_{\mathcal{A}(\varphi)}$.
Since $\mathcal{A}(\varphi)$ is standard, and P-views in $!A \multimap \ !B$ are those in $!A \multimap B$, we define $\mathcal{A}(\varphi^\dagger)_{\boldsymbol{s}}$'s calculation of inner elements to be exactly the same as that of $\mathcal{A}(\varphi)_{\boldsymbol{s}}$, where we may assume that the computation of each inner element takes at most $8$ moves since inputs for instruction strategies are ternary.
Thus, below we focus on $\mathcal{A}(\varphi^\dagger)_{\boldsymbol{s}}$'s calculation of outer tags.

\item Let $\boldsymbol{t} [m]_{\boldsymbol{\tilde{e}}} [n]_{\boldsymbol{e}} \in \varphi$ and $\boldsymbol{t'} [m]_{\boldsymbol{\tilde{e}'}} [n]_{\boldsymbol{e'}} \in \varphi^\dagger$ such that $\mathcal{Q}_{\mathcal{A}(\varphi^\dagger)}^\star(\lceil \boldsymbol{t'}[m]_{\boldsymbol{\tilde{e}}} \rceil) = \boldsymbol{s}$ and $\boldsymbol{t'} [m]_{\boldsymbol{\tilde{e}'}} [n]_{\boldsymbol{e'}} \upharpoonright \boldsymbol{e^{(0)}} = \boldsymbol{t} [m]_{\boldsymbol{\tilde{e}}} [n]_{\boldsymbol{e}}$, where the opening occurrence of the current thread in $\varphi^\dagger$ is of the form $[b^{(0)}]_{\Lbag \boldsymbol{e^{(0)}} \Rbag \hbar \boldsymbol{f^{(0)}}}$. 
Let us describe how $\mathcal{A}(\varphi^\dagger)_{\boldsymbol{s}}$ calculates $\mathscr{C}^\ast(\boldsymbol{e'})$ by a case analysis on $m$ and $n$:

\begin{itemize}

\item If $m$ and $n$ both belong to $A$, which $\mathcal{A}(\varphi^\dagger)_{\boldsymbol{s}}$ may recognize by the method described below, then $\boldsymbol{\tilde{e}}$, $\boldsymbol{e}$, $\boldsymbol{\tilde{e}'}$ and $\boldsymbol{e'}$ are respectively of the form $\Lbag \boldsymbol{\tilde{g}} \Rbag \hbar \boldsymbol{\tilde{h}}$, $\Lbag \boldsymbol{g} \Rbag \hbar \boldsymbol{h}$, $\Lbag \Lbag \boldsymbol{e^{(0)}} \Rbag \hbar \Lbag \boldsymbol{\tilde{g}} \Rbag \Rbag \hbar \boldsymbol{\tilde{h}}$ and $\Lbag \Lbag \boldsymbol{e^{(0)}} \Rbag \hbar \Lbag \boldsymbol{g} \Rbag \Rbag \hbar \boldsymbol{h}$.
Note that the outer tags which $\mathcal{A}(\varphi)_{\boldsymbol{s}}$ refers to for computing $\boldsymbol{e}$ are only $\boldsymbol{\tilde{e}}$ as $\mathcal{A}(\varphi)$ is standard.
Then, with the help of m-views, $\mathcal{A}(\varphi^\dagger)_{\boldsymbol{s}}$ first calculates $\langle \langle \mathscr{C}^\ast(\boldsymbol{e^{(0)}}) \rangle \sharp$ by copying it from $\mathscr{C}^\ast(\boldsymbol{\tilde{e}'})$, and then computes the remaining $\langle \mathscr{C}^\ast(\boldsymbol{g}) \rangle \rangle \sharp \mathscr{C}^\ast(\boldsymbol{h})$ by simulating the computation of $\mathscr{C}^\ast(\boldsymbol{e})$ by $\mathcal{A}(\varphi)_{\boldsymbol{s}}$, inserting $\rangle$ before $\sharp$. 
This computation of $\mathcal{A}(\varphi^\dagger)_{\boldsymbol{s}}$ is clearly standard.

\item If $m$ and $n$ belong to $A$ and $B$, respectively, then $\boldsymbol{\tilde{e}}$, $\boldsymbol{e}$, $\boldsymbol{\tilde{e}'}$ and $\boldsymbol{e'}$ are of the form $\Lbag \boldsymbol{\tilde{g}} \Rbag \hbar \boldsymbol{\tilde{h}}$, $\boldsymbol{h}, \Lbag \Lbag \boldsymbol{e^{(0)}} \Rbag \hbar \Lbag \boldsymbol{\tilde{g}} \Rbag \Rbag \hbar \boldsymbol{\tilde{h}}$ and $\Lbag \boldsymbol{e^{(0)}} \Rbag \hbar \boldsymbol{h}$, respectively.
Again, with the help of m-views, $\mathcal{A}(\varphi^\dagger)_{\boldsymbol{s}}$ first calculates $\langle \mathscr{C}^\ast(\boldsymbol{e^{(0)}}) \rangle \sharp$ by copying $\mathscr{C}^\ast(\boldsymbol{e^{(0)}})$ from $\mathscr{C}^\ast(\boldsymbol{\tilde{e}'})$, and then computes $\mathscr{C}^\ast(\boldsymbol{h})$ by simulating the computation of $\mathscr{C}^\ast(\boldsymbol{e})$ by $\mathcal{A}(\varphi)_{\boldsymbol{s}}$.
This computation of $\mathcal{A}(\varphi^\dagger)_{\boldsymbol{s}}$ is standard.

\item The remaining two cases are completely analogous.

\end{itemize}

\item Now, it remains to stipulate how $\mathcal{A}(\varphi^\dagger)_{\boldsymbol{s}}$ distinguishes the above four cases.
Assume that $(q_{!A \multimap B})^{[3]} \boldsymbol{v} (n)^{[3]}$ occurs when $\mathcal{A}(\varphi)_{\boldsymbol{s}}$ computes the inner element $n$.
Note that $\boldsymbol{v}$ has enough information to identify $n$.
Thus, by simulating $\mathcal{A}(\varphi)_{\boldsymbol{s}}$'s computation for inner elements and replacing $(n)^{[3]}$ with $(q_{!A \multimap !B})^{[2]}$ to learn about $m$, $\mathcal{A}(\varphi^\dagger)_{\boldsymbol{s}}$ may recognize the current case out of the four described above.
Specifically, if $(q_{\mathcal{T}})^{[3]} \mapsto x$ is the first step for $\mathcal{A}(\varphi)_{\boldsymbol{s}}$ to compute $\boldsymbol{e}$, then correspondingly $\mathcal{A}(\varphi^\dagger)_{\boldsymbol{s}}$ computes as $(q_{\mathcal{T}})^{[3]} \mapsto v_1$, $(q_{\mathcal{T}})^{[3]} v_1 v_2 \mapsto v_3$, \dots, $(q_{\mathcal{T}})^{[3]} \boldsymbol{v} \mapsto (q_{!A \multimap !B})^{[2]}$, $(q_{\mathcal{T}})^{[3]} \boldsymbol{v}  (q_{!A \multimap !B})^{[2]} (\mathsf{m})^{[2]} \mapsto x'$, where $x'$ is the first step for $\mathcal{A}(\varphi^\dagger)_{\boldsymbol{s}}$ to compute $\boldsymbol{e'}$, and if $\mathcal{A}(\varphi)_{\boldsymbol{s}}$ next computes $(q_{\mathcal{T}})^{[3]} x y \mapsto z$, then correspondingly $\mathcal{A}(\varphi^\dagger)_{\boldsymbol{s}}$ computes as $(q_{\mathcal{T}})^{[3]} \boldsymbol{v}  (q_{!A \multimap !B})^{[2]} (\mathsf{m})^{[2]} x' y' \mapsto v_1$, $(q_{\mathcal{T}})^{[3]} \boldsymbol{v}  (q_{!A \multimap !B})^{[2]} (\mathsf{m})^{[2]} x' y' v_1 v_2 \mapsto v_3$, \dots, $(q_{\mathcal{T}})^{[3]} \boldsymbol{v}  (q_{!A \multimap !B})^{[2]} (\mathsf{m})_2 x' y' \boldsymbol{v} \mapsto (q_{!A \multimap !B})^{[2]}$, $(q_{\mathcal{T}})^{[3]} \boldsymbol{v}  (q_{!A \multimap !B})^{[2]} (\mathsf{m})^{[2]} x' y' \boldsymbol{v}  (q_{!A \multimap !B})^{[2]} (\mathsf{m})^{[2]} \mapsto z'$, where $y', z'$ are the second and third steps for $\mathcal{A}(\varphi^\dagger)_{\boldsymbol{s}}$ to compute $\boldsymbol{e'}$, and so on. That is, $\mathcal{A}(\varphi^\dagger)_{\boldsymbol{s}}$ remembers the current case by inserting the sequence $\boldsymbol{v}  (q_{!A \multimap !B})^{[2]} (\mathsf{m})^{[2]}$ (of length $\leqslant 8$) between each computational step. 

\end{enumerate}
It should be clear from the above description how to construct $\mathcal{A}(\varphi^\dagger)_{\boldsymbol{s}}$ from $\mathcal{A}(\varphi)_{\boldsymbol{s}}$. 
Finally, note that $\mathcal{A}(\varphi^\dagger)$ is clearly standard.
\end{proof}

\subsection{Examples of viable strategies}
\label{Examples}
This section presents various examples of a viable strategy realized by a standard st-algorithm.
These strategies except \emph{fixed-point strategies} $\mathit{fix}_A$ are actually finitary; thus, we need the notion of viable strategies only for promotion $(\_)^\dagger$ and $\mathit{fix}_A$, both of which are necessary for the proof of Turing completeness in Section~\ref{TuringCompleteness}.

\begin{example}
Given a normalized game $A$, we define an st-algorithm $\mathcal{A}(\mathit{cp}_A)$ that realizes the copy-cat strategy $\mathit{cp}_A : [A_\mathscr{W}]_{\boldsymbol{e}} \multimap [A_\mathscr{E}]_{\boldsymbol{e}}$ by $\mathcal{Q}_{\mathit{cp}_A}(m) \stackrel{\mathrm{df. }}{=} \begin{cases} \top &\text{if $m \in \pi_1(M_{A \multimap A}^{\mathsf{Init}})$} \\ \bot &\text{otherwise} \end{cases}$, $\mathcal{S}_{\mathit{cp}_A} \stackrel{\mathrm{df. }}{=} \pi_1(M_{A \multimap A}^{\mathsf{Init}})$, $|\mathcal{A}(\mathit{cp}_A)_m| \stackrel{\mathrm{df. }}{=} 3$ and $\|\mathcal{A}(\mathit{cp}_A)_m\| \stackrel{\mathrm{df. }}{=} 0$ for all $m \in \mathcal{S}_{\mathit{cp}_A}$, and the table is as given in Appendix~\ref{TableCP}.
$\mathcal{Q}_{\mathit{cp}_A}$ clearly satisfies the required condition. 
Accordingly, $\mathcal{A}(\mathit{cp}_A)_m^\circledS$ computes as in the following diagrams:
\begin{center}
\begin{tabular}{ccccccc}
$\mathcal{G}(M_{A \multimap A})^{[0]}$ & $\&$ & $\mathcal{G}(M_{A \multimap A})^{[1]}$ & $\&$ & $\mathcal{G}(M_{A \multimap A})^{[2]}$ & $\stackrel{\mathcal{A}(\mathit{cp}_A)_m^\circledS}{\Rightarrow}$ & $\mathcal{G}(M_{A \multimap A})^{[3]}$ \\
\cline{1-7} 
&&&&&&$(\hat{q}_{\mathcal{T}})^{[3]}$ \\
&&&&$(\hat{q}_{\mathcal{T}})^{[2]}$&& \\
&&&&$(e_1)^{[2]}$&& \\
&&&&&&$(e_1)^{[3]}$ \\
&&&&&&$(q_{\mathcal{T}})^{[3]}$ \\
&&&&$(q_{\mathcal{T}})^{[2]}$&& \\
&&&&$(e_2)^{[2]}$&& \\
&&&&&&$(e_2)^{[3]}$ \\
&&&&&$\vdots$& \\
&&&&&&$(q_{\mathcal{T}})^{[3]}$ \\
&&&&$(q_{\mathcal{T}})^{[2]}$&& \\
&&&&$(e_k)^{[2]}$&& \\
&&&&&&$(e_k)^{[3]}$ \\
&&&&&&$(q_{\mathcal{T}})^{[3]}$ \\
&&&&$(q_{\mathcal{T}})^{[2]}$&& \\
&&&&$(\checkmark)^{[2]}$&& \\
&&&&&&$(\checkmark)^{[3]}$
\end{tabular}
\end{center}
\begin{center}
\begin{tabular}{ccccccc}
$\mathcal{G}(M_{A \multimap A})^{[0]}$ & $\&$ & $\mathcal{G}(M_{A \multimap A})^{[1]}$ & $\&$ & $\mathcal{G}(M_{A \multimap A})^{[2]}$ & $\stackrel{\mathcal{A}(\mathit{cp}_A)_m^\circledS}{\Rightarrow}$ & $\mathcal{G}(M_{A \multimap A})^{[3]}$ \\
\cline{1-7} 
&&&&&&$(\hat{q}_{A \multimap A})^{[3]}$ \\
&&&&$(\hat{q}_{A \multimap A})^{[2]}$&& \\
&&&&$(\mathsf{a_W})^{[2]}$ ($(\mathsf{a_E})^{[2]}$)&& \\
&&&&&&$(\mathsf{a_E})^{[3]}$ ($(\mathsf{a_W})^{[3]}$)
\end{tabular}
\end{center}
Then it is easy to see that $\mathsf{st}(\mathcal{A}(\mathit{cp}_A)) = \mathit{cp}_A$ holds, showing viability of $\mathit{cp}_A$.
Also, $\mathcal{A}(\mathit{cp}_A)$ is clearly standard.
In a completely analogous (but slightly more complex for outer tags) manner, we may show that the dereliction $\mathit{der}_A : A \Rightarrow A$ for each normalized game $A$ is viable with a standard st-algorithm realizing it as well.
\end{example}

\begin{example}
The \emph{\bfseries case strategy} $\mathit{case}_A : [A_\mathscr{WWW}]_{\Lbag \boldsymbol{e''} \Rbag \hbar \boldsymbol{f}} \& [A_{\mathscr{EWW}}]_{\Lbag \boldsymbol{e'} \Rbag \hbar \boldsymbol{f}} \& [\boldsymbol{2}_{\mathscr{EW}}]_{\Lbag \boldsymbol{e} \Rbag \hbar} \Rightarrow [A_{\mathscr{E}}]_{\boldsymbol{f}}$ for each normalized game $A$ is defined by:
\begin{align*}
\mathit{case}_A &\stackrel{\mathrm{df. }}{=} \mathsf{Pref}(\{ [a_{\mathscr{E}}]_{\boldsymbol{e}} [\hat{q}_{\mathscr{EW}}]_{\Lbag \boldsymbol{e} \Rbag \hbar} [\mathit{tt}_{\mathscr{EW}}]_{\Lbag \boldsymbol{e} \Rbag \hbar} [a_\mathscr{WWW}]_{\Lbag \Rbag \hbar \boldsymbol{e}} . \boldsymbol{s} \mid [a_{\mathscr{E}}]_{\boldsymbol{e}} [a_\mathscr{WWW}]_{\Lbag \Rbag \hbar \boldsymbol{e}} . \boldsymbol{s} \in \mathit{der}^\mathscr{W}_A \} \\
& \ \ \ \ \cup \{ [a_{\mathscr{E}}]_{\boldsymbol{f}} [q_{\mathscr{EW}}]_{\Lbag \boldsymbol{f} \Rbag \hbar} [\mathit{ff}_{\mathscr{EW}}]_{\Lbag \boldsymbol{f} \Rbag \hbar} [a_{\mathscr{EWW}}]_{\Lbag \Rbag \hbar \boldsymbol{f}} . \boldsymbol{t} \mid [a_{\mathscr{E}}]_{\boldsymbol{f}} [a_{\mathscr{EWW}}]_{\Lbag \Rbag \hbar \boldsymbol{f}} . \boldsymbol{t} \in \mathit{der}^{\mathscr{E}}_A \})^{\mathsf{Even}}
\end{align*}
where $\mathit{der}^\mathscr{W}_A : [A_\mathscr{WWW}]_{\Lbag \boldsymbol{e''} \Rbag \hbar \boldsymbol{f}} \Rightarrow  [A_{\mathscr{E}}]_{\boldsymbol{f}}$ and $\mathit{der}^{\mathscr{E}}_A : [A_{\mathscr{EWW}}]_{\Lbag \boldsymbol{e'} \Rbag \hbar \boldsymbol{f}} \Rightarrow  [A_{\mathscr{E}}]_{\boldsymbol{f}}$ are the same as the usual dereliction $\mathit{der}_A : [A_\mathscr{W}]_{\Lbag \boldsymbol{e'} \Rbag \hbar \boldsymbol{f}} \Rightarrow [A_\mathscr{E}]_{\boldsymbol{f}}$ up to inner tags.
Given input strategies $\sigma_1, \sigma_2 : T \multimap A$ and $\beta : T \multimap \boldsymbol{2}$, the composition $\mathit{case}_A \circ \langle \langle \sigma_1, \sigma_2 \rangle, \beta \rangle^\dagger$ is $\sigma_1$ (resp. $\sigma_2$) if $\beta$ is $\{ \boldsymbol{\epsilon}, [\hat{q}_{\mathscr{E}}] . [\mathit{tt}_{\mathscr{E}}] \}$ (resp. $\{ \boldsymbol{\epsilon}, [\hat{q}_{\mathscr{E}}] . [\mathit{ff}_{\mathscr{E}}] \}$).
We give an st-algorithm $\mathcal{A}(\mathit{case}_A)$ that realizes $\mathit{case}_A$ whose states and query function are similar to those of $\mathcal{A}(\mathit{cp}_A)$, and for all $m \in \mathcal{S}_{\mathcal{A}(\mathit{case}_A)}$ the instruction strategy $\mathcal{A}(\mathit{case}_A)_m^\circledS$ is as follows (again, we skip formally writing down the table of $\mathcal{A}(\mathit{case}_A)_m$), where we write $A^{A^2 \& \boldsymbol{2}}$ for $A \& A \& \boldsymbol{2} \Rightarrow A$:
\begin{center}
\begin{tabular}{ccccccc}
$\mathcal{G}(M_{A^{A^2 \& \boldsymbol{2}}})^{[0]}$ & $\&$ & $\mathcal{G}(M_{A^{A^2 \& \boldsymbol{2}}})^{[1]}$ & $\&$ & $\mathcal{G}(M_{A^{A^2 \& \boldsymbol{2}}})^{[2]}$ & $\stackrel{\mathcal{A}(\mathit{case}_A)_m^\circledS}{\Rightarrow}$ & $\mathcal{G}(M_{A^{A^2 \& \boldsymbol{2}}})^{[3]}$ \\
\cline{1-7} 
&&&&&&$(\hat{q}_{A^{A^2 \& \boldsymbol{2}}})^{[3]}$ \\
&&&&$(\hat{q}_{A^{A^2 \& \boldsymbol{2}}})^{[2]}$&& \\
&&&&$(\mathsf{a_{E}})^{[2]}$&& \\
&&&&&&$(\mathsf{\hat{q}_{EW}})^{[3]}$ 
\end{tabular}
\end{center}
where $a_{\mathscr{E}} \in \pi_1(M_{A \multimap A}^{\mathsf{Init}})$, which can be recognized as it has no justifier.

\begin{center}
\begin{tabular}{ccccccc}
$\mathcal{G}(M_{A^{A^2 \& \boldsymbol{2}}})^{[0]}$ & $\&$ & $\mathcal{G}(M_{A^{A^2 \& \boldsymbol{2}}})^{[1]}$ & $\&$ & $\mathcal{G}(M_{A^{A^2 \& \boldsymbol{2}}})^{[2]}$ & $\stackrel{\mathcal{A}(\mathit{case}_A)_m^\circledS}{\Rightarrow}$ & $\mathcal{G}(M_{A^{A^2 \& \boldsymbol{2}}})^{[3]}$ \\
\cline{1-7} 
&&&&&&$(q_{A^{A^2 \& \boldsymbol{2}}})^{[3]}$ \\
&&&&$(q_{A^{A^2 \& \boldsymbol{2}}})^{[2]}$&& \\
&&&&$(\mathsf{tt_{EW}})^{[2]}$ ($(\mathsf{ff_{EW}})^{[2]}$)&& \\
$(q_{A^{A^2 \& \boldsymbol{2}}})^{[0]}$&&&&&& \\
$(\mathsf{a_{E}})^{[0]}$&&&&&& \\
&&&&&&$(\mathsf{a_{WWW}})^{[3]}$ ($(\mathsf{a_{EWW}})^{[3]}$) 
\end{tabular}
\end{center}

\begin{center}
\begin{tabular}{ccccccc}
$\mathcal{G}(M_{A^{A^2 \& \boldsymbol{2}}})^{[0]}$ & $\&$ & $\mathcal{G}(M_{A^{A^2 \& \boldsymbol{2}}})^{[1]}$ & $\&$ & $\mathcal{G}(M_{A^{A^2 \& \boldsymbol{2}}})^{[2]}$ & $\stackrel{\mathcal{A}(\mathit{case}_A)_m^\circledS}{\Rightarrow}$ & $\mathcal{G}(M_{A^{A^2 \& \boldsymbol{2}}})^{[3]}$ \\
\cline{1-7} 
&&&&&&$(q_{A^{A^2 \& \boldsymbol{2}}})^{[3]}$ \\
&&&&$(q_{A^{A^2 \& \boldsymbol{2}}})^{[2]}$&& \\
&&&&$(\mathsf{a_{WWW}})^{[2]}$ ($(\mathsf{a_{EWW}})^{[2]}$)&& \\
&&&&&&$(\mathsf{a_{E}})^{[3]}$ 
\end{tabular}
\end{center}

\begin{center}
\begin{tabular}{ccccccc}
$\mathcal{G}(M_{A^{A^2 \& \boldsymbol{2}}})^{[0]}$ & $\&$ & $\mathcal{G}(M_{A^{A^2 \& \boldsymbol{2}}})^{[1]}$ & $\&$ & $\mathcal{G}(M_{A^{A^2 \& \boldsymbol{2}}})^{[2]}$ & $\stackrel{\mathcal{A}(\mathit{case}_A)_m^\circledS}{\Rightarrow}$ & $\mathcal{G}(M_{A^{A^2 \& \boldsymbol{2}}})^{[3]}$ \\
\cline{1-7} 
&&&&&&$(q_{A^{A^2 \& \boldsymbol{2}}})^{[3]}$ \\
&&&&$(q_{A^{A^2 \& \boldsymbol{2}}})^{[2]}$&& \\
&&&&$(\mathsf{a_{E}})^{[2]}$ && \\
$(q_{A^{A^2 \& \boldsymbol{2}}})^{[0]}$&&&&&& \\
$(\mathsf{a'_{WWW}})^{[0]}$ ($(\mathsf{a'_{EWW}})^{[0]}$)&&&&&& \\
&&&&&&$(\mathsf{a_{WWW}})^{[3]}$ ($(\mathsf{a_{EWW}})^{[3]}$)
\end{tabular}
\end{center}
where $a_{\mathscr{E}} \not \in \pi_1(M_{A \multimap A}^{\mathsf{Init}})$.

\begin{center}
\begin{tabular}{ccccccc}
$\mathcal{G}(M_{A^{A^2 \& \boldsymbol{2}}})^{[0]}$ & $\&$ & $\mathcal{G}(M_{A^{A^2 \& \boldsymbol{2}}})^{[1]}$ & $\&$ & $\mathcal{G}(M_{A^{A^2 \& \boldsymbol{2}}})^{[2]}$ & $\stackrel{\mathcal{A}(\mathit{case}_A)_m^\circledS}{\Rightarrow}$ & $\mathcal{G}(M_{A^{A^2 \& \boldsymbol{2}}})^{[3]}$ \\
\cline{1-7} 
&&&&&&$(\hat{q}_{\mathcal{T}})^{[3]}$ \\
&&&&$(\hat{q}_{A^{A^2 \& \boldsymbol{2}}})^{[2]}$&& \\
&&&&$(\mathsf{a_{E}})^{[2]}$&& \\
&&&&&&$(\langle)^{[3]}$ \\
&&&&&&$(q_{\mathcal{T}})^{[3]}$ \\
&&&&$(\hat{q}_{A^{A^2 \& \boldsymbol{2}}})^{[2]}$&& \\
&&&&$(\mathsf{a_{E}})^{[2]}$&& \\
&&&&$(\hat{q}_{\mathcal{T}})^{[2]}$&& \\
&&&&$(e_1)^{[2]}$&& \\
&&&&&&$(e_1)^{[3]}$ \\
&&&&&&$(q_{\mathcal{T}})^{[3]}$ \\
&&&&$(\hat{q}_{A^{A^2 \& \boldsymbol{2}}})^{[2]}$&& \\
&&&&$(\mathsf{a_{E}})^{[2]}$&& \\
&&&&$(q_{\mathcal{T}})^{[2]}$&& \\
&&&&$(e_2)^{[2]}$&& \\
&&&&&&$(e_2)^{[3]}$ \\
&&&&&$\vdots$& \\
&&&&&&$(q_{\mathcal{T}})^{[3]}$ \\
&&&&$(\hat{q}_{A^{A^2 \& \boldsymbol{2}}})^{[2]}$&& \\
&&&&$(\mathsf{a_{E}})^{[2]}$&& \\
&&&&$(q_{\mathcal{T}})^{[2]}$&& \\
&&&&$(e_k)^{[2]}$&& \\
&&&&&&$(e_k)^{[3]}$ \\
&&&&&&$(q_{\mathcal{T}})^{[3]}$ \\
&&&&$(\hat{q}_{A^{A^2 \& \boldsymbol{2}}})^{[2]}$&& \\
&&&&$(\mathsf{a_{E}})^{[2]}$&& \\
&&&&$(q_{\mathcal{T}})^{[2]}$&& \\
&&&&$(\checkmark)^{[2]}$&& \\
&&&&&&$(\rangle)^{[3]}$ \\
&&&&&&$(q_{\mathcal{T}})^{[3]}$ \\
&&&&&&$(\sharp)^{[3]}$ \\
&&&&&&$(q_{\mathcal{T}})^{[3]}$ \\
&&&&&&$(\checkmark)^{[3]}$
\end{tabular}
\end{center}
where $a_{\mathscr{E}} \in M_{A \multimap A}^{\mathsf{Init}}$. Note that the iteration of $(\hat{q}_{A^{A^2 \& \boldsymbol{2}}})_2 . (\mathsf{a_{E}})_2$ is to distinguish this case from other cases.
This remark is applied to the remaining diagrams below.

\begin{center}
\begin{tabular}{ccccccc}
$\mathcal{G}(M_{A^{A^2 \& \boldsymbol{2}}})^{[0]}$ & $\&$ & $\mathcal{G}(M_{A^{A^2 \& \boldsymbol{2}}})^{[1]}$ & $\&$ & $\mathcal{G}(M_{A^{A^2 \& \boldsymbol{2}}})^{[2]}$ & $\stackrel{\mathcal{A}(\mathit{case}_A)_m^\circledS}{\Rightarrow}$ & $\mathcal{G}(M_{A^{A^2 \& \boldsymbol{2}}})^{[3]}$ \\
\cline{1-7} 
&&&&&&$(\hat{q}_{\mathcal{T}})^{[3]}$ \\
&&&&$(\hat{q}_{A^{A^2 \& \boldsymbol{2}}})^{[2]}$&& \\
&&&&$(\mathsf{tt_{EW}})^{[2]}$ ($(\mathsf{ff_{EW}})^{[2]}$)&& \\
&&&&&&$(\langle)^{[3]}$ \\
&&&&&&$(q_{\mathcal{T}})^{[3]}$ \\
&&&&&&$(\rangle)^{[3]}$ \\
&&&&&&$(q_{\mathcal{T}})^{[3]}$ \\
&&&&&&$(\sharp)^{[3]}$ \\
&&&&&&$(q_{\mathcal{T}})^{[3]}$ \\
&&&&$(\hat{q}_{\mathcal{T}})^{[2]}$&& \\
&&&&$(\langle)^{[2]}$&& \\
&&&&$(\hat{q}_{A^{A^2 \& \boldsymbol{2}}})^{[2]}$&& \\
&&&&$(\mathsf{tt_{EW}})^{[2]}$ ($(\mathsf{ff_{EW}})^{[2]}$)&& \\
&&&&$(q_{\mathcal{T}})^{[2]}$&& \\
&&&&$(e_1)^{[2]}$&& \\
&&&&&&$(e_1)^{[3]}$ \\
&&&&&&$(q_{\mathcal{T}})^{[3]}$ \\
&&&&$(\hat{q}_{A^{A^2 \& \boldsymbol{2}}})^{[2]}$&& \\
&&&&$(\mathsf{tt_{EW}})^{[2]}$ ($(\mathsf{ff_{EW}})^{[2]}$)&& \\
&&&&$(q_{\mathcal{T}})^{[2]}$&& \\
&&&&$(e_2)^{[2]}$&& \\
&&&&&&$(e_2)^{[3]}$ \\
&&&&&$\vdots$& \\
&&&&&&$(q_{\mathcal{T}})^{[3]}$ \\
&&&&$(\hat{q}_{A^{A^2 \& \boldsymbol{2}}})^{[2]}$&& \\
&&&&$(\mathsf{tt_{EW}})^{[2]}$ ($(\mathsf{ff_{EW}})^{[2]}$)&& \\
&&&&$(q_{\mathcal{T}})^{[2]}$&& \\
&&&&$(e_k)^{[2]}$&& \\
&&&&&&$(e_k)^{[3]}$ \\
&&&&&&$(q_{\mathcal{T}})^{[3]}$ \\
&&&&$(\hat{q}_{A^{A^2 \& \boldsymbol{2}}})^{[2]}$&& \\
&&&&$(\mathsf{tt_{EW}})^{[2]}$ ($(\mathsf{ff_{EW}})^{[2]}$)&& \\
&&&&$(q_{\mathcal{T}})^{[2]}$&& \\
&&&&$(\rangle)^{[2]}$&& \\
&&&&$(q_{\mathcal{T}})^{[2]}$&& \\
&&&&$(\sharp)^{[2]}$&& \\
&&&&$(q_{\mathcal{T}})^{[2]}$&& \\
&&&&$(\checkmark)^{[2]}$&& \\
&&&&&&$(\checkmark)^{[3]}$
\end{tabular}
\end{center}

\begin{center}
\begin{tabular}{ccccccc}
$\mathcal{G}(M_{A^{A^2 \& \boldsymbol{2}}})^{[0]}$ & $\&$ & $\mathcal{G}(M_{A^{A^2 \& \boldsymbol{2}}})^{[1]}$ & $\&$ & $\mathcal{G}(M_{A^{A^2 \& \boldsymbol{2}}})^{[2]}$ & $\stackrel{\mathcal{A}(\mathit{case}_A)_m^\circledS}{\Rightarrow}$ & $\mathcal{G}(M_{A^{A^2 \& \boldsymbol{2}}})^{[3]}$ \\
\cline{1-7} 
&&&&&&$(\hat{q}_{\mathcal{T}})^{[3]}$ \\
&&&&$(\hat{q}_{A^{A^2 \& \boldsymbol{2}}})^{[2]}$&& \\
&&&&$(\mathsf{a_{WWW}})^{[2]}$ ($(\mathsf{a_{EWW}})^{[2]}$)&& \\
&&&&$(q_{\mathcal{T}})^{[2]}$&& \\
&&&&$(\langle)^{[2]}$&& \\
&&&&$(q_{\mathcal{T}})^{[2]}$&& \\
&&&&$(\rangle)_2$&& \\
&&&&$(q_{\mathcal{T}})^{[2]}$&& \\
&&&&$(\sharp)^{[2]}$&& \\
&&&&$(q_{\mathcal{T}})^{[2]}$&& \\
&&&&$(e_1)^{[2]}$&& \\
&&&&&&$(e_1)^{[3]}$ \\
&&&&&&$(q_{\mathcal{T}})^{[3]}$ \\
&&&&$(\hat{q}_{A^{A^2 \& \boldsymbol{2}}})^{[2]}$&& \\
&&&&$(\mathsf{a_{WWW}})^{[2]}$ ($(\mathsf{a_{EWW}})^{[2]}$)&& \\
&&&&$(q_{\mathcal{T}})^{[2]}$&& \\
&&&&$(e_2)^{[2]}$&& \\
&&&&&&$(e_2)^{[3]}$ \\
&&&&&$\vdots$& \\
&&&&&&$(q_{\mathcal{T}})^{[3]}$ \\
&&&&$(\hat{q}_{A^{A^2 \& \boldsymbol{2}}})^{[2]}$&& \\
&&&&$(\mathsf{a_{WWW}})^{[2]}$ ($(\mathsf{a_{EWW}})^{[2]}$)&& \\
&&&&$(q_{\mathcal{T}})^{[2]}$&& \\
&&&&$(e_k)^{[2]}$&& \\
&&&&&&$(e_k)^{[3]}$ \\
&&&&&&$(q_{\mathcal{T}})^{[3]}$ \\
&&&&$(\hat{q}_{A^{A^2 \& \boldsymbol{2}}})^{[2]}$&& \\
&&&&$(\mathsf{a_{WWW}})^{[2]}$ ($(\mathsf{a_{EWW}})^{[2]}$)&& \\
&&&&$(q_{\mathcal{T}})^{[2]}$&& \\
&&&&$(\checkmark)^{[2]}$&& \\
&&&&&&$(\checkmark)^{[3]}$ 
\end{tabular}
\end{center}

\begin{center}
\begin{tabular}{ccccccc}
$\mathcal{G}(M_{A^{A^2 \& \boldsymbol{2}}})^{[0]}$ & $\&$ & $\mathcal{G}(M_{A^{A^2 \& \boldsymbol{2}}})^{[1]}$ & $\&$ & $\mathcal{G}(M_{A^{A^2 \& \boldsymbol{2}}})^{[2]}$ & $\stackrel{\mathcal{A}(\mathit{case}_A)_m^\circledS}{\Rightarrow}$ & $\mathcal{G}(M_{A^{A^2 \& \boldsymbol{2}}})^{[3]}$ \\
\cline{1-7} 
&&&&&&$(\hat{q}_{\mathcal{T}})^{[3]}$ \\
&&&&$(\hat{q}_{A^{A^2 \& \boldsymbol{2}}})^{[2]}$&& \\
&&&&$(a_{\mathsf{E}})^{[2]}$&& \\
&&&&&&$(\langle)^{[3]}$ \\
&&&&&&$(q_{\mathcal{T}})^{[3]}$ \\
&&&&&&$(\rangle)^{[3]}$ \\
&&&&&&$(q_{\mathcal{T}})^{[3]}$ \\
&&&&&&$(\sharp)^{[3]}$ \\
&&&&&&$(q_{\mathcal{T}})^{[3]}$ \\
&&&&$(\hat{q}_{A^{A^2 \& \boldsymbol{2}}})^{[2]}$&& \\
&&&&$(a_{\mathsf{E}})^{[2]}$&& \\
&&&&$(q_{\mathcal{T}})^{[2]}$&& \\
&&&&$(e_1)^{[2]}$&& \\
&&&&&&$(e_1)^{[3]}$ \\
&&&&&&$(q_{\mathcal{T}})^{[3]}$ \\
&&&&$(\hat{q}_{A^{A^2 \& \boldsymbol{2}}})^{[2]}$&& \\
&&&&$(a_{\mathsf{E}})^{[2]}$&& \\
&&&&$(q_{\mathcal{T}})^{[2]}$&& \\
&&&&$(e_2)^{[2]}$&& \\
&&&&&&$(e_2)^{[3]}$ \\
&&&&&$\vdots$& \\
&&&&&&$(q_{\mathcal{T}})^{[3]}$ \\
&&&&$(\hat{q}_{A^{A^2 \& \boldsymbol{2}}})^{[2]}$&& \\
&&&&$(a_{\mathsf{E}})^{[2]}$&& \\
&&&&$(q_{\mathcal{T}})^{[2]}$&& \\
&&&&$(e_k)^{[2]}$&& \\
&&&&&&$(e_k)^{[3]}$ \\
&&&&&&$(q_{\mathcal{T}})^{[3]}$ \\
&&&&$(q_{\mathcal{T}})^{[2]}$&& \\
&&&&$(\checkmark)^{[2]}$&& \\
&&&&&&$(\checkmark)^{[3]}$ 
\end{tabular}
\end{center}
where $a_{\mathscr{E}} \not \in \pi_1(M_{A \multimap A}^{\mathsf{Init}})$.
Clearly, $\mathsf{st}(\mathcal{A}(\mathit{case}_A)) = \mathit{case}_A$, and so $\mathit{case}_A$ is viable.
And again, it is easy to see that $\mathcal{A}(\mathit{case}_A)$ is standard. 
\end{example}

\begin{example}
Consider the \emph{\bfseries ifzero strategy} $\mathit{zero?} : [\mathcal{N}_\mathscr{W}]_{\Lbag \boldsymbol{e} \Rbag \hbar} \Rightarrow [\boldsymbol{2}_\mathscr{E}]$ defined by $\mathit{zero?} \stackrel{\mathrm{df. }}{=} \mathsf{Pref}(\{ [\hat{q}_\mathscr{E}] [\hat{q}_\mathscr{W}]_{\Lbag \Rbag \hbar} [\mathit{no}_\mathscr{W}]_{\Lbag \Rbag \hbar} [\mathit{tt}_\mathscr{E}], [\hat{q}_\mathscr{E}] [\hat{q}_\mathscr{W}]_{\Lbag \Rbag \hbar} [\mathit{yes}_\mathscr{W}]_{\Lbag \Rbag \hbar} [\mathit{ff}_\mathscr{E}] \})^{\mathsf{Even}}$.
It outputs $\mathit{tt}$ (resp. $\mathit{ff}$) if the input is $\underline{0}$ (resp. $\underline{n+1}$ for some $n \in \mathbb{N}$).
Let us give an st-algorithm $\mathcal{A}(\mathit{zero?})$ that realizes $\mathit{zero?}$ as follows.
Define $\mathcal{Q}_{\mathcal{A}(\mathit{zero?})}(m) \stackrel{\mathrm{df. }}{=} \begin{cases} \top &\text{if $m = \hat{q}_\mathscr{E}$;} \\ \bot &\text{otherwise} \end{cases}$, $\mathcal{S}_{\mathcal{A}(\mathit{zero?})} \stackrel{\mathrm{df. }}{=} \{ \hat{q}_\mathscr{E} \}$, $|\mathcal{A}(\mathit{zero?})_{\hat{q}_\mathscr{E}}| \stackrel{\mathrm{df. }}{=} 3$, $\| \mathcal{A}(\mathit{zero?})_{\hat{q}_\mathscr{E}} \| \stackrel{\mathrm{df. }}{=} 0$, and the instruction strategy $\mathcal{A}(\mathit{zero?})_{\hat{q}_\mathscr{E}}^\circledS$ is as depicted in the following diagrams (again, we omit the formal description of $\mathcal{A}(\mathit{zero?})_{\hat{q}_\mathscr{E}}$ as it should be clear at this point):
\begin{center}
\begin{tabular}{ccccccc}
$\mathcal{G}(M_{\mathcal{N} \Rightarrow \boldsymbol{2}})^{[0]}$ & $\&$ & $\mathcal{G}(M_{\mathcal{N} \Rightarrow \boldsymbol{2}})^{[1]}$ & $\&$ & $\mathcal{G}(M_{\mathcal{N} \Rightarrow \boldsymbol{2}})^{[2]}$ & $\stackrel{\mathcal{A}(\mathit{zero?})_{\hat{q}_\mathscr{E}}^\circledS}{\Rightarrow}$ & $\mathcal{G}(M_{\mathcal{N} \Rightarrow \boldsymbol{2}})^{[3]}$ \\
\cline{1-7} 
&&&&&&$(\hat{q}_{\mathcal{N} \Rightarrow \boldsymbol{2}})^{[3]}$ \\
&&&&$(\hat{q}_{\mathcal{N} \Rightarrow \boldsymbol{2}})^{[2]}$&& \\
&&&&$(\mathsf{\hat{q}_E})^{[2]}$&& \\
&&&&&&$(\mathsf{\hat{q}_W})^{[3]}$ 
\end{tabular}
\end{center}

\begin{center}
\begin{tabular}{ccccccc}
$\mathcal{G}(M_{\mathcal{N} \Rightarrow \boldsymbol{2}})^{[0]}$ & $\&$ & $\mathcal{G}(M_{\mathcal{N} \Rightarrow \boldsymbol{2}})^{[1]}$ & $\&$ & $\mathcal{G}(M_{\mathcal{N} \Rightarrow \boldsymbol{2}})^{[2]}$ & $\stackrel{\mathcal{A}(\mathit{zero?})_{\hat{q}_\mathscr{E}}^\circledS}{\Rightarrow}$ & $\mathcal{G}(M_{\mathcal{N} \Rightarrow \boldsymbol{2}})^{[3]}$ \\
\cline{1-7} 
&&&&&&$(\hat{q}_{\mathcal{T}})^{[3]}$ \\
&&&&$(\hat{q}_{\mathcal{N} \Rightarrow \boldsymbol{2}})^{[2]}$&& \\
&&&&$(\mathsf{\hat{q}_E})^{[2]}$&& \\
&&&&&&$(\langle)^{[3]}$ \\
&&&&&&$(q_{\mathcal{T}})^{[3]}$ \\
&&&&&&$(\rangle)^{[3]}$ \\
&&&&&&$(q_{\mathcal{T}})^{[3]}$ \\
&&&&&&$(\sharp)^{[3]}$ \\
&&&&&&$(q_{\mathcal{T}})^{[3]}$ \\
&&&&&&$(\checkmark)^{[3]}$ \\
\end{tabular}
\end{center}

\begin{center}
\begin{tabular}{ccccccc}
$\mathcal{G}(M_{\mathcal{N} \Rightarrow \boldsymbol{2}})^{[0]}$ & $\&$ & $\mathcal{G}(M_{\mathcal{N} \Rightarrow \boldsymbol{2}})^{[1]}$ & $\&$ & $\mathcal{G}(M_{\mathcal{N} \Rightarrow \boldsymbol{2}})^{[2]}$ & $\stackrel{\mathcal{A}(\mathit{zero?})_{\hat{q}_\mathscr{E}}^\circledS}{\Rightarrow}$ & $\mathcal{G}(M_{\mathcal{N} \Rightarrow \boldsymbol{2}})^{[3]}$ \\
\cline{1-7} 
&&&&&&$(\hat{q}_{\mathcal{N} \Rightarrow \boldsymbol{2}})^{[3]}$ \\
&&&&$(\hat{q}_{\mathcal{N} \Rightarrow \boldsymbol{2}})^{[2]}$&& \\
&&&&$(\mathsf{no_W})^{[2]}$ ($(\mathsf{yes_W})^{[2]}$) && \\
&&&&&&$(\mathsf{tt_E}]^{[3]}$ ($[\mathsf{ff_E})^{[3]}$)
\end{tabular}
\end{center}
\begin{center}
\begin{tabular}{ccccccc}
$\mathcal{G}(M_{\mathcal{N} \Rightarrow \boldsymbol{2}})^{[0]}$ & $\&$ & $\mathcal{G}(M_{\mathcal{N} \Rightarrow \boldsymbol{2}})^{[1]}$ & $\&$ & $\mathcal{G}(M_{\mathcal{N} \Rightarrow \boldsymbol{2}})^{[2]}$ & $\stackrel{\mathcal{A}(\mathit{zero?})_{\hat{q}_\mathscr{E}}^\circledS}{\Rightarrow}$ & $\mathcal{G}(M_{\mathcal{N} \Rightarrow \boldsymbol{2}})^{[3]}$ \\
\cline{1-7} 
&&&&&&$(\hat{q}_{\mathcal{T}})^{[3]}$ \\
&&&&$(\hat{q}_{\mathcal{N} \Rightarrow \boldsymbol{2}})^{[2]}$&& \\
&&&&$(\mathsf{no_W})^{[2]}$ ($(\mathsf{yes_W})^{[2]}$)&& \\
&&&&&&$(\checkmark)^{[3]}$
\end{tabular}
\end{center}
We clearly have $\mathsf{st}(\mathcal{A}(\mathit{zero?})) = \mathit{zero?}$, and $\mathcal{A}(\mathit{zero?})$ is standard.
\end{example}

\begin{example}
Consider the \emph{\bfseries fixed-point strategy} $\mathit{fix}_A : ([A_{\mathscr{WW}}]_{\Lbag \boldsymbol{e'} \Rbag \hbar \Lbag \boldsymbol{e} \Rbag \hbar \boldsymbol{f}} \Rightarrow [A_{\mathscr{EW}}]_{\Lbag \boldsymbol{e'} \Rbag \hbar \boldsymbol{f}}) \Rightarrow [A_{\mathscr{E}}]_{\boldsymbol{f}}$ for each normalized game $A$ \cite{abramsky2000full,hyland2000full,mccusker1998games}.
We already described $\mathit{fix}_A$ informally in Example~\ref{ExInformalFix}; here we give a more detailed account, but again, it should suffice to just give diagrams for $\mathcal{A}(\mathit{fix}_A)_m^\circledS$ (where $m \in \mathcal{S}_{\mathit{fix}_A}$):

\begin{center}
\begin{tabular}{ccccccc}
$\mathcal{G}(M_{A \Rightarrow A \Rightarrow A})^{[0]}$ & $\&$ & $\mathcal{G}(M_{A \Rightarrow A \Rightarrow A})^{[1]}$ & $\&$ & $\mathcal{G}(M_{A \Rightarrow A \Rightarrow A})^{[2]}$ & $\stackrel{\mathcal{A}(\mathit{fix}_A)_m^\circledS}{\Rightarrow}$ & $\mathcal{G}(M_{A \Rightarrow A \Rightarrow A})^{[3]}$ \\
\cline{1-7} 
&&&&&&$(\hat{q}_{A \Rightarrow A \Rightarrow A})^{[3]}$ \\
&&&&$(\hat{q}_{A \Rightarrow A \Rightarrow A})^{[2]}$&& \\
&&&&$(\mathsf{a_{WW}})^{[2]}$ ($(\mathsf{a_{E}})^{[2]}$)&& \\
&&&&&&$(\mathsf{a_{EW}})^{[3]}$
\end{tabular}
\end{center}

\begin{center}
\begin{tabular}{ccccccc}
$\mathcal{G}(M_{A \Rightarrow A \Rightarrow A})^{[0]}$ & $\&$ & $\mathcal{G}(M_{A \Rightarrow A \Rightarrow A})^{[1]}$ & $\&$ & $\mathcal{G}(M_{A \Rightarrow A \Rightarrow A})^{[2]}$ & $\stackrel{\mathcal{A}(\mathit{fix}_A)_m^\circledS}{\Rightarrow}$ & $\mathcal{G}(M_{A \Rightarrow A \Rightarrow A})^{[3]}$ \\
\cline{1-7} 
&&&&&&$(\hat{q}_{A \Rightarrow A \Rightarrow A})^{[3]}$ \\
&&&&$(\hat{q}_{A \Rightarrow A \Rightarrow A})^{[2]}$&& \\
&&&&$(\mathsf{a_{EW}})^{[2]}$&& \\
$(\hat{q}_{A \Rightarrow A \Rightarrow A})^{[0]}$&&&&&& \\
$(\mathsf{a'_E})^{[0]}$ ($(\mathsf{a'_{WW}})^{[0]}$)&&&&&& \\
&&&&&&$(\mathsf{a_{E}})^{[3]}$ ($(\mathsf{a_{WW}})^{[3]}$)
\end{tabular}
\end{center}

\begin{center}
\begin{tabular}{ccccccc}
$\mathcal{G}(M_{A \Rightarrow A \Rightarrow A})^{[0]}$ & $\&$ & $\mathcal{G}(M_{A \Rightarrow A \Rightarrow A})^{[1]}$ & $\&$ & $\mathcal{G}(M_{A \Rightarrow A \Rightarrow A})^{[2]}$ & $\stackrel{\mathcal{A}(\mathit{fix}_A)_m^\circledS}{\Rightarrow}$ & $\mathcal{G}(M_{A \Rightarrow A \Rightarrow A})^{[3]}$ \\
\cline{1-7} 
&&&&&&$(\hat{q}_{\mathcal{T}})^{[3]}$ \\
&&&&$(\hat{q}_{A \Rightarrow A \Rightarrow A})^{[2]}$&& \\
&&&&$(\mathsf{a_E})^{[2]}$&& \\
&&&&&&$(\langle)^{[3]}$ \\
&&&&&&$(q_{\mathcal{T}})^{[3]}$ \\
&&&&&&$(\rangle)^{[3]}$ \\
&&&&&&$(q_{\mathcal{T}})^{[3]}$ \\
&&&&&&$(\sharp)^{[3]}$ \\
&&&&&&$(q_{\mathcal{T}})^{[3]}$ \\
&&&&$(q_{\mathcal{T}})^{[2]}$&& \\
&&&&$(e_1)^{[2]}$&& \\
&&&&&&$(e_1)^{[3]}$ \\
&&&&&&$(q_{\mathcal{T}})^{[3]}$ \\
&&&&$(q_{\mathcal{T}})^{[2]}$&& \\
&&&&$(e_2)^{[2]}$&& \\
&&&&&&$(e_2)^{[3]}$ \\
&&&&&$\vdots$& \\
&&&&&&$(q_{\mathcal{T}})^{[3]}$ \\
&&&&$(q_{\mathcal{T}})^{[2]}$&& \\
&&&&$(e_k)^{[2]}$&& \\
&&&&&&$(e_k)^{[3]}$ \\
&&&&&&$(q_{\mathcal{T}})^{[3]}$ \\
&&&&$(q_{\mathcal{T}})^{[2]}$&& \\
&&&&$(\checkmark)^{[2]}$&& \\
&&&&&&$(\checkmark)^{[3]}$
\end{tabular}
\end{center}

\begin{center}
\begin{tabular}{ccccccc}
$\mathcal{G}(M_{A \Rightarrow A \Rightarrow A})^{[0]}$ & $\&$ & $\mathcal{G}(M_{A \Rightarrow A \Rightarrow A})^{[1]}$ & $\&$ & $\mathcal{G}(M_{A \Rightarrow A \Rightarrow A})^{[2]}$ & $\stackrel{\mathcal{A}(\mathit{fix}_A)_m^\circledS}{\Rightarrow}$ & $\mathcal{G}(M_{A \Rightarrow A \Rightarrow A})^{[3]}$ \\
\cline{1-7} 
&&&&&&$(\hat{q}_{\mathcal{T}})^{[3]}$ \\
&&&&$(\hat{q}_{A \Rightarrow A \Rightarrow A})^{[2]}$&& \\
&&&&$(\mathsf{a_{EW}})^{[2]}$&& \\
$(\hat{q}_{A \Rightarrow A \Rightarrow A})^{[0]}$&&&&&& \\
$(\mathsf{a'_{E}})^{[0]}$&&&&&& \\
&&&&$(\hat{q}_{\mathcal{T}})^{[2]}$&& \\
&&&&$(\langle)^{[2]}$&& \\
&&&&$(q_{\mathcal{T}})^{[2]}$&& \\
&&&&$(\rangle)^{[2]}$&& \\
&&&&$(q_{\mathcal{T}})^{[2]}$&& \\
&&&&$(\sharp)^{[2]}$&& \\
&&&&$(\hat{q}_{A \Rightarrow A \Rightarrow A})^{[2]}$&& \\
&&&&$(\mathsf{a_{EW}})^{[2]}$&& \\
$(\hat{q}_{A \Rightarrow A \Rightarrow A})^{[0]}$&&&&&& \\
$(\mathsf{a'_{E}})^{[0]}$&&&&&& \\
&&&&$(q_{\mathcal{T}})^{[2]}$&& \\
&&&&$(e_1)^{[2]}$&& \\
&&&&&&$(e_1)^{[3]}$ \\
&&&&&&$(q_{\mathcal{T}})^{[3]}$ \\
&&&&$(\hat{q}_{A \Rightarrow A \Rightarrow A})^{[2]}$&& \\
&&&&$(\mathsf{a_{EW}})^{[2]}$&& \\
$(\hat{q}_{A \Rightarrow A \Rightarrow A})^{[0]}$&&&&&& \\
$(\mathsf{a'_{E}})^{[0]}$&&&&&& \\
&&&&$(q_{\mathcal{T}})^{[2]}$&& \\
&&&&$(e_2)^{[2]}$&& \\
&&&&&&$(e_2)^{[3]}$ \\
&&&&&$\vdots$& \\
&&&&&&$(q_{\mathcal{T}})^{[3]}$ \\
&&&&$(\hat{q}_{A \Rightarrow A \Rightarrow A})^{[2]}$&& \\
&&&&$(\mathsf{a_{EW}})^{[2]}$&& \\
$(\hat{q}_{A \Rightarrow A \Rightarrow A})^{[0]}$&&&&&& \\
$(\mathsf{a'_{E}})^{[0]}$&&&&&& \\
&&&&$(q_{\mathcal{T}})^{[2]}$&& \\
&&&&$(e_k)^{[2]}$&& \\
&&&&&&$(e_k)^{[3]}$ \\
&&&&&&$(q_{\mathcal{T}})^{[3]}$ \\
&&&&$(\hat{q}_{A \Rightarrow A \Rightarrow A})^{[2]}$&& \\
&&&&$(\mathsf{a_{EW}})^{[2]}$&& \\
$(\hat{q}_{A \Rightarrow A \Rightarrow A})^{[0]}$&&&&&& \\
$(\mathsf{a'_{E}})^{[0]}$&&&&&& \\
&&&&$(q_{\mathcal{T}})^{[2]}$&& \\
&&&&$(\checkmark)^{[2]}$&& \\
&&&&&&$(\checkmark)^{[3]}$
\end{tabular}
\end{center}

\begin{center}
\begin{tabular}{ccccccc}
$\mathcal{G}(M_{A \Rightarrow A \Rightarrow A})^{[0]}$ & $\&$ & $\mathcal{G}(M_{A \Rightarrow A \Rightarrow A})^{[1]}$ & $\&$ & $\mathcal{G}(M_{A \Rightarrow A \Rightarrow A})^{[2]}$ & $\stackrel{\mathcal{A}(\mathit{fix}_A)_m^\circledS}{\Rightarrow}$ & $\mathcal{G}(M_{A \Rightarrow A \Rightarrow A})^{[3]}$ \\
\cline{1-7} 
&&&&&&$(\hat{q}_{\mathcal{T}})^{[3]}$ \\
&&&&$(\hat{q}_{A \Rightarrow A \Rightarrow A})^{[2]}$&& \\
&&&&$(\mathsf{a_{EW}})^{[2]}$&& \\
$(\hat{q}_{A \Rightarrow A \Rightarrow A})^{[0]}$&&&&&& \\
$(\mathsf{a'_{WW}})^{[0]}$&&&&&& \\
&&&&$(\hat{q}_{\mathcal{T}})^{[2]}$&& \\
&&&&$(\langle)^{[2]} @0$&& \\
&&&&$(q_{\mathcal{T}})^{[2]}$&& \\
&&&&$(\langle)^{[2]} @1$&& \\
&&&&&&$(\langle)^{[3]} @0$ \\
&&&&&$\vdots$& \\
&&&&&&$(q_{\mathcal{T}})^{[3]}$ \\
&&&&$(q_{\mathcal{T}})^{[2]}$&& \\
&&&&$(\rangle)^{[2]} @1$&& \\
&&&&&&$(\rangle)^{[3]} @0$ \\
&&&&&&$(q_{\mathcal{T}})^{[3]}$ \\
&&&&$(q_{\mathcal{T}})^{[2]}$&& \\
&&&&$(\sharp)^{[2]}$&& \\
&&&&&&$(\sharp)^{[3]}$ \\
&&&&&&$(q_{\mathcal{T}})^{[3]}$ \\
&&&&$(q_{\mathcal{T}})^{[2]}$&& \\
&&&&$(\langle)^{[2]} @1$&& \\
&&&&&&$(\langle)^{[3]} @0$ \\
&&&&&$\vdots$& \\
&&&&&&$(q_{\mathcal{T}})^{[3]}$ \\
&&&&$(q_{\mathcal{T}})^{[2]}$&& \\
&&&&$(\rangle)^{[2]} @1$&& \\
&&&&&&$(\rangle)^{[3]} @0$ \\
&&&&&&$(q_{\mathcal{T}})^{[3]}$ \\
&&&&$(q_{\mathcal{T}})^{[2]}$&& \\
&&&&$(\rangle)^{[2]}@0$&& \\
&&&&$(q_{\mathcal{T}})^{[2]}$&& \\
&&&&$(\sharp)^{[2]}$&& \\
&&&&&&$(\sharp)^{[3]}$ \\
&&&&&&$(q_{\mathcal{T}})^{[3]}$ \\
&&&&$(q_{\mathcal{T}})^{[2]}$&& \\
&&&&$(e)^{[2]}$&& \\
&&&&&&$(e)^{[3]}$ \\
&&&&&$\vdots$& \\
&&&&&&$(q_{\mathcal{T}})^{[3]}$ \\
&&&&$(q_{\mathcal{T}})^{[2]}$&& \\
&&&&$(\checkmark)^{[2]}$&& \\
&&&&&&$(\checkmark)^{[3]}$ \\
\end{tabular}
\end{center}

\begin{center}
\begin{tabular}{ccccccc}
$\mathcal{G}(M_{A \Rightarrow A \Rightarrow A})^{[0]}$ & $\&$ & $\mathcal{G}(M_{A \Rightarrow A \Rightarrow A})^{[1]}$ & $\&$ & $\mathcal{G}(M_{A \Rightarrow A \Rightarrow A})^{[2]}$ & $\stackrel{\mathcal{A}(\mathit{fix}_A)_m^\circledS}{\Rightarrow}$ & $\mathcal{G}(M_{A \Rightarrow A \Rightarrow A})^{[3]}$ \\
\cline{1-7} 
&&&&&&$(\hat{q}_{\mathcal{T}})^{[3]}$ \\
&&&&$(\hat{q}_{A \Rightarrow A \Rightarrow A})^{[2]}$&& \\
&&&&$(\mathsf{a_{WW}})^{[2]}$&& \\
&&&&&&$(\langle)^{[3]} @0$ \\
&&&&&&$(q_{\mathcal{T}})^{[3]}$ \\
&&&&$(\hat{q}_{\mathcal{T}})^{[2]}$&& \\
&&&&$(\langle)^{[2]} @0$&& \\
&&&&&&$(\langle)^{[3]} @1$ \\
&&&&&$\vdots$& \\
&&&&&&$(q_{\mathcal{T}})^{[3]}$ \\
&&&&$(q_{\mathcal{T}})^{[2]}$&& \\
&&&&$(\rangle)^{[2]} @0$&& \\
&&&&&&$(\rangle)^{[3]} @1$ \\
&&&&&&$(q_{\mathcal{T}})^{[3]}$ \\
&&&&$(q_{\mathcal{T}})^{[2]}$&& \\
&&&&$(\sharp)^{[2]}$&& \\
&&&&&&$(\sharp)^{[3]}$ \\
&&&&&&$(q_{\mathcal{T}})^{[3]}$ \\
&&&&$(q_{\mathcal{T}})^{[2]}$&& \\
&&&&$(\langle)^{[2]} @0$&& \\
&&&&&&$(\langle)^{[3]} @1$ \\
&&&&&$\vdots$& \\
&&&&&&$(q_{\mathcal{T}})^{[3]}$ \\
&&&&$(q_{\mathcal{T}})^{[2]}$&& \\
&&&&$(\rangle)^{[2]} @0$&& \\
&&&&&&$(\rangle)^{[3]} @1$ \\
&&&&&&$(q_{\mathcal{T}})^{[3]}$ \\
&&&&&&$(\rangle)^{[3]} @0$ \\
&&&&&&$(q_{\mathcal{T}})^{[3]}$ \\
&&&&$(q_{\mathcal{T}})^{[2]}$&& \\
&&&&$(e_1)^{[2]}$&& \\
&&&&&&$(e_1)^{[3]}$ \\
&&&&&&$(q_{\mathcal{T}})^{[3]}$ \\
&&&&$(q_{\mathcal{T}})^{[2]}$&& \\
&&&&$(e_2)^{[2]}$&& \\
&&&&&&$(e_2)^{[3]}$ \\
&&&&&$\vdots$& \\
&&&&&&$(q_{\mathcal{T}})^{[3]}$ \\
&&&&$(q_{\mathcal{T}})^{[2]}$&& \\
&&&&$(e_k)^{[2]}$&& \\
&&&&&&$(e_k)^{[3]}$ \\
&&&&&&$(q_{\mathcal{T}})^{[3]}$ \\
&&&&$(q_{\mathcal{T}})^{[2]}$&& \\
&&&&$(\checkmark)^{[2]}$&& \\
&&&&&&$(\checkmark)^{[3]}$ \\
\end{tabular}
\end{center}
where $(\langle)_i@d$ (resp. $(\rangle)_i@d$) in the diagrams indicates that $(\langle)_i$ (resp. $(\rangle)_i$) is of depth $d$.
We have omitted the iteration of $(\hat{q}_{A \Rightarrow A \Rightarrow})^{[2]} . (\mathsf{a_{EW}})^{[2]} (\hat{q}_{A \Rightarrow A \Rightarrow})^{[0]} . (\mathsf{a'_{WW}})^{[0]}$ and $(\hat{q}_{A \Rightarrow A \Rightarrow})^{[2]} . (\mathsf{a_{WW}})^{[2]}$ respectively in the last two diagrams for the lack of space. Also, the vertical dots abbreviate `copy-cat' between $(\_)^{[3]}$- and $(\_)^{[2]}$-moves.
With m-views, there is clearly a finite table $\mathcal{A}(\mathit{fix}_A)_m$ that implements the instruction strategy $\mathcal{A}(\mathit{fix}_A)^\circledS_m$.
It is then not hard to see that $\mathsf{st}(\mathcal{A}(\mathit{fix}_A)) = \mathit{fix}_A$ holds, showing that $\mathit{fix}_A$ is viable.
Also, it is easy to see that $\mathcal{A}(\mathit{fix}_A)$ is standard. 
\end{example}

\subsection{Turing completeness}
\label{TuringCompleteness}
In the last two sections, we have seen through examples that each `atomic' strategy \emph{definable} in PCF \cite{abramsky1999game} is viable, and it is realized by a standard st-algorithm. 
In addition, Theorem~\ref{ThmPreservationOfViability} shows that constructions on strategies preserve this property.
From this fact, our main theorem immediately follows:
\begin{theorem}[Main theorem]
\label{ThmMainTheorem}
Every normalized strategy $\sigma : S_\sigma$ definable in PCF has a viable strategy $\phi_\sigma : D_\sigma$ that satisfies $\sigma = \mathcal{H}^\omega(\phi_\sigma) : \mathcal{H}^\omega(D_\sigma) \trianglelefteqslant S_\sigma$.
\end{theorem}
\begin{proof}
First, see \cite{abramsky1999game} for normalized strategies definable in the language PCF.
We enumerate these normalized strategies by the following construction of a set $\mathcal{PCF}$ (which also contains strategies not definable in PCF):
\begin{enumerate}

\item $(\sigma : S_\sigma) \in \mathcal{PCF}$ if $\sigma : S_\sigma$ is `atomic', i.e., $\mathit{der}_A : A \Rightarrow A$, $\mathit{zero}_A : A \Rightarrow \mathcal{N}$, $\mathit{succ} : \mathcal{N} \Rightarrow \mathcal{N}$, $\mathit{pred} : \mathcal{N} \Rightarrow \mathcal{N}$, $\mathit{zero}? : \mathcal{N} \Rightarrow \boldsymbol{2}$, $\mathit{case}_{\mathcal{N}} : A \Rightarrow A \Rightarrow \boldsymbol{2} \Rightarrow A$ or $\mathit{fix}_A : (A \Rightarrow A) \Rightarrow A$, where $A$ is a normalized game generated from $\mathcal{N}$ by $\&$ and $\Rightarrow$ (n.b. the construction of $A$ is `orthogonal' to that of $\sigma : S_\sigma$);

\item $(\Lambda(\sigma) : A \Rightarrow (B \Rightarrow C)) \in \mathcal{PCF}$ if $\sigma \in \mathcal{PCF}$ and $S_\sigma = A \& B \Rightarrow C$ for some normalized games $A$, $B$ and $C$;

\item $(\langle \varphi, \psi \rangle : C \Rightarrow A \& B) \in \mathcal{PCF}$ if $\varphi, \psi \in \mathcal{PCF}$, $S_\varphi = C \Rightarrow A$ and $S_\psi = C \Rightarrow B$ for some normalized games $A$, $B$ and $C$;

\item $(\iota^\dagger ; \kappa : A \Rightarrow C) \in \mathcal{PCF}$ if $\iota, \kappa \in \mathcal{PCF}$, $S_\iota = A \Rightarrow B$ and $S_\kappa = B \Rightarrow C$ for some normalized games $A$, $B$ and $C$

\end{enumerate}
where since \emph{projections} and \emph{evaluation} are derelictions up to inner tags, we count them as `atomic' ones as well.
Note that the strategy $\underline{n} : T \Rightarrow \mathcal{N}$ (i.e., the one defined below Example~\ref{ExLazyNaturalNumbers}) may be obtained by $\mathit{zero}^\dagger ; \underbrace{\mathit{succ}^\dagger ; \mathit{succ}^\dagger ; \cdots ; \mathit{succ}}_n$ for each $n \in \mathbb{N}$, which interprets numerals.

We assign a strategy $\phi_\sigma$ to each $\sigma \in \mathcal{PCF}$ along with the construction: 1. $D_\sigma \stackrel{\mathrm{df. }}{=} S_\sigma$ and $\phi_\sigma \stackrel{\mathrm{df. }}{=} \sigma$ if $\sigma$ is `atomic'; 2. $D_{\Lambda(\sigma)} \stackrel{\mathrm{df. }}{=} \Lambda(D_\sigma)$ and $\phi_{\Lambda(\sigma)} \stackrel{\mathrm{df. }}{=} \Lambda(\phi_\sigma)$; 3. $D_{\langle \varphi, \psi \rangle} \stackrel{\mathrm{df. }}{=} \langle D_\varphi, D_\psi \rangle$ and $\phi_{\langle \varphi, \psi \rangle} \stackrel{\mathrm{df. }}{=} \langle \phi_\varphi, \phi_\psi \rangle$; 4. $D_{\iota^\dagger ; \kappa} \stackrel{\mathrm{df. }}{=} D_\iota^\dagger \ddagger D_\kappa$ and $\phi_{\iota^\dagger ; \kappa}  \stackrel{\mathrm{df. }}{=} \phi_\iota^\dagger \ddagger \phi_\kappa$.

We have shown in the previous sections that `atomic' strategies are all viable, realized by standard st-algorithms.  
Moreover, the operations $\Lambda$, $\langle \_, \_ \rangle$ and $(\_)^\dagger \ddagger (\_)$ on strategies have been shown to preserve viability of strategies and standardness of the underlying st-algorithms in Theorem~\ref{ThmPreservationOfViability}.
Therefore, we may conclude that $\phi_\sigma$ is viable (and realized by a standard st-algorithm) for each $\sigma \in \mathcal{PCF}$.

It remains to show $\mathcal{H}^\omega(\phi_\sigma) = \sigma$ and $\mathcal{H}^\omega(D_\sigma) \trianglelefteqslant S_\sigma$ for all $\sigma \in \mathcal{PCF}$; we show it by induction on the construction of $\sigma$:
\begin{enumerate}

\item $\mathcal{H}^\omega(\phi_\sigma) = \mathcal{H}^\omega(\sigma) = \sigma$ and $\mathcal{H}^\omega(D_\sigma) = \mathcal{H}^\omega(S_\sigma) = S_\sigma$ if $\sigma$ is `atomic' since in this case $\sigma$ and $S_\sigma$ are both normalized;

\item $\mathcal{H}^\omega(\phi_{\Lambda(\sigma)}) = \mathcal{H}^\omega(\Lambda(\phi_\sigma)) = \Lambda(\mathcal{H}^\omega(\phi_\sigma)) = \Lambda(\sigma)$ and $\mathcal{H}^\omega(D_{\Lambda(\sigma)}) = \mathcal{H}^\omega(\Lambda(D_\sigma)) \trianglelefteqslant \Lambda (\mathcal{H}^\omega(D_\sigma)) \trianglelefteqslant \Lambda(S_\sigma)$ by the induction hypothesis, Theorem~\ref{ThmConstructionsOnGames} and Lemmata~\ref{LemHidingLemmaOnGames} and \ref{LemHidingLemmaOnStrategies};

\item $\mathcal{H}^\omega(\phi_{\langle \varphi, \psi \rangle}) =  \mathcal{H}^\omega(\langle \phi_\varphi, \phi_\psi \rangle) = \langle \mathcal{H}^\omega(\phi_\varphi), \mathcal{H}^\omega(\phi_\psi) \rangle = \langle \varphi, \psi \rangle$ and $\mathcal{H}^\omega(D_{\langle \varphi, \psi \rangle}) = \mathcal{H}^\omega(\langle D_\varphi, D_\psi \rangle) \trianglelefteqslant \langle \mathcal{H}^\omega(D_\varphi), \mathcal{H}^\omega(D_\psi) \rangle \trianglelefteqslant \langle S_\varphi, S_\psi \rangle = S_{\langle \varphi, \psi \rangle}$ by the induction hypothesis, Theorem~\ref{ThmConstructionsOnGames} and Lemmata~\ref{LemHidingLemmaOnGames} and \ref{LemHidingLemmaOnStrategies};

\item $\mathcal{H}^\omega(\phi_{\iota^\dagger ; \kappa}) = \mathcal{H}^\omega(\phi_\iota^\dagger \ddagger \phi_\kappa) = \mathcal{H}^\omega(\phi_\iota)^\dagger ; \mathcal{H}^\omega(\phi_\kappa) = \iota^\dagger ; \kappa$ and $\mathcal{H}^\omega(D_{\iota^\dagger ; \kappa}) = \mathcal{H}^\omega(D_{\iota^\dagger} \ddagger D_\kappa) = \mathcal{H}^1(\mathcal{H}^\omega(D_{\iota^\dagger}) \ddagger \mathcal{H}^\omega(D_\kappa)) \trianglelefteqslant \mathcal{H}^1(S_{\iota^\dagger} \ddagger S_\kappa) \trianglelefteqslant A \Rightarrow C = S_{\iota^\dagger ; \kappa}$ by the induction hypothesis, Theorem~\ref{ThmConstructionsOnGames} and Lemmata~\ref{LemHidingLemmaOnGames} and \ref{LemHidingLemmaOnStrategies}

\end{enumerate}
which completes the proof.
\end{proof}

Since PCF is \emph{Turing complete} \cite{gunter1992semantics,longley2015higher}, this result particularly implies:
\begin{corollary}[Turing completeness]
\label{CoroTuringCompleteness}
Every partial recursive function $f : \mathbb{N}^k \rightharpoonup \mathbb{N}$, where $k \in \mathbb{N}$, has a viable strategy $\phi_f : D_f$ such that $\mathcal{H}^\omega(D_f) \trianglelefteqslant \mathcal{N}^k \Rightarrow \mathcal{N}$ and $\mathcal{H}^\omega(\langle \underline{n_1}, \underline{n_2}, \dots, \underline{n_k} \rangle^\dagger \ddagger \phi_f) \simeq \underline{f(n_1, n_2, \dots, n_k)}$ for all $(n_1, n_2, \dots, n_k) \in \mathbb{N}^k$, where the strategy $\underline{n} : T \Rightarrow \mathcal{N}$ for each $n \in \mathbb{N}$ is the one below Example~\ref{ExLazyNaturalNumbers} up to tags.\footnote{Recall that $\langle \_, \_ \rangle$, $(\_)^\dagger$ and $\ddagger$ are \emph{pairing}, \emph{promotion} and \emph{concatenation} of strategies defined in Section~\ref{ConstructionsOnStrategies}.}
\end{corollary}
\begin{proof}
Let $\mathsf{x_1 : N, x_2 : N, \dots, x_k : N \vdash F : N}$ be a term of PCF that implements a given partial recursive function $f : \mathbb{N}^k \rightharpoonup \mathbb{N}$, i.e., $\mathsf{F [n_1/x_1, n_2/x_2, \dots, n_k/x_k]}$ evaluates to $\mathsf{f(n_1, n_2, \dots, n_k)}$ if $f(n_1, n_2, \dots, n_k)$ is defined and diverges otherwise, for all $n_i \in \mathbb{N}$ ($i = 1, 2, \dots, k$), where $\mathsf{n : N}$ is the $n^{\text{th}}$-numeral, and $\mathsf{F [n_1/x_1, n_2/x_2, \dots, n_k/x_k]}$ is the result of substituting $\mathsf{n_i}$ for $\mathsf{x_i}$ in $\mathsf{F}$ for $i = 1, 2, \dots, k$ (see, e.g., \cite{gunter1992semantics,longley2015higher} for how to construct $\mathsf{F}$ from $f$).
Then, there exists a normalized strategy $\sigma_f : \mathcal{N}^k \Rightarrow \mathcal{N}$ in $\mathcal{PCF}$ that interprets $\mathsf{F}$ in the game semantics of PCF in \cite{abramsky1999game}.
By Theorem~\ref{ThmMainTheorem}, there exists a viable strategy $\phi_f : D_f$ such that $\mathcal{H}^\omega(\phi_f) = \sigma_f$ and $\mathcal{H}^\omega(D_f) \trianglelefteqslant \mathcal{N}^k \Rightarrow \mathcal{N}$.
Hence, $\mathcal{H}^\omega(\langle \underline{n_1}, \underline{n_2}, \dots, \underline{n_k} \rangle^\dagger \ddagger \phi_f) = \mathcal{H}^\omega(\langle \underline{n_1}, \underline{n_2}, \dots, \underline{n_k} \rangle^\dagger) ; \mathcal{H}^\omega(\phi_f) = \langle \underline{n_1}, \underline{n_2}, \dots, \underline{n_k} \rangle^\dagger ; \sigma_f \simeq \underline{f(n_1, n_2, \dots, n_k)}$.
\end{proof}

\begin{remark}
Crucially, there is clearly a partial recursive function $f : \mathbb{N}^k \rightharpoonup \mathbb{N}$ such that $\sigma_f$ is \emph{not} viable (but $\phi_f$ is viable) by the finitary nature of tables for st-algorithms.
\end{remark}

As our game-semantic model of computation is Turing complete, some of the well-known theorems in computability theory \cite{cutland1980computability,rogers1967theory} are immediately generalized (in the sense that they are not restricted to computation on natural numbers):
\begin{corollary}[Generalized smn-theorem]
\label{CoroGeneralizedSMN}
If strategies $\sigma_i : T \Rightarrow A_i$, $i = 1, 2, \dots, n$, and $\phi : D$ with $\mathcal{H}^\omega(D) \trianglelefteqslant A_1 \& A_2 \& \dots \& A_n \& B_1 \& B_2 \& \dots \& B_m \Rightarrow C$ are realized by standard st-algorithms, then we may compute a standard st-algorithm that realizes a viable strategy $\phi_{\sigma_1, \sigma_2, \dots, \sigma_n} : D_{A_1, A_2, \dots, A_n}$ such that:

\begin{enumerate}

\item $\mathcal{H}^\omega(\phi_{\sigma_1, \sigma_2, \dots, \sigma_n}) : \mathcal{H}^\omega(D_{A_1, A_2, \dots, A_n}) \trianglelefteqslant T \& B_1 \& B_2 \& \dots \& B_m \Rightarrow C$; 

\item $\langle \{ \boldsymbol{\epsilon} \}, \tau_1, \tau_2, \dots, \tau_m \rangle^\dagger \ddagger \phi_{\sigma_1, \sigma_2, \dots, \sigma_n} \simeq \langle \sigma_1, \sigma_2, \dots, \sigma_n, \tau_1, \tau_2, \dots, \tau_m \rangle^\dagger \ddagger \phi$ for any strategies $\tau_j : T \Rightarrow B_j$, $j = 1, 2, \dots, m$.

\end{enumerate}
\end{corollary}

\begin{proof}
We define $\phi_{\sigma_1, \sigma_2, \dots, \sigma_n} \stackrel{\mathrm{df. }}{=} \underbrace{\Lambda^\circleddash(\cdots \Lambda^\circleddash}_m(\langle \sigma_1, \sigma_2, \dots, \sigma_n \rangle^\dagger \ddagger \underbrace{\Lambda( \cdots \Lambda}_m(\phi) \cdots)) \cdots )$, where the proof of Theorem~\ref{ThmPreservationOfViability} describes how to `effectively' obtain a standard st-algorithm that realizes $\phi_{\sigma_1, \sigma_2, \dots, \sigma_n}$ in an informal sense\footnote{It is interesting future work to formalize this informal `effective computability' by certain viable strategies.}. Note that Corollary~\ref{CoroTuringCompleteness} implies that it is in fact a generalization of the conventional smn-theorem \cite{cutland1980computability}.
\end{proof}

\begin{corollary}[Generalized FRT]
\label{CoroGeneralizedFRT}
Given a viable strategy $\varphi : D$ such that $\mathcal{H}^\omega(D) \trianglelefteqslant T \Rightarrow (A \Rightarrow A)$ realized by a standard st-algorithm, there exists another viable strategy $\sigma_\varphi : D_\varphi$ with $\mathcal{H}^\omega(D_\varphi) \trianglelefteqslant T \Rightarrow A$ realized by a standard st-algorithm such that $\mathcal{H}^\omega (\sigma_\varphi^\dagger \ddagger \varphi) = \mathcal{H}^\omega(\sigma_\varphi) : T \Rightarrow A$.
\end{corollary}
\begin{proof}
Just define $\sigma_\varphi \stackrel{\mathrm{df. }}{=} \varphi^\dagger \ddagger \mathit{fix}_A$.
\end{proof}

Finally, let us show that the converse of the main theorem also holds for classical computation because it would then give further naturality and/or reasonability of our definition of `effective computability': 
\begin{theorem}[Conservativeness]
\label{ThmConservativeness}
Any viable strategy $\sigma : G$ with $\mathcal{H}^\omega(G) \trianglelefteqslant \mathcal{N}^k \Rightarrow \mathcal{N}$, where $k \in \mathbb{N}$, can be simulated by a partial recursive function $f_\sigma : \mathbb{N} \rightharpoonup \mathbb{N}$ in the sense that $\mathcal{H}^\omega(\langle \underline{n_1}, \underline{n_2}, \dots, \underline{n_k} \rangle^\dagger \ddagger \sigma) \simeq \underline{f_\sigma(n_1, n_2, \dots, n_k)}$ for all $n_1, n_2, \dots, n_k \in \mathbb{N}$.
\end{theorem}
\begin{proof}[Proof sketch]
This theorem is not as surprising as Theorem~\ref{ThmMainTheorem} for one may just employ \emph{Church's thesis} \cite{cutland1980computability}.
Nevertheless, let us give a proof sketch for the theorem, by which an ardent reader can construct a detailed proof if she or he wishes to. 

The idea is to simulate a given viable strategy $\sigma : G$ by a 5-tape TM $\mathscr{M}$ by writing down an input on the first tape, the entire history of previous occurrences of moves, i.e., each position during a play, on the second tape, and the last three moves in each P-view as well as the next P-move on the third tape, where the fourth and fifth tapes are used for auxiliary computations (specified below).

Let us first specify the `format' of the first and second tapes.
On each of these tapes, moves are separated by an occurrence of a distinguished symbol $\mathsf{\$}$, and each move $[m]_{e_1 e_2 \dots e_k}$ together with the \emph{identifier} $j \in \mathbb{N}$ of its justifier defined below is written as the sequence
\begin{center}
\begin{tikzpicture}[every node/.style={block},
        block/.style={minimum height=1.5em,outer sep=0pt,draw,rectangle,node distance=0pt}]
   \node (A) {$\mathsf{m}$};
   \node (B) [left=of A] {$\prime$};
   \node (C) [left=of B] {$\ldots$};
   \node (G) [left=of C] {$\prime$};
   \node (H) [left=of G] {$\prime$};
   \node (K) [left=of H] {$\$$};
   \node (D) [right=of A] {$\mathsf{e_1}$};
   \node (E) [right=of D] {$\mathsf{e_2}$};
   \node (L) [right=of E] {$\ldots$};
   \node (M) [right=of L] {$\mathsf{e_k}$};
   \node (N) [right=of M] {$\$$};
   \draw (K.north west) -- ++(-1cm,0) (K.south west) -- ++ (-1cm,0) 
                 (N.north east) -- ++(1cm,0) (N.south east) -- ++ (1cm,0);
\end{tikzpicture}
\end{center}
where there are $j$-many ($\prime$)'s (which represents $j$) between the left occurrence of $\mathsf{\$}$ and the occurrence of $\mathsf{m}$.
We define the \emph{\bfseries identifier} of each occurrence of a move $[n]_{f_1 f_2 \dots f_l}$ on the tape to be the number $i \in \mathbb{N}$ such that $\mathsf{n}$ is written on the $i^{\textit{th}}$-cell of the tape, where note that $i \geqslant 1$.
Then, the number $j$ of ($\prime$)'s displayed above is defined to be the identifier of the justifier of $[m]_{e_1 e_2 \dots e_k}$ if it exists, and $0$ if $[m]_{e_1 e_2 \dots e_k}$ is initial.
In this manner, we encode pointers on the tape.\footnote{Note that pointers in instruction games are rather trivial, and thus we omit them.} 
We also assume without loss of generality that the symbol $\mathsf{m}$ contains information for the I/E-parity of each move $[m]_{e_1 e_2 \dots e_k}$; an obvious (though not canonical) way to achieve it is to use two different fonts for $\mathsf{m}$.

On the other hand, the last three moves in the P-view of the current odd-length position (in the `format' described above but without identifiers), their identifiers and m-views are written respectively on the third, fourth and fifth tapes, where each occurrence of a move, an identifier or an m-view is separated again by $\mathsf{\$}$.

Next, note that computation of the next move is trivial if it is an O-move because external O-moves in the output $\mathcal{N}$ are all obvious questions, external O-moves in the input $\mathcal{N}$ are already given as an input on the first tape, and internal O-moves are just `dummies' of internal P-moves by the axiom DP2 (Definition~\ref{DefGames}).
Note also that $\mathscr{M}$ may recognize the O/P-parity of the next move by its state (and the I/E-parity by the symbolic information on the tape assumed above). 
Hence, it suffices to focus on computation of the next P-move; $\mathscr{M}$ computes it as follows:
\begin{enumerate}

\item Copy the last occurrence $[m_1]_{\boldsymbol{e^{(1)}}}$ of the current P-view on the second tape onto the initial cells of the third tape, compute its identifier and m-view in the obvious manner, and write them on the fourth and fifth tapes, respectively;

\item Locate the second last occurrence $[m_2]_{\boldsymbol{e^{(2)}}}$ of the current P-view on the second tape by the identifier associated to the occurrence $[m_1]_{\boldsymbol{e^{(1)}}}$, and then execute the same computation as the one on $[m_1]_{\boldsymbol{e^{(1)}}}$, where the new content prefixed with $\mathsf{\$}$ on each tape is just concatenated to the existing one;

\item Similarly, locate the third last occurrence $[m_3]_{\boldsymbol{e^{(3)}}}$ of the current P-view on the second tape (which is easy as it locates next to $[m_2]_{\boldsymbol{e^{(2)}}}$), and execute the same computation on it (so that the third, fourth and fifth tapes contain all information of the last three moves in the P-view);

\item With the current contents on the third, fourth and fifth tapes, compute the next P-move $[m]_{\boldsymbol{e}}$ and the identifier of its justifier, write them on the second tape, and erase all contents on the third, fourth and fifth tapes.

\end{enumerate}

Note that $\mathscr{M}$ is clearly able to execute the last step $([m_3]_{\boldsymbol{e^{(3)}}}, [m_2]_{\boldsymbol{e^{(2)}}}, [m_1]_{\boldsymbol{e^{(1)}}}) \mapsto [m]_{\boldsymbol{e}}$ by basically simulating the computation of an instruction strategy for $\sigma$, completing the proof.
\end{proof}

\begin{remark}
Theorem~\ref{ThmConservativeness} does \emph{not} hold for higher-order computation because TMs cannot take additional inputs from O during the course of computation. 
Of course, one may consider TMs that interact with O, like the model of computation employed in computability logic, but it is no longer TMs in the usual sense.
\end{remark}

As an immediate corollary, we have:
\begin{corollary}[Universality]
\label{CoroUniversality}
Let $\mathsf{A}$ be a type of PCF, and $A$ be the game that interprets $\mathsf{A}$ (as defined in \cite{yamada2016dynamic}).
Then, any viable strategy $\alpha$ on $A$ is the denotation $a$ of a term $\mathsf{a}$ of PCF (as defined in \cite{yamada2016dynamic}) up to internal moves, i.e., $\mathcal{H}^\omega(\alpha) = \mathcal{H}^\omega(a)$.
\end{corollary}
\begin{proof}
First, recall that a strategy $\sigma : G$ is \emph{total} if $\boldsymbol{s} \in \sigma \wedge \boldsymbol{s} . o \in P_G^{\mathsf{Odd}} \Rightarrow \exists p \in M_G . \ \! \boldsymbol{s} . o . p \in \sigma$ \cite{abramsky1997semantics}.
Next, let us identify the total strategy $\underline{n} : \mathcal{N}$ with $n : N$ for each $n \in \mathbb{N}$, and any non-total strategy on $\mathcal{N}$ with $\bot \stackrel{\mathrm{df. }}{=} \{ \boldsymbol{\epsilon} \} : N$; in this way, the normalized one $\mathcal{H}^\omega(\alpha) : \mathcal{H}^\omega(A)$ for each \emph{viable} strategy $\alpha : A$ such that $A$ is the dynamic game semantics of a type $\mathsf{A}$ of PCF \cite{yamada2016dynamic} may be regarded as a static strategy in the conventional game semantics of PCF \cite{mccusker1998games}.
Then, by the proof of Theorem~\ref{ThmConservativeness}, $\mathcal{H}^\omega(\alpha) : \mathcal{H}^\omega(A)$ is \emph{recursive}, i.e., computable by a TM, and thus it is the conventional game semantics of a term $\mathsf{a}$ of PCF by the univarsality theorem of the conventional game semantics of PCF \cite{abramsky2000full,hyland2000full}.
Therefore, we may conclude that $\alpha$ coincides with the dynamic game semantics $a$ of $\mathsf{a}$ up to internal moves, i.e., $\mathcal{H}^\omega(\alpha) = \mathcal{H}^\omega(a)$, completing the proof. 
\end{proof}

\begin{remark}
Our universality theorem (Corollary~\ref{CoroUniversality}) holds only up to internal moves because (dynamic) strategies may compute internal moves in such a way that does not correspond to computation of PCF.
Certainly, it would be interesting as future work to refine this result so that it holds \emph{on the nose}.
\end{remark}

\section{Conclusion and future work}
\label{ConclusionAndFutureWork}
We have given a novel notion of `effective computability' in game semantics, namely \emph{viable} strategies.
Due to its  \emph{intrinsic}, \emph{non-inductive} and \emph{non-axiomatic} natures, it can be seen as a fundamental investigation of `effective' computation beyond classical one, where note that viability of strategies makes sense \emph{universally}, i.e., regardless of the underlying games (e.g., games do not have to correspond to types of PCF).
Furthermore, our game-semantic model of computation formulates both `high-level' and `low-level' computational processes, and defines  `computability' of the former in terms of the latter, which sheds new light on the very notion of computation.
For instance, strategies $\underline{n} : \mathcal{N}$ may be seen as the \emph{definition} of natural numbers, and thus a viable strategy of the form $\phi : \mathcal{N}^k \Rightarrow \mathcal{N}$ can be regarded as `high-level' computation on natural numbers, not on their representations, and (the table of) an st-algorithm that realizes $\phi$ can be seen as its symbolic implementation.

There are various directions for further work. First, we need to analyze the exact computational power of viable strategies, in comparison with other known notion of higher-order computability \cite{longley2015higher}.
Also, as an application, the present framework may give an accurate measure for computational complexity \cite{kozen2006theory}, where note that the work on dynamic games and strategies \cite{yamada2016dynamic} has already given such a measure via internal moves, but the present work may refine it further since two single steps in a game $G$ may take different numbers of steps in the instruction game $\mathcal{G}(M_G)^3 \Rightarrow \mathcal{G}(M_G)$. 
Moreover, it is of theoretical interest to see which theorems in computability theory can be generalized by the present framework in addition to the smn- and the first recursion theorems.
However, the most imminent future work is perhaps, by exploiting the flexibility of game semantics, to enlarge the scope of the present work (i.e., not only the language PCF) in order to establish a computational model of various (constructive) logics and programming languages.
We are particularly interested in how to apply our approach to \emph{non-innocent} strategies.

Finally, let us propose two open questions.
Since the definition of viable strategies is somewhat reflexive (as it is via strategies), we may naturally consider strategies \emph{that can be realized by a viable strategy}.
Let us define such strategies to be \emph{2-viable}. 
More generally, rephrasing viability as \emph{1-viability}, we define a strategy to be \emph{$(n+1)$-viable} if it can be realized by an $n$-viable strategy for each $n \in \mathbb{N}$.
Clearly, any $n$-viable strategy is `effective' in an informal sense.
Then, the first questions is:
\begin{quote}
Is the class of all $(n+1)$-viable strategies strictly larger than that of all $n$-viable strategies for each $n \in \mathbb{N}$?
\end{quote}
This question seems highly interesting from a theoretical perspective.\footnote{Note that this question would not have arised if we had not defined `effective computability' \emph{solely in terms of games and strategies}.}
If the answer is positive, then there would be an infinite hierarchy of generalized viable strategies.
It is then natural to ask the following second question:
\begin{quote}
Does the hierarchy, if it exists, correspond to any known hierarchy (perhaps in computability theory or proof theory)?
\end{quote}
We shall aim to answer these questions as future work as well.



\section*{Acknowledgements}
The author acknowleges support from Funai Overseas Scholarship, and also he is grateful to Samson Abramsky and Robin Piedeleu for fruitful discussions.

\bibliographystyle{alpha} 
\bibliography{CategoricalLogic,GamesAndStrategies,Recursion,PCF,TypeTheoriesAndProgrammingLanguages,HoTT,GoI,LinearLogic}      

\newcommand{\etalchar}[1]{$^{#1}$}
\begin{thebibliography}{TVD14}

\bibitem[A{\etalchar{+}}97]{abramsky1997semantics}
Samson Abramsky et~al.
\newblock Semantics of {Interaction: An Introduction to Game Semantics}.
\newblock {\em Semantics and Logics of Computation, Publications of the Newton
  Institute}, pages 1--31, 1997.

\bibitem[Abr14]{abramsky2014intensionality}
Samson Abramsky.
\newblock Intensionality, {Definability and Computation}.
\newblock In {\em Johan van Benthem on Logic and Information Dynamics}, pages
  121--142. Springer, 2014.

\bibitem[AC98]{amadio1998domains}
Roberto~M Amadio and Pierre-Louis Curien.
\newblock {\em Domains and {L}ambda-calculi}.
\newblock Number~46. Cambridge University Press, 1998.

\bibitem[AJ94]{abramsky1994games}
Samson Abramsky and Radha Jagadeesan.
\newblock Games and {Full Completeness for Multiplicative Linear Logic}.
\newblock {\em The Journal of Symbolic Logic}, 59(02):543--574, 1994.

\bibitem[AJM00]{abramsky2000full}
Samson Abramsky, Radha Jagadeesan, and Pasquale Malacaria.
\newblock Full {A}bstraction for {PCF}.
\newblock {\em Information and Computation}, 163(2):409--470, 2000.

\bibitem[AJV15]{abramsky2015games}
Samson Abramsky, Radha Jagadeesan, and Matthijs V{\'a}k{\'a}r.
\newblock Games for {Dependent Types}.
\newblock In {\em Automata, Languages, and Programming}, pages 31--43.
  Springer, 2015.

\bibitem[AM99a]{abramsky1999game}
Samson Abramsky and Guy McCusker.
\newblock Game {S}emantics.
\newblock In {\em Computational logic}, pages 1--55. Springer, 1999.

\bibitem[AM99b]{abramsky1999concurrent}
Samson Abramsky and P-A Mellies.
\newblock Concurrent games and full completeness.
\newblock In {\em Logic in Computer Science, 1999. Proceedings. 14th Symposium
  on}, pages 431--442. IEEE, 1999.

\bibitem[B{\etalchar{+}}84]{barendregt1984lambda}
Hendrik~Pieter Barendregt et~al.
\newblock {\em The {L}ambda {C}alculus}, volume~3.
\newblock North-Holland Amsterdam, 1984.

\bibitem[BC82]{berry1982sequential}
G{\'e}rard Berry and Pierre-Louis Curien.
\newblock Sequential algorithms on concrete data structures.
\newblock {\em Theoretical Computer Science}, 20(3):265--321, 1982.

\bibitem[Bla92]{blass1992game}
Andreas Blass.
\newblock A game semantics for linear logic.
\newblock {\em Annals of Pure and Applied logic}, 56(1):183--220, 1992.

\bibitem[Blo17]{blot2017realizability}
Valentin Blot.
\newblock Realizability for peano arithmetic with winning conditions in hon
  games.
\newblock {\em Annals of Pure and Applied Logic}, 168(2):254--277, 2017.

\bibitem[BO08]{blum2008concrete}
William Blum and CHL Ong.
\newblock {A Concrete Presentation of Game Semantics}, 2008.

\bibitem[Buc94]{bucciarelli1994another}
Antonio Bucciarelli.
\newblock Another approach to sequentiality: Kleene's unimonotone functions.
\newblock In {\em Mathematical Foundations of Programming Semantics}, pages
  333--358. Springer, 1994.

\bibitem[Cur07]{curien2007definability}
Pierre-Louis Curien.
\newblock Definability and full abstraction.
\newblock {\em Electronic Notes in Theoretical Computer Science}, 172:301--310,
  2007.

\bibitem[Cut80]{cutland1980computability}
Nigel Cutland.
\newblock {\em Computability: An {Introduction to Recursive Function Theory}}.
\newblock Cambridge university press, 1980.

\bibitem[DGL05]{dimovski2005data}
Aleksandar Dimovski, Dan~R Ghica, and Ranko Lazi{\'c}.
\newblock {Data-abstraction Refinement: A Game Semantic Approach}.
\newblock In {\em International Static Analysis Symposium}, pages 102--117.
  Springer, 2005.

\bibitem[DLL08]{dal2008quantitative}
Ugo Dal~Lago and Olivier Laurent.
\newblock Quantitative game semantics for linear logic.
\newblock In {\em International Workshop on Computer Science Logic}, pages
  230--245. Springer, 2008.

\bibitem[Fel86]{felscher1986dialogues}
Walter Felscher.
\newblock Dialogues as a foundation for intuitionistic logic.
\newblock In {\em Handbook of philosophical logic}, pages 341--372. Springer,
  1986.

\bibitem[F{\'e}r17]{feree2017game}
Hugo F{\'e}r{\'e}e.
\newblock Game semantics approach to higher-order complexity.
\newblock {\em Journal of Computer and System Sciences}, 87:1--15, 2017.

\bibitem[Gan80]{gandy1980church}
Robin Gandy.
\newblock Church's thesis and principles for mechanisms.
\newblock {\em Studies in Logic and the Foundations of Mathematics},
  101:123--148, 1980.

\bibitem[Gan93]{gandy1993dialogues}
RO~Gandy.
\newblock Dialogues, blass games and sequentiality for objects of finite type.
\newblock {\em Unpublished manuscript}, 1993.

\bibitem[Ghi05]{ghica2005slot}
Dan~R Ghica.
\newblock Slot games: a quantitative model of computation.
\newblock In {\em ACM SIGPLAN Notices}, volume~40, pages 85--97. ACM, 2005.

\bibitem[Gir87]{girard1987linear}
Jean-Yves Girard.
\newblock Linear logic.
\newblock {\em Theoretical computer science}, 50(1):1--101, 1987.

\bibitem[Gir89]{girard1989geometry}
Jean-Yves Girard.
\newblock Geometry of {I}nteraction {I}: {I}nterpretation of {S}ystem {F}.
\newblock {\em Studies in Logic and the Foundations of Mathematics},
  127:221--260, 1989.

\bibitem[Gir90]{girard1990geometry}
Jean-Yves Girard.
\newblock Geometry of interaction {II}: Deadlock-free algorithms.
\newblock In {\em COLOG-88}, pages 76--93. Springer, 1990.

\bibitem[Gir95]{girard1995geometry}
Jean-Yves Girard.
\newblock Geometry of interaction {III}: accommodating the additives.
\newblock {\em London Mathematical Society Lecture Note Series}, pages
  329--389, 1995.

\bibitem[Gir03]{girard2003geometry}
Jean-Yves Girard.
\newblock Geometry of interaction {IV}: the feedback equation.
\newblock In {\em Logic Colloquium}, volume~3, pages 76--117. Citeseer, 2003.

\bibitem[Gir11]{girard2011geometry}
Jean-Yves Girard.
\newblock Geometry of interaction {V}: logic in the hyperfinite factor.
\newblock {\em Theoretical Computer Science}, 412(20):1860--1883, 2011.

\bibitem[Gir13]{girard2013geometry}
Jean-Yves Girard.
\newblock Geometry of interaction {VI}: a blueprint for transcendental syntax.
\newblock {\em preprint}, 2013.

\bibitem[Gre05]{greenland2005game}
William~Edward Greenland.
\newblock {\em Game {Semantics for Region Analysis}}.
\newblock PhD thesis, University of Oxford, 2005.

\bibitem[Gun92]{gunter1992semantics}
Carl~A Gunter.
\newblock {\em Semantics of {P}rogramming {L}anguages: {S}tructures and
  {T}echniques}.
\newblock MIT press, 1992.

\bibitem[Gur04]{gurevich2004abstract}
Yuri Gurevich.
\newblock Abstract state machines: An overview of the project.
\newblock {\em Foundations of Information and Knowledge Systems}, pages 6--13,
  2004.

\bibitem[HO93]{hyland1993fair}
J~Martin~E Hyland and C-H~Luke Ong.
\newblock Fair games and full completeness for multiplicative linear logic
  without the mix-rule.
\newblock {\em preprint}, 190, 1993.

\bibitem[HO00]{hyland2000full}
J~Martin~E Hyland and C-HL Ong.
\newblock On {F}ull {A}bstraction for {PCF}: {I}, {II}, and {III}.
\newblock {\em Information and computation}, 163(2):285--408, 2000.

\bibitem[Hoa78]{hoare1978communicating}
Charles Antony~Richard Hoare.
\newblock Communicating sequential processes.
\newblock In {\em The origin of concurrent programming}, pages 413--443.
  Springer, 1978.

\bibitem[Hyl97]{hyland1997game}
Martin Hyland.
\newblock Game {S}emantics.
\newblock {\em Semantics and logics of computation}, 14:131, 1997.

\bibitem[Jap03]{japaridze2003introduction}
Giorgi Japaridze.
\newblock Introduction to computability logic.
\newblock {\em Annals of Pure and Applied Logic}, 123(1-3):1--99, 2003.

\bibitem[Kle59]{kleene1959recursive}
Stephen~Cole Kleene.
\newblock Recursive functionals and quantifiers of finite types i.
\newblock {\em Transactions of the American Mathematical Society}, 91(1):1--52,
  1959.

\bibitem[Kle63]{kleene1963recursive}
Stephen~Cole Kleene.
\newblock Recursive functionals and quantifiers of finite types ii.
\newblock {\em Transactions of the American Mathematical Society},
  108(1):106--142, 1963.

\bibitem[Kle78]{kleene1978recursive}
Stephen~Cole Kleene.
\newblock Recursive functionals and quantifiers of finite types revisited {I}.
\newblock {\em Studies in Logic and the Foundations of Mathematics},
  94:185--222, 1978.

\bibitem[Kle80]{kleene1980recursive}
Stephen~Cole Kleene.
\newblock Recursive functionals and quantifiers of finite types revisited {II}.
\newblock {\em Studies in Logic and the Foundations of Mathematics}, 101:1--29,
  1980.

\bibitem[Kle82]{kleene1982recursive}
Stephen~Cole Kleene.
\newblock Recursive functionals and quantifiers of finite types revisited
  {III}.
\newblock {\em Studies in Logic and the Foundations of Mathematics}, 109:1--40,
  1982.

\bibitem[Kle85]{kleene1985unimonotone}
SC~Kleene.
\newblock Unimonotone functions of finite types (recursive functionals and
  quantifiers of finite types revisited {IV}).
\newblock {\em Recursion Theory}, 42:119--138, 1985.

\bibitem[Kle91]{kleene1991recursive}
Stephen~Cole Kleene.
\newblock Recursive functionals and quantifiers of finite types revisited. {V}.
\newblock {\em Transactions of the American Mathematical Society},
  325(2):593--630, 1991.

\bibitem[Koz06]{kozen2006theory}
Dexter~C Kozen.
\newblock {\em Theory of {C}omputation}.
\newblock Springer Science \& Business Media, 2006.

\bibitem[Koz12]{kozen2012automata}
Dexter~C Kozen.
\newblock {\em {A}utomata and {C}omputability}.
\newblock Springer Science \& Business Media, 2012.

\bibitem[Lau04]{laurent2004polarized}
Olivier Laurent.
\newblock Polarized games.
\newblock {\em Annals of Pure and Applied Logic}, 130(1-3):79--123, 2004.

\bibitem[LN15]{longley2015higher}
John Longley and Dag Normann.
\newblock {\em Higher-{O}rder {C}omputability}.
\newblock Springer, 2015.

\bibitem[McC98]{mccusker1998games}
Guy McCusker.
\newblock {\em Games and {Full Abstraction for a Functional Metalanguage with
  Recursive Types}}.
\newblock Springer Science \& Business Media, 1998.

\bibitem[Mil]{milner2005software}
Robin Milner.
\newblock Software science: From virtual to reality.
\newblock {\em Bulletin of EATCS}, 87.

\bibitem[Mil80]{milner1980calculus}
Robin Milner.
\newblock A calculus of communicating systems.
\newblock {\em Lecture Notes in Comput. Sci. 92}, 1980.

\bibitem[Mit96]{mitchell1996foundations}
John~C Mitchell.
\newblock {\em Foundations for {P}rogramming {L}anguages}, volume~1.
\newblock MIT press Cambridge, 1996.

\bibitem[Mos98]{moschovakis1998founding}
Yiannis~N Moschovakis.
\newblock On founding the theory of algorithms.
\newblock {\em Truth in mathematics}, pages 71--104, 1998.

\bibitem[Nic94]{nickau1994hereditarily}
Hanno Nickau.
\newblock Hereditarily {S}equential {F}unctionals.
\newblock In {\em Logical Foundations of Computer Science}, pages 253--264.
  Springer, 1994.

\bibitem[Ong06]{ong2006model}
C-HL Ong.
\newblock On {Model-checking Trees Generated by Higher-order Recursion
  Schemes}.
\newblock In {\em 21st Annual IEEE Symposium on Logic in Computer Science
  (LICS'06)}, pages 81--90. IEEE, 2006.

\bibitem[Oua97]{ouaknine1997two}
Jo{\"e}l Ouaknine.
\newblock {\em A {T}wo-{D}imensional {E}xtension of {L}ambekps {C}ategorical
  {P}roof {T}heory}.
\newblock PhD thesis, McGill University, Montr{\'e}al, 1997.

\bibitem[Plo77]{plotkin1977lcf}
Gordon~D. Plotkin.
\newblock Lcf considered as a {P}rogramming {L}anguage.
\newblock {\em Theoretical computer science}, 5(3):223--255, 1977.

\bibitem[RLL80]{requena1980dialogische}
Est{\'e}ban Requena, Paul LORENZEN, and Kuno LORENZ.
\newblock Dialogische logik, 1980.

\bibitem[RR67]{rogers1967theory}
Hartley Rogers and H~Rogers.
\newblock {\em Theory of {Recursive Functions and Effective Computability}},
  volume~5.
\newblock McGraw-Hill New York, 1967.

\bibitem[Sho67]{shoenfield1967mathematical}
Joseph~R Shoenfield.
\newblock {\em Mathematical {L}ogic}, volume~21.
\newblock Addison-Wesley Reading, 1967.

\bibitem[Sip12]{sipser2012introduction}
Michael Sipser.
\newblock {\em Introduction to the Theory of Computation}.
\newblock Cengage Learning, 2012.

\bibitem[SS71]{scott1971toward}
Dana~S Scott and Christopher Strachey.
\newblock {\em Toward a mathematical semantics for computer languages},
  volume~1.
\newblock Oxford University Computing Laboratory, Programming Research Group,
  1971.

\bibitem[T{\etalchar{+}}98]{troelstra1998realizability}
Anne~Sjerp Troelstra et~al.
\newblock Realizability.
\newblock 1998.

\bibitem[Tur36]{turing1936computable}
Alan~Mathison Turing.
\newblock On {Computable Numbers, with an Application to the
  Entscheidungsproblem}.
\newblock {\em J. of Math}, 58(345-363):5, 1936.

\bibitem[TVD14]{troelstra2014constructivism}
Anne~Sjerp Troelstra and Dirk Van~Dalen.
\newblock {\em Constructivism in {M}athematics}, volume~2.
\newblock Elsevier, 2014.

\bibitem[VO08]{van2008realizability}
Jaap Van~Oosten.
\newblock {\em Realizability: an introduction to its categorical side}, volume
  152.
\newblock Elsevier, 2008.

\bibitem[Wei12]{weihrauch2012computable}
Klaus Weihrauch.
\newblock {\em Computable analysis: an introduction}.
\newblock Springer Science \& Business Media, 2012.

\bibitem[Win93]{winskel1993formal}
Glynn Winskel.
\newblock {\em The {F}ormal {S}emantics of {P}rogramming {L}anguages: {A}n
  {I}ntroduction}.
\newblock MIT press, 1993.

\bibitem[YA16]{yamada2016dynamic}
Norihiro Yamada and Samson Abramsky.
\newblock {Dynamic Games and Strategies}.
\newblock {\em arXiv preprint arXiv:1601.04147}, 2016.

\bibitem[Yam16]{yamada2016game}
Norihiro Yamada.
\newblock {Game Semantics for Martin-L\"{o}f Type Theory}.
\newblock {\em arXiv preprint arXiv:1610.01669}, 2016.

\end{thebibliography}







\appendix
\section{A Finite Table for Successor Strategy}
\label{AppSucc}
The finite table for an st-algorithm $\mathcal{A}(\mathit{succ})$ for the successor strategy $\mathit{succ} : \mathcal{N} \Rightarrow \mathcal{N}$ is as follows: 
\begin{align*}
\mathcal{A}(\mathit{succ})_{\hat{q}_{\mathscr{E}}} : \ &(\hat{q}_{\mathcal{N} \Rightarrow \mathcal{N}})^{[3]} \mapsto (\hat{q}_{\mathcal{N} \Rightarrow \mathcal{N}})^{[2]} \mid (\hat{q}_{\mathcal{N} \Rightarrow \mathcal{N}})^{[3]} (\hat{q}_{\mathcal{N} \Rightarrow \mathcal{N}})^{[2]} (\mathsf{\hat{q}_E})^{[2]} \mapsto (\mathsf{\hat{q}_W})^{[3]} \mid \\ 
&(\hat{q}_{\mathcal{T}})^{[3]} \mapsto (\hat{q}_{\mathcal{N} \Rightarrow \mathcal{N}})^{[2]} \mid (\hat{q}_{\mathcal{T}})^{[3]} (\hat{q}_{\mathcal{N} \Rightarrow \mathcal{N}})^{[2]} (\mathsf{\hat{q}_E})^{[2]} \mapsto (\langle)^{[3]} \mid \\ 
&(\hat{q}_{\mathcal{T}})^{[3]} (\hat{q}_{\mathcal{N} \Rightarrow \mathcal{N}})^{[2]} (\mathsf{\hat{q}_E})^{[2]} (\langle)^{[3]} (q_{\mathcal{T}})^{[3]} \mapsto (\rangle)^{[3]} \mid \\
&(\hat{q}_{\mathcal{T}})^{[3]} (\hat{q}_{\mathcal{N} \Rightarrow \mathcal{N}})^{[2]} (\mathsf{\hat{q}_E})^{[2]} (\langle)^{[3]} (q_{\mathcal{T}})^{[3]} (\rangle)^{[3]} (q_{\mathcal{T}})^{[3]} \mapsto (\sharp)^{[3]} \mid \\
&(\hat{q}_{\mathcal{T}})^{[3]} (\hat{q}_{\mathcal{N} \Rightarrow \mathcal{N}})^{[2]} (\mathsf{\hat{q}_E})^{[2]} (\langle)^{[3]} (q_{\mathcal{T}})^{[3]} (\rangle)^{[3]} (q_{\mathcal{T}})^{[3]} (\sharp)^{[3]} (q_{\mathcal{T}})^{[3]} \mapsto (\checkmark)^{[3]} \mid \\
&(\hat{q}_{\mathcal{N} \Rightarrow \mathcal{N}})^{[3]} (\hat{q}_{\mathcal{N} \Rightarrow \mathcal{N}})^{[2]} (\mathsf{q_E})^{[2]} \mapsto (\hat{q}_{\mathcal{N} \Rightarrow \mathcal{N}})^{[0]} \mid \\
&(\hat{q}_{\mathcal{N} \Rightarrow \mathcal{N}})^{[3]} (\hat{q}_{\mathcal{N} \Rightarrow \mathcal{N}})^{[2]} (\mathsf{q_E})^{[2]} (\hat{q}_{\mathcal{N} \Rightarrow \mathcal{N}})^{[0]} (\mathsf{no_W})^{[0]} \mapsto (\mathsf{no_E})^{[3]} \mid \\ 
&(\hat{q}_{\mathcal{T}})^{[3]} (\hat{q}_{\mathcal{N} \Rightarrow \mathcal{N}})^{[2]} (\mathsf{q_E})^{[2]} \mapsto (\hat{q}_{\mathcal{N} \Rightarrow \mathcal{N}})^{[0]} \mid \\
&(\hat{q}_{\mathcal{T}})^{[3]} (\hat{q}_{\mathcal{N} \Rightarrow \mathcal{N}})^{[2]} (\mathsf{q_E})^{[2]} (\hat{q}_{\mathcal{N} \Rightarrow \mathcal{N}})^{[0]} (\mathsf{no_W})^{[0]} \mapsto (\checkmark)^{[3]} \mid \\
&(\hat{q}_{\mathcal{N} \Rightarrow \mathcal{N}})^{[3]} (\hat{q}_{\mathcal{N} \Rightarrow \mathcal{N}})^{[2]} (\mathsf{q_E})^{[2]} (\hat{q}_{\mathcal{N} \Rightarrow \mathcal{N}})^{[0]} (\mathsf{yes_W})^{[0]} \mapsto (\mathsf{q_W})^{[3]} \mid \\
&(\hat{q}_{\mathcal{T}})^{[3]} (\hat{q}_{\mathcal{N} \Rightarrow \mathcal{N}})^{[2]} (\mathsf{q_E})^{[2]} (\hat{q}_{\mathcal{N} \Rightarrow \mathcal{N}})^{[0]} (\mathsf{yes_W})^{[0]} \mapsto (\langle)^{[3]} \mid \\
&(\hat{q}_{\mathcal{T}})^{[3]} (\hat{q}_{\mathcal{N} \Rightarrow \mathcal{N}})^{[2]} (\mathsf{q_E})^{[2]} (\hat{q}_{\mathcal{N} \Rightarrow \mathcal{N}})^{[0]} (\mathsf{yes_W})^{[0]} (\langle)^{[3]} (q_{\mathcal{T}})^{[3]} \mapsto (\rangle)^{[3]} \mid \\
&(\hat{q}_{\mathcal{T}})^{[3]} (\hat{q}_{\mathcal{N} \Rightarrow \mathcal{N}})^{[2]} (\mathsf{q_E})^{[2]} (\hat{q}_{\mathcal{N} \Rightarrow \mathcal{N}})^{[0]} (\mathsf{yes_W})^{[0]} (\langle)^{[3]} (q_{\mathcal{T}})^{[3]} (\rangle)^{[3]} (q_{\mathcal{T}})^{[3]} \mapsto (\sharp)^{[3]} \mid \\
&(\hat{q}_{\mathcal{T}})^{[3]} (\hat{q}_{\mathcal{N} \Rightarrow \mathcal{N}})^{[2]} (\mathsf{q_E})^{[2]} (\hat{q}_{\mathcal{N} \Rightarrow \mathcal{N}})^{[0]} (\mathsf{yes_W})^{[0]} (\langle)^{[3]} (\hat{q}_{\mathcal{T}})^{[3]} (\rangle)^{[3]} (q_{\mathcal{T}})^{[3]} (\sharp)^{[3]} (q_{\mathcal{T}})^{[3]} \mapsto (\checkmark)^{[3]} \mid \\
&(\hat{q}_{\mathcal{N} \Rightarrow \mathcal{N}})^{[3]} (\hat{q}_{\mathcal{N} \Rightarrow \mathcal{N}})^{[2]} (\mathsf{x_W})^{[2]} \mapsto (\mathsf{yes_E})^{[3]} \mid (\hat{q}_{\mathcal{T}})^{[3]} (\hat{q}_{\mathcal{N} \Rightarrow \mathcal{N}})^{[2]} (\mathsf{x_W})^{[2]} \mapsto (\checkmark)^{[3]} 
\end{align*}
where $\mathsf{x_W} \in \{ \mathsf{yes_W}, \mathsf{no_W} \}$.

\section{A Finite Table for Copy-Cat Strategy}
\label{TableCP}
The finite table for an st-algorithm $\mathcal{A}(\mathit{cp}_A)$ for the copy-cat strategy $\mathit{cp} : A \multimap A$ is as follows: 
\begin{align*}
\mathcal{A}(\mathit{cp}_A)_m : \ &(\hat{q}_{A \multimap A})^{[3]} \mapsto (\hat{q}_{A \multimap A})^{[2]} \mid (\hat{q}_{A \multimap A})^{[3]} (\hat{q}_{A \multimap A})^{[2]} (\mathsf{a_W})^{[2]} \mapsto (\mathsf{a_E})^{[3]} \mid \\
&(\hat{q}_{A \multimap A})^{[3]} (\hat{q}_{A \multimap A})^{[2]} (\mathsf{a_E})^{[2]} \mapsto (\mathsf{a_W})^{[3]} \mid \\
&(\hat{q}_{\mathcal{T}})^{[3]} \mapsto (\hat{q}_{\mathcal{T}})^{[2]} \mid (\mathsf{x})^{[3]} (\mathsf{x})^{[2]} (\mathsf{y})^{[2]} \mapsto (\mathsf{y})^{[3]} \mid (\mathsf{x})^{[2]} (\mathsf{x})^{[3]} (\mathsf{y})^{[3]} \mapsto (\mathsf{y})^{[2]}
\end{align*}
where $a \in \pi_1(M_A)$, $\mathsf{x}, \mathsf{y} \in \pi_1(M_{\mathcal{G}(\mathcal{T})})$.

\end{document}